\documentclass[a4paper]{article}

\usepackage[T1]{fontenc}
\usepackage{lmodern}
\usepackage{microtype}
\usepackage{enumitem}
\usepackage{amsmath}
\usepackage{amssymb}
\usepackage{amsthm}
\usepackage{mathrsfs}
\usepackage{comment}
\usepackage[numbers,sort]{natbib}
\usepackage[colorlinks=true,urlcolor=blue,citecolor=red,linkcolor=blue]{hyperref}
\usepackage{todonotes,booktabs,tabularx}

\newtheorem{theorem}{Theorem}[section]
\newtheorem{corollary}[theorem]{Corollary}
\newtheorem{lemma}[theorem]{Lemma}
\newtheorem{assumption}[theorem]{Assumption}
\newtheorem{remark}[theorem]{Remark}

\newcommand{\C}{\mathbb{C}}
\newcommand{\R}{\mathbb{R}}
\newcommand{\N}{\mathbb{N}}
\newcommand{\erfc}{{\rm erfc}}
\newcommand{\bigO}{\mathcal{O}}

\allowdisplaybreaks
\numberwithin{equation}{section}

\title{Asymptotics for partition functions of random normal matrix models with singularities}

\author{
Kohei Noda\\[2mm]
\small Institut de Recherche en Math\'ematique et Physique,\\
\small Universit\'e catholique de Louvain, Louvain-la-Neuve, B--1348, Belgium\\
\small \texttt{kohei.noda@uclouvain.be}
}

\date{}

\begin{document}

\maketitle

\begin{abstract}
We study determinantal planar Coulomb gases with a smooth radial potential, a smooth radial perturbation, jump-type singularities (including the bulk, semi-hard, hard, and soft-edge regimes),
and a hard wall region given by a centered disk, a disk complement, or
an annulus.  We obtain explicit asymptotic expansions of the partition functions up to and including the term of order $1$. Our results unify and generalize many of the earlier works.
\end{abstract}




\section{Introduction and statement of results}\label{section: introduction}

Two-dimensional Coulomb gases at inverse temperature $\beta=2$ form determinantal point processes and describe the eigenvalue distributions of random normal matrices, with the complex Ginibre ensemble as a basic example \cite{Ginibre,Forrester,BFreview}. 

There is now a substantial body of work on partition functions with radially symmetric weights with singularities. However, the results available in the literature are scattered and typically concern specific singular regimes, such as merging singularities in the bulk or the (soft) edge \cite{ABE2023,ABES2023,BC2022,C2021 FH,CL2023,FenzlLambert,L et al 2019 b,MMO25,Noda2025,CMV2016}, singularities at a hard edge \cite{ACCL1,ACCL2,C2021}, hole probabilities \cite{AFLS25,ForresterHoleProba,L et al 2019,SF2026}, and related special configurations \cite{ACC2023c,BKS2023}. Many of these works treat special types of potentials and focus on only one regime at a time. One of the main goals of this paper is to provide a unified framework that both recovers and generalizes many of the results established in the works cited above. In particular, we allow for general (smooth) potentials and for several types of singularities to occur simultaneously. As corollaries, we derive precise hole probabilities, as well as fluctuations of smooth radial statistics and joint disk counting statistics in various regimes.

\subsection{Background and related work}
\label{subsection:background}

\paragraph{Hole probabilities and free energy.}
The probability that a macroscopic region contains no particles is a basic
large-deviation observable for planar Coulomb gases.  Forrester's analysis
of the complex Ginibre ensemble \cite{ForresterHoleProba} provided early
precise disk-gap asymptotics.  Potential-theoretic approaches to more general
holes were developed by Adhikari and Reddy \cite{AR2017}, Adhikari
\cite{A2018}, and Charlier \cite{C2023}, see also \cite{ASZ2014,BLY2026,BF25} and the references therein. The results in \cite{AR2017, A2018, C2023} only concern the leading term.

For radially symmetric potentials, there is a product representation of the partition function, providing a starting point for precise free-energy and fluctuation asymptotics \cite{C2021 FH,C2021, BKS2023,ACC2023c}. In particular, when an annular hole region separates the occupied region into two components, the expansion of the partition function also contains a bounded Jacobi-theta term, as first proved in \cite{C2021}, see also \cite{ACCL2, ACC2023c}.

For unconstrained radial gases, Byun, Kang, and
Seo \cite{BKS2023} established precise free-energy expansions for
determinantal and Pfaffian ensembles.  Their work gives a rigorous radial
counterpart to the large-$n$ expansion studied by Zabrodin and Wiegmann
\cite{ZW2006}.  Ameur, Charlier, and Cronvall \cite{ACC2023c} subsequently
treated radial droplets with spectral gaps and outposts; see also
\cite{ACCL2,Noda2026}. The latter, corresponding to the ``birth of a new
connected component'', are not considered here. Also outside our assumptions
is the distinct ``birth of a gap'' regime caused by a vanishing equilibrium
density, studied by Allard and Lahiry \cite{AL25}.

Beyond the rotation invariant potentials, there has also been substantial recent progress in the study of two-dimensional logarithmic Coulomb gases with more general potentials, see \cite{R2025,Ser2024,Ser2023,AS2021,LS2017,B2025,BCMS2025,BSY2025,BYY2026,BFKL25,BFL25,HW2021,HW2024,HWpreprint,Jo2022,JV2023}.

\paragraph{Counting statistics (jump-type singularities).}
Counting statistics are equivalent to partition functions with jump-type singularities, since their moment generating functions can be expressed as ratios of determinants with discontinuous weights. Their fluctuations and large deviations were studied by Jancovici, Lebowitz, and Manificat \cite{JLM1993}. For the Ginibre ensemble, Lacroix-A-Chez-Toine et al.\ \cite{L et al 2019} studied intermediate deviations, while Fenzl and Lambert \cite{FenzlLambert} obtained precise deviation estimates and functional limit theorems. More recently, Byun and Lee \cite{BL2026} studied disk counting statistics for the real Ginibre ensemble.

Precise asymptotics for moment generating functions with circular jump singularities were obtained by Charlier \cite{C2021 FH}, and Charlier and Lenells \cite{CL2023} studied merging singularities at critical bulk and soft-edge scales. Hard-edge and multi-component counting statistics were analyzed by Ameur, Charlier, Cronvall, and Lenells \cite{ACCL1,ACCL2}, revealing hard and semi-hard regimes as well as oscillatory order-one corrections. 

There is also a complementary universality theory for the number variance.
Akemann, Byun, and Ebke \cite{ABE2023} established bulk and edge
universality for disk counting statistics under general radial potentials, see also \citep{ABES2023}.
Levi, Marzo, and Ortega-Cerd'{a} \cite{LMO2024} studied the counting statistics of general bulk regions for the Ginibre ensemble using norm representation formulas for Sobolev and bounded variation functions. Marzo, Molag, and Ortega-Cerd'{a} \cite{MMO25} extended this to nonradial potentials and more general counting domains. For a regular bulk domain $A$, their result reads
\[
\operatorname{Var} N_A
=\frac{\sqrt n}{2\pi\sqrt\pi}
  \int_{\partial A}\sqrt{\Delta Q(z)}\,|dz|+o(\sqrt n),
\]
for $N_A:=\#\{j:z_j\in A\subset\C\}$.
It explains the $n^{1/4}$ scale of bulk fluctuations for the counting statistics.  
In contrast to \cite{MMO25}, which only treats the second moment (the variance), our rotationally invariant setting allows us to derive the joint generating function up to and including the order-one term, and hence asymptotics for all cumulants of the counting statistics in the bulk, semi-hard-edge, and hard-edge regimes, including hard-wall effects and interactions between several counting statistics, see Theorem~\ref{theorem:annular-multivariate-mgf}.

For one-dimensional Coulomb gases, precise asymptotics for weights with root- and jump-type singularities have been extensively studied; see, for example, \cite{C2019,CG2021,FisherHartwig,DIK} and the references therein. In the planar setting, point-charge insertions were studied in \cite{B2025,BCMS2025,BSY2025,BYY2026,BFKL25,BFL25}, while Byun and Charlier \cite{BC2022} obtained the asymptotics of the partition functions for simultaneous circular root- and jump-type singularities in the Mittag--Leffler ensemble. The author \cite{Noda2025} extended their results to general radial potentials.

\paragraph{Fluctuations of smooth linear statistics.}
For the Ginibre ensemble, Rider and Vir\'{a}g \cite{RV2007} proved a central limit theorem for smooth linear statistics, whose variance consists of a bulk Dirichlet term and a boundary term. Ameur, Hedenmalm, and Makarov \cite{AHM2011,AHM2015} extended these results to general random normal matrices and obtained corrections to the mean via Ward identities. De~Bruyne et al.\ \cite{DDMS2024} derived asymptotic formulas for the higher-order cumulants of smooth linear statistics for rotationally invariant potentials.

Ameur, Charlier, and Cronvall \cite{ACC2023a,ACC2023c} studied fluctuations of linear statistics for two-dimensional Coulomb gases with spectral gaps, where Gaussian and Heine-type fluctuations arise. Ameur and Cronvall \cite{AC2026} further established such fluctuations in more general settings. These results show that smooth linear statistics in multi-component models need not have standard Gaussian fluctuations. For the annular hard region studied here, the difference between the values of the test function at the two walls gives rise to an additional fluctuation.

The recent work \cite{SF2026} by Shen and Forrester concerns Coulomb gases on a cylinder at $\beta=2$ with radially symmetric potentials and hard walls. Their results contain an analogue of \cite[Theorem~1.9 with $g=2$]{C2021} for a particular quadratic potential.

\subsection{The model and three hard-wall regions}

We normalize the area measure by
\[
dA(z)=\frac{dx\,dy}{\pi},\qquad z=x+iy\in\C,
\]
and use the normalized Laplacian
\[
\Delta=\partial_z\partial_{\bar z}
=\frac14\bigl(\partial_x^2+\partial_y^2\bigr).
\]
For a real-valued external potential $Q:\C\to\R$, set
\begin{equation}
\label{def of det partition function}
Z_n[Q]
:=
\int_{\C^n}
\prod_{1\leq j<k\leq n}|z_k-z_j|^2
\prod_{j=1}^n e^{-nQ(z_j)}\,dA(z_j),
\end{equation}
and define
\begin{equation}
\label{def of Pn soft}
d\mathbb P_n(\boldsymbol z_n)
:=
\frac{1}{Z_n[Q]}
\prod_{1\leq j<k\leq n}|z_k-z_j|^2
\prod_{j=1}^n e^{-nQ(z_j)}\,dA(z_j),
\qquad \boldsymbol z_n=(z_1,\ldots,z_n)\in\C^n.
\end{equation}

Throughout the paper, we assume that the function $Q$ satisfies to the following assumption.

\begin{assumption}\label{Assumption_Q}
The potential is radial, $Q(z)=q(|z|)$, and satisfies the following conditions.
\begin{enumerate}[label=\textup{(\arabic*)}]
\item\label{Assumption 1}
$\displaystyle \liminf_{|z|\to\infty}\frac{Q(z)}{2\log|z|}>1$.
\item\label{Assumption 2}
$Q$ is globally subharmonic, is $C^\infty$-smooth in a neighborhood of its droplet, and is strictly subharmonic there.
\item The droplet is a disk centered at the origin.
\end{enumerate}
\end{assumption}

The strict positivity of $\Delta Q$ at the origin is part of this
assumption.  In particular, the Ginibre potential $Q(z)=|z|^2$ is
included, whereas the Mittag--Leffler potential $|z|^{2b}$ for $b>0,b\neq 1$ is
not.

Under Assumption~\ref{Assumption_Q}, the empirical measure converges to the equilibrium measure
\[
d\mu_Q(z)=\Delta Q(z)\mathbf 1_S(z)\,dA(z).
\]
More generally, before assuming the part (3) of Assumption~\ref{Assumption_Q}, \cite[Section IV.6]{SaTo} reads
\begin{equation}
\label{def of droplet}
S=\{z\in\C:R_0\leq |z|\leq R_1\},
\end{equation}
where $R_0$ is the largest solution of $rq'(r)=0$ and $R_1$ is the smallest solution of $rq'(r)=2$, see also \cite{BKS2023}.  We take $R_0=0$ and write $R:=R_1$, so $S=\overline{\mathbb D}_R:=\{z\in\C:0\leq |z|\leq R\}$.  It is useful to introduce the equilibrium mass function
\begin{equation}
\tau_\rho:=\frac12\rho q'(\rho)
=\int_{\mathbb D_\rho}\Delta Q\,dA,
\qquad 0\leq \rho\leq R.
\end{equation}

For a hard wall region $U$ whose boundary lies in the interior of $S$, define
\begin{equation}
Q_U(z):=
\begin{cases}
Q(z),&z\notin U,\\
+\infty,&z\in U.
\end{cases}
\end{equation}
We consider the following three choices:
\begin{equation}
\label{def of hole region}
U_{\mathrm{an}}=\mathbb A_{\rho_1,\rho_2}:=\{\rho_1<|z|<\rho_2\},
\qquad
U_{\mathrm{dc}}=\C\setminus\overline{\mathbb D}_{\rho_1}
=\{z:|z|>\rho_1\},
\qquad
U_{\mathrm{cd}}=\mathbb D_{\rho_2},
\end{equation}
where $0<\rho_1<\rho_2<R$ whenever both radii occur.  The abbreviations $\mathrm{an}$, $\mathrm{dc}$, and $\mathrm{cd}$ stand for annulus, disk complement, and centered disk, respectively.  Notice that $U$ is the \emph{forbidden} set; the particles live in $\C\setminus U$.

Let $\nu_\rho$ denote the uniform probability measure on $\partial\mathbb D_\rho$.  The constrained equilibrium measures are obtained by balayage \cite{C2023}.  For the disk-complement and centered-disk holes they are, respectively,
\begin{align}
d\mu_{\mathrm{dc}}
&=\Delta Q\,\mathbf 1_{\{|z|<\rho_1\}}\,dA
+(1-\tau_{\rho_1})\,d\nu_{\rho_1},
\\
d\mu_{\mathrm{cd}}
&=\Delta Q\,\mathbf 1_{\{\rho_2<|z|<R\}}\,dA
+\tau_{\rho_2}\,d\nu_{\rho_2},
\end{align} 
see \cite[Theorem 2.10 and 2.14]{C2023}.
For an annular hole, put
\begin{equation}
\label{def of sigma star}
\sigma_\star
:=
\frac{q(\rho_2)-q(\rho_1)}{2\log(\rho_2/\rho_1)}
\end{equation}
and
\begin{equation}
\label{def of sigma1 sigma2}
\sigma_1:=\sigma_\star-\tau_{\rho_1},
\qquad
\sigma_2:=\tau_{\rho_2}-\sigma_\star.
\end{equation}
Then
\begin{equation}
\label{def of equi measure annulus hole}
d\mu_{\mathrm{an}}
=
\Delta Q\,\mathbf 1_{\{|z|<\rho_1\}\cup\{\rho_2<|z|<R\}}\,dA
+\sigma_1\,d\nu_{\rho_1}
+\sigma_2\,d\nu_{\rho_2},
\end{equation}
see \cite[Theorem 2.12]{C2023}.
Thus $\sigma_1$ and $\sigma_2$ are precisely the masses accumulated at the
two hard walls.  Indeed, since
\[
\sigma_\star
=\frac{1}{\log(\rho_2/\rho_1)}
  \int_{\rho_1}^{\rho_2}\tau_r\,\frac{dr}{r},
\]
the strict increase of $\tau_r$ implies
$\tau_{\rho_1}<\sigma_\star<\tau_{\rho_2}$ and hence
$\sigma_1,\sigma_2>0$.

\subsection{Partition functions of random normal matrix models with singularities}

Let
\[
\mathrm N(r):=\#\{j:|z_j|<r\}.
\]
Let $\alpha>-1$ and $s\in\R$. 
Let $\lambda\in C_c^6(\C)$ be a real-valued and rotation invariant, i.e., $\lambda(z)=\lambda(|z|)$, and let $I$ be a finite index set.  For vectors
$\boldsymbol s=(s_\ell)_{\ell\in I}\in\R^I$ and
$\boldsymbol r=(r_\ell)_{\ell\in I}$, define
\begin{align}
\label{def of DnQ}
D_{n,\alpha}(Q;s\lambda,\boldsymbol s,\boldsymbol r,U)
&:=
\int_{\C^n}
\prod_{1\leq j<k\leq n}|z_j-z_k|^2
\exp\bigg(\sum_{\ell\in I}s_\ell\mathrm N(r_\ell)\bigg)
\nonumber\\
&\hspace{3.7cm}\times
\prod_{j=1}^n |z_j|^{2\alpha}
e^{-nQ_U(z_j)+s\lambda(z_j)}\,dA(z_j).
\end{align}
With the smooth linear statistics and the jump-type singularities to be zero, we write
\[
Z_{n,U}^{(\alpha)}[Q]:=D_{n,\alpha}(Q;0,\boldsymbol 0,\boldsymbol r,U).
\]
The corresponding hard-wall ensemble is
\begin{equation}
\label{def of PnU}
d\mathbb P_{n,U}^{(\alpha)}(\boldsymbol z_n)
:=
\frac{1}{Z_{n,U}^{(\alpha)}[Q]}
\prod_{1\leq j<k\leq n}|z_j-z_k|^2
\prod_{j=1}^n|z_j|^{2\alpha}e^{-nQ_U(z_j)}\,dA(z_j),
\end{equation}
and its expectation is denoted by $\mathbb E_{n,U}^{(\alpha)}$.  Superscripts are omitted when $\alpha=0$.

Three important quantities are ratios of the determinants in \eqref{def of DnQ}.

\subsubsection*{Hole probabilities}

For the unconstrained ensemble \eqref{def of Pn soft},
\begin{equation}
\label{def of hole probability U}
\mathcal P_n[Q,U]
:=\mathbb P_n\bigl(\#\{j:z_j\in U\}=0\bigr)
=\frac{Z_{n,U}^{(0)}[Q]}{Z_n[Q]}.
\end{equation}

\subsubsection*{Moment generating function of smooth radial statistics}

For $X_n[\lambda]:=\sum_{j=1}^n\lambda(z_j)$,
\begin{equation}
\label{def of MGF smooth test f}
\mathcal G_{n,U}^{(\alpha)}(s;\lambda)
:=\mathbb E_{n,U}^{(\alpha)}\!\left[e^{sX_n[\lambda]}\right]
=\frac{D_{n,\alpha}(Q;s\lambda,\boldsymbol0,\boldsymbol r,U)}{D_{n,\alpha}(Q;0,\boldsymbol0,\boldsymbol r,U)}.
\end{equation}

\subsubsection*{Multivariate moment generating function of disk counting statistics}

For radii $\boldsymbol r=(r_\ell)_{\ell\in I}$,
\begin{equation}
\label{def of multivariate MGF disk counting statistics}
\mathcal E_{n,U}^{(\alpha)}(\boldsymbol s;\boldsymbol r)
:=\mathbb E_{n,U}^{(\alpha)}\!\left[
\exp\bigg(\sum_{\ell\in I}s_\ell\mathrm N(r_\ell)\bigg)
\right]
=\frac{D_{n,\alpha}(Q;0,\boldsymbol s,\boldsymbol r,U)}{D_{n,\alpha}(Q;0,\boldsymbol0,\boldsymbol r,U)}.
\end{equation}
Derivatives with respect to the $s_\ell$ yield joint cumulants of the disk counting statistics.

\subsection{Critical radii for several regimes}

Fix $m\in\N$.  To treat all relevant locations simultaneously, split the counting statistics into seven blocks:
\begin{align*}
\boldsymbol s_{\mathrm{b,in}}&=(s_1,\ldots,s_m),
&\boldsymbol s_{\mathrm{sh,in}}&=(s_{m+1},\ldots,s_{2m}),
&\boldsymbol s_{\mathrm{h,in}}&=(s_{2m+1},\ldots,s_{3m}),
\\
\boldsymbol s_{\mathrm{h,out}}&=(s_{3m+1},\ldots,s_{4m}),
&\boldsymbol s_{\mathrm{sh,out}}&=(s_{4m+1},\ldots,s_{5m}),
\\
\boldsymbol s_{\mathrm{b,out}}&=(s_{5m+1},\ldots,s_{6m}),
&\boldsymbol s_{\mathrm{s}}&=(s_{6m+1},\ldots,s_{7m}).
\end{align*}
Define
\begin{align}
\vec{\boldsymbol{s}}_{\mathrm{an}}
&:=(\boldsymbol s_{\mathrm{b,in}},\boldsymbol s_{\mathrm{sh,in}},\boldsymbol s_{\mathrm{h,in}},\boldsymbol s_{\mathrm{h,out}},\boldsymbol s_{\mathrm{sh,out}},\boldsymbol s_{\mathrm{b,out}},\boldsymbol s_{\mathrm{s}})\in\R^{7m},
\\
\vec{\boldsymbol{s}}_{\mathrm{dc}}
&:=(\boldsymbol s_{\mathrm{b,in}},\boldsymbol s_{\mathrm{sh,in}},\boldsymbol s_{\mathrm{h,in}})\in\R^{3m},
\\
\vec{\boldsymbol{s}}_{\mathrm{cd}}
&:=(\boldsymbol s_{\mathrm{h,out}},\boldsymbol s_{\mathrm{sh,out}},\boldsymbol s_{\mathrm{b,out}},\boldsymbol s_{\mathrm{s}})\in\R^{4m}.
\end{align}
Choose fixed points $0<a_1<\rho_1<\rho_2<a_2<R$.  The inner and outer bulk radii are
\begin{equation}
\label{def of mergin radii inside bulk}
r_\ell=a_1+\frac{t_\ell}{\sqrt{2n\Delta Q(a_1)}},
\qquad \ell=1,\ldots,m,
\qquad t_1<\cdots<t_m,
\end{equation}
and
\begin{equation}
\label{def of mergin radii outside bulk}
r_\ell=a_2+\frac{t_\ell}{\sqrt{2n\Delta Q(a_2)}},
\qquad \ell=5m+1,\ldots,6m,
\qquad t_{5m+1}<\cdots<t_{6m}.
\end{equation}
The semi-hard radii approach the walls from the allowed sides:
\begin{equation}
\label{def of merging radii inside semi hard}
r_\ell=\rho_1-\frac{t_\ell}{\sqrt{2n\Delta Q(\rho_1)}},
\quad \ell=m+1,\ldots,2m,
\quad t_{m+1}>\cdots>t_{2m}>0,
\end{equation}
and
\begin{equation}
\label{def of merging radii outside semi hard}
r_\ell=\rho_2+\frac{t_\ell}{\sqrt{2n\Delta Q(\rho_2)}},
\quad \ell=4m+1,\ldots,5m,
\quad 0<t_{4m+1}<\cdots<t_{5m}.
\end{equation}
The hard-edge radii lie at distance $n^{-1}$ from the walls:
\begin{equation}
\label{def of merging radii inside hard}
r_\ell=\rho_1\left(1-\frac{t_\ell}{n}\right),
\quad \ell=2m+1,\ldots,3m,
\quad t_{2m+1}>\cdots>t_{3m}\geq0,
\end{equation}
and
\begin{equation}
\label{def of merging radii outside hard}
r_\ell=\rho_2\left(1+\frac{t_\ell}{n}\right),
\quad \ell=3m+1,\ldots,4m,
\quad 0\leq t_{3m+1}<\cdots<t_{4m}.
\end{equation}
Finally, the natural soft-edge radii are
\begin{equation}
\label{def of merging radii soft edge}
r_\ell=R+\frac{t_\ell}{\sqrt{2n\Delta Q(R)}},
\qquad \ell=6m+1,\ldots,7m,
\qquad t_{6m+1}<\cdots<t_{7m}.
\end{equation}


Let $\boldsymbol r_{\mathrm{an}}=(r_1,\ldots,r_{7m})$, $\boldsymbol r_{\mathrm{dc}}=(r_1,\ldots,r_{3m})$, and $\boldsymbol r_{\mathrm{cd}}=(r_{3m+1},\ldots,r_{7m})$.  The three deformed partition functions studied in this paper are
\begin{align}
\label{def of Dnc annulus}
D_n^{(\mathrm{an})}
&:=D_{n,\alpha}(Q;s\lambda,\vec{\boldsymbol s}_{\mathrm{an}},\boldsymbol r_{\mathrm{an}},\mathbb A_{\rho_1,\rho_2}),
\\
\label{def of Dnc disk complement}
D_n^{(\mathrm{dc})}
&:=D_{n,\alpha}(Q;s\lambda,\vec{\boldsymbol s}_{\mathrm{dc}},\boldsymbol r_{\mathrm{dc}},U_{\mathrm{dc}}),
\\
\label{def of Dnc centered disk}
D_n^{(\mathrm{cd})}
&:=D_{n,\alpha}(Q;s\lambda,\vec{\boldsymbol s}_{\mathrm{cd}},\boldsymbol r_{\mathrm{cd}},\mathbb D_{\rho_2}).
\end{align}
We shall repeatedly use the combined smooth and point-charge perturbation
\begin{equation}
\label{def of mathsfkr}
\mathsf k(r):=s\lambda(r)+\alpha\ell(r),
\qquad \ell(r):=2\log r.
\end{equation}

\subsection{Main results}

The coefficients in the following theorems are listed explicitly in
Section~2.  Their notation records three independent sources:
$C_k^{(\mathrm n)}$ is the unconstrained random-normal-matrix contribution,
$C_{k,\mathrm{hole}}^{(\cdot)}$ is produced by the hard wall, and
$C_{k,\#}^{(\cdot)}$ is produced by the disk-counting insertions.  Under
Assumption~\ref{Assumption_Q} and the radii scalings above, the expansions
hold for fixed $\alpha$, $s$, $\boldsymbol s$, and $t_\ell$; the constants
implicit in the remainder may depend on these fixed parameters.

\begin{theorem}[Centered annular hard wall]
\label{theorem:asymptotic determinant of annulus}
As $n\to\infty$,
\begin{align}
\log D_n^{(\mathrm{an})}
&=C_1^{(\mathrm{an})}n^2
+C_2^{(\mathrm{an})}n\log n
+C_3^{(\mathrm{an})}n
+C_4^{(\mathrm{an})}\sqrt n
\nonumber\\
&\quad
+C_5^{(\mathrm{an})}\log n
+C_6^{(\mathrm{an})}
+\mathcal F_n[\lambda]
+\mathcal F_{n,\#}
+\mathcal O\!\left(\frac{(\log n)^3}{n^{1/12}}\right),
\end{align}
where
\begin{align}
\label{def of C1 Dn an}
C_1^{(\mathrm{an})}
&:=C_{1,\mathrm{hole}}^{(\mathrm{an})}+C_1^{(\mathrm n)},
\\
\label{def of C2 Dn an}
C_2^{(\mathrm{an})}
&:=C_{2,\mathrm{hole}}^{(\mathrm{an})}+C_2^{(\mathrm n)},
\\
\label{def of C3 Dn an}
C_3^{(\mathrm{an})}
&:=C_{3,\#}^{(\mathrm{an})}+C_{3,\mathrm{hole}}^{(\mathrm{an})}+C_3^{(\mathrm n)},
\\
\label{def of C4 Dn an}
C_4^{(\mathrm{an})}
&:=C_{4,\#}^{(\mathrm{an})}+C_{4,\mathrm{hole}}^{(\mathrm{an})}+C_4^{(\mathrm n)},
\\
\label{def of C5 Dn an}
C_5^{(\mathrm{an})}
&:=C_{5,\#}^{(\mathrm{an})}+C_{5,\mathrm{hole}}^{(\mathrm{an})}+C_5^{(\mathrm n)},
\\
\label{def of C6 Dn an}
C_6^{(\mathrm{an})}
&:=C_{6,\#}^{(\mathrm{an})}+C_{6,\mathrm{hole}}^{(\mathrm{an})}+C_6^{(\mathrm n)}.
\end{align}
The hole coefficients and the baseline filling-fraction term
$\mathcal F_n[\lambda]$ are given in
\eqref{def of C1 hole an}--\eqref{def of calF hole an}; the counting
coefficients and the counting-induced theta increment $\mathcal F_{n,\#}$
are given in \eqref{def of C3sharp an}--\eqref{def of C6sharp an} and
\eqref{def of calFNsharp}.  These bounded order-one terms are expressed
through the Jacobi theta function.
\end{theorem}

\begin{theorem}[Centered disk-complement hard wall]
\label{theorem:asymptotic determinant of disk complement}
As $n\to\infty$,
\begin{equation}
\log D_n^{(\mathrm{dc})}
=C_1^{(\mathrm{dc})}n^2
+C_2^{(\mathrm{dc})}n\log n
+C_3^{(\mathrm{dc})}n
+C_4^{(\mathrm{dc})}\sqrt n
+C_5^{(\mathrm{dc})}\log n
+C_6^{(\mathrm{dc})}
+\mathcal O\!\left(\frac{(\log n)^3}{n^{1/12}}\right),
\end{equation}
where
\begin{align}
\label{def of C1 Dn dc}
C_1^{(\mathrm{dc})}
&:=C_{1,\mathrm{hole}}^{(\mathrm{dc})}+C_1^{(\mathrm n)},
\\
\label{def of C2 Dn dc}
C_2^{(\mathrm{dc})}
&:=C_{2,\mathrm{hole}}^{(\mathrm{dc})}+C_2^{(\mathrm n)},
\\
\label{def of C3 Dn dc}
C_3^{(\mathrm{dc})}
&:=C_{3,\#}^{(\mathrm{dc})}+C_{3,\mathrm{hole}}^{(\mathrm{dc})}+C_3^{(\mathrm n)},
\\
\label{def of C4 Dn dc}
C_4^{(\mathrm{dc})}
&:=C_{4,\#}^{(\mathrm{dc})}+C_{4,\mathrm{hole}}^{(\mathrm{dc})}+C_4^{(\mathrm n)},
\\
\label{def of C5 Dn dc}
C_5^{(\mathrm{dc})}
&:=C_{5,\#}^{(\mathrm{dc})}+C_{5,\mathrm{hole}}^{(\mathrm{dc})}+C_5^{(\mathrm n)},
\\
\label{def of C6 Dn dc}
C_6^{(\mathrm{dc})}
&:=C_{6,\#}^{(\mathrm{dc})}+C_{6,\mathrm{hole}}^{(\mathrm{dc})}+C_6^{(\mathrm n)}.
\end{align}
The hole coefficients are given in \eqref{def of C1dc hole}--\eqref{def of C6dc hole}, and the counting coefficients are given in \eqref{def of C3sharp dc}--\eqref{def of C6sharp dc}.
\end{theorem}

\begin{theorem}[Centered-disk hard wall]
\label{theorem:asymptotic determinant of centered disk}
As $n\to\infty$,
\begin{equation}
\log D_n^{(\mathrm{cd})}
=C_1^{(\mathrm{cd})}n^2
+C_2^{(\mathrm{cd})}n\log n
+C_3^{(\mathrm{cd})}n
+C_4^{(\mathrm{cd})}\sqrt n
+C_5^{(\mathrm{cd})}\log n
+C_6^{(\mathrm{cd})}
+\mathcal O\!\left(\frac{(\log n)^3}{n^{1/12}}\right),
\end{equation}
where
\begin{align}
\label{def of C1 Dn cd}
C_1^{(\mathrm{cd})}
&:=C_{1,\mathrm{hole}}^{(\mathrm{cd})}+C_1^{(\mathrm n)},
\\
\label{def of C2 Dn cd}
C_2^{(\mathrm{cd})}
&:=C_{2,\mathrm{hole}}^{(\mathrm{cd})}+C_2^{(\mathrm n)},
\\
\label{def of C3 Dn cd}
C_3^{(\mathrm{cd})}
&:=C_{3,\#}^{(\mathrm{cd})}+C_{3,\mathrm{hole}}^{(\mathrm{cd})}+C_3^{(\mathrm n)},
\\
\label{def of C4 Dn cd}
C_4^{(\mathrm{cd})}
&:=C_{4,\#}^{(\mathrm{cd})}+C_{4,\mathrm{hole}}^{(\mathrm{cd})}+C_4^{(\mathrm n)},
\\
\label{def of C5 Dn cd}
C_5^{(\mathrm{cd})}
&:=C_{5,\#}^{(\mathrm{cd})}+C_{5,\mathrm{hole}}^{(\mathrm{cd})}+C_5^{(\mathrm n)},
\\
\label{def of C6 Dn cd}
C_6^{(\mathrm{cd})}
&:=C_{6,\#}^{(\mathrm{cd})}+C_{6,\mathrm{hole}}^{(\mathrm{cd})}+C_6^{(\mathrm n)}.
\end{align}
The hole coefficients are given in \eqref{def of C1cd hole}--\eqref{def of C6cd hole}, and the counting coefficients are given in \eqref{def of Csharp3 cd}--\eqref{def of Csharp6 cd}.
\end{theorem}

\begin{remark}
Theorems~\ref{theorem:asymptotic determinant of annulus}, \ref{theorem:asymptotic determinant of disk complement}, and \ref{theorem:asymptotic determinant of centered disk} generalize and unify some of the results in \cite{C2021,C2021 FH,ACCL1,ACCL2,CL2023}.  
\end{remark}

\subsection{Asymptotics of hole probabilities}
\label{subsection:precise-hole-probabilities}

We state the resulting hole
probabilities separately.  Thus $\alpha=s=0$, all counting statistics vanish,
and
\[
\mathcal P_n[Q,U]=\frac{Z_{n,U}^{(0)}[Q]}{Z_n[Q]}.
\]
For $\mathrm g\in\{\mathrm{an},\mathrm{dc},\mathrm{cd}\}$, set
\begin{equation}
\label{def:H-hole-specialization}
H_k^{(\mathrm g)}
:=\left.C_{k,\mathrm{hole}}^{(\mathrm g)}\right|_{\alpha=0,\,s=0},
\qquad 1\leq k\leq6.
\end{equation}
The coefficients are therefore completely explicit from
\eqref{def of C1 hole an}--\eqref{def of C6 hole an},
\eqref{def of C1dc hole}--\eqref{def of C6dc hole}, and
\eqref{def of C1cd hole}--\eqref{def of C6cd hole}.  For the annular
geometry, write
\begin{equation}
\label{def:L-and-hole-phase}
\xi_n^{(0)}:=n\sigma_\star+\frac12
+\frac{\log(\sigma_2/\sigma_1)}{2\log\frac{\rho_2}{\rho_1}},
\end{equation}
and use the Jacobi-theta convention
\begin{equation}
\label{def:Jacobi-theta-convention}
\theta(z\mid\tau):=\sum_{k\in\mathbb Z}
e^{\pi i\tau k^2+2\pi i k z},
\qquad \operatorname{Im}\tau>0.
\end{equation}

\begin{theorem}[Annular hole probability]
\label{theorem:precise-annular-hole-probability}
Let $0<\rho_1<\rho_2<R$.  Under Assumption~\ref{Assumption_Q}, as
$n\to\infty$,
\begin{align}
\log\mathbb P_n\!\left(\#\{z_j:z_j\in \mathbb{A}_{\rho_1,\rho_2}\}=0\right)
&=H_1^{(\mathrm{an})}n^2
+H_2^{(\mathrm{an})}n\log n
+H_3^{(\mathrm{an})}n
+H_4^{(\mathrm{an})}\sqrt n
\nonumber\\
&\quad
+H_6^{(\mathrm{an})}
+\log\theta\!\left(\xi_n^{(0)}\,\middle|\,\frac{\pi i}{\log\frac{\rho_2}{\rho_1}}\right)
+\mathcal O\!\left(\frac{(\log n)^3}{n^{1/12}}\right),
\label{eq:precise-annular-hole-probability}
\end{align}
where $H_5^{(\mathrm{an})}=0$. 
\end{theorem}

\begin{proof}
Set $\alpha=s=0$ and $\boldsymbol s=\boldsymbol0$ in
Theorem~\ref{theorem:asymptotic determinant of annulus}, and subtract
Theorem~\ref{theorem:unconstrained-free-energy}.  The unconstrained
coefficients cancel, the counting coefficients and $\mathcal F_{n,\#}$
vanish, and \eqref{def of calF hole an} becomes
$\log\theta(\xi_n^{(0)}\mid\pi i/\log\frac{\rho_2}{\rho_1})$.  Equation
\eqref{def of hole probability U} then gives the result.
\end{proof}

\begin{theorem}[Disk-complement hole probability]
\label{theorem:precise-disk-complement-hole-probability}
Let $0<\rho_1<R$.  Under Assumption~\ref{Assumption_Q}, as $n\to\infty$,
\begin{align}
\log\mathbb P_n\!\left(\#\{z_j:z_j\in \mathbb{D}_{\rho_1}^c\}=0\right)
&=H_1^{(\mathrm{dc})}n^2
+H_2^{(\mathrm{dc})}n\log n
+H_3^{(\mathrm{dc})}n
+H_4^{(\mathrm{dc})}\sqrt n
\nonumber\\
&\quad
-\frac14\log n+H_6^{(\mathrm{dc})}
+\mathcal O\!\left(\frac{(\log n)^3}{n^{1/12}}\right).
\label{eq:precise-disk-complement-hole-probability}
\end{align}
\end{theorem}

\begin{proof}
Set $\alpha=s=0$ and $\boldsymbol s=\boldsymbol0$ in
Theorem~\ref{theorem:asymptotic determinant of disk complement}, divide by
$Z_n[Q]$, and use Theorem~\ref{theorem:unconstrained-free-energy}.  All
unconstrained and counting contributions cancel.  The remaining terms are
\eqref{def:H-hole-specialization}, and
\eqref{def of C5dc hole} gives $H_5^{(\mathrm{dc})}=-1/4$.
\end{proof}

\begin{theorem}[Centered-disk hole probability]
\label{theorem:precise-centered-disk-hole-probability}
Let $0<\rho_2<R$.  Under Assumption~\ref{Assumption_Q}, as $n\to\infty$,
\begin{align}
\log\mathbb P_n\!\left(\#\{z_j:z_{j}\in\mathbb{D}_{\rho_2}\}=0\right)
&=H_1^{(\mathrm{cd})}n^2
+H_2^{(\mathrm{cd})}n\log n
+H_3^{(\mathrm{cd})}n
+H_4^{(\mathrm{cd})}\sqrt n
\nonumber\\
&\quad
+\frac13\log n+H_6^{(\mathrm{cd})}
+\mathcal O\!\left(\frac{(\log n)^3}{n^{1/12}}\right).
\label{eq:precise-centered-disk-hole-probability}
\end{align}
\end{theorem}

\begin{proof}
Set $\alpha=s=0$ and $\boldsymbol s=\boldsymbol0$ in
Theorem~\ref{theorem:asymptotic determinant of centered disk}, divide by
$Z_n[Q]$, and use Theorem~\ref{theorem:unconstrained-free-energy}.  The
unconstrained and counting contributions cancel, while
\eqref{def of C5cd hole} gives $H_5^{(\mathrm{cd})}=1/3$.
\end{proof}

\begin{remark}
Theorems~\ref{theorem:precise-annular-hole-probability}, \ref{theorem:precise-disk-complement-hole-probability}, and \ref{theorem:precise-centered-disk-hole-probability} generalize and unify some of the results in \cite{C2021}.  
\end{remark}

\subsection{Fluctuations of smooth rotation invariant linear statistics}

The fluctuation of smooth rotation invariant linear statistics is encoded in $s\lambda$. 
Let
\[
A_{\rho_2,R}:=\mathbb D_{\rho_2}^{\mathrm c}\cap S
=\{z:\rho_2\leq |z|\leq R\},
\]
and let $\partial_{\mathrm n}$ denote differentiation along the outward unit
normal to $A_{\rho_2,R}$; thus $\partial_{\mathrm n}=\partial_r$ on
$|z|=R$ and $\partial_{\mathrm n}=-\partial_r$ on $|z|=\rho_2$.

\begin{theorem}[Fluctuations of smooth rotation invariant linear statistics for a centered-disk hole]
\label{theorem:smooth centered disk}
For each fixed $s\in\R$, as $n\to\infty$,
\begin{align}
\log \mathbb E_{n,\mathbb D_{\rho_2}}^{(\alpha)}
\!\left[e^{s\sum_{j=1}^n\lambda(z_j)}\right]
&=s\mu_{\mathbb D_{\rho_2}}(\lambda)n
+s\widetilde\mu_{\mathbb D_{\rho_2}}(\lambda)\log n
+s\,\mathsf e_\lambda[A_{\rho_2,R}]
\nonumber\\
&\quad
+\alpha s\bigl(\lambda(R)-\lambda(\rho_2)\bigr)
+\frac{s^2}{2}\mathsf v_\lambda[A_{\rho_2,R}]
\nonumber\\
&\quad
-\frac{s}{2}\rho_2\lambda'(\rho_2)
\log\!\left(
\frac{\rho_2\sqrt{\Delta Q(\rho_2)}}{\sqrt{2\pi}\,\tau_{\rho_2}}
\right)
+\mathcal O\!\left(\frac{(\log n)^3}{n^{1/12}}\right),
\end{align}
where
\begin{align}
\mu_{\mathbb D_{\rho_2}}(\lambda)
&:=\int_{\C}\lambda\,d\mu_{\mathrm{cd}},
\qquad
\widetilde\mu_{\mathbb D_{\rho_2}}(\lambda)
:=\frac{\rho_2}{4}\lambda'(\rho_2),
\\
\mathsf e_\lambda[A_{\rho_2,R}]
&:=\frac12\int_{A_{\rho_2,R}}\lambda\,\Delta\log\Delta Q\,dA
+\frac{1}{8\pi}\int_{\partial A_{\rho_2,R}}\partial_{\mathrm n}\lambda\,|dz|
-\frac{1}{8\pi}\int_{\partial A_{\rho_2,R}}
\lambda\,\frac{\partial_{\mathrm n}\Delta Q}{\Delta Q}\,|dz|,
\label{def:smooth-e-functional}
\\
\mathsf v_\lambda[A_{\rho_2,R}]
&:=\frac14\int_{A_{\rho_2,R}}|\nabla\lambda|^2\,dA.
\label{def:smooth-v-functional}
\end{align}
\end{theorem}

\begin{proof}
Set all counting parameters equal to zero in
Theorem~\ref{theorem:asymptotic determinant of centered disk}, and subtract
the resulting expansion at $s=0$ from the expansion at $s$.  The counting
coefficients vanish, and the coefficients of $n^2$, $n\log n$, and
$\sqrt n$ are independent of $s$.  At order $n$, the remaining contribution
is
\[
s\int_S\lambda\,d\mu_Q
-s\int_{\mathbb D_{\rho_2}}\lambda\,d\mu_Q
+s\tau_{\rho_2}\lambda(\rho_2)
=s\int_{\C}\lambda\,d\mu_{\mathrm{cd}}.
\]
The $s$-dependent part of the coefficient of $\log n$ is
$s\rho_2\lambda'(\rho_2)/4$.

It remains to identify the order-one term.  Write
$\mathsf e_\lambda[B]$ and $\mathsf v_\lambda[B]$ for the same bulk and
outward-normal boundary functionals on a radial region $B$.  The difference of the
unconstrained coefficients is
\[
s\,\mathsf e_\lambda[S]
+\frac{s^2}{2}\mathsf v_\lambda[S]
+\alpha s\bigl(\lambda(R)-\lambda(0)\bigr),
\]
whereas the difference of the centered-disk hole coefficients is
\[
-\frac{s\rho_2\lambda'(\rho_2)}{2}
 \log\!\left(
 \frac{\rho_2\sqrt{\Delta Q(\rho_2)}}
 {\sqrt{2\pi}\,\tau_{\rho_2}}
 \right)
-\alpha s\bigl(\lambda(\rho_2)-\lambda(0)\bigr)
-\Lambda_Q(0,\rho_2).
\]
The second representation in the definition of $\Lambda_Q$ gives
\[
\Lambda_Q(0,\rho_2)
=s\,\mathsf e_\lambda[\mathbb D_{\rho_2}]
+\frac{s^2}{2}\mathsf v_\lambda[\mathbb D_{\rho_2}].
\]
Therefore we obtain the asserted functional on $A_{\rho_2,R}$.  
The remainder follows by subtracting the two determinant expansions.
\end{proof}

In particular, with
\[
\mathsf m_\lambda^{(\mathrm{cd})}
:=\mathsf e_\lambda[A_{\rho_2,R}]
+\alpha\bigl(\lambda(R)-\lambda(\rho_2)\bigr)
-\frac{\rho_2\lambda'(\rho_2)}{2}
\log\!\left(\frac{\rho_2\sqrt{\Delta Q(\rho_2)}}{\sqrt{2\pi}\,\tau_{\rho_2}}\right),
\]
we have the convergence in distribution
\[
\sum_{j=1}^n\lambda(z_j)
-n\mu_{\mathbb D_{\rho_2}}(\lambda)
-\widetilde\mu_{\mathbb D_{\rho_2}}(\lambda)\log n
\ \xrightarrow{\ d\ }\ 
\mathcal N\!\left(\mathsf m_\lambda^{(\mathrm{cd})},
\mathsf v_\lambda[A_{\rho_2,R}]\right),
\]
where $\mathcal{N}(m,v)$ denotes a Gaussian random variable with mean $m$ and variance $v$. 
This follows from the preceding logarithmic MGF convergence.

We next give the corresponding statement for the disk-complement hole region. Put
\[
A_{0,\rho_1}:=S\cap\overline{\mathbb D}_{\rho_1}.
\]
Its outward normal is $+\partial_r$ on $|z|=\rho_1$.  In the next theorem,
$\mathsf e_\lambda[A_{0,\rho_1}]$ and
$\mathsf v_\lambda[A_{0,\rho_1}]$ are defined by
\eqref{def:smooth-e-functional} and \eqref{def:smooth-v-functional}, with
$A_{\rho_2,R}$ replaced by $A_{0,\rho_1}$.

\begin{theorem}[Fluctuations of smooth rotation invariant linear statistics for a disk-complement hole]
\label{theorem:smooth-disk-complement}
For each fixed $s\in\R$, as $n\to\infty$,
\begin{align}
\log \mathbb E_{n,U_{\mathrm{dc}}}^{(\alpha)}
 \!\left[e^{sX_n[\lambda]}\right]
&=sn\mu_{\mathrm{dc}}(\lambda)
-\frac{s}{4}\rho_1\lambda'(\rho_1)\log n
+s\mathsf e_\lambda[A_{0,\rho_1}]
\nonumber\\
&\quad
+\alpha s\bigl(\lambda(\rho_1)-\lambda(0)\bigr)
+\frac{s^2}{2}\mathsf v_\lambda[A_{0,\rho_1}]
\nonumber\\
&\quad
+\frac{s}{2}\rho_1\lambda'(\rho_1)
 \log\!\left(
 \frac{\rho_1\sqrt{\Delta Q(\rho_1)}}
 {\sqrt{2\pi}\,(1-\tau_{\rho_1})}
 \right)
+\mathcal O\!\left(\frac{(\log n)^3}{n^{1/12}}\right),
\label{eq:smooth-disk-complement-mgf}
\end{align}
where
\begin{equation}
\label{def:mu-disk-complement-smooth}
\mu_{\mathrm{dc}}(\lambda)
:=\int_{\C}\lambda\,d\mu_{\mathrm{dc}}
=2\int_0^{\rho_1}\lambda(r)r\Delta Q(r)\,dr
+(1-\tau_{\rho_1})\lambda(\rho_1).
\end{equation}
Consequently, if
\begin{equation}
\label{def:m-disk-complement-smooth}
\mathsf m_\lambda^{(\mathrm{dc})}
:=\mathsf e_\lambda[A_{0,\rho_1}]
+\alpha\bigl(\lambda(\rho_1)-\lambda(0)\bigr)
+\frac{\rho_1\lambda'(\rho_1)}{2}
 \log\!\left(
 \frac{\rho_1\sqrt{\Delta Q(\rho_1)}}
 {\sqrt{2\pi}\,(1-\tau_{\rho_1})}
 \right),
\end{equation}
then
\begin{equation}
\label{eq:smooth-disk-complement-clt}
X_n[\lambda]-n\mu_{\mathrm{dc}}(\lambda)
+\frac{\rho_1\lambda'(\rho_1)}4\log n
\ \xrightarrow{\ d\ }\
\mathcal N\!\left(
\mathsf m_\lambda^{(\mathrm{dc})},
\mathsf v_\lambda[A_{0,\rho_1}]
\right).
\end{equation}
\end{theorem}

\begin{proof}
Set all counting parameters to zero in
Theorem~\ref{theorem:asymptotic determinant of disk complement} and subtract
the same expansion at $s=0$.  The $s$-parts of
\eqref{def of C3dc hole} and \eqref{def of C5dc hole}, together with the
unconstrained contribution, give respectively
$s\mu_{\mathrm{dc}}(\lambda)n$ and
$-s\rho_1\lambda'(\rho_1)\log n/4$.  At order one,
\eqref{def of C6dc hole} contributes
\begin{align*}
&\frac{s}{2}\rho_1\lambda'(\rho_1)
 \log\!\left(
 \frac{\rho_1\sqrt{\Delta Q(\rho_1)}}
 {\sqrt{2\pi}\,(1-\tau_{\rho_1})}
 \right)
-\alpha s\bigl(\lambda(R)-\lambda(\rho_1)\bigr)
-\Lambda_Q(\rho_1,R).
\end{align*}
Combining this with the unconstrained order-one coefficient, and using the
second representation of $\Lambda_Q$, leaves precisely the two functionals
on $A_{0,\rho_1}$ and the root term displayed in
\eqref{eq:smooth-disk-complement-mgf}.
\end{proof}

The annular hole region is different because its allowed set has two connected components.  Define
\begin{equation}
\label{def:allowed-annular-set}
\Omega_{\mathrm{an}}
:=A_{0,\rho_1}\cup A_{\rho_2,R}.
\end{equation}
The outward normals are $+\partial_r$ at $\rho_1$ and $R$, and
$-\partial_r$ at $\rho_2$.  We use the functionals
$\mathsf e_\lambda[\Omega_{\mathrm{an}}]$ and
$\mathsf v_\lambda[\Omega_{\mathrm{an}}]$ from
\eqref{def:smooth-e-functional}--\eqref{def:smooth-v-functional}.

\begin{theorem}[Fluctuations of smooth rotation invariant linear statistics for an annular hole]
\label{theorem:smooth-annular-hole}
Let
\begin{align}
\delta_\lambda:=\lambda(\rho_2)-\lambda(\rho_1),
\label{def:annular-smooth-shortcuts},\qquad
x_n:=n\sigma_\star-\alpha+\frac12-\frac{\log\frac{\sigma_1}{\sigma_2}}{2\log\frac{\rho_2}{\rho_1}},
\qquad c_\lambda:=\frac{\delta_\lambda}{2\log\frac{\rho_2}{\rho_1}}.
\end{align}
For each fixed $s\in\R$, as $n\to\infty$,
\begin{align}
\log \mathbb E_{n,\mathbb A_{\rho_1,\rho_2}}^{(\alpha)}
 \!\left[e^{sX_n[\lambda]}\right]
&=sn\mu_{\mathrm{an}}(\lambda)
+s\widetilde\mu_{\mathrm{an}}(\lambda)\log n
+s\mathsf m_\lambda^{(\mathrm{an})}
+\frac{s^2}{2}\mathsf v_\lambda^{(\mathrm{an})}
\nonumber\\
&\quad
+\log\frac{
 \theta\!\left(x_n-c_\lambda s\,\middle|\,\frac{\pi i}{\log\frac{\rho_2}{\rho_1}}\right)}{
 \theta\!\left(x_n\,\middle|\,\frac{\pi i}{\log\frac{\rho_2}{\rho_1}}\right)}
+\mathcal O\!\left(\frac{(\log n)^3}{n^{1/12}}\right),
\label{eq:smooth-annular-mgf}
\end{align}
where
\begin{align}
\mu_{\mathrm{an}}(\lambda)
&:=\int_{\C}\lambda\,d\mu_{\mathrm{an}}
=2\int_0^{\rho_1}\lambda(r)r\Delta Q(r)\,dr
+2\int_{\rho_2}^{R}\lambda(r)r\Delta Q(r)\,dr
+\sigma_1\lambda(\rho_1)+\sigma_2\lambda(\rho_2),
\label{def:annular-smooth-leading-mean}\\
\widetilde\mu_{\mathrm{an}}(\lambda)
&:=\frac14\bigl(\rho_2\lambda'(\rho_2)
-\rho_1\lambda'(\rho_1)\bigr),
\label{def:annular-smooth-log-mean}\\
\mathsf h_\lambda^{(\mathrm{an})}
&:=\frac{\rho_1\lambda'(\rho_1)}2
 \log\!\left(
 \frac{\rho_1\sqrt{\Delta Q(\rho_1)}}{\sqrt{2\pi}\,\sigma_1}
 \right)
-\frac{\rho_2\lambda'(\rho_2)}2
 \log\!\left(
 \frac{\rho_2\sqrt{\Delta Q(\rho_2)}}{\sqrt{2\pi}\,\sigma_2}
 \right),
\label{def:annular-smooth-hard-mean}\\
\mathsf m_\lambda^{(\mathrm{an})}
&:=\mathsf e_\lambda[\Omega_{\mathrm{an}}]
+\alpha\bigl(\lambda(R)-\lambda(0)\bigr)
+\mathsf h_\lambda^{(\mathrm{an})}
+\frac{\delta_\lambda}{2}
+\frac{\delta_\lambda\log\frac{\sigma_1}{\sigma_2}}{2\log\frac{\rho_2}{\rho_1}},
\label{def:annular-smooth-order-one-mean}\\
\mathsf v_\lambda^{(\mathrm{an})}
&:=\mathsf v_\lambda[\Omega_{\mathrm{an}}]
+\frac{\delta_\lambda^2}{2\log\frac{\rho_2}{\rho_1}}.
\label{def:annular-smooth-variance}
\end{align}
\end{theorem}

\begin{proof}
Set all counting parameters to zero in
Theorem~\ref{theorem:asymptotic determinant of annulus} and subtract its
value at $s=0$.  The coefficients of $n$ and $\log n$ are
\eqref{def:annular-smooth-leading-mean} and
\eqref{def:annular-smooth-log-mean}.  At order one, the difference of
\eqref{def of C6 hole an} between $s$ and zero is
\[
s\mathsf h_\lambda^{(\mathrm{an})}
+\frac{s\delta_\lambda}{2}
+\frac{(\log\frac{\sigma_1}{\sigma_2}+s\delta_\lambda)^2-(\log\frac{\sigma_1}{\sigma_2})^2}{4\log\frac{\rho_2}{\rho_1}}
-\Lambda_Q(\rho_1,\rho_2).
\]
Adding the unconstrained coefficient and using the second representation of
$\Lambda_Q$ gives
\eqref{def:annular-smooth-order-one-mean} and
\eqref{def:annular-smooth-variance}.  Finally,
\eqref{def of calF hole an} at $s$ minus its value at zero is exactly the
theta ratio in \eqref{eq:smooth-annular-mgf}.
\end{proof}

The theta ratio in \eqref{eq:smooth-annular-mgf} is not a negligible. 
If $\delta_{\lambda}=0$, the centered linear statistic converges in distribution to a Gaussian random variable, whereas if $\delta_{\lambda}\neq 0$, its limiting behavior is non-Gaussian due to the additional theta-function contribution. The following comparison makes this distinction explicit.

\begin{corollary}
\label{theorem:annular-smooth-limit-law}
If $\delta_\lambda=0$, then as $n\to+\infty$, we have
\begin{equation}
Z_n[\lambda]:=X_n[\lambda]-n\mu_{\mathrm{an}}(\lambda)
-\widetilde\mu_{\mathrm{an}}(\lambda)\log n
\ \xrightarrow{\ d\ }\
\mathcal N\!\left(
\mathsf e_\lambda[\Omega_{\mathrm{an}}]
+\alpha(\lambda(R)-\lambda(0))
+\mathsf h_\lambda^{(\mathrm{an})},
\mathsf v_\lambda[\Omega_{\mathrm{an}}]
\right).
\label{eq:annular-smooth-gaussian-limit}
\end{equation}
If $\delta_\lambda\ne0$, let 
\[
\Upsilon_n:=\sigma_{\ast}n-\frac{1}{2}-\alpha+\frac{\log\frac{\sigma_2}{\sigma_1}}{2\log \frac{\rho_2}{\rho_1}}
\]
and $\langle\Upsilon_n\rangle=\Upsilon_n-\lfloor \Upsilon_n\rfloor$ be the fraction part of $\Upsilon_n$.
Let $\mathcal{Y}_n$ be the discrete Gaussian random variable on $\mathbb{Z}$ defined by 
\begin{equation}
\mathbb P(\mathcal{Y}_n=y)  =\frac{e^{-(y-\langle\Upsilon_n\rangle)^2\log\frac{\rho_2}{\rho_1}}}{\displaystyle\sum_{\ell\in\mathbb Z}
e^{-(\ell-\langle\Upsilon_n\rangle)^2\log\frac{\rho_2}{\rho_1}}},
\qquad y\in\mathbb Z. 
\end{equation}
Let $G$ be a random variable independent of $\mathcal{Y}_n$ with
\[
G\overset{d}{\sim}\mathcal N\!\left(
\mathsf e_\lambda[\Omega_{\mathrm{an}}]
+\alpha(\lambda(R)-\lambda(0))
+\mathsf h_\lambda^{(\mathrm{an})},
\mathsf v_\lambda[\Omega_{\mathrm{an}}]
\right).
\]
Let
\begin{equation}
    W_n[\lambda]:=G+\delta_{\lambda}\Bigl(\frac{1}{2}-\frac{\log\frac{\sigma_2}{\sigma_1}}{2\log\frac{\rho_2}{\rho_1}}+\langle\Upsilon_n\rangle-\mathcal{Y}_n\Bigr).
\end{equation}
Then as $n\to+\infty$, we have 
\begin{align}
\label{def of Zn Wn}
Z_n[\lambda]
-W_n[\lambda]
\ \xrightarrow{\ d\ }\
0.
\end{align}
\end{corollary}

\begin{proof}
Poisson summation gives, for $x\in \R$,
\begin{equation}
\label{eq:theta-discrete-gaussian-identity}
\theta\!\left(x\,\middle|\,\frac{\pi i}{\log\frac{\rho_2}{\rho_1}}\right)
=\sqrt{\frac{\log\frac{\rho_2}{\rho_1}}{\pi}}
 \sum_{k\in\mathbb Z}e^{-\log\frac{\rho_2}{\rho_1}(k-x)^2}.
\end{equation}
If $\delta_\lambda=0$, the theta ratio in
\eqref{eq:smooth-annular-mgf} is one and this yields
\eqref{eq:annular-smooth-gaussian-limit}. 
If $\delta_\lambda\ne0$, we use
\eqref{eq:theta-discrete-gaussian-identity} in
\eqref{eq:smooth-annular-mgf}. 
By the periodicity of $\theta$, we get 
\[
\frac{\theta\!\left(\Upsilon_n-\frac{s\delta_{\lambda}}{2\log\frac{\rho_2}{\rho_1}}\,\middle|\,\frac{\pi i}{\log\frac{\rho_2}{\rho_1}}\right)}
{\theta\!\left(\Upsilon_n\,\middle|\,\frac{\pi i}{\log\frac{\rho_2}{\rho_1}}\right)}
=
e^{-\frac{s^2\delta_{\lambda}^2}{4\log\frac{\rho_2}{\rho_1}}}\mathbb{E}\Bigl[e^{-s\delta_{\lambda}(\mathcal{Y}_n-\langle\mathcal{Y}_n\rangle)}\Bigr].
\]
Therefore, by the same argument of \cite[Proof of Corollary 1.7]{ACCL2}, as $n\to+\infty$, we obtain \eqref{def of Zn Wn}. 

\end{proof}


The three results on the fluctuations of smooth rotation invariant linear statistics can be summarized as follows.  In every
row the order-$n$ mean is $n\int\lambda\,d\mu_{\mathrm g}$ for $g=\mathrm{cd,dc,an}$.

For the three cases considered above, the logarithmic correction, the Dirichlet variance, and the form of the limiting law can be summarized as follows. In the centered-disk hole case, the coefficient of $\log n$ is
\[
\frac{\rho_2\lambda'(\rho_2)}{4},
\]
the variance is given by $\mathsf v_\lambda[A_{\rho_2,R}]$, and the limiting law is Gaussian. For the disk-complement hole case, the corresponding coefficient is
\[
-\frac{\rho_1\lambda'(\rho_1)}{4},
\]
while the variance is $\mathsf v_\lambda[A_{0,\rho_1}]$; the limiting law is again Gaussian. Finally, in the annular hole case, the coefficient of $\log n$ is
\[
\frac{\rho_2\lambda'(\rho_2)-\rho_1\lambda'(\rho_1)}{4},
\]
and the variance is $\mathsf v_\lambda[\Omega_{\mathrm{an}}]$. In this case, the limiting law is Gaussian when $\delta_\lambda=0$, whereas for $\delta_\lambda\neq0$ an additional theta-function contribution appears.

\subsection{Multivariate counting statistics for an annular hard wall region}
\label{subsection:annular-multivariate-counting}

We finally specialize the smooth parameter to zero and retain all annular
counting parameters.  This gives both an anisotropically scaled multivariate
CLT and a fixed-parameter refinement in which the theta term remains visible.
Set
\begin{equation}
\label{def:annular-counting-phase}
\Phi_n:=n\sigma_\star-\alpha+\frac12+\frac{\log\frac{\sigma_2}{\sigma_1}}{2\log\frac{\rho_2}{\rho_1}},
\qquad
\widehat \tau_{\log\frac{\rho_2}{\rho_1}}:=\frac{\pi i}{\log\frac{\rho_2}{\rho_1}},
\end{equation}
and abbreviate
\begin{equation}
\label{def:annular-q-sharp}
\mathfrak q(\boldsymbol s):=\log\mathsf Q,
\end{equation}
where $\mathsf Q$ is given by \eqref{def of mathsf Q}.

\begin{theorem}[Multivariate counting statistics for an annular hard wall region]
\label{theorem:annular-multivariate-mgf}
Let $s=0$.
There exists $\varepsilon>0$ such that the following expansion holds locally uniformly for $\boldsymbol s$ in the complex polydisc $\max_\ell|s_\ell|<\varepsilon$:
\begin{align}
\log\mathcal E_{n,\mathbb A_{\rho_1,\rho_2}}^{(\alpha)}
(\boldsymbol s;\boldsymbol r_{\mathrm{an}})
&=C_{3,\#}^{(\mathrm{an})}(\boldsymbol s) n
+\,C_{4,\#}^{(\mathrm{an})}(\boldsymbol s)\sqrt{n}
+C_{5,\#}^{(\mathrm{an})}(\boldsymbol s)\log n
\nonumber\\
&\quad
+C_{6,\#}^{(\mathrm{an})}(\boldsymbol s)
+\mathcal F_{n,\#}(\boldsymbol s)
+\mathcal O(n^{-1/12}),
\label{eq:annular-multivariate-mgf}
\end{align}
where $\mathcal F_{n,\#}(\boldsymbol s)$ is given by
\begin{align}
\mathcal F_{n,\#}(\boldsymbol s)
=\frac{\mathfrak q(\boldsymbol s)\log\frac{\sigma_2}{\sigma_1}}{2\log\frac{\rho_2}{\rho_1}}
+\frac{\mathfrak q(\boldsymbol s)^2}{4\log\frac{\rho_2}{\rho_1}}
+\log\frac{
\theta\!\left(\Phi_n+\frac{\mathfrak q(\boldsymbol s)}{2\log\frac{\rho_2}{\rho_1}}
\,\middle|\,\widehat \tau_{\log\frac{\rho_2}{\rho_1}}\right)}
{\theta\!\left(\Phi_n\,\middle|\,\widehat \tau_{\log\frac{\rho_2}{\rho_1}}\right)}.
\label{eq:explicit-annular-counting-theta}
\end{align}
\end{theorem}

\begin{proof}
Equation \eqref{eq:annular-multivariate-mgf} is
Theorem~\ref{theorem:counting statistics of annulus case}, written for the
multivariate MGF \eqref{def of multivariate MGF disk counting statistics}, with the smooth
parameter set to zero.  Substituting
$\log(\rho_2/\rho_1)$ and $\mathfrak q=\log\mathsf Q$ into
\eqref{def of calFNsharp} gives
\eqref{eq:explicit-annular-counting-theta}.  
\end{proof}

\begin{remark}
Theorem~\ref{theorem:annular-multivariate-mgf} generalizes and unifies some of the results in \cite{CL2023,ACCL1,ACCL2}.    
\end{remark}

To state the CLT, put
\begin{equation}
I_{\mathrm h}:=\{2m+1,\ldots,4m\},
\qquad
I_{\mathrm r}:=\{1,\ldots,2m\}\cup\{4m+1,\ldots,7m\},
\end{equation}
and write
$\boldsymbol N_n=(\mathrm N(r_1),\ldots,\mathrm N(r_{7m}))^{\mathsf T}$.
Let $\mathsf D_n$ be diagonal with
\begin{equation}
(\mathsf D_n)_{\ell\ell}
:=\begin{cases}
\sqrt n,&\ell\in I_{\mathrm h},\\
n^{1/4},&\ell\in I_{\mathrm r}.
\end{cases}
\label{def:annular-counting-normalization}
\end{equation}

\begin{theorem}[Multivariate CLT for annular disk counting statistics]
\label{theorem:annular-multivariate-clt}
As $n\to\infty$,
\begin{equation}
\mathsf D_n^{-1}
\bigl(\boldsymbol N_n-\mathbb E\boldsymbol N_n\bigr)
\ \xrightarrow{\ d\ }\
\mathcal N_{7m}(\boldsymbol0,\boldsymbol\Sigma),
\label{eq:annular-multivariate-clt}
\end{equation}
where
\begin{equation}
\Sigma_{\ell k}
:=\begin{cases}
\partial_{s_\ell}\partial_{s_k}
C_{3,\#}^{(\mathrm{an})}(\boldsymbol0),
&\ell,k\in I_{\mathrm h},\\[2pt]
\partial_{s_\ell}\partial_{s_k}
C_{4,\#}^{(\mathrm{an})}(\boldsymbol0),
&\ell,k\in I_{\mathrm r},\\[2pt]
0,&\text{otherwise}.
\end{cases}
\label{def:annular-counting-limit-covariance}
\end{equation}
Each mean can be read directly from Subsection~\ref{subsection:counting statistics functionals}:
\begin{align}
\mathbb E\,\mathrm N(r_\ell)
&=n\,\partial_{s_\ell}C_{3,\#}^{(\mathrm{an})}(\boldsymbol0)
+\sqrt n\,\partial_{s_\ell}C_{4,\#}^{(\mathrm{an})}(\boldsymbol0)
+(\log n)\,\partial_{s_\ell}C_{5,\#}^{(\mathrm{an})}(\boldsymbol0)
+\mathcal O(1).
\label{eq:annular-counting-means}
\end{align}
For $c>0$, define
\begin{equation}
\mathsf A_c(t):=
\begin{cases}
\dfrac{1-e^{-2ct}}{2t},&t>0,\\[5pt]
c,&t=0.
\end{cases}
\label{def:hard-edge-covariance-function}
\end{equation}
If $p,q\in\{2m+1,\ldots,3m\}$, then
\begin{equation}
\Sigma_{pq}
=\mathsf A_{\sigma_1}\!\bigl(\max\{t_p,t_q\}\bigr)
-\mathsf A_{\sigma_1}(t_p+t_q),
\label{eq:inner-hard-edge-covariance}
\end{equation}
whereas for $p,q\in\{3m+1,\ldots,4m\}$,
\begin{equation}
\Sigma_{pq}
=\mathsf A_{\sigma_2}\!\bigl(\max\{t_p,t_q\}\bigr)
-\mathsf A_{\sigma_2}(t_p+t_q).
\label{eq:outer-hard-edge-covariance}
\end{equation}
The covariance between the two hard-wall blocks is zero.  The remaining
entries are the Hessian at zero of the explicit $C_{4,\#}^{(\mathrm{an})}$ integrals in Subsection~\ref{subsection:counting statistics functionals}.
\end{theorem}

\begin{proof}
For a fixed vector $\boldsymbol u$, substitute
$s_\ell=u_\ell/\sqrt n$ on $I_{\mathrm h}$ and
$s_\ell=u_\ell/n^{1/4}$ on $I_{\mathrm r}$ in
\eqref{eq:annular-multivariate-mgf}.  To make the block structure explicit,
put $I_-:=\{2m+1,\ldots,3m\}$,
$I_+:=\{3m+1,\ldots,4m\}$, and
$S_+:=\sum_{\ell\in I_+}s_\ell$.  The telescoping weights give
\begin{align*}
1+\widehat{\mathsf T}_0^{(3m+1,7m)}(\tau_{\rho_2})
&=\exp\!\left(\sum_{\ell=3m+1}^{7m}s_\ell\right),\\
\frac{\mathsf T_0^{(2m+1,3m)}(x)}
{1+\widehat{\mathsf T}_0^{(3m+1,7m)}(\tau_{\rho_2})}
&=\sum_{\ell\in I_-}\omega_\ell^{(3m)}
e^{-2t_\ell(x-\tau_{\rho_1})},\\
\frac{\widehat{\mathsf T}_0^{(3m+1,4m)}(x)}
{1+\widehat{\mathsf T}_0^{(3m+1,7m)}(\tau_{\rho_2})}
&=e^{-S_+}\sum_{\ell\in I_+}\omega_\ell^{(4m)}
e^{-2t_\ell(\tau_{\rho_2}-x)}.
\end{align*}
Consequently,
$C_{3,\#}^{(\mathrm{an})}$ is the sum of a linear function of the regular
parameters and two functions depending separately on the inner and outer
hard blocks.  Moreover, $C_{4,\#}^{(\mathrm{an})}$ depends only on the
regular parameters.  Hence all mixed Hessians invoked below vanish.

Local uniformity in the complex polydisc and Cauchy estimates allow Taylor
expansion and differentiation of the remainder at the origin.  After centering,
we obtain \eqref{eq:annular-multivariate-clt} by L\'evy’s continuity theorem.
\end{proof}

The theta term is therefore invisible at the central-limit scales, but if
$\{\Phi_{n_j}\}\to\phi\in[0,1)$ along some subsequences, then
\begin{align}
&\exp\!\left(
-n_jC_{3,\#}^{(\mathrm{an})}
-\sqrt{n_j}C_{4,\#}^{(\mathrm{an})}
-(\log n_j)C_{5,\#}^{(\mathrm{an})}
\right)
\mathcal E_{n_j,\mathbb A_{\rho_1,\rho_2}}^{(\alpha)}
(\boldsymbol s;\boldsymbol r_{\mathrm{an}})
\nonumber\\
&\qquad\longrightarrow
\exp\!\left(
C_{6,\#}^{(\mathrm{an})}
+\frac{\log(\sigma_2/\sigma1)\mathfrak q}{2L}+\frac{\mathfrak q^2}{4\log(\rho_2/\rho_1)}
\right)
\frac{\theta\!\left(\phi+\frac{\mathfrak q}{2\log(\rho_2/\rho_1)}
\,\middle|\,\widehat \tau_{\log(\rho_2/\rho_1)}\right)}
{\theta(\phi\mid \widehat \tau_{\log(\rho_2/\rho_1)})}.
\label{eq:annular-fixed-parameter-mod-theta}
\end{align}


Then we have the following.

\begin{corollary}
\label{theorem:annular-filling-number-law}
For every fixed $u\in\mathbb R$, as $n\to+\infty$, we have 
\begin{align}
\log\mathbb E_{n,\mathbb A_{\rho_1,\rho_2}}^{(\alpha)}
\!\left[e^{u\mathrm N(\rho_1)}\right]
&=n\sigma_\star u-\left(\alpha+\frac12\right)u
+\frac{\log(\sigma_2/\sigma_1)u}{2\log(\rho_2/\rho_1)}+\frac{u^2}{4\log(\rho_2/\rho_1)}
\nonumber\\
&\quad
+\log\frac{\theta\!\left(
\Phi_n+\frac{u}{2\log(\rho_2/\rho_1)}\,\middle|\,\widehat \tau_{\log(\rho_2/\rho_1)}\right)}
{\theta(\Phi_n\mid \widehat \tau_{\log(\rho_2/\rho_1)})}
+\mathcal O(n^{-1/12}).
\label{eq:annular-filling-number-mgf}
\end{align}
Let 
\[
\Lambda_n:=\sigma_{\ast}n-\frac{1}{2}-\alpha+\frac{\log\frac{\sigma_2}{\sigma_1}}{2\log\frac{\rho_2}{\rho_1}},
\]
and $\langle \Lambda_n \rangle:=\Lambda_n-\lfloor \Lambda_n \rfloor$ be the fractional part of $\Lambda_n$. 
Let $\mathcal{X}_n$ be the discrete Gaussian random variable on $\mathbb Z$ defined by 
\[
\mathbb{P}(\mathcal{X}_n=x)=\frac{e^{-(x-\langle\Lambda_n\rangle)^2}\log(\rho_2/\rho_1)}{\sum_{x\in\mathbb Z}
e^{-(x-\langle\Lambda_n\rangle)^2}\log(\rho_2/\rho_1)},\qquad x\in \mathbb Z. 
\]
Then as $n\to+\infty$, we have 
\begin{equation}
\label{def of convergence of fraction}    
\mathbb{P}(\mathrm{N}(\rho_1)=\lfloor\Lambda_n\rfloor+x)=
\mathbb{P}(\mathcal{X}_n=x)+o(1).
\end{equation}
\end{corollary}

\begin{proof}
In \eqref{eq:annular-multivariate-mgf}, retain one inner hard-edge insertion
at $t=0$ with parameter $u$ and set all other counting parameters to zero.
Then
\[
C_{3,\#}^{(\mathrm{an})}=\sigma_\star u,
\qquad C_{4,\#}^{(\mathrm{an})}=C_{5,\#}^{(\mathrm{an})}=0,
\qquad C_{6,\#}^{(\mathrm{an})}
=-\left(\alpha+\frac12\right)u,
\qquad \mathsf Q=e^u.
\]
This proves \eqref{eq:annular-filling-number-mgf}.  
\eqref{def of convergence of fraction} follows from the same argument in \cite[Proof of Corollary 1.7]{ACCL2}

\end{proof}

We conclude this section by summarizing our results in Table~\ref{tab:hard-wall-results-summary}.

\begin{table}[p]
\centering
\small
\setlength{\tabcolsep}{4pt}
\renewcommand{\arraystretch}{1.20}
\caption{Summary of the hard-wall asymptotics and fluctuations.
The set $U_{\mathrm g}$ is the hard wall region, whereas
$\Omega_{\mathrm g}=S\cap U_{\mathrm g}^{\mathrm c}$ is the support of the equilibrium measure. }
\label{tab:hard-wall-results-summary}
\vspace{3pt}
\begin{tabularx}{\textwidth}{@{}>{\raggedright\arraybackslash}p{0.24\textwidth}
  *{3}{>{\centering\arraybackslash}X}@{}}
\toprule
& Centered disk & Disk complement & Annulus \\
& $\mathrm{cd}$ & $\mathrm{dc}$ & $\mathrm{an}$ \\
\midrule
Hard wall region $U_{\mathrm g}$
& $\mathbb D_{\rho_2}$
& $\mathbb D_{\rho_1}^{\mathrm c}$
& $\mathbb A_{\rho_1,\rho_2}$ \\
Support of equilibrium measure $\Omega_{\mathrm g}$
& $A_{\rho_2,R}$
& $A_{0,\rho_1}$
& $A_{0,\rho_1}\cup A_{\rho_2,R}$ \\
\midrule
\multicolumn{4}{@{}l}{\textit{A. Determinants and hole probabilities}}\\[2pt]
Order in the asymptotic expansion
& \multicolumn{3}{c}{$n^2,\ n\log n,\ n,\ \sqrt n,\ \log n,\ 1$} \\
Theta term in $\log D_n^{(\mathrm g)}$
& $0$ & $0$
& $\mathcal F_n[\lambda]+\mathcal F_{n,\#}$ \\
Coefficient of $\log n$ in $\log\mathcal P_n[Q,U_{\mathrm g}]$
& $\dfrac13$ & $-\dfrac14$ & $0$ \\
Theta term in $\log\mathcal P_n[Q,U_{\mathrm g}]$
& $0$ & $0$
& $\log\theta(\xi_n^{(0)}\mid\pi i/\log(\rho_2/\rho_1))$ \\
Remainder
& \multicolumn{3}{c}{$\mathcal O(\varepsilon_n)$} \\
\midrule
\multicolumn{4}{@{}l}{\textit{B. Smooth rotation invariant linear statistics}}\\[2pt]
Order-$n$ centering
& \multicolumn{3}{c}{$n\mu_{\mathrm g}(\lambda)
   =n\int_{\mathbb C}\lambda\,d\mu_{\mathrm g}$} \\
Coefficient $\widetilde\mu_{\mathrm g}(\lambda)$ of the $\log n$ centering
& $\dfrac{\rho_2\lambda'(\rho_2)}4$
& $-\dfrac{\rho_1\lambda'(\rho_1)}4$
& $\dfrac{\rho_2\lambda'(\rho_2)-\rho_1\lambda'(\rho_1)}4$ \\
Gaussian-component variance
& $\mathsf v_\lambda[A_{\rho_2,R}]$
& $\mathsf v_\lambda[A_{0,\rho_1}]$
& $\mathsf v_\lambda[\Omega_{\mathrm{an}}]$ \\
Additional centered log-MGF term
& $0$ & $0$ & $\mathcal T_n(s)$ \\
Fluctuation law
& Gaussian
& Gaussian
& Gaussian if $\delta_\lambda=0$; non-Gaussian otherwise \\
Remainder
& \multicolumn{3}{c}{$\mathcal O(\varepsilon_n)$ for fixed $s\in\mathbb R$} \\
\midrule
\multicolumn{4}{@{}l}{\textit{C. Counting statistics for an annular hole region ($s=0$)}}\\[2pt]
Joint log-MGF, fixed parameters
& \multicolumn{3}{c}{$\begin{gathered}
 nC_{3,\#}^{(\mathrm{an})}+\sqrt n\,C_{4,\#}^{(\mathrm{an})}
 +(\log n)C_{5,\#}^{(\mathrm{an})}\\
 {}+C_{6,\#}^{(\mathrm{an})}+\mathcal F_{n,\#}
 +\mathcal O(n^{-1/12})
 \end{gathered}$} \\
Multivariate CLT
& \multicolumn{3}{c}{$\mathsf D_n^{-1}
 (\boldsymbol N_n-\mathbb E\boldsymbol N_n)
 \xrightarrow{\ d\ }\mathcal N_{7m}(\boldsymbol0,\boldsymbol\Sigma)$} \\
Normalization
& \multicolumn{3}{c}{$\sqrt n\ \text{on }I_{\mathrm h},
 \qquad n^{1/4}\ \text{on }I_{\mathrm r}$} \\
 $\mathrm N(\rho_1)$
& \multicolumn{3}{c}{Order-one discrete Gaussian fluctuation} \\
\bottomrule
\end{tabularx}

\vspace{5pt}
\begin{minipage}{\textwidth}
\footnotesize
\textit{Conventions.}
Hole probabilities are specialized to $\alpha=s=0$ and
$\boldsymbol s=\boldsymbol0$.
$A_{a,b}=\{z:a\leq |z|\leq b\}$,
$\varepsilon_n=(\log n)^3n^{-1/12}$,
$\delta_\lambda=\lambda(\rho_2)-\lambda(\rho_1)$, and
$\mathsf v_\lambda[B]=\frac14\int_B|\nabla\lambda|^2\,dA$.
We write $\langle x\rangle=x-\lfloor x\rfloor$ for the fractional part.
The phases are
\[
 \xi_n^{(0)}=n\sigma_\star+\frac12
  +\frac{\log(\sigma_2/\sigma_1)}{2\log(\rho_2/\rho_1)},
 \qquad x_n=\xi_n^{(0)}-\alpha=\Phi_n.
\]
In panel B, write
$Z_n^{(\mathrm g)}=X_n[\lambda]-n\mu_{\mathrm g}(\lambda)
-\widetilde\mu_{\mathrm g}(\lambda)\log n$. The convention is
\[
 \log\mathbb E[e^{sZ_n^{(\mathrm g)}}]
 =s\mathsf m_\lambda^{(\mathrm g)}
 +\frac{s^2}{2}\mathsf v_\lambda[\Omega_{\mathrm g}]
 +\mathbf 1_{\{\mathrm g=\mathrm{an}\}}\mathcal T_n(s)
 +\mathcal O(\varepsilon_n),\qquad
 \mathcal T_n(s)=\frac{s^2\delta_\lambda^2}{4\log(\rho_2/\rho_1)}
 +\log\frac{\theta(x_n-\delta_\lambda s/(2\log(\rho_2/\rho_1))\mid\pi i/\log(\rho_2/\rho_1))}
 {\theta(x_n\mid\pi i/\log(\rho_2/\rho_1))}.
\]
The order-one coefficients $\mathsf m_\lambda^{(\mathrm g)}$ are those in the
smooth-statistics theorems. The displayed Gaussian-component variance is
not the full annular variance when $\delta_\lambda\ne0$.
In panel C,
$I_{\mathrm h}=\{2m+1,\ldots,4m\}$ and
$I_{\mathrm r}=\{1,\ldots,2m\}\cup\{4m+1,\ldots,7m\}$.
Local uniformity of the joint log-MGF in a complex neighbourhood of zero
also gives asymptotics for all mixed cumulants by differentiation.
\end{minipage}
\end{table}

\subsection*{Organization of the paper}

Section~\ref{section: introduction} introduces the model, states the three
determinant expansions, and derives their consequences for hole
probabilities, smooth statistics, and annular disk counting statistics.
Section~\ref{section:coefficient-catalogue} collects the universal functions
and explicit coefficients.
Section~\ref{section:preliminary} gives the exact radial factorization and
the uniform asymptotics of the norms.
Section~\ref{section:hole-proofs} proves the hard-wall ratio expansions,
including the annular theta term.
Section~\ref{section:counting-proofs} establishes the counting expansions
in all four local regimes and completes the proofs of the main determinant
theorems.

\section{Lists of coefficients}
\label{section:coefficient-catalogue}

This section collects the several functionals to describe the results and explicit coefficients used
in the main theorems.  
Here and below, the superscripts $\mathrm{an}$, $\mathrm{dc}$, and
$\mathrm{cd}$ refer, respectively, to the annular-hole,
disk-complement-hole, and centered-disk-hole regions. For the coefficients of the disc counting statistics in each regime, $\mathrm b$, $\mathrm{se}$, $\mathrm h$, and $\mathrm s$
denote the bulk, semi-hard, hard, and soft regimes, while $\mathrm{in}$ and
$\mathrm{out}$ distinguish the two sides of a hard wall.
The coefficient index follows the convention
$k=1,\ldots,6$ for the orders
$n^2,n\log n,n,\sqrt n,\log n,1$, respectively.

\subsection{Hole-probability functionals}
We begin with the universal functionals and numerical constants entering the
hole-probability expansions.
\begin{description}
\item[(Entropy-type functional)] Define
\begin{equation}
    E_{U}[\mu_{Q}]:=\int_{U}\log \Delta Q\,d\mu_{Q}. 
\end{equation}
For example,
\begin{align*}
E_{\mathbb{D}_{\rho_2}}[\mu_Q]&=\int_{0}^{\rho_2}2u\Delta Q(u)\log\Delta Q(u)\,du,\\
E_{\mathbb{A}_{\rho_1,\rho_2}}[\mu_Q]&=\int_{\rho_1}^{\rho_2}2u\Delta Q(u)\log\Delta Q(u)\,du.
\end{align*}
\item[(Complementary-error-function constants)] Define
\begin{align}
    \Phi(x)&:=\log\Bigl(\frac{1}{2}\erfc(x)\Bigr),\qquad
    \Phi'(x)=-\frac{2e^{-x^2}}{\sqrt{\pi}\erfc(x)},
    \\
\mathcal{I}&:=
\int_{-\infty}^0 \log\Bigl( 
\frac{1}{2}\mathrm{erfc}(y)
\Bigr)\,dy
+
\int_{0}^{\infty}
\log\Bigl( 
\sqrt{\pi}ye^{y^2}\mathrm{erfc}(y)
\Bigr) \,dy.
\end{align}
We also define    
\begin{align}
\alpha_{\mathrm{in}}&:=\int_{-\infty}^{0}\frac{x^2-2}{3}\Phi'(x)\,dx, 
\\
\alpha_{\mathrm{out}}&:=\int_{0}^{+\infty}\bigg[
\frac{1}{3}\Bigl( 
(x^2-2)\Phi'(x)+x(2x^2-3)
\Bigr)-\frac{1}{1+x}
\bigg]\,dx, 
\\
\beta_{\mathrm{in}}&:=-\int_{-\infty}^{0}\frac{x^2+1}{6}\Phi'(x)\,dx,
\\
\beta_{\mathrm{out}}&:=-\frac{1}{6}\int_0^{+\infty}\Bigl( 
(x^2+1)\Phi'(x)+x(2x^2+3)
\Bigr)\,dx, 
\end{align}
and 
\begin{equation}
\mathcal{J}(\rho)
:=
-\bigl(\alpha_{\mathrm{in}}+\alpha_{\mathrm{out}}\bigr)
+\bigl(\beta_{\mathrm{in}}+\beta_{\mathrm{out}}\bigr)\frac{\rho\,\partial_r\Delta Q(\rho)}{\Delta Q(\rho)}. 
\end{equation}
\item[(Cumulative mass coordinate)] 
For $0\leq\rho\leq R$, the radial mass coordinate admits the equivalent representations
\begin{equation}
    \tau_{\rho}=\frac{1}{2}\rho q'(\rho)=\int_0^{\rho}2r\Delta Q(r)\,dr
    =\int_{\mathbb{D}_{\rho}}\Delta Q(z)\,dA(z)
    =\int_{S\cap\mathbb{D}_{\rho}}d\mu_{Q}(z).
\end{equation}
Since $\mu_Q(S)=1$, we likewise have
\begin{equation}
    1-\tau_{\rho}=\int_{\rho}^{R}2r\Delta Q(r)\,dr=\int_{S\cap\mathbb{D}_{\rho}^{\mathrm{c}}}\Delta Q(z)\,dA(z)
    =\int_{S\cap\mathbb{D}_{\rho}^{\mathrm{c}}}d\mu_{Q}(z). 
\end{equation}
For later use, we adopt the shorthand
\begin{equation}
    \mu_Q[U]:=\int_{S\cap U}d\mu_{Q}(z). 
\end{equation}
\item[(Potential and smooth-statistic functionals)]
Define
\begin{align}
\begin{split}
\label{def of FQab}
F_{Q}(a,b)
&:=\frac{1}{12}\log\Bigl(\frac{b^2\Delta Q(b)}{a^2\Delta Q(a)}\Bigr)
-\frac{1}{16}\Bigl(\frac{b\partial_r\Delta Q(b)}{\Delta Q(b)}-\frac{a\partial_r\Delta Q(a)}{\Delta Q(a)}\Bigr)
\\
&\quad 
+\frac{1}{24}\int_a^{b}\Bigl(\frac{\partial_r\Delta Q(u)}{\Delta Q(u)}\Bigr)^2u\,du,   
\end{split}    
\\
\begin{split}
\Lambda_Q(a,b)
&:=
\int_{a}^{b}
\bigg[
\frac{s^2}{4}u\lambda'(u)^2
-\frac{s}{4}u\lambda'(u)\frac{\partial_u\Delta Q(u)}{\Delta Q(u)}
\bigg]\,du
+\frac{s}{4}\bigl(b\lambda'(b)-a\lambda'(a)\bigr)
\\
&\,=\frac{s^2}{2}\cdot\frac{1}{4}\int_{\mathbb{A}_{a,b}}|\nabla\lambda(z)|^2\,dA(z)
+\frac{s}{2}\int_{\mathbb{A}_{a,b}}\lambda(z)\Delta \log \Delta Q(z)\,dA(z)\\
&\quad
+\frac{s}{8\pi}\int_{\partial\mathbb{A}_{a,b}}\partial_{\mathrm{n}}\lambda(z)\,|dz|-\frac{s}{8\pi}\int_{\partial\mathbb{A}_{a,b}}\lambda(z)\frac{\partial_{\mathrm{n}}\Delta Q(z)}{\Delta Q(z)}\,|dz|.
\end{split}
\end{align}
For the centered-disk geometry, we use the following one-endpoint version of
\eqref{def of FQab}: for $0<b<R$, define
\begin{equation}
\widetilde{F}_{Q}(b):=\frac{1}{12}\log\big(b^2\Delta Q(b)\bigr)-\frac{1}{16}\frac{b\,\partial_r\Delta Q(b)}{\Delta Q(b)}+\frac{1}{24}\int_0^b \Bigl(\frac{\partial_u\Delta Q(u)}{\Delta Q(u)}\Bigr)^2u\,du.
\end{equation}
\end{description}

\subsubsection{Annular hole}
The annular hard-wall coefficients and their bounded filling-fraction correction
are defined by
\begin{align}
\begin{split}
\label{def of C1 hole an}
C_{1,\mathrm{hole}}^{(\mathrm{an})}
&:=
\frac{(q(\rho_2)-q(\rho_1))^2}{4\log\frac{\rho_2}{\rho_1}}
+\tau_{\rho_1}q(\rho_1)
-\tau_{\rho_1}^2\log \rho_1
-\tau_{\rho_2}q(\rho_2)
+\tau_{\rho_2}^2\log \rho_2
\\
&\quad
+2\int_{\rho_1}^{\rho_2}\Bigl(q(u)-uq'(u)\log u\Bigr)u\Delta Q(u)\,du,
\end{split}
\\
\label{def of C2 hole an}
C_{2,\mathrm{hole}}^{(\mathrm{an})}
&:=
-\frac{1}{2}\int_{\mathbb{A}_{\rho_1,\rho_2}}\,d\mu_Q(z),
\\
\begin{split}
\label{def of C3 hole an}
C_{3,\mathrm{hole}}^{(\mathrm{an})}
&:=
\int_{\rho_1}^{\rho_2}
\Bigl(\log \Delta Q(u)-2s\lambda(u)-2\log u \Bigr)u\Delta Q(u)\,du
\\
&\quad 
+(\tau_{\rho_2}-\sigma_{\star})\Bigl(2\log\rho_2
+s\lambda(\rho_2)-\log (\tau_{\rho_2}-\sigma_{\star})
\Bigr)
\\
&\quad
+(\sigma_{\star}-\tau_{\rho_1})
\Bigl(2\log \rho_1
+s\lambda(\rho_1)
-\log (\sigma_{\star}-\tau_{\rho_1})\Bigr)
+(\tau_{\rho_2}-\tau_{\rho_1})(1+\log\frac{1}{\sqrt{2\pi}} ) , 
\end{split}
\\
\label{def of C4 hole an}
C_{4,\mathrm{hole}}^{(\mathrm{an})}
&:=
\mathcal{I}\sum_{k=1}^{2}\sqrt{2}\rho_k\sqrt{\Delta Q(\rho_k)},
\\
\label{def of C5 hole an}
C_{5,\mathrm{hole}}^{(\mathrm{an})}
&:=
\frac{s}{4}\rho_2\lambda'(\rho_2)-\frac{s}{4}\rho_1\lambda'(\rho_1), 
\\
\begin{split}
\label{def of C6 hole an}
C_{6,\mathrm{hole}}^{(\mathrm{an})}
&:=
\frac{1}{2}\log\Bigl(\frac{\pi}{\log\frac{\rho_2}{\rho_1}}\Bigr)
+\frac{1}{12}\log\Bigl(\frac{\rho_2}{\rho_1}\frac{\Delta Q(\rho_1)}{\Delta Q(\rho_2)}\Bigr)
+\frac{s}{2}(\lambda(\rho_2)-\lambda(\rho_1))
\\
&\quad
+\frac{1}{4\log\frac{\rho_2}{\rho_1}}
\Bigl(\log\Bigl(\frac{\sigma_{\star}-\tau_{\rho_1}}{\tau_{\rho_2}-\sigma_{\star}}e^{s(\lambda(\rho_2)-\lambda(\rho_1))}\Bigr)\Bigr)^2
-\sum_{j=1}^{\infty}\log\Bigl(1-\Bigl(\frac{\rho_1}{\rho_2}\Bigr)^{2j}\Bigr)
\\
&\quad
+
\frac{s}{2}\rho_1\lambda'(\rho_1)
\log\Bigl(\frac{\rho_1\sqrt{\Delta Q(\rho_1)}}{\sqrt{2\pi}(\sigma_{\star}-\tau_{\rho_1})}\Bigr)
-
\frac{s}{2}\rho_2\lambda'(\rho_2)
\log\Bigl(\frac{\rho_2\sqrt{\Delta Q(\rho_2)}}{\sqrt{2\pi}(\tau_{\rho_2}-\sigma_{\star})}\Bigr)
\\
&\quad
+\frac{\rho_1^2\Delta Q(\rho_1)}{\sigma_{\star}-\tau_{\rho_1}}
+\frac{\rho_2^2\Delta Q(\rho_2)}{\tau_{\rho_2}-\sigma_{\star}}
-F_Q(\rho_1,\rho_2)
-\Lambda_Q(\rho_1,\rho_2)
+\mathcal{J}(\rho_1)-\mathcal{J}(\rho_2), 
\end{split}
\\
\label{def of calF hole an}
\mathcal{F}_{n}[\lambda]&:=   
\log\theta
\Bigl( 
n\sigma_{\star}-\alpha+\frac{1}{2}
+
\frac{\log(\frac{\sigma_{\star}-\tau_{\rho_1}}{\tau_{\rho_2}-\sigma_{\star}})}{2\log\frac{\rho_1}{\rho_2}}
+\frac{s(\lambda(\rho_2)-\lambda(\rho_1))}{2\log\frac{\rho_1}{\rho_2}}
\Bigr|\frac{\pi i}{\log\frac{\rho_2}{\rho_1}}
\Bigr). 
\end{align}
\subsubsection{Disk-complement hole}
The disk-complement hard-wall coefficients are defined by
\begin{align}
\label{def of C1dc hole}
C_{1,\mathrm{hole}}^{(\mathrm{dc})}
&:=
-(1-\tau_{\rho_1})q(\rho_1)+(1-\tau_{\rho_1}^2)\log\rho_1
+2\int_{\rho_1}^{R}\Bigl(q(u)-uq'(u)\log u\Bigr)u\Delta Q(u)\,du,
\\
\label{def of C2dc hole}
C_{2,\mathrm{hole}}^{(\mathrm{dc})}
&:=
-\frac{1}{2}(1-\tau_{\rho_1}), 
\\
\begin{split}
\label{def of C3dc hole}
C_{3,\mathrm{hole}}^{(\mathrm{dc})}
&:=
-\frac{1}{2}(q(R)-q(\rho_1))
+\log \frac{R}{\rho_1}
\\
&\quad
+(1-\tau_{\rho_1})\Bigl(1-\log(1-\tau_{\rho_1})-\log\sqrt{2\pi}+2\log \rho_1+\mathsf{k}(\rho_1)\Bigr)
\\
&\quad
+\int_{\rho_1}^{R}
\Bigl(\log \Delta Q(u)-2s\lambda(u)-2(2\alpha+1)\log u\Bigr)u\Delta Q(u)\,du,
\end{split}
\\
\label{def of C4dc hole}
C_{4,\mathrm{hole}}^{(\mathrm{dc})}
&:=
\rho_1\sqrt{2\Delta Q(\rho_1)}\mathcal{I},
\\
\label{def of C5dc hole}
C_{5,\mathrm{hole}}^{(\mathrm{dc})}
&:=
-s\frac{\rho_1}{4}\lambda'(\rho_1)
-\frac{2\alpha+1}{4}, 
\\
\begin{split}
\label{def of C6dc hole}
C_{6,\mathrm{hole}}^{(\mathrm{dc})}
&:=
\Bigl(\frac{1}{2}+\frac{\rho_1}{2}\mathsf{k}'(\rho_1)\Bigr)
\log\Bigl(\frac{\rho_1\sqrt{\Delta Q(\rho_1)}}{\sqrt{2\pi}(1-\tau_{\rho_1})}\Bigr)
+\log(2\sqrt{\pi})
+\frac{\rho_1^2\Delta Q(\rho_1)}{1-\tau_{\rho_1}}
\\
&\quad
+\Bigl(\frac{\alpha}{2}+\frac{1}{6}\Bigr)\log \Bigl(\frac{\Delta Q(R)}{\Delta Q(\rho_1)}\Bigr)
+\Bigl(-\alpha^2+\frac{1}{3}\Bigr)\log\frac{R}{\rho_1}
-s\alpha(\lambda(R)-\lambda(\rho_1))
\\
&\quad
-F_Q(\rho_1,R)-\Lambda_Q(\rho_1,R)
+\mathcal{J}(\rho_1). 
\end{split}
\end{align}

\subsubsection{Centered-disk hole}
The centered-disk hard-wall coefficients are defined by

\begin{align}
\label{def of C1cd hole}
C_{1,\mathrm{hole}}^{(\mathrm{cd})}
&:=
-\tau_{\rho_2}q(\rho_2)
+\tau_{\rho_2}^2\log \rho_2
+2\int_{0}^{\rho_2}\Bigl(q(u)-uq'(u)\log u
\Bigr)u\Delta Q(u)\,du,
\\
\label{def of C2cd hole}
C_{2,\mathrm{hole}}^{(\mathrm{cd})}
&:=
-\frac{1}{2}\tau_{\rho_2},
\\
\begin{split}
\label{def of C3cd hole}
C_{3,\mathrm{hole}}^{(\mathrm{cd})}
&:=
-\frac{1}{2}q(\rho_2)+\frac{1}{2}q(0)
 +\tau_{\rho_2}\Bigl(1-\log(\tau_{\rho_2}\sqrt{2\pi})  
+s\lambda(\rho_2)+2(\alpha+1)\log\rho_2\Bigr)
\\
&\quad
-2\int_{0}^{\rho_2}\log u\cdot u\Delta Q(u)\,du
-2\int_{0}^{\rho_2}\mathsf{k}(u)u\Delta Q(u)\,du
+\int_{0}^{\rho_2}u\Delta Q(u)\log \Delta Q(u)\,du,
\end{split}
\\
\label{def of C4cd hole}
C_{4,\mathrm{hole}}^{(\mathrm{cd})}
&:=
\sqrt{2}\rho_2\sqrt{\Delta Q(\rho_2)}\mathcal{I}
,
\\
\label{def of C5cd hole}
C_{5,\mathrm{hole}}^{(\mathrm{cd})}
&:=
s\frac{\rho_2}{4}\lambda'(\rho_2)+\frac{\alpha}{2}-\frac{\alpha^2}{2}+\frac{1}{3}, 
\\
\begin{split}
\label{def of C6cd hole}
C_{6,\mathrm{hole}}^{(\mathrm{cd})}
&:=
-\frac{\alpha}{2}\log(2\pi)-\log(2\sqrt{\pi})-\zeta'(-1)+\log G(\alpha+1)+\Bigl(-\alpha^2+\frac{1}{3}\Bigr)\log\rho_2
\\
&
\quad
+\Bigl(\frac{\alpha}{2}+\frac{1}{6}\Bigr)\log\Delta Q(\rho_2)
-\Bigl(\frac{\rho_2}{2}\mathsf{k}'(\rho_2)+\frac{1}{2}\Bigr)\log\Bigl(\frac{\rho_2\sqrt{\Delta Q(\rho_2)}}{\sqrt{2\pi}\tau_{\rho_2}}\Bigr)
+\frac{\rho_2^2\Delta Q(\rho_2)}{\tau_{\rho_2}}
\\
&\quad 
-\frac{\alpha(\alpha+1)}{2}\log\Delta Q(0)
-\alpha s(\lambda(\rho_2)-\lambda(0))
-\widetilde{F}_Q(\rho_2)-\Lambda_Q(0,\rho_2)-\mathcal{J}(\rho_2).
\end{split}
\end{align}

\subsection{counting statistics functionals}
\label{subsection:counting statistics functionals}
We next introduce the weights and local functions used in the counting
moment-generating functions.  For the annular and centered-disk geometries,
define
\begin{equation}
\label{def of omega ell}
 \omega_{\ell}   
 :=
 \begin{cases}
     e^{s_{\ell}+\cdots+s_{7m}}-e^{s_{\ell+1}+\cdots+s_{7m}}, & \text{if } \ell<7m, \\
     e^{s_{7m}}-1, & \text{if } \ell=7m, \\
     1, & \text{if } \ell=7m+1.
 \end{cases}.
\end{equation}
For the disk-complement geometry, replace $7m$ by $3m$ throughout this
definition.  In the centered-disk geometry, impose
$s_{\ell}=0$ for $1\leq \ell\leq 3m$.  More generally, for an endpoint $p$
and $\ell\leq p$, set
\begin{equation}
\omega_{\ell}^{(p)}:=e^{s_{\ell}+\cdots+s_{p}}-e^{s_{\ell+1}+\cdots+s_{p}},
\qquad \omega_{p+1}^{(p)}:=1,
\end{equation}
where an empty sum in the exponent is zero.  Thus
\begin{equation}
\Omega_{\ell}^{(p)}:=\sum_{j=\ell}^{p+1}\omega_j^{(p)}
=\begin{cases}
e^{s_{\ell}+\cdots+s_{p}}, & \text{if } \ell\leq p, \\
1, & \text{if } \ell=p+1.
\end{cases}.
\end{equation}
For example, $\Omega_{m+1}^{(5m)}=e^{s_{m+1}+\cdots+s_{5m}}$; moreover,
$\omega_\ell^{(7m)}=\omega_\ell$.  For $j\geq0$ and integers $a\leq b$,
define
\begin{equation}
\mathsf{T}_{j}^{(a,b)}(x)
:=
\sum_{\ell=a}^{b}\omega_{\ell}t_{\ell}^j e^{-2t_{\ell}(x-\tau_{\rho_1})},
\qquad
\widehat{\mathsf{T}}_j^{(a,b)}(x):=\sum_{\ell=a}^{b}\omega_{\ell}t_{\ell}^j
e^{-2t_{\ell}(\tau_{\rho_2}-x)}.
\end{equation}
\subsubsection{Terms of order \texorpdfstring{$n$}{n}}
Because $\mu_Q$ is supported on $S$, the shorthand introduced above is
equivalently
\begin{equation}
 \mu_{Q}[B]:=\int_{B}d\mu_{Q}(z).   
\end{equation}
For a bounded measurable set $B\subset\C$, define
\begin{equation}
C_{3,p,q}[B]:=\mu_Q[B]\sum_{\ell=pm+1}^{qm}s_{\ell},    
\end{equation}
We shall also use the following hard-wall integrals:
\begin{align}
C_{3,\mathrm{in}}^{(\mathrm{an})}&:=\int_{\tau_{\rho_1}}^{\sigma_{\star}}\log\Bigl[1+\frac{\mathsf{T}_{0}^{(2m+1,3m)}(x)}{1+\widehat{\mathsf{T}}_0^{(3m+1,7m)}(\tau_{\rho_2})}\Bigr]\,dx,\qquad    
C_{3,\mathrm{out}}^{(\mathrm{an})}:=\int_{\sigma_{\star}}^{\tau_{\rho_2}}
\log\Bigl[1-\frac{\widehat{\mathsf{T}}_0^{(3m+1,4m)}(x)}{1+\widehat{\mathsf{T}}_0^{(3m+1,7m)}(\tau_{\rho_2})}\Bigr]\,dx,
\\
C_{3,\mathrm{in}}^{(\mathrm{dc})}&:=
\int_{\tau_{\rho_1}}^1
\log\bigl(1+\mathsf{T}_{0}^{(2m+1,3m)}(x)\bigr)\,dx,
\qquad
C_{3,\mathrm{out}}^{(\mathrm{cd})}:=\int_{0}^{\tau_{\rho_2}}
\log\Bigl(1-\frac{\widehat{\mathsf{T}}_0^{(3m+1,4m)}(x)}{1+\widehat{\mathsf{T}}_0^{(3m+1,7m)}(\tau_{\rho_2})}\Bigr)\,dx.
\end{align}
The local order-$n$ coefficients are
\begin{align}
C_3^{\#(\mathrm{b,in})}
&:=\mu_Q[\mathbb{D}_{a_1}]\sum_{\ell=1}^m s_{\ell},
\\
C_{3}^{\#(\mathrm{se,in})}
&:=
\mu_Q[\mathbb{D}_{\rho_1}]\sum_{j=m+1}^{2m}s_j,
\\
C_{3,(\mathrm{an})}^{\#(\mathrm{h,in})}
&:=
\mu_Q[\mathbb{D}_{\rho_1}]\sum_{j=2m+1}^{3m}s_j
+\int_{\tau_{\rho_1}}^{\sigma_{\star}}
\log\Bigl(1+\frac{\mathsf{T}_{0}^{(2m+1,3m)}(x)}{1+\widehat{\mathsf{T}}_0^{(3m+1,7m)}(\tau_{\rho_2})}\Bigr)\,dx,
\\
C_{3,(\mathrm{dc})}^{\#(\mathrm{h,in})}
&:=
\mu_Q[\mathbb{D}_{\rho_1}]\sum_{j=2m+1}^{3m}s_j
+
\int_{\tau_{\rho_1}}^1
\log\bigl(1+\mathsf{T}_{0}^{(2m+1,3m)}(x)\bigr)\,dx, 
\\
C_{3,(\mathrm{an})}^{\#(\mathrm{h,out})}
&:=
\mu_Q[\mathbb{D}_{\rho_2}]\sum_{j=3m+1}^{4m}s_j
+
\int_{\sigma_{\star}}^{\tau_{\rho_2}}
\log\Bigl(1-\frac{\widehat{\mathsf{T}}_0^{(3m+1,4m)}(x)}{1+\widehat{\mathsf{T}}_0^{(3m+1,7m)}(\tau_{\rho_2})}\Bigr)\,dx,
\\
C_{3,(\mathrm{cd})}^{\#(\mathrm{h,out})}
&:=
\mu_Q[\mathbb{D}_{\rho_2}]\sum_{j=3m+1}^{4m}s_j
+
\int_{0}^{\tau_{\rho_2}}
\log\Bigl(1-\frac{\widehat{\mathsf{T}}_0^{(3m+1,4m)}(x)}{1+\widehat{\mathsf{T}}_0^{(3m+1,7m)}(\tau_{\rho_2})}\Bigr)\,dx,
\\
C_{3}^{\#(\mathrm{se,out})}
&:=
\mu_Q[\mathbb{D}_{\rho_2}]\sum_{j=4m+1}^{5m}s_j,
\\
C_3^{\#(\mathrm{b,out})}
&:=\mu_Q[\mathbb{D}_{a_2}]\sum_{j=5m+1}^{6m}s_j,
\\
C_3^{\#(\mathrm{s})}
&:=\mu_Q[\mathbb{D}_{R}]\sum_{j=6m+1}^{7m}s_j.
\end{align}
\paragraph{Bulk regime, inner side.}
Define
\begin{align}
\mathcal{H}_1^{(\mathrm{b,in})}(x)
&:=
1+\sum_{\ell=1}^{m}\frac{e^{s_{\ell}}-1}{2}
\exp\Bigl( 
\sum_{p=\ell+1}^{m}s_{p}
\Bigr)
\erfc(x-t_{\ell}),
\\
\mathcal{H}_2^{(\mathrm{b,in})}(x)
&:=
1+\sum_{\ell=1}^{m}\frac{e^{-s_{\ell}}-1}{2}
\exp\Bigl( 
-\sum_{p=1}^{\ell-1}s_{p}
\Bigr)
\erfc(x+t_{\ell}), 
\\
\mathcal{G}_1^{(\mathrm{b,in})}(x)
&:=\frac{1}{\mathcal{H}_1^{(\mathrm{b,in})}(x)}
\sum_{\ell=1}^{m}
(e^{s_{\ell}}-1)
\exp\Bigl(\sum_{p=\ell+1}^{m}s_{p}\Bigr)
\frac{e^{-(t_{\ell}-x)^2}}{\sqrt{2\pi}}
\frac{1-2t_{\ell}^2+t_{\ell}x-5x^2}{3}, 
\\
\widetilde{\mathcal{G}}_1^{(\mathrm{b,in})}(x)
&:=
\frac{1}{\mathcal{H}_1^{(\mathrm{b,in})}(x)}
\sum_{\ell=1}^{m}
(e^{s_{\ell}}-1)
e^{\sum_{p=\ell+1}^{m}s_{p}}
\frac{e^{-(t_{\ell}-x)^2}}{\sqrt{2\pi}}
t_{\ell}^2. 
\end{align}
\paragraph{Semi-hard regime, inner side.}
Define
\begin{align}
\mathcal{H}_1^{(\mathrm{se,in})}(x)
&:=
1+\sum_{\ell=m+1}^{2m}\omega_{\ell}^{(2m)}\frac{\erfc(t_{\ell}+x)}{\erfc(x)}, 
\\
\begin{split}
\mathcal{H}_2^{(\mathrm{se,in})}(x)
&:=
\frac{1}{\mathcal{H}_1^{(\mathrm{se,in})}(x)}
\sum_{\ell=m+1}^{2m} \frac{1}{3\sqrt{\pi}}\omega_{\ell}^{(2m)}
\\
&\quad
\times
\bigg[
(5x^2-1)
\frac{e^{-x^2}}{\erfc(x)}
\frac{\erfc(t_{\ell}+x)}{\erfc(x)}
-
(5x^2+t_{\ell}x+2t_{\ell}^2-1)\frac{e^{-(t_{\ell}+x)^2}}{\erfc(x)}
\bigg], 
\end{split}
\\
\widetilde{\mathcal{H}}_2^{(\mathrm{se,in})}(x)
&:=
\frac{1}{\mathcal{H}_1^{(\mathrm{se,in})}(x)}
\sum_{\ell=m+1}^{2m}\omega_{\ell}^{(2m)}t_{\ell}^2\frac{e^{-(t_{\ell}+x)^2}}{3\sqrt{\pi}\erfc(x)}.
\end{align}
\paragraph{Hard regime, inner side.}
Define
\begin{align}
\mathcal{H}_{1}^{(\mathrm{h,in})}(x)
&:=
\begin{cases}
-\frac{2\rho_1^2\Delta Q(\rho_1)}{x-\tau_{\rho_1}}
\frac{\mathsf{T}_{1}^{(2m+1,3m)}(x)}{1+\mathsf{T}_{0}^{(2m+1,3m)}(x)}
-
x\frac{\mathsf{T}_{2}^{(2m+1,3m)}(x)}{1+\mathsf{T}_{0}^{(2m+1,3m)}(x)}, & \text{in the disk-complement geometry}, \\
-\frac{\frac{2\rho_1^2\Delta Q(\rho_1)}{x-\tau_{\rho_1}}
\mathsf{T}_{1}^{(2m+1,3m)}(x)
+
x
\mathsf{T}_{2}^{(2m+1,3m)}(x)
}{
1+\mathsf{T}_{0}^{(2m+1,3m)}(x)+\widehat{\mathsf{T}}_0^{(3m+1,7m)}(\tau_{\rho_2})}, & \text{in the annular geometry},
\end{cases},
\\
\widetilde{\mathcal{H}}_{k}^{(\mathrm{h,in})}(x)&:=\frac{\mathsf{T}_{k}^{(2m+1,3m)}(x)}{1+\mathsf{T}_{0}^{(2m+1,3m)}(x)+\widehat{\mathsf{T}}_{0}^{(3m+1,7m)}(\tau_{\rho_2})}.
\end{align}
\paragraph{Hard regime, outer side.}
Define
\begin{align}
\mathcal{H}_1^{(\mathrm{h,out})}(x)&:=
\frac{\frac{2\rho_2^2\Delta Q(\rho_2)}{\tau_{\rho_2}-x}
\widehat{\mathsf{T}}_{1}^{(3m+1,4m)}(x)
+x\widehat{\mathsf{T}}_{2}^{(3m+1,4m)}(x)}{1-\widehat{\mathsf{T}}_0^{(3m+1,4m)}(x)+\widehat{\mathsf{T}}_0^{(3m+1,7m)}(\tau_{\rho_2})}, 
\\
\widetilde{\mathcal{H}}_{k}^{(\mathrm{h,out})}(x)&:=\frac{\widehat{\mathsf{T}}_{k}^{(3m+1,4m)}(x)}{1-\widehat{\mathsf{T}}_0^{(3m+1,4m)}(x)+\widehat{\mathsf{T}}_0^{(3m+1,7m)}(\tau_{\rho_2})}.
\end{align}
\paragraph{Semi-hard regime, outer side.}
Define
\begin{align}
\mathcal{H}_{1}^{(\mathrm{se,out})}
(x)
&:=
1-\frac{1}{\Omega_{4m+1}^{(7m)}}\sum_{\ell=4m+1}^{5m}\omega_{\ell}
\frac{1-\frac{1}{2}\erfc(-t_{\ell}+x)}{1-\frac{1}{2}\erfc(x)}, 
\\
\begin{split}
\mathcal{H}_{2}^{(\mathrm{se,out})}
(x)
&:=
\frac{1}{\mathcal{H}_{1}^{(\mathrm{se,out})}
(x)}
\sum_{\ell=4m+1}^{5m}\frac{\omega_{\ell}}{\Omega_{4m+1}^{(7m)}}
\frac{1-\frac{1}{2}\erfc(-t_{\ell}+x)}{1-\frac{1}{2}\erfc(x)}
\frac{e^{-x^2}}{1-\frac{1}{2}\erfc(x)}
\frac{(5x^2-1)}{6\sqrt{\pi}}
\\
&\quad
-
\frac{1}{\mathcal{H}_{1}^{(\mathrm{se,out})}
(x)}
\sum_{\ell=4m+1}^{5m}\frac{\omega_{\ell}}{\Omega_{4m+1}^{(7m)}}
\frac{e^{-(t_{\ell}-x)^2}}{1-\frac{1}{2}\erfc(x)}
\frac{5x^2-1+2t_{\ell}^2-xt_{\ell}}{6\sqrt{\pi}},
\end{split}
\\
\widetilde{\mathcal{H}}_{2}^{(\mathrm{se,out})}
(x)
&:=
\frac{1}{\mathcal{H}_{1}^{(\mathrm{se,out})}
(x)}
\sum_{\ell=4m+1}^{5m}\frac{\omega_{\ell}}{\Omega_{4m+1}^{(7m)}}
\frac{e^{-(t_{\ell}-x)^2}}{(1-\frac{1}{2}\erfc(x))}
\frac{t_{\ell}^2}{6\sqrt{\pi}}, 
\end{align}
\paragraph{Bulk regime, outer side.}
Define
\begin{align}
\mathcal{H}_1^{(\mathrm{b,out})}(x)
&:=
1+\sum_{\ell=5m+1}^{6m}\frac{e^{s_{\ell}}-1}{2}
\exp\Bigl( 
\sum_{p=\ell+1}^{6m}s_{p}
\Bigr)
\erfc(x-t_{\ell}), 
\\
\mathcal{H}_2^{(\mathrm{b,out})}(x)
&:=
1+\sum_{\ell=5m+1}^{6m}\frac{e^{-s_{\ell}}-1}{2}
\exp\Bigl( 
-\sum_{p=5m+1}^{\ell-1}s_{p}
\Bigr)
\erfc(x+t_{\ell}), 
\\
\begin{split}
\mathcal{G}_1^{(\mathrm{b,out})}(x)
&:=
\frac{1}{\mathcal{H}_1^{(\mathrm{b,out})}(x)}
\sum_{\ell=5m+1}^{6m}
(e^{s_{\ell}}-1)\exp\Bigl( 
\sum_{p=\ell+1}^{6m}s_{p}
\Bigr)
\frac{e^{-(t_{\ell}-x)^2}}{\sqrt{2\pi}}
\frac{1-2t_{\ell}^2+t_{\ell}x-5x^2}{3}, 
\end{split}
\\
\widetilde{\mathcal{G}}_1^{(\mathrm{b,out})}(x)
&:=
\frac{1}{\mathcal{H}_1^{(\mathrm{b,out})}(x)}
\sum_{\ell=5m+1}^{6m}
(e^{s_{\ell}}-1)\exp\Bigl( 
\sum_{p=\ell+1}^{6m}s_{p}
\Bigr)
\frac{e^{-(t_{\ell}-x)^2}}{\sqrt{2\pi}}
t_{\ell}^2. 
\end{align}
\paragraph{Soft regime.}
Define
\begin{align}
\mathcal{H}_1^{(\mathrm{s})}(x)
&:= 
1+\sum_{\ell=6m+1}^{7m}
\frac{\omega_{\ell}}{2}\erfc(-t_{\ell}+x),
\\
\mathcal{G}_1^{(\mathrm{s})}(x)
&:=
\frac{1}{\mathcal{H}_1^{(\mathrm{s})}(x)}
\sum_{\ell=6m+1}^{7m}
\frac{\omega_{\ell}e^{-(t_{\ell}-x)^2}}{3\sqrt{2\pi}}
(1-2t_{\ell}^2+t_{\ell}x-5x^2), 
\\
\mathcal{G}_2^{(\mathrm{s})}(x)
&:=
\frac{1}{\mathcal{H}_1^{(\mathrm{s})}(x)}
\sum_{\ell=6m+1}^{7m}
\frac{\omega_{\ell}t_{\ell}^2e^{-(t_{\ell}-x)^2}}{6\sqrt{\pi}}.
\end{align}

\subsubsection{Terms of order \texorpdfstring{$\sqrt n$}{sqrt(n)}}
The $\sqrt n$ counting coefficients are
\begin{align}
C_4^{\#(\mathrm{b,in})}&:=
a_1\sqrt{2\Delta Q(a_1)}
\int_0^{+\infty}
\Bigl(
\log\mathcal{H}_1^{(\mathrm{b,in})}(x)
+
\log\mathcal{H}_2^{(\mathrm{b,in})}(x)
\Bigr)
\,dx,
\\
C_4^{\#(\mathrm{se,in})}
&:= 
\rho_1\sqrt{2\Delta Q(\rho_1)}\int_{-\infty}^{+\infty}
\bigg(
\log\mathcal{H}_1^{(\mathrm{se,in})}(x)
-
\mathbf{1}_{(-\infty,0)}
\sum_{j=m+1}^{2m}s_{j}
\bigg)
\,dx,
\\
C_4^{\#(\mathrm{b,out})}
&:= 
a_2\sqrt{2\Delta Q(a_2)}
\int_0^{+\infty}
\Bigl(
\log\mathcal{H}_{1}^{(\mathrm{b,out})}(x)
+
\log\mathcal{H}_{2}^{(\mathrm{b,out})}(x)
\Bigr)
\,dx, 
\\
C_4^{\#(\mathrm{se,out})}
&:= 
\rho_2\sqrt{2\Delta Q(\rho_2)}
\int_{-\infty}^{+\infty}
\Bigl( \log \mathcal{H}_{1}^{(\mathrm{se,out})}
(x)
+\mathbf{1}_{[0,+\infty)}(x)\sum_{\ell=4m+1}^{5m}s_{\ell}\Bigr)\,dx,
\\
C_4^{\#(\mathrm{s})}
&:= 
R\sqrt{2\Delta Q(R)}\int_{-\infty}^{0}
\Bigl( \log\mathcal{H}_{1}^{(\mathrm{s})}(x) -\sum_{\ell=6m+1}^{7m}s_{\ell} \Bigr)\,dx.
\end{align}

\subsubsection{Terms of order \texorpdfstring{$\log n$}{log(n)}}
The $\log n$ counting coefficients are
\begin{align}
C_{5,(\mathrm{dc})}^{\#(\mathrm{h,in})}
&:=-
\frac{\rho_1^2\Delta Q(\rho_1)\mathsf{T}_{1}^{(2m+1,3m)}(\tau_{\rho_1})}{1+\mathsf{T}_{0}^{(2m+1,3m)}(\tau_{\rho_1})},
\\
  C_{5,(\mathrm{an})}^{\#(\mathrm{h,in/out})}
    &:=
    -
\rho_1^2\Delta Q(\rho_1)\widetilde{\mathcal{H}}_{1}^{(\mathrm{h,in})}(\tau_{\rho_1})
+
\rho_2^2\Delta Q(\rho_2)\widetilde{\mathcal{H}}_{1}^{(\mathrm{h,out})}(\tau_{\rho_2}),
\\
C_{5,(\mathrm{cd})}^{\#(\mathrm{h,out})}
&:=\rho_2^2\Delta Q(\rho_2)\widetilde{\mathcal{H}}_1^{(\mathrm{h,out})}(\tau_{\rho_2}).
\end{align}

\subsubsection{Terms of order one}
The order-one counting coefficients are
\begin{align}
C_6^{\#(\mathrm{b,in})}&:=
\mathcal{D}_6^{(\mathrm{b,in})}-\Bigl( 
\frac{1}{2}a_1\mathsf{k}'(a_1)+\frac{1}{2}
\Bigr)\sum_{\ell=1}^{m}s_{\ell},
\\
C_6^{\#(\mathrm{se,in})}&:=
\mathcal{D}_6^{(\mathrm{se,in})}
-
\Bigl(\frac{1}{2}\rho_1\mathsf{k}'(\rho_1)+\frac{1}{2}\Bigr)
\sum_{\ell=m+1}^{2m}s_{\ell},
\\
C_{6,(\mathrm{dc})}^{\#(\mathrm{h,in})}&:=    
\mathcal{D}_{6,(\mathrm{dc})}^{(\mathrm{h,in})}-\Bigl(\frac{1}{2}\rho_1\mathsf{k}'(\rho_1)+\frac{1}{2}\Bigr)
\sum_{\ell=2m+1}^{3m}s_{\ell}+\Bigl(\frac{\rho_1\mathsf{k}'(\rho_1)}{2}+\frac{1}{2}\Bigr)\log(1+\mathsf{T}_0^{(2m+1,3m)}(1)),
\\
\begin{split}
C_{6,(\mathrm{an})}^{\#(\mathrm{h,in})}&:=
\mathcal{D}_{6,(\mathrm{an})}^{(\mathrm{h,in})}
\\
&\quad
- 
\Bigl(\frac{1}{2}+\frac{\rho_1\mathsf{k}'(\rho_1)}{2}\Bigr)\sum_{\ell=2m+1}^{3m}s_{\ell}
+
\frac{s\rho_1\lambda'(\rho_1)}{2}
\log\bigg[1+\frac{\mathsf{T}_{0}^{(2m+1,3m)}(\sigma_{\star})}{1+\widehat{\mathsf{T}}_0^{(3m+1,7m)}(\tau_{\rho_2})}\bigg], 
\end{split}
\\
\begin{split}
C_{6,(\mathrm{an})}^{\#(\mathrm{h,out})}&:=
\mathcal{D}_{6,(\mathrm{an})}^{(\mathrm{h,out})}
\\
&\quad
-\Bigl(\frac{1}{2}+\frac{\rho_2\mathsf{k}'(\rho_2)}{2}\Bigr)\sum_{\ell=3m+1}^{4m}s_{\ell} 
-\frac{s\rho_2\lambda'(\rho_2)}{2}\log\bigg[1-\frac{\widehat{\mathsf{T}}_0^{(3m+1,4m)}(\sigma_{\star})}{1+\widehat{\mathsf{T}}_0^{(3m+1,7m)}(\tau_{\rho_2})}\bigg],
\end{split}
\\
\begin{split}
C_{6,(\mathrm{cd})}^{\#(\mathrm{h,out})}
&:=\mathcal{D}_{6,(\mathrm{cd})}^{(\mathrm{h,out})}
\\
&\quad
-\Bigl( \frac{1}{2}+\frac{\rho_2\mathsf{k}'(\rho_2)}{2}\Bigr)\sum_{j=3m+1}^{4m}s_j
-\Bigl( \frac{1}{2}+\frac{\rho_2\mathsf{k}'(\rho_2)}{2}\Bigr)\log\bigg[1-\frac{\widehat{\mathsf{T}}_0^{(3m+1,4m)}(0)}{1+\widehat{\mathsf{T}}_0^{(3m+1,7m)}(\tau_{\rho_2})}\bigg],
\end{split}
\\
C_{6}^{\#(\mathrm{se,out})}&:=
\mathcal{D}_6^{(\mathrm{se,out})}
-\Bigl(\alpha+\frac{1}{2}+\frac{s}{2}\rho_2\lambda'(\rho_2)\Bigr)\sum_{\ell=4m+1}^{5m}s_{\ell},
\\
C_{6}^{\#(\mathrm{b,out})}
&:=
\mathcal{D}_6^{(\mathrm{b,out})}
- 
\Bigl(\alpha+\frac{1}{2}+\frac{s}{2}a_2\lambda'(a_2)\Bigr)\sum_{\ell=5m+1}^{6m}s_{\ell},
\\
C_{6}^{\#(\mathrm{s})}
&:=
\mathcal{D}_6^{(\mathrm{s})}
-\Bigl(\alpha+\frac{1}{2}+\frac{s}{2}R\lambda'(R)\Bigr)\sum_{\ell=6m+1}^{7m}s_{\ell}.
\end{align}
where
\begin{align}
\begin{split}
\mathcal{D}_6^{(\mathrm{b,in})}
&:=
\frac{1}{\sqrt{2}}\Bigl(2+\frac{a_1\partial_r\Delta Q(a_1)}{\Delta Q(a_1)}\Bigr)
\int_{-\infty}^{+\infty}
\mathcal{G}_1^{(\mathrm{b,in})}(x)\,dx
\\
&\quad
+
\frac{1}{\sqrt{2}}
\Bigl(1+\frac{a_1\partial_r\Delta Q(a_1)}{\Delta Q(a_1)}\Bigr)
\int_{-\infty}^{+\infty}
\widetilde{\mathcal{G}}_1^{(\mathrm{b,in})}(x)
\,dx  
\\
&\quad
+2\Bigl(
2+\frac{a_1\partial_r\Delta Q(a_1)}{\Delta Q(a_1)}
\Bigr)
\int_{0}^{+\infty}
x\Bigl( 
\log\mathcal{H}_1^{(\mathrm{b,in})}(x)-\log\mathcal{H}_2^{(\mathrm{b,in})}(x)
\Bigr)\,dx, 
\end{split}
\\
\begin{split}
\mathcal{D}_6^{(\mathrm{b,out})}
&:=
\frac{1}{\sqrt{2}}\Bigl(2+\frac{a_2\partial_r\Delta Q(a_2)}{\Delta Q(a_2)}\Bigr)
\int_{-\infty}^{+\infty}
\mathcal{G}_1^{(\mathrm{b,out})}(x)
\,dx
\\
&\quad
+
\frac{1}{\sqrt{2}}
\Bigl(1+\frac{a_2\partial_r\Delta Q(a_2)}{\Delta Q(a_2)}\Bigr)
\int_{-\infty}^{+\infty}
\widetilde{\mathcal{G}}_1^{(\mathrm{b,out})}(x)
\,dx  
\\
&\quad
+2\Bigl(
2+\frac{a_2\partial_r\Delta Q(a_2)}{\Delta Q(a_2)}
\Bigr)
\int_{0}^{+\infty}
x\Bigl( 
\log\mathcal{H}_1^{(\mathrm{b,out})}(x)-\log\mathcal{H}_2^{(\mathrm{b,out})}(x)
\Bigr)\,dx, 
\end{split}
\\
\begin{split}
\mathcal{D}_6^{(\mathrm{s})}
&:=
\frac{1}{\sqrt{2}}\Bigl(2+\frac{R\partial_r\Delta Q(R)}{\Delta Q(R)}\Bigr)
\int_{-\infty}^{0}
\mathcal{G}_1^{(\mathrm{s})}(x)\,dx
+
3\Bigl(1+\frac{R\partial_r\Delta Q(R)}{\Delta Q(R)}\Bigr)
\int_{-\infty}^{0}
\mathcal{G}_2^{(\mathrm{s})}(x)
\,dx
\\
&\quad
+2\Bigl(
2+\frac{R\partial_r\Delta Q(R)}{\Delta Q(R)}
\Bigr)
\int_{-\infty}^{0}
x\Bigl(\log\mathcal{H}_1^{(\mathrm{s})}(x)-\sum_{\ell=6m+1}^{7m}s_{\ell}\Bigr)
\,dx
\\
&\quad
+\Bigl(\frac{1}{2}R\,\mathsf{k}'(R)+\frac{1}{2}\Bigr)
\log\Bigl(1+\sum_{\ell=6m+1}^{7m}
\frac{\omega_{\ell}}{2}\erfc(-t_{\ell})\Bigr),
\end{split}
\\
\begin{split}
\mathcal{D}_{6,(\mathrm{dc})}^{(\mathrm{h,in})}&:=    
\int_{\tau_{\rho_1}}^{1}
\Bigl(
\mathcal{H}_1^{(\mathrm{h,in})}(x)
+\frac{2\rho_1^2\Delta Q(\rho_1)}{x-\tau_{\rho_1}}\frac{\mathsf{T}_{1}^{(2m+1,3m)}(\tau_{\rho_1})}{1+\mathsf{T}_{0}^{(2m+1,3m)}(\tau_{\rho_1})}\Bigr)\,dx
\\
&\quad 
-
\frac{2\rho_1^2\Delta Q(\rho_1)\mathsf{T}_{1}^{(2m+1,3m)}(\tau_{\rho_1})}{1+\mathsf{T}_{0}^{(2m+1,3m)}(\tau_{\rho_1})}\log(1-\tau_{\rho_1})
\\
&\quad
-
\bigl(2\rho_1^2\Delta Q(\rho_1)-\tau_{\rho_1}\bigr)
\int_{\tau_{\rho_1}}^{1}
\frac{\mathsf{T}_{2}^{(2m+1,3m)}(x)}{1+\mathsf{T}_{0}^{(2m+1,3m)}(x)}\,dx,
\end{split}
\\
\begin{split}
\mathcal{D}_{6,(\mathrm{an})}^{(\mathrm{h,in})}
&:=
\int_{\tau_{\rho_1}}^{\sigma_{\star}}
\bigl(
\mathcal{H}_1^{(\mathrm{h,in})}(x)
+
\frac{2\rho_1^2\Delta Q(\rho_1)}{x-\tau_{\rho_1}}
\widetilde{\mathcal{H}}_{1}^{(\mathrm{h,in})}(\tau_{\rho_1})
\bigr)\,dx
\\
&\quad
-\bigl(2\rho_1^2\Delta Q(\rho_1)-\tau_{\rho_1}\bigr)
\int_{\tau_{\rho_1}}^{\sigma_{\star}}
\widetilde{\mathcal{H}}_{2}^{(\mathrm{h,in})}(x)\,dx
-
2\rho_1^2\Delta Q(\rho_1)
\widetilde{\mathcal{H}}_{1}^{(\mathrm{h,in})}(\tau_{\rho_1})
\log\bigl(\sigma_{\star}-\tau_{\rho_1}\bigr), 
\end{split}
\\
\begin{split}
\mathcal{D}_{6,(\mathrm{an})}^{(\mathrm{h,out})}
&:=
\int_{\sigma_{\star}}^{\tau_{\rho_2}}
\bigl(
\mathcal{H}_1^{(\mathrm{h,out})}(x)
-
\frac{2\rho_2^2\Delta Q(\rho_2)}{\tau_{\rho_2}-x}
\widetilde{\mathcal{H}}_{1}^{(\mathrm{h,out})}(\tau_{\rho_2})
\bigr)\,dx
\\
&
+
2\rho_2^2\Delta Q(\rho_2)
\widetilde{\mathcal{H}}_{1}^{(\mathrm{h,out})}(\tau_{\rho_2})
\log\bigl(\tau_{\rho_2}-\sigma_{\star}\bigr)
+\bigl( 
2\rho_2^2\Delta Q(\rho_2)-\tau_{\rho_2}
\bigr)\int_{\sigma_{\star}}^{\tau_{\rho_2}}
\widetilde{\mathcal{H}}_{2}^{(\mathrm{h,out})}(x)\,dx, 
\end{split}
\\
\begin{split}
\mathcal{D}_{6,(\mathrm{cd})}^{(\mathrm{h,out})}
&:=
\int_{0}^{\tau_{\rho_2}}
\Bigl(\mathcal{H}_1^{(\mathrm{h,out})}(x)-\frac{2\rho_2^2\Delta Q(\rho_2)}{\tau_{\rho_2}-x}\widetilde{\mathcal{H}}_1^{(\mathrm{h,out})}(\tau_{\rho_2})\Bigr)\,dx
\\
&
+2\rho_2^2\Delta Q(\rho_2)\widetilde{\mathcal{H}}_1^{(\mathrm{h,out})}(\tau_{\rho_2})\log(\tau_{\rho_2})+\bigl( 2\rho_2^2\Delta Q(\rho_2)-\tau_{\rho_2}\bigr)\int_{0}^{\tau_{\rho_2}}\widetilde{\mathcal{H}}_2^{(\mathrm{h,out})}(x)\,dx, 
\end{split}
\\
\begin{split}
\mathcal{D}_6^{(\mathrm{se,in})}
&:=
\Bigl(\frac{\rho_1\partial_r\Delta Q(\rho_1)}{\Delta Q(\rho_1)}+2\Bigr)
\int_{-\infty}^{+\infty}
2x\Bigl(\log \mathcal{H}_1^{(\mathrm{se,in})}(x)-\mathbf{1}_{(-\infty,0)}(x)\sum_{j=m+1}^{2m}s_j\Bigr)
\,dx
\\
&
+\Bigl(\frac{\rho_1\partial_r\Delta Q(\rho_1)}{\Delta Q(\rho_1)}+2\Bigr)
\int_{-\infty}^{+\infty}
\mathcal{H}_2^{(\mathrm{se,in})}(x)
\,dx
\\
&
+
3\Bigl(1+\frac{\rho_1\partial_r\Delta Q(\rho_1)}{\Delta Q(\rho_1)}\Bigr)
\int_{-\infty}^{+\infty}
\widetilde{\mathcal{H}}_2^{(\mathrm{se,in})}(x)\,dx
\\
&
+\frac{\rho_1^2\Delta Q(\rho_1)\mathsf{T}_{1}^{(2m+1,3m)}(\tau_{\rho_1})\log (2\rho_1^2\Delta Q(\rho_1))}{1+\mathsf{T}_{0}^{(2m+1,3m)}(\tau_{\rho_1})+\widehat{\mathsf{T}}_{0}^{(3m+1,7m)}(\tau_{\rho_2})}
\\
&
-4\rho_1^2\Delta Q(\rho_1)
\sum_{\ell=2m+1}^{3m}
\frac{\omega_{\ell}t_{\ell}}{\Omega_{2m+1}^{(7m)}}
\int_{-\infty}^{+\infty}
\bigg[
\frac{\frac{e^{-x^2}}{\sqrt{\pi}\erfc(x)}}{
\mathcal{H}_1^{(\mathrm{se,in})}(x)
}
-
\Bigl( 
x+\frac{x}{2(1+x^2)}
\Bigr)
\mathbf{1}_{[0,+\infty)}(x)
\bigg]
\,dx,
\end{split}
\\
\begin{split}
\mathcal{D}_6^{(\mathrm{se,out})}
&:=
\Bigl(2+\frac{\rho_2\partial_r\Delta Q(\rho_2)}{\Delta Q(\rho_2)}\Bigr)
\int_{-\infty}^{+\infty}
\mathcal{H}_{2}^{(\mathrm{se,out})}(x)
\,dx
\\
&
+3\Bigl(1+\frac{\rho_2\partial_r\Delta Q(\rho_2)}{\Delta Q(\rho_2)}\Bigr)
\int_{-\infty}^{+\infty}
\widetilde{\mathcal{H}}_{2}^{(\mathrm{se,out})}
(x)
\,dx
\\
&
+\Bigl(2+\frac{\rho_2\partial_r\Delta Q(\rho_2)}{\Delta Q(\rho_2)}\Bigr)
\int_{-\infty}^{+\infty}
2x\Bigl( \log \mathcal{H}_{1}^{(\mathrm{se,out})}
(x)
+\mathbf{1}_{[0,+\infty)}(x)\sum_{\ell=4m+1}^{5m}s_{\ell}\Bigr)
\,dx
\\
&
+
\frac{4\rho_2^2\Delta Q(\rho_2)\widehat{\mathsf{T}}_1^{(3m+1,4m)}(\tau_{\rho_2})}{1+\widehat{\mathsf{T}}_0^{(4m+1,7m)}(\tau_{\rho_2})}
\\
&\times
\int_{-\infty}^{+\infty}
\Bigl( 
\frac{\frac{e^{-x^2}}{2\sqrt{\pi}(1-\frac{1}{2}\erfc(x))}}{\mathcal{H}_{1}^{(\mathrm{se,out})}
(x)}
+
\Bigl(x+\frac{x}{2(x^2+1)}\Bigr)\mathbf{1}_{(-\infty,0]}(x)
\Bigr)\,dx
\\
&
-\frac{\rho_2^2\Delta Q(\rho_2)\widehat{\mathsf{T}}_1^{(3m+1,4m)}(\tau_{\rho_2})\log (2\rho_2^2\Delta Q(\rho_2))}{1+\widehat{\mathsf{T}}_0^{(4m+1,7m)}(\tau_{\rho_2})}.
\end{split}
\end{align}

\subsubsection{Jacobi-theta corrections}
The annular geometry has an additional bounded correction.  Set
\begin{equation}
\label{def of mathsf Q}
\mathsf{Q}
:=
\frac{1+\mathsf{T}_{0}^{(2m+1,3m)}(\sigma_{\star})
+\widehat{\mathsf{T}}_0^{(3m+1,7m)}(\tau_{\rho_2})}{1-\widehat{\mathsf{T}}_0^{(3m+1,4m)}(\sigma_{\star})+\widehat{\mathsf{T}}_0^{(3m+1,7m)}(\tau_{\rho_2})},
\end{equation}
and define
\begin{align}
    \begin{split}
    \label{def of calFNsharp}
\mathcal{F}_{n,\#}
&:=
-\frac{\log \mathsf{Q}}{2}\Bigl(\frac{s(\lambda(\rho_1)-\lambda(\rho_2))}{\log \frac{\rho_1}{\rho_2}}+\frac{2\log (\tfrac{\sigma_2}{\sigma_1})+\log \mathsf{Q}}{2\log \frac{\rho_1}{\rho_2}}\Bigr)
    \\
    &\quad 
    +\log \frac{\theta( \tfrac{\log (\tfrac{\sigma_2}{\sigma_1})+\log \mathsf{Q}}{2\log \tfrac{\rho_2}{\rho_1}}+\frac{s(\lambda(\rho_1)-\lambda(\rho_2))}{2\log\frac{\rho_2}{\rho_1}}+n\sigma_{\star}-\alpha +\frac{1}{2}; \frac{\pi i }{\log \frac{\rho_2}{\rho_1}})}{\theta( \tfrac{\log (\tfrac{\sigma_2}{\sigma_1})}{2\log \tfrac{\rho_2}{\rho_1}}+\frac{s(\lambda(\rho_1)-\lambda(\rho_2))}{2\log\frac{\rho_2}{\rho_1}}+n\sigma_{\star}-\alpha+\frac{1}{2}; \frac{\pi i }{\log \frac{\rho_2}{\rho_1}})}.
\end{split}   
\end{align}
This term records the change in the discrete filling-fraction modulation caused
by the counting insertions.
\subsection{The counting statistics coefficients}
The geometry-dependent coefficients are obtained by summing the admissible
local regimes.
\subsubsection{Annular hole}
\begin{align}
\label{def of C3sharp an}
C_{3,\#}^{(\mathrm{an})}&:=
C_3^{\#(\mathrm{b,in})}+C_3^{\#(\mathrm{se,in})}+C_{3,(\mathrm{an})}^{\#(\mathrm{h,in})}
+C_{3,(\mathrm{an})}^{\#(\mathrm{h,out})}+C_{3}^{\#(\mathrm{se,out})}+C_{3}^{\#(\mathrm{b,out})}+C_{3}^{\#(\mathrm{s})}, 
\\
\label{def of C4sharp an}
C_{4,\#}^{(\mathrm{an})}&:=
C_4^{\#(\mathrm{b,in})}+C_4^{\#(\mathrm{se,in})}+C_{4}^{\#(\mathrm{se,out})}+C_{4}^{\#(\mathrm{b,out})}+C_{4}^{\#(\mathrm{s})},
\\
\label{def of C5sharp an}
C_{5,\#}^{(\mathrm{an})}&:=C_{5,(\mathrm{an})}^{\#(\mathrm{h,in/out})}, 
\\
\label{def of C6sharp an}
C_{6,\#}^{(\mathrm{an})}&:=C_6^{\#(\mathrm{b,in})}+C_6^{\#(\mathrm{se,in})}+C_{6,(\mathrm{an})}^{\#(\mathrm{h,in})}+C_{6,(\mathrm{an})}^{\#(\mathrm{h,out})}+C_{6}^{\#(\mathrm{se,out})}+C_6^{\#(\mathrm{b,out})}+C_{6}^{\#(\mathrm{s})}.
\end{align}
\subsubsection{Disk-complement hole}
\begin{align}
\label{def of C3sharp dc}
C_{3,\#}^{(\mathrm{dc})}&:=
C_3^{\#(\mathrm{b,in})}+C_3^{\#(\mathrm{se,in})}+C_{3,(\mathrm{dc})}^{\#(\mathrm{h,in})}, 
\\
\label{def of C4sharp dc}
C_{4,\#}^{(\mathrm{dc})}
&:=
C_4^{\#(\mathrm{b,in})}
+
C_4^{\#(\mathrm{se,in})},
\\
\label{def of C5sharp dc}
C_{5,\#}^{(\mathrm{dc})}
&:=
C_{5,(\mathrm{dc})}^{\#(\mathrm{h,in})},
\\
\label{def of C6sharp dc}
C_{6,\#}^{(\mathrm{dc})}&:=
C_6^{\#(\mathrm{b,in})}+C_6^{\#(\mathrm{se,in})}+C_{6,(\mathrm{dc})}^{\#(\mathrm{h,in})}. 
\end{align}
\subsubsection{Centered-disk hole}
\begin{align}
\label{def of Csharp3 cd}
C_{3,\#}^{(\mathrm{cd})}&:=
C_{3,(\mathrm{cd})}^{\#(\mathrm{h,out})}+C_{3}^{\#(\mathrm{se,out})}+C_{3}^{\#(\mathrm{b,out})}+C_{3}^{\#(\mathrm{s})}, 
\\
\label{def of Csharp4 cd}
C_{4,\#}^{(\mathrm{cd})}&:=C_{4}^{\#(\mathrm{se,out})}+C_{4}^{\#(\mathrm{b,out})}+C_{4}^{\#(\mathrm{s})},
\\
\label{def of Csharp5 cd}
C_{5,\#}^{(\mathrm{cd})}&:=C_{5,(\mathrm{cd})}^{\#(\mathrm{h,out})}, 
\\
\label{def of Csharp6 cd}
C_{6,\#}^{(\mathrm{cd})}&:=C_{6,(\mathrm{cd})}^{\#(\mathrm{h,out})}+C_{6}^{\#(\mathrm{se,out})}+C_6^{\#(\mathrm{b,out})}+C_{6}^{\#(\mathrm{s})}. 
\end{align}

\subsection{Unconstrained partition function}
For completeness, we recall the special case of \cite{ACC2023c} used to
assemble the constrained partition functions.
\begin{theorem}[Unconstrained free-energy expansion; \cite{ACC2023c}]
\label{theorem:unconstrained-free-energy}
Let $\alpha>-1$ be fixed and let $\lambda$ satisfy the standing smoothness
assumptions.  Uniformly for real $s$ with $|s|\leq\log n$, as
$n\to\infty$,
\begin{equation}
\log Z_{n,s\lambda}^{(\alpha)}[Q]
=
C_1^{(\mathrm{n})}n^2
+C_2^{(\mathrm{n})}n\log n
+C_3^{(\mathrm{n})}n
+C_4^{(\mathrm{n})}\sqrt{n}
+C_5^{(\mathrm{n})}\log n
+C_6^{(\mathrm{n})}+\mathcal{O}\Bigl(\frac{(\log n)^3}{n^{\frac{1}{12}}}\Bigr),
\end{equation}
where 
\begin{align}
    C_1^{(\mathrm{n})}&:=-I_Q[\mu_Q]=-\Bigl(q(R)-\log R-\frac{1}{4}\int_{0}^{R}rq'(r)^2\,dr\Bigr),
    \\
    C_2^{(\mathrm{n})}&:=\frac{1}{2},
    \\
    C_3^{(\mathrm{n})}&:=\frac{\log 2\pi}{2}-1-\frac{E_Q[\mu_Q]}{2}+\int_{\C}\mathsf{k}(z)\,d\mu_Q(z),
    \\
    C_4^{(\mathrm{n})}&:=0,
    \\
    C_5^{(\mathrm{n})}&:=\frac{5}{12}+\frac{\alpha^2}{2},
    \\
    \begin{split}
    C_6^{(\mathrm{n})}&:=\zeta'(-1)-\log G(1+\alpha)+F_Q[\mu_Q]+\frac{1+\alpha}{2}\log(2\pi)+\mathsf{e}_{\mathsf{k}}+\frac{s^2}{2}\mathsf{v}_{0,\lambda}
    \\
    &\quad
    +\frac{\alpha^2}{2}\log\bigl(R^2\Delta Q(0)\bigr)+\alpha s\bigl(\lambda(R)-\lambda(0)\bigr).
\end{split}
\end{align}
The superscript $\mathrm n$ refers to the unconstrained random normal matrix
model.  Here $G$ is the Barnes $G$-function, $\zeta$ is the Riemann zeta
function, and
\begin{align*}
E_{Q}[\mu_{Q}]&:=\int_{\C}\log \Delta Q\,d\mu_{Q}, 
\\
F_{Q}[\mu_{Q}]&:=\frac{1}{12}\log\frac{1}{R^2\Delta Q(R)}-\frac{1}{16}\frac{R\partial_r\Delta Q(R)}{\Delta Q(R)}+\frac{1}{24}\int_0^{R}\Bigl(\frac{\partial_r\Delta Q(r)}{\Delta Q(r)}\Bigr)^2r\,dr, 
\\
\mathsf{e}_{\mathsf{k}}
&:=
\frac{1}{2}\int_{S}\mathsf{k}(z)\Delta \log \Delta Q(z)\,dA(z)+\frac{1}{8\pi}\int_{\partial S}\partial_{\mathrm{n}}\mathsf{k}(z)|dz|-\frac{1}{8\pi}\int_{\partial S}\mathsf{k}(z)\frac{\partial_{\mathrm{n}}\Delta Q(z)}{\Delta Q(z)}\,|dz|, 
\\
\mathsf{v}_{0,\lambda}
&:=\frac{1}{4}\int_{S}|\nabla \lambda(z)|^2\,dA(z)
=\frac{1}{2}\int_{0}^{R}r\lambda'(r)^2\,dr. 
\end{align*}
The symbol $\partial_{\mathrm n}$ denotes the outward normal derivative on
$\partial S$.
\end{theorem}

\section{Proof strategy and preliminary asymptotic expansions}
\label{section:preliminary}
\subsection{Exact factorization}
Radial symmetry makes the monomials orthogonal.  Andr\'{e}ief's identity
(see, for example, \cite{C2021}) therefore factorizes each deformed
partition function into one-dimensional radial norms.  Separating the
unconstrained normalization, the hard-wall ratio, and the jump deformation
gives the exact identities
\begin{align}
\label{def of log Dnan Dndc Dncd}
\log D_{n}^{(\mathrm{an})}
&=\log\mathcal{E}_{n,s\lambda,\alpha}^{(\mathrm{an})}
+\log\mathcal{P}_{n,s\lambda,\alpha}^{(\mathrm{an})}
+\log Z_{n,s\lambda}^{(\alpha)}[Q],
\\
\log D_{n}^{(\mathrm{dc})}
&=\log\mathcal{E}_{n,s\lambda,\alpha}^{(\mathrm{dc})}
+\log\mathcal{P}_{n,s\lambda,\alpha}^{(\mathrm{dc})}
+\log Z_{n,s\lambda}^{(\alpha)}[Q],
\\
\log D_{n}^{(\mathrm{cd})}
&=\log\mathcal{E}_{n,s\lambda,\alpha}^{(\mathrm{cd})}
+\log\mathcal{P}_{n,s\lambda,\alpha}^{(\mathrm{cd})}
+\log Z_{n,s\lambda}^{(\alpha)}[Q].
\end{align}
The hard-wall factors are
\begin{equation}
\label{def of mathcalPn an dc cd}
\mathcal{P}_{n,s\lambda,\alpha}^{(\mathrm{an})}
:=\prod_{j=0}^{n-1}\frac{h_{n,j}^{(\mathrm{an})}}{h_{n,j}}, 
\qquad
\mathcal{P}_{n,s\lambda,\alpha}^{(\mathrm{dc})}
:=\prod_{j=0}^{n-1}\frac{h_{n,j}^{(\mathrm{dc})}}{h_{n,j}}, 
\qquad
\mathcal{P}_{n,s\lambda,\alpha}^{(\mathrm{cd})}
:=\prod_{j=0}^{n-1}\frac{h_{n,j}^{(\mathrm{cd})}}{h_{n,j}}, 
\end{equation}
where 
\begin{align}
\begin{split}
\label{def of hjnU}
h_{n,j}^{(\mathrm{an})}&:=
\int_{0}^{\rho_1}
2re^{\mathsf{k}(r)} e^{-nV_{\tau}(r)}\,dr
+
\int_{\rho_2}^{+\infty}
2re^{\mathsf{k}(r)} e^{-nV_{\tau}(r)}\,dr,
\\
h_{n,j}^{(\mathrm{dc})}&:=
\int_{0}^{\rho_1}2re^{\mathsf{k}(r)} e^{-nV_{\tau}(r)}\,dr,\quad
h_{n,j}^{(\mathrm{cd})}:=
\int_{\rho_2}^{+\infty}2re^{\mathsf{k}(r)} e^{-nV_{\tau}(r)}\,dr,
\end{split}
\end{align}
where, with $\tau_j:=j/n$,
\begin{equation}
\label{def of V tau}
V_{\tau}(r):=q(r)-2\tau\log r.
\end{equation}
Thus $V_\tau=V_{\tau_j}$ whenever it occurs in the $j$th norm, and
$\mathsf{k}$ is defined in \eqref{def of mathsfkr}.  We also set
\begin{equation}
 h_{n,j}:=
\int_{0}^{+\infty}2r e^{\mathsf{k}(r)}e^{-nV_{\tau_j}(r)}\,dr.   
\end{equation}
We next record the exact jump factors.  Recall the telescoping weights
$\omega_\ell$ from \eqref{def of omega ell}.  In the annular case,
with $r_{7m+1}=+\infty$,
\begin{equation}
\label{def of calEnan}
\log\mathcal{E}_{n,s\lambda,\alpha}^{(\mathrm{an})}=
\sum_{j=0}^{n-1}\log\Big( 
1+\sum_{\ell=1}^{7m}\omega_{\ell}F_{n,j,\ell}^{(\mathrm{an})}
\Bigr), 
\end{equation}
where 
\begin{equation}
\label{def of norming constant + counting on annulus}
   F_{n,j,\ell}^{(\mathrm{an})}
   :=
   \begin{cases}
       \frac{2\int_0^{r_{\ell}}ue^{-nV_{\tau}(u)}e^{\mathsf{k}(u)}\,du}{2(\int_{0}^{\rho_1}+\int_{\rho_2}^{+\infty})ue^{\mathsf{k}(u)}e^{-nV_{\tau}(u)}\,du}, & \ell=1,\dots, 3m, \\
       \frac{2\int_0^{\rho_1}ue^{-nV_{\tau}(u)}e^{\mathsf{k}(u)}\,du+2\int_{\rho_2}^{r_{\ell}}ue^{-nV_{\tau}(u)}e^{\mathsf{k}(u)}\,du}{2(\int_{0}^{\rho_1}+\int_{\rho_2}^{+\infty})ue^{\mathsf{k}(u)}e^{-nV_{\tau}(u)}\,du}, & \ell=3m+1,\dots,7m. 
    \end{cases}
\end{equation}
We use the principal branch of the logarithm.  Equivalently, set
$h_{n,j,0}^{(\mathrm{an})}:=h_{n,j}^{(\mathrm{an})}$ and
\begin{equation}
h_{n,j,\ell}^{(\mathrm{an})}
:=
\begin{cases}
2\int_0^{r_{\ell}}u e^{-nV_{\tau}(u)}e^{\mathsf{k}(u)}\,du,     & \ell=1,2,\dots,3m, \\
2\int_0^{\rho_1}u e^{-nV_{\tau}(u)}e^{\mathsf{k}(u)}\,du
+2\int_{\rho_2}^{r_{\ell}}u e^{-nV_{\tau}(u)}e^{\mathsf{k}(u)}\,du,       & \ell=3m+1,\dots, 7m.
\end{cases}
\end{equation}
Then
\[
    F_{n,j,\ell}^{(\mathrm{an})}= \frac{h_{n,j,\ell}^{(\mathrm{an})}}{h_{n,j}^{(\mathrm{an})}}. 
\]
For the disk-complement geometry, put $r_{3m+1}=+\infty$ and write
\begin{equation}
\label{def of calEndc}
\log \mathcal{E}_{n,s\lambda,\alpha}^{(\mathrm{dc})}=
\sum_{j=0}^{n-1}\log\Big( 
1+\sum_{\ell=1}^{3m}\omega_{\ell}F_{n,j,\ell}^{(\mathrm{dc})}
\Bigr), 
\end{equation}
where we set $h_{n,j,0}^{(\mathrm{dc})}=h_{n,j}^{(\mathrm{dc})}$ and define
\begin{equation}
   F_{n,j,\ell}^{(\mathrm{dc})}
   :=
       \frac{h_{n,j,\ell}^{(\mathrm{dc})}}{h_{n,j}^{(\mathrm{dc})}}, \qquad
h_{n,j,\ell}^{(\mathrm{dc})}
:=
2\int_0^{r_{\ell}}u e^{-nV_{\tau}(u)}e^{\mathsf{k}(u)}\,du,    \qquad \ell=1,2,\dots,3m. 
\end{equation}
For the centered-disk geometry, put $r_{7m+1}=+\infty$.  Then
\begin{equation}
\label{def of calEncd}
\log \mathcal{E}_{n,s\lambda,\alpha}^{(\mathrm{cd})}=
\sum_{j=0}^{n-1}\log\Big( 
1+\sum_{\ell=3m+1}^{7m}\omega_{\ell}F_{n,j,\ell}^{(\mathrm{cd})}
\Bigr), 
\end{equation}
where 
\begin{equation}
   F_{n,j,\ell}^{(\mathrm{cd})}
   :=
       \frac{2\int_{\rho_2}^{r_{\ell}}ue^{-nV_{\tau}(u)}e^{\mathsf{k}(u)}\,du}{2\int_{\rho_2}^{+\infty}ue^{\mathsf{k}(u)}e^{-nV_{\tau}(u)}\,du}, \qquad
       h_{n,j,\ell}^{(\mathrm{cd})}
:=
2\int_{\rho_2}^{r_{\ell}}u e^{-nV_{\tau}(u)}e^{\mathsf{k}(u)}\,du,   
       \qquad
    \ell=3m+1,\dots,7m. 
\end{equation}

The proof consists of three parts: asymptotics of the hard-wall
products \eqref{def of mathcalPn an dc cd}, asymptotics of the jump sums
\eqref{def of calEnan}--\eqref{def of calEncd}, and the unconstrained
expansion recalled above.  The next six theorems state the first two inputs.

\begin{theorem}[Annular hard-wall ratio]
\label{theorem:hole probability of annulus case}
Under the standing assumptions, as $n\to\infty$,
\[
\log \mathcal{P}_{n,s\lambda,\alpha}^{(\mathrm{an})}
=
C_{1,\mathrm{hole}}^{(\mathrm{an})}n^2
+C_{2,\mathrm{hole}}^{(\mathrm{an})}n\log n
+C_{3,\mathrm{hole}}^{(\mathrm{an})}n
+C_{4,\mathrm{hole}}^{(\mathrm{an})}\sqrt{n}
+C_{5,\mathrm{hole}}^{(\mathrm{an})}\log n
+C_{6,\mathrm{hole}}^{(\mathrm{an})}
+\mathcal{F}_{n}[\lambda]
+\mathcal{O}(n^{-\frac{1}{12}}),
\] 
where $C_{k,\mathrm{hole}}^{(\mathrm{an})}$, $1\leq k\leq6$, and
$\mathcal F_n[\lambda]$ are given in
\eqref{def of C1 hole an}--\eqref{def of calF hole an}.
In particular, setting $\alpha=0$ and $\lambda\equiv0$ gives the precise
asymptotic expansion of the annular hole probability
\eqref{def of hole probability U}.
\end{theorem}

\begin{theorem}[Disk-complement hard-wall ratio]
\label{theorem:hole probability of disk complement case}
Under the standing assumptions, as $n\to\infty$,
\[
\log \mathcal{P}_{n,s\lambda,\alpha}^{(\mathrm{dc})}
=
C_{1,\mathrm{hole}}^{(\mathrm{dc})}n^2
+C_{2,\mathrm{hole}}^{(\mathrm{dc})}n\log n
+C_{3,\mathrm{hole}}^{(\mathrm{dc})}n
+C_{4,\mathrm{hole}}^{(\mathrm{dc})}\sqrt{n}
+C_{5,\mathrm{hole}}^{(\mathrm{dc})}\log n
+C_{6,\mathrm{hole}}^{(\mathrm{dc})}+\mathcal{O}(n^{-\frac{1}{12}}),
\]
where $C_{k,\mathrm{hole}}^{(\mathrm{dc})}$, $1\leq k\leq6$, are
given in \eqref{def of C1dc hole}--\eqref{def of C6dc hole}.
Setting $\alpha=0$ and $\lambda\equiv0$ gives the corresponding
disk-complement hole probability.
\end{theorem}

\begin{theorem}[Centered-disk hard-wall ratio]
\label{theorem:hole probability of centered disk case}
Under the standing assumptions, as $n\to\infty$,
\[
\log \mathcal{P}_{n,s\lambda,\alpha}^{(\mathrm{cd})}
=
C_{1,\mathrm{hole}}^{(\mathrm{cd})}n^2
+C_{2,\mathrm{hole}}^{(\mathrm{cd})}n\log n
+C_{3,\mathrm{hole}}^{(\mathrm{cd})}n
+C_{4,\mathrm{hole}}^{(\mathrm{cd})}\sqrt{n}
+C_{5,\mathrm{hole}}^{(\mathrm{cd})}\log n
+C_{6,\mathrm{hole}}^{(\mathrm{cd})}+\mathcal{O}(n^{-\frac{1}{12}}),
\]
where $C_{k,\mathrm{hole}}^{(\mathrm{cd})}$, $1\leq k\leq6$, are
given in \eqref{def of C1cd hole}--\eqref{def of C6cd hole}.
Setting $\alpha=0$ and $\lambda\equiv0$ gives the corresponding
centered-disk hole probability.
\end{theorem}
The proofs are given in Section~\ref{section:proof of hole probability part}.
The remaining three auxiliary theorems concern
\eqref{def of calEnan}--\eqref{def of calEncd}.

\begin{theorem}[Disk counting statistics with a disk-complement wall]
\label{theorem:counting statistics of disk complement case}
Under the standing assumptions, as $n\to\infty$,
\[
\log \mathcal{E}_{n,s\lambda,\alpha}^{(\mathrm{dc})}
=
C_{3,\#}^{(\mathrm{dc})}n 
+
C_{4,\#}^{(\mathrm{dc})}\sqrt{n}
+
C_{5,\#}^{(\mathrm{dc})}\log n
+
C_{6,\#}^{(\mathrm{dc})}
+
\mathcal{O}(n^{-\frac{1}{12}}),
\]
where $C_{k,\#}^{(\mathrm{dc})}$, $3\leq k\leq6$, are given in
\eqref{def of C3sharp dc}--\eqref{def of C6sharp dc}.
\end{theorem}

\begin{theorem}[Disk counting statistics with an annular wall]
\label{theorem:counting statistics of annulus case}
Under the standing assumptions, as $n\to\infty$,
\[
\log \mathcal{E}_{n,s\lambda,\alpha}^{(\mathrm{an})}
=C_{3,\#}^{(\mathrm{an})}n
+C_{4,\#}^{(\mathrm{an})}\sqrt n
+C_{5,\#}^{(\mathrm{an})}\log n
+C_{6,\#}^{(\mathrm{an})}
+\mathcal{F}_{n,\#}
+\mathcal{O}(n^{-1/12}),
\]
where $C_{k,\#}^{(\mathrm{an})}$, $3\leq k\leq6$, are given in
\eqref{def of C3sharp an}--\eqref{def of C6sharp an}, and
$\mathcal F_{n,\#}$ is given in \eqref{def of calFNsharp}.
For counting parameters in a sufficiently small complex polydisc about the
origin, with all logarithmic branches continued from the origin, the
expansion and its remainder are locally uniform.  In particular, all
parameter derivatives of the remainder satisfy the corresponding Cauchy
estimates on smaller polydiscs.
\end{theorem}

\begin{theorem}[Disk counting statistics with a centered-disk wall]
\label{theorem:counting statistics of centered disk case}
Under the standing assumptions, as $n\to\infty$,
\[
\log \mathcal{E}_{n,s\lambda,\alpha}^{(\mathrm{cd})}
=C_{3,\#}^{(\mathrm{cd})}n
+C_{4,\#}^{(\mathrm{cd})}\sqrt n
+C_{5,\#}^{(\mathrm{cd})}\log n
+C_{6,\#}^{(\mathrm{cd})}
+\mathcal{O}(n^{-1/12}),
\]
where $C_{k,\#}^{(\mathrm{cd})}$, $3\leq k\leq6$, are given in
\eqref{def of Csharp3 cd}--\eqref{def of Csharp6 cd}.
\end{theorem}
Their proofs occupy Section~\ref{section:proof of calen part}.  Substitution
of the six auxiliary expansions and the unconstrained theorem into
\eqref{def of log Dnan Dndc Dncd} proves the three main theorems.

We begin in Subsection~\ref{subsection:asymptotics of norming constant annulus}
with uniform asymptotics for $h_{n,j}^{(\mathrm{an})}$.  The analysis combines
the radial Laplace method and Euler--Maclaurin summation used in
\cite{BKS2023,ACC2023c,Seo,C2021}; the principal additional issue is matching
the two competing endpoint saddles across the annular spectral gap.

\subsection{Uniform asymptotics of the norming constants}
\label{subsection:asymptotics of norming constant annulus}
For $k\in\{1,2\}$, set
\begin{equation}
\label{def of tau k ast}
    \tau_{\rho_k}:=\frac{1}{2}\rho_kq'(\rho_{k}). 
\end{equation}
We first record several elementary identities. For a radially symmetric potential $Q(z)=q(|z|)$, the Laplacian in polar coordinates is
\[
4\Delta Q(z)\bigr|_{z=r}\equiv 4\Delta Q(r)=\frac{1}{r}(rq'(r))'=q''(r)+\frac{q'(r)}{r}. 
\]
Differentiating \eqref{def of V tau} with respect to $r$ gives (cf.\ \cite[Eq.~(2.4)]{BKS2023})
\begin{align}
\begin{split}
\label{def of V tau relationship 1}
V_{\tau}'(r)&=q'(r)-\frac{2\tau}{r},\qquad V_{\tau}''(r)=4\Delta Q(r)-\frac{1}{r}V_{\tau}'(r), \\ 
V_{\tau}^{(3)}(r)&=4\partial_{r}\Delta Q(r)-\frac{4}{r}\Delta Q(r)+\frac{2}{r^2}V_{\tau}'(r), \\
V_{\tau}^{(4)}(r)&=4\partial_r^2\Delta Q(r)+\frac{12}{r^2}\Delta Q(r)-\frac{4}{r}\partial_{r}\Delta Q(r)-\frac{6}{r^3}V_{\tau}'(r). 
\end{split}
\end{align} 
By Assumption~\ref{Assumption_Q}(2), $rq'(r)$ is strictly increasing
on the smooth neighborhood of $[0,R]$.  Hence, for each $\tau\in(0,1]$,
there is a unique $r(\tau)\in(0,R]$ such that
\begin{equation}
\label{def of V tau relationship 2}
V_{\tau}'(r(\tau))=0,\qquad V_{\tau}''(r(\tau))>0. 
\end{equation}
For $\tau>0$, implicit differentiation yields
\[
\frac{dr(\tau)}{d\tau}=\frac{2}{(rq'(r))'}\Bigr|_{r=r(\tau)}=\frac{1}{2r(\tau)\Delta Q(r(\tau))}>0,
\]
Extend $r$ continuously by $r(0)=0$.  Then $r(\tau)$ is strictly increasing
on $[0,1]$.  By \eqref{def of V tau relationship 1} and \eqref{def of V tau relationship 2}, $r_{\tau}:=r(\tau)$ satisfies
\begin{equation}
\label{def of rtau q'rtau = 2tau}
r_{\tau}q'(r_{\tau})=2\tau.
\end{equation}
In particular, $r(0)q'(r(0))=0$ and $r(1)=R$, with $Rq'(R)=2$, by Assumption~\ref{Assumption_Q}(2).
Let $\epsilon>0$ be a small constant independent of $n$. 
Define 
\[
j_{k,-}:=\Bigl\lceil\frac{n\tau_{\rho_k}}{1+\epsilon}\Bigr\rceil, \qquad 
j_{k,+}:=\Bigl\lfloor\frac{n\tau_{\rho_k}}{1-\epsilon}\Bigr\rfloor,\qquad k\in\{1,2\},
\]
where $\lceil x\rceil$ and $\lfloor x\rfloor$ denote the least integer not smaller than $x$ and the greatest integer not larger than $x$, respectively.
Choose $\epsilon>0$ sufficiently small that
\begin{equation}
\label{def of tau k ast epsilon inequality}
    \frac{\tau_{\rho_1}}{1-\epsilon}<\frac{\tau_{\rho_2}}{1+\epsilon},\qquad \frac{\tau_{\rho_2}}{1-\epsilon}<1. 
\end{equation}
This choice is possible by the definition of $\tau$.
Then
\[
\frac{n\tau_{\rho_k}}{1+\epsilon}\leq j_{k,-}<j_{k,+}\leq \frac{n\tau_{\rho_k}}{1-\epsilon}. 
\]
Set
\begin{equation}
\label{def of M parameter}
M=n^{\frac{1}{12}}.
\end{equation}
Define
\begin{equation}
g_{k,-}:=\Bigl\lceil \frac{n\tau_{\rho_k}}{1+\frac{M}{\sqrt{n}}} \Bigr\rceil, \qquad 
g_{k,+}:=\Bigl\lfloor \frac{n\tau_{\rho_k}}{1-\frac{M}{\sqrt{n}}}\Bigr\rfloor,\qquad k\in\{1,2\}. 
\end{equation}
Shrinking $\epsilon$ further if necessary, we also require
\begin{equation}
\label{def of sigma ast interrelashion epsilon}
\frac{\tau_{\rho_1}}{1-\epsilon}<\sigma_{\star}<\frac{\tau_{\rho_2}}{1+\epsilon},
\end{equation}
which is possible by Assumption~\ref{Assumption_Q}.
We further set
\begin{equation}
\label{def of index jstar}
    j_{\star}:=n\sigma_{\star}.
\end{equation}
By the definition of $\sigma_{\star}$,
\begin{align*}
V_{\sigma_{\star}}(\rho_{2})-V_{\sigma_{\star}}(\rho_{1})
&=q(\rho_{2})-q(\rho_{1})-2\log\frac{\rho_{2}}{\rho_{1}}\frac{q(\rho_{2})-q(\rho_{1})}{2\log\frac{\rho_{2}}{\rho_{1}}}=0.
\end{align*}
We shall repeatedly use the identities
\begin{align}
\label{def of difference between Vtau V rho1 jstar}
V_{\tau}(r)-V_{\tau}(\rho_{1})
&=V_{\tau}(r)-V_{\tau}(\rho_{2})+2(\tau-\sigma_{\star})\log\frac{\rho_{1}}{\rho_{2}}, 
\\
\label{def of difference between Vtau V rho2 jstar}
V_{\tau}(r)-V_{\tau}(\rho_{2})&=
V_{\tau}(r)-V_{\tau}(\rho_{1})+2(\tau-\sigma_{\star})\log\frac{\rho_{2}}{\rho_{1}}. 
\end{align}
Throughout this section, we use the scales
\begin{equation}
\delta_n:=\frac{M}{\sqrt{n}},\qquad \delta_n':=\frac{\log n}{\sqrt{n}}.  
\end{equation}

\begin{lemma}\label{lemma: 0 j j1-}
Let $D_n:=\lceil n^{\frac{1}{6}}\rceil$.
Uniformly for $0\leq j\leq D_n-1$, there exists a constant $c>0$, independent of $n$ and $j$, such that
\begin{align}
\begin{split}
h_{n,j}^{(\mathrm{an})}
&=h_{n,j}\cdot(1+\mathcal{O}(e^{-c (\log n)^2})).
\end{split}
\end{align}
The same conclusion holds uniformly for $D_n\leq j\leq g_{1,-}-1$:
\begin{align}
\begin{split}
h_{n,j}^{(\mathrm{an})}&
=h_{n,j}\cdot(1+\mathcal{O}(e^{-c (\log n)^2})).
\end{split}
\end{align}
More precisely, uniformly for $0\leq j\leq D_n-1$,
\begin{equation}
\label{def of 0 leq j leq Dn}
h_{n,j}=\frac{\Gamma(j+\alpha+1)}{(n\Delta Q(0))^{j+\alpha+1}}e^{-nq(0)+s\lambda(0)}\cdot\Bigl(1+\mathcal{O}(\frac{(j+1)^{\frac{3}{2}}(\log n)^3}{\sqrt{n}})\Bigr),    
\end{equation}
whereas, with $\mathsf{k}(r)=s\lambda(r)+2\alpha\log r$, uniformly for $D_n\leq j\leq n-1$,
\begin{equation}
\label{def of Dn leq j leq n-1}
h_{n,j}
=\sqrt{\frac{2\pi}{n}} \frac{r_{\tau}}{\sqrt{\Delta Q(r_{\tau})}}e^{\mathsf{k}(r_{\tau})}e^{-nV_{\tau}(r_{\tau})}\cdot\Bigl(1+\frac{\mathcal{A}(r_{\tau})}{n}+\mathcal{O}(\frac{(\log n)^{\nu}}{j^{3/2}})\Bigr)    
\end{equation}
for some $\nu>0$, where
\begin{align}
\mathcal{A}(r_{\tau})&:=\mathcal{B}(r_{\tau})+\frac{\mathsf{k}'(r_{\tau})^2}{2}\frac{1}{d_2}+\frac{\mathsf{k}''(r_{\tau})}{2}\frac{1}{d_2}+\frac{\mathsf{k}'(r_{\tau})}{r_{\tau}}\frac{1}{d_2}-\frac{\mathsf{k}'(r_{\tau})}{2}\frac{d_3}{d_2^2},   
\\
\mathcal{B}(r)&:=-\frac{1}{32}\frac{\partial_r^2\Delta Q(r)}{(\Delta Q(r))^2}-\frac{19}{96r}\frac{\partial_r\Delta Q(r)}{(\Delta Q(r))^2}+\frac{5}{96}\frac{(\partial_r \Delta Q(r))^2}{(\Delta Q(r))^3}+\frac{1}{12}\frac{1}{r^2\Delta Q(r)},      \\
d_m&:=V_{\tau}^{(m)}(r_{\tau}).
\end{align}
Here the error $\mathcal{O}(j^{-3/2}(\log n)^{\nu})$ may be replaced by $\mathcal{O}(n^{-2})$ whenever $j\geq c_0n$, for any fixed $c_0>0$.
\end{lemma}

\begin{proof}[Proof of Lemma~\ref{lemma: 0 j j1-}]
Recall that $r_0=0$. We first consider $0\leq j\leq D_n-1$. Applying \cite[Lemma~4.1]{ACC2023c} to
\[
\int_{0}^{\rho_1}2r e^{\mathsf{k}(r)}e^{-nV_{\tau}(r)}\,dr
\]
gives the asserted local expansion; the contribution from the complementary component of $\mathbb A_{\rho_1,\rho_2}^{\rm c}$ is exponentially small by the estimate below.
It remains to treat $D_n\leq j\leq g_{1,-}-1$.
Since $rq'(r)$ is strictly increasing, the unique global minimizer $r_\tau$ of $V_\tau$ lies in $(0,\rho_1)$ throughout this range. We decompose
\begin{align*}
h_{n,j}^{(\mathrm{an})}
&=\int_{|r-r_{\tau}|<\delta_n}2r e^{\mathsf{k}(r)}e^{-nV_{\tau}(r)}\,dr
+\int_0^{r_{\tau}-\delta_n}2r e^{\mathsf{k}(r)}e^{-nV_{\tau}(r)}\,dr
\\
&\quad
+\int_{r_{\tau}+\delta_n}^{\rho_1}2r e^{\mathsf{k}(r)}e^{-nV_{\tau}(r)}\,dr
+\int_{\rho_2}^{+\infty}2r e^{\mathsf{k}(r)}e^{-nV_{\tau}(r)}\,dr.
\end{align*}
By strict convexity near $r_\tau$ and monotonicity away from the minimizer, there exists $c>0$, uniform in the stated range of $j$, such that, whenever $|r-r_\tau|>\delta_n'$,
\[
|V_{\tau}(r)-V_{\tau}(r_{\tau})|\geq c V_{\tau}''(r_{\tau}){\delta_n'}^2. 
\]
Since $\lambda$ is smooth and compactly supported, $C_s:=2\sup_{t\geq0}e^{s\lambda(t)}<\infty$ for each fixed $s\in\mathbb R$. Choosing a sufficiently large fixed $p>0$ and using the preceding bound, we obtain
\begin{align*}
\int_0^{r_{\tau}-\delta_n'}2r e^{\mathsf{k}(r)}e^{-nV_{\tau}(r)}\,dr
&=e^{-nV_{\tau}(r_{\tau})}\int_0^{r_{\tau}-\delta_n'}2r e^{\mathsf{k}(r)} e^{-n(V_{\tau}(r)-V_{\tau}(r_{\tau}))}\,dr
\\
&\leq C_se^{-nV_{\tau}(r_{\tau})}e^{-c(n-p){\delta_n'}^2}
\int_0^{r_{\tau}-\delta_n'}r^{2\alpha+1}e^{-p(V_{\tau}(r)-V_{\tau}(r_{\tau}))}\,dr
\\
&=e^{-nV_{\tau}(r_{\tau})}\cdot\mathcal{O}(e^{-c(\log n)^2}),
\end{align*}
where the implicit constant is uniform for $D_n\leq j\leq g_{1,-}-1$. The same argument gives, for some uniform $c>0$,
\begin{align*}
\int_{r_{\tau}+\delta_n'}^{\rho_1}2r e^{\mathsf{k}(r)}e^{-nV_{\tau}(r)}\,dr&=e^{-nV_{\tau}(r_{\tau})}\cdot\mathcal{O}(e^{-c(\log n)^2}),\\
\int_{\rho_2}^{+\infty}2r e^{\mathsf{k}(r)}e^{-nV_{\tau}(r)}\,dr&=e^{-nV_{\tau}(r_{\tau})}\cdot\mathcal{O}(e^{-c(\log n)^2}).
\end{align*}
Applying \cite[Lemma~4.3]{ACC2023c} therefore proves the first assertion. For later use, we record the Laplace expansion explicitly. Taylor expansion about $r_\tau$ gives
\begin{align*}
&\quad2\int_{|r-r_{\tau}|<\delta_n'}re^{\mathsf{k}(r)}e^{-nV_{\tau}(r)}\,dr
\\
&=    
\frac{2r_{\tau}e^{\mathsf{k}(r_{\tau})}e^{-nV_{\tau}(r_{\tau})}}{\sqrt{n}}\int_{-\sqrt{n}\delta_n'}^{\sqrt{n}\delta_n'}\Bigl(1+\frac{u}{r_{\tau}\sqrt{n}}\Bigr)
e^{-\frac{1}{2}V_{\tau}''(r_{\tau})u^2}
\\
&\quad\times
e^{\mathsf{k}(r_{\tau}+\frac{u}{\sqrt{n}})-\mathsf{k}(r_{\tau})}
e^{-\frac{1}{\sqrt{n}}\frac{1}{6}V_{\tau}^{(3)}(r_{\tau})u^3-\frac{1}{n}\frac{1}{24}V_{\tau}^{(4)}(r_{\tau})u^4-\frac{1}{n^{3/2}}\frac{1}{120}V_{\tau}^{(5)}(r_{\tau})u^5-\frac{1}{n^2}\frac{1}{720}V_{\tau}^{(6)}(r_{\tau})u^6+\mathcal{O}(\frac{u^7}{n^{5/2}})}\,du
\\
&=
\frac{2r_{\tau}e^{\mathsf{k}(r_{\tau})}e^{-nV_{\tau}(r_{\tau})}}{\sqrt{nd_2}}\int_{-\sqrt{d_2n}\delta_n'}^{\sqrt{d_2n}\delta_n'}
e^{-\frac{1}{2}u^2}
\Bigl( 
1+\frac{c_1(u)}{\sqrt{n}}+\frac{c_2(u)}{n}+\frac{c_3(u)}{n^{3/2}}+\mathcal{O}(\frac{c_4(u)}{n^2})
\Bigr)\,du
\end{align*}
where, with $d_j:=V_{\tau}^{(j)}(r_\tau)$,
\begin{align}
\label{def of c_1}
c_1(u)&:=-\frac{d_3}{6d_2^{3/2}}u^3+\Bigl(\mathsf{k}'(r_{\tau})+\frac{1}{r_{\tau}}\Bigr)\frac{1}{d_2^{1/2}}u, 
\\
\label{def of c_2}
c_2(u)&:=\frac{d_3^2}{72d_2^3}u^6-\Bigl(d_4+4d_3\mathsf{k}'(r_{\tau})+\frac{4d_3}{r_{\tau}}\Bigr)\frac{u^4}{24d_2^2}+\Bigl(\frac{\mathsf{k}''(r_{\tau})}{2}+\frac{\mathsf{k}'(r_{\tau})^2}{2}+\frac{\mathsf{k}'(r_{\tau})}{r_{\tau}}\Bigr)\frac{u^2}{d_2}
\\
\begin{split}
\label{def of c_3}
c_3(u)&:=-\frac{d_3^3}{1296d_2^{9/2}}u^9+\Bigl(d_3d_4+2d_3^2\mathsf{k}'(r_{\tau})+\frac{2d_3^2}{r_{\tau}} \Bigr)\frac{u^7}{144d_2^{7/2}}
\\
&\quad
-\Bigl(
d_5+\frac{5}{r_{\tau}}d_4+5d_4\mathsf{k}'(r_{\tau})+\frac{20d_3}{r_{\tau}}\mathsf{k}'(r_{\tau})+10d_3\mathsf{k}'(r_{\tau})^2+10d_3\mathsf{k}''(r_{\tau})
\Bigr)
\frac{u^5}{120d_2^{5/2}}
\\
&\quad
+
\Bigl(\frac{3\mathsf{k}'(r_{\tau})^2}{r_{\tau}}
+\frac{3\mathsf{k}''(r_{\tau})}{r_{\tau}}+\mathsf{k}'(r_{\tau})^3+3\mathsf{k}'(r_{\tau})\mathsf{k}''(r_{\tau})+\mathsf{k}^{(3)}(r_{\tau})\Bigr)\frac{u^3}{6d_2^{3/2}},
\end{split}
\end{align}
and $c_4(u)$ is a polynomial in $u^{12},u^{10},u^8,u^6,u^4$. For every $m\geq1$, the derivative bound $|V_\tau^{(m)}(r_\tau)|\leq C(1+\tau^{-m/2+1})$ follows, for example, from \cite[proof of Lemma~3.2]{BKS2023}. Termwise integration, together with the Gaussian tail estimate, yields the stated expansion; see also \cite[Lemmas~4.1 and 4.3]{ACC2023c}.
This completes the proof. 
\end{proof}

\begin{lemma}\label{lemma: j1- j g1-}
Uniformly for $j_{1,-}\leq j\leq g_{1,-}-1$, there exists $c>0$, independent of $n$ and $j$, such that
\begin{align}
h_{n,j}^{(\mathrm{an})}
=h_{n,j}\cdot(1+\mathcal{O}(e^{-c M^2})).
\end{align}
\end{lemma}
\begin{proof}
For $j_{1,-}\leq j\leq g_{1,-}-1$, the definitions of the two cutoffs give
\[
\frac{\tau_{\rho_1}}{1+\epsilon}+\mathcal O(n^{-1})
\leq \tau=\frac{j}{n}
\leq \frac{\tau_{\rho_1}}{1+M/\sqrt n},
\]
and hence
\[
\tau_{\rho_1}-\tau
\geq \frac{\tau_{\rho_1}M/\sqrt n}{1+M/\sqrt n}.
\]
Since $r\mapsto rq'(r)$ has strictly positive derivative near $\rho_1$, the inverse-function theorem and
$r_\tau q'(r_\tau)=2\tau$ show that $r_\tau<\rho_1$ and
$\rho_1-r_\tau\geq cM/\sqrt n$, uniformly in this range. The quadratic lower bound for
$V_\tau(r)-V_\tau(r_\tau)$ used in Lemma~\ref{lemma: 0 j j1-} therefore makes the portion beyond the hard edge relatively $\mathcal O(e^{-cM^2})$. This proves the claim.
\end{proof}

We next analyze the transition range $g_{1,-}\leq j\leq g_{1,+}$, where the saddle approaches the hard edge.

\begin{lemma}\label{lemma:g1- j g1+}
Let
\begin{equation}
\mathfrak{p}_{\ell}(x):=
\sum_{m=1}^{\ell}\frac{x^{2m-1}}{(2m-1)!!}
\qquad
\mathfrak{q}_{\ell}(x):=\sum_{m=0}^{\ell}\frac{x^{2m}}{(2m)!!}.
\end{equation}
Recall that $d_j:=V_{\tau}^{(j)}(r_{\tau})$, and define
\begin{align}
    \xi_{j,1}&:=\sqrt{nd_2}(\rho_1-r_{\tau}), 
    \\
    \mathfrak{c}_1(\xi_{j,1})&:=\frac{e^{-\frac{1}{2}\xi_{j,1}^2}}{2\sqrt{\Delta Q(r_{\tau})}}\biggl[
\frac{r_{\tau}\partial_r\Delta Q(r_{\tau})-\Delta Q(r_{\tau})}{3r_{\tau}\Delta Q(r_{\tau})}\mathfrak{q}_1(\xi_{j,1})-\Bigl( \mathsf{k}'(r_{\tau})+\frac{1}{r_{\tau}}\Bigr)
\bigg],
    \\
    \begin{split}
    \mathfrak{c}_2(\xi_{j,1})&:=\bigg[
\frac{5d_3^2}{24d_2^3}-\frac{1}{8d_2^2}\Bigl(d_4+4d_3\mathsf{k}'(r_{\tau})+\frac{4d_3}{r_{\tau}}\Bigr)
\\
&\quad+\frac{1}{d_2}\Bigl(\frac{\mathsf{k}''(r_{\tau})}{2}+\frac{\mathsf{k}'(r_{\tau})^2}{2}+\frac{\mathsf{k}'(r_{\tau})}{r_{\tau}}\Bigr)
\bigg]\sqrt{\frac{\pi}{2}}\mathrm{erfc}\Bigl(-\frac{\xi_{j,1}}{\sqrt{2}}\Bigr)
\\
&-e^{-\frac{\xi_{j,1}^2}{2}}\bigg[ 
\frac{5d_3^2}{24d_2^3}\mathfrak{p}_{3}(\xi_{j,1})-\frac{1}{8d_2^2}\Bigl(d_4+4d_3\mathsf{k}'(r_{\tau})+\frac{4d_3}{r_{\tau}}\Bigr)\mathfrak{p}_{2}(\xi_{j,1})
\\
&\quad
+\frac{1}{d_2}\Bigl(\frac{\mathsf{k}''(r_{\tau})}{2}+\frac{\mathsf{k}'(r_{\tau})^2}{2}+\frac{\mathsf{k}'(r_{\tau})}{r_{\tau}}\Bigr)\mathfrak{p}_{1}(\xi_{j,1})
\bigg].
    \end{split}
\end{align}
Then, uniformly for $g_{1,-}\leq j\leq g_{1,+}$,
\begin{align}
\begin{split}
h_{n,j}^{(\mathrm{an})}
&=h_{\#}(\xi_{j,1})\cdot\bigl(1+\mathcal{O}(e^{-cM^2})\bigr)
\\
&:=
\frac{r_{\tau}e^{\mathsf{k}(r_{\tau})}e^{-nV_{\tau}(r_{\tau})}}{\sqrt{n\Delta Q(r_{\tau})}}
\bigg[
\sqrt{\frac{\pi}{2}}\mathrm{erfc}\Bigl( -\frac{\xi_{j,1}}{\sqrt{2}}\Bigr)
+\frac{\mathfrak{c}_1(\xi_{j,1})}{\sqrt{n}}+\frac{\mathfrak{c}_2(\xi_{j,1})}{n}+\mathcal{O}(\frac{\xi_{j,1}^8}{n^{3/2}})
\bigg].
\end{split}    
\end{align}

\end{lemma}

\begin{proof}
The function $V_\tau$ has its unique global minimum at $r=r_\tau$. Since $Q$ is strictly subharmonic near $|z|=\rho_1$, there is a neighborhood of $\rho_1$ and a constant $\delta>0$ such that $r_\tau q'(r_\tau)=2\tau$ and $\Delta Q(r_\tau)>\delta$ whenever $r_\tau$ belongs to that neighborhood. For $g_{1,-}\leq j\leq g_{1,+}$,
\[
-\frac{M}{\sqrt{n}}\frac{\tau_{\rho_1}}{1-\frac{M}{\sqrt{n}}}
=
\tau_{\rho_1}-\frac{\tau_{\rho_1}}{1-\frac{M}{\sqrt{n}}}
\leq
\tau_{\rho_1}-\tau\leq \tau_{\rho_1}-\frac{\tau_{\rho_1}}{1+\frac{M}{\sqrt{n}}}=\frac{M}{\sqrt{n}}\frac{\tau_{\rho_1}}{1+\frac{M}{\sqrt{n}}}.
\]
The inverse-function theorem, applied to $r\mapsto rq'(r)$ near $\rho_1$, consequently gives
$|\rho_1-r_\tau|\leq CM/\sqrt n$ uniformly in $j$. We first suppose that $\tau\leq\tau_{\rho_1}$ and decompose
\begin{align*}
h_{n,j}^{(\mathrm{an})}
&=\int_{0}^{r_{\tau}-\delta_n}2re^{\mathsf{k}(r)}e^{-nV_{\tau}(r)}\,dr
+\int_{r_{\tau}-\delta_n}^{\rho_1}2re^{\mathsf{k}(r)}e^{-nV_{\tau}(r)}\,dr
+\int_{\rho_2}^{+\infty}2re^{\mathsf{k}(r)}e^{-nV_{\tau}(r)}\,dr. 
\end{align*}
The first and third integrals are exponentially smaller than the second, by the tail estimates in the proof of Lemma~\ref{lemma: 0 j j1-}. Expanding the second integral about $r_\tau$ as in that proof gives
\[
\int_{r_{\tau}-\delta_n}^{\rho_1}2re^{\mathsf{k}(r)}e^{-nV_{\tau}(r)}\,dr
=
\frac{2r_{\tau}e^{\mathsf{k}(r_{\tau})}e^{-nV_{\tau}(r_{\tau})}}{\sqrt{nd_2}}
\int_{-\infty}^{\xi_{j,1}}
e^{-\frac{1}{2}u^2}
\Bigl( 
1+\frac{c_1(u)}{\sqrt{n}}+\frac{c_2(u)}{n}+\frac{c_3(u)}{n^{3/2}}+\mathcal{O}(\frac{c_4(u)}{n^2})
\Bigr)\,du,
\]
where $c_1,c_2,c_3$ are defined in \eqref{def of c_1}--\eqref{def of c_3}, and $c_4$ is the polynomial described immediately thereafter. We use the Gaussian moment identities
\begin{align*}
\int_{-\infty}^{\xi_{j,1}}x^{2\ell}e^{-\frac{1}{2}x^2}
&=(2\ell-1)!!\Bigl(\sqrt{\frac{\pi}{2}}\mathrm{erfc}\Bigl(-\frac{\xi_{j,1}}{\sqrt{2}}\Bigr)
-e^{-\frac{\xi_{j,1}^2}{2}}
\mathfrak{p}_{\ell}(\xi_{j,1})\Bigr),
\\
\int_{-\infty}^{\xi_{j,1}}x^{2\ell+1}e^{-\frac{1}{2}x^2}
&=
-(2\ell)!!e^{-\frac{\xi_{j,1}^2}{2}}\mathfrak{q}_{\ell}(\xi_{j,1}). 
\end{align*}
Termwise integration proves the asserted expansion when $\tau\leq\tau_{\rho_1}$, uniformly throughout the stated transition range.

Suppose now that $\tau\geq\tau_{\rho_1}$. Then
\begin{align*}
h_{n,j}^{(\mathrm{an})}
&=
\int_0^{+\infty}
2re^{\mathsf{k}(r)}e^{-nV_{\tau}(r)}\,dr
-\int_{\rho_1}^{+\infty}
2re^{\mathsf{k}(r)}e^{-nV_{\tau}(r)}\,dr
\\
&=
\int_0^{+\infty}
2re^{\mathsf{k}(r)}e^{-nV_{\tau}(r)}\,dr
-\int_{\rho_1}^{r_{\tau}+\delta_n}
2re^{\mathsf{k}(r)}e^{-nV_{\tau}(r)}\,dr
-
\int_{r_{\tau}+\delta_n}^{+\infty}
2re^{\mathsf{k}(r)}e^{-nV_{\tau}(r)}\,dr
\end{align*}
The final integral is exponentially small by the same tail estimate, while the middle integral is evaluated by the preceding Laplace expansion. Combining it with the full-line Laplace expansion yields
\begin{align*}
h_{n,j}^{(\mathrm{an})}
= 
\frac{r_{\tau}e^{\mathsf{k}(r_{\tau})}e^{-nV_{\tau}(r_{\tau})}}{\sqrt{n\Delta Q(r_{\tau})}}
\bigg[
\sqrt{\frac{\pi}{2}}\mathrm{erfc}\Bigl( -\frac{\xi_{j,1}}{\sqrt{2}}\Bigr)
+\frac{\mathfrak{c}_1(\xi_{j,1})}{\sqrt{n}}+\frac{\mathfrak{c}_2(\xi_{j,1})}{n}+\mathcal{O}\Bigl(\frac{\xi_{j,1}^8}{n^{3/2}}\Bigr)
\bigg],
\end{align*}
This is the same formula as in the first case and completes the proof.
\end{proof}

We next describe the behavior of $h_{n,j}[\mathbb{A}_{\rho_1,\rho_2}^{\rm c}]$ when the dominant point is a hard endpoint. For $g_{1,+}+1\leq j\leq j_{1,+}$, set
$\eta_1\equiv\eta_1(j):=\frac{2}{\rho_1}(\tau-\tau_{\rho_1})>0$ and $v_{\ell,1}:=V_{\tau}^{(\ell)}(\rho_1)$, and define
\begin{align}
\label{def of mathfraka1}
\mathfrak{a}_1
&=
\frac{v_{2,1}}{\eta_1^2}+\Bigl(\mathsf{k}'(\rho_1)+\frac{1}{\rho_1}\Bigr)\frac{1}{\eta_1},
\\
\begin{split}
\label{def of mathfraka2}
\mathfrak{a}_2
&=
\frac{3v_{2,1}^2}{\eta_1^4}+\Bigl(\frac{3v_{2,1}}{\rho_1}+3v_{2,1}\mathsf{k}'(\rho_1)+v_{3,1}\Bigr)\frac{1}{\eta_1^3}
+\Bigl(\frac{2\mathsf{k}'(\rho_1)}{\rho_1}+\mathsf{k}'(\rho_1)^2+\mathsf{k}''(\rho_1)\Bigr)\frac{1}{\eta_1^2},
\end{split}
\\
\begin{split}
\label{def of mathfraka3}
\mathfrak{a}_3
&=
\frac{15}{\eta_1^6}v_{2,1}^3
+\Bigl(\frac{3v_{2,1}^2}{\rho_1}+3v_{2,1}^2\mathsf{k}'(\rho_1)+2v_{2,1}v_{3,1}\Bigr)\frac{5}{\eta_1^5}
\\
&\quad
+\Bigl(v_{4,1}+\frac{4v_{3,1}}{\rho_1}+4v_{3,1}\mathsf{k}'(\rho_1)+6\mathsf{k}''(\rho_1)v_{2,1}+6\mathsf{k}'(\rho_1)^2v_{2,1}+\frac{12v_{2,1}}{\rho_1}\mathsf{k}'(\rho_1)\Bigr)\frac{1}{\eta_1^4}
\\
&\quad
+\Bigl(\frac{3\mathsf{k}'(\rho_1)^2}{\rho_1}+\frac{3\mathsf{k}''(\rho_1)}{\rho_1}+\mathsf{k}'(\rho_1)^3+3\mathsf{k}'(\rho_1)\mathsf{k}''(\rho_1)+\mathsf{k}^{(3)}(\rho_1)\Bigr)
\frac{1}{\eta_1^3}.
\end{split}
\end{align}
Similarly, set $\eta_2\equiv\eta_2(j):=V_{\tau}'(\rho_2)=\frac{2}{\rho_2}(\tau_{\rho_2}-\tau_j)$ and $v_{\ell,2}:=V_{\tau}^{(\ell)}(\rho_2)$, and define
\begin{align}
\label{def of tilde frak a1}
\widetilde{\mathfrak{a}}_1&:=-\frac{v_{2,2}}{\eta_2^2}+\Bigl(\mathsf{k}'(\rho_2)+\frac{1}{\rho_2}\Bigr)\frac{1}{\eta_2}, 
\\
\begin{split}
\label{def of tilde frak a2}
\widetilde{\mathfrak{a}}_2&:=
\frac{3v_{2,2}^2}{\eta_2^4}
 -\Bigl(\frac{3v_{2,2}}{\rho_2}+3\mathsf{k}'(\rho_2)v_{2,2}+v_{3,2}\Bigr)\frac{1}{\eta_2^3}
 +\Bigl(\frac{2\mathsf{k}'(\rho_2)}{\rho_2}+\mathsf{k}'(\rho_2)^2+\mathsf{k}''(\rho_2)\Bigr)\frac{1}{\eta_2^2},
\end{split}
\\
\begin{split}
\label{def of tilde frak a3}
\widetilde{\mathfrak{a}}_3&:=
-\frac{15v_{2,2}^3}{\eta_2^6}+\Bigl(\frac{3v_{2,2}^2}{\rho_2}+3\mathsf{k}'(\rho_2)v_{2,2}^2+2v_{2,2}v_{3,2}\Bigr)\frac{5}{\eta_2^5}
 \\
 &\quad -\Bigl(\frac{12\mathsf{k}'(\rho_2)v_{2,2}}{\rho_2}+6\mathsf{k}'(\rho_2)^2v_{2,2}+6\mathsf{k}''(\rho_2)v_{2,2}+\frac{4v_{3,2}}{\rho_2}+4\mathsf{k}'(\rho_2)v_{3,2}+v_{4,2}\Bigr)\frac{1}{\eta_2^4}
 \\
 &\quad
 +\Bigl(\frac{3\mathsf{k}'(\rho_2)^2}{\rho_2}+\mathsf{k}'(\rho_2)^3+\frac{3\mathsf{k}''(\rho_2)}{\rho_2}+3\mathsf{k}'(\rho_2)\mathsf{k}''(\rho_2)+\mathsf{k}^{(3)}(\rho_2)\Bigr)\frac{1}{\eta_2^3}.
\end{split}
\end{align}

\begin{lemma}\label{lemma: g1+leq j leq j+1}
Uniformly for $g_{1,+}+1\leq j\leq j_{1,+}$,
\begin{align}
\begin{split}
h_{n,j}^{(\mathrm{an})}
&=
\frac{2\rho_1e^{\mathsf{k}(\rho_1)}e^{-nV_{\tau}(\rho_1)}}{\eta_1 n}
\bigg[1-\frac{\mathfrak{a}_1}{n}+\frac{\mathfrak{a}_2}{n^2}-\frac{\mathfrak{a}_3}{n^3}+\mathcal{O}(\frac{1}{\eta_1^8n^4})\bigg].
\end{split}    
\end{align}
\end{lemma}

\begin{proof}
Split the integral into
\begin{align*}
h_{n,j}^{(\mathrm{an})}
&=\int_{0}^{\rho_1-\delta_n}2re^{\mathsf{k}(r)}e^{-nV_{\tau}(r)}\,dr
+\int_{\rho_1-\delta_n}^{\rho_1}2re^{\mathsf{k}(r)}e^{-nV_{\tau}(r)}\,dr
+\int_{\rho_2}^{+\infty}2re^{\mathsf{k}(r)}e^{-nV_{\tau}(r)}\,dr. 
\end{align*}
For $g_{1,+}+1\leq j\leq j_{1,+}$, we have
\[
V_{\tau}'(r)=\frac{1}{r}(rq'(r)-2\tau_j)\leq \frac{1}{r}\Bigl(rq'(r)-\frac{2}{n}\frac{n\tau_{\rho_1}}{1-\frac{M}{\sqrt{n}}}\Bigr).
\]
Because $Q$ is subharmonic in $\mathbb C$, the function $r\mapsto rq'(r)$ is increasing; moreover, $\rho_1q'(\rho_1)=2\tau_{\rho_1}$. Hence
\begin{equation}
\label{def of Vtau r leq rho1ast}
V_{\tau}'(r)
\leq \frac{1}{r}\Bigl(\rho_1q'(\rho_1)-\frac{2\tau_{\rho_1}}{1-\frac{M}{\sqrt n}}\Bigr)
=-\frac{2\tau_{\rho_1}}{r}\frac{\frac{M}{\sqrt n}}{1-\frac{M}{\sqrt n}},
\qquad r\leq \rho_1.
\end{equation}
Thus $V_\tau$ is decreasing on $(0,\rho_1]$, and $V_\tau(r)\geq V_\tau(\rho_1-\delta_n)$ for $r\leq\rho_1-\delta_n$. By the mean-value theorem,
\begin{equation}
    V_{\tau}(\rho_1-\delta_n)=V_{\tau}(\rho_1)-V_{\tau}'(r_1^{\ast})\delta_n,
\end{equation}
for some $r_{1}^{\ast}\in[\rho_1-\delta_n,\rho_1]$. 
Consequently, the first integral satisfies
\begin{align*}
\int_{0}^{\rho_1-\delta_n}2re^{\mathsf{k}(r)}e^{-nV_{\tau}(r)}\,dr
&\leq e^{-nV_{\tau}(\rho_1)}\mathcal O(e^{-cM^2})
\end{align*}
for some $c>0$ independent of $n$ and $j$. Thus this contribution is negligible uniformly in $j$.
 
We next show that the third integral is also negligible. Since
\begin{align*}
V_{\tau}(\rho_2)-V_{\tau}(\rho_1)
&=V_{\tau_{\rho_1}}(\rho_2)-V_{\tau_{\rho_1}}(\rho_1)+2(\tau_{\rho_1}-\tau)\log \frac{\rho_2}{\rho_1},
\end{align*}
we have 
\begin{equation}
\label{def of Vtau minus Vtau rho1ast}
V_{\tau}(r)-V_{\tau}(\rho_1) 
=
V_{\tau}(r)-V_{\tau}(\rho_2)
+
V_{\tau_{\rho_1}}(\rho_2)-V_{\tau_{\rho_1}}(\rho_1)+2(\tau_{\rho_1}-\tau)\log \frac{\rho_2}{\rho_1}.
\end{equation}
For $g_{1,+}+1\leq j\leq j_{1,+}$, we have $\tau_{\rho_1}-\tau>-\frac{\epsilon}{1-\epsilon}\tau_{\rho_1}$,
and hence
\begin{align}
    \begin{split}
\label{def of Vtau1ast rho2 - Vtau1ast rho1}
V_{\tau_{\rho_1}}(\rho_2)-V_{\tau_{\rho_1}}(\rho_1)+2(\tau_{\rho_1}-\tau)\log\frac{\rho_2}{\rho_1}
&\geq q(\rho_2)-q(\rho_1)-\frac{2\tau_{\rho_1}}{1-\epsilon}\log\frac{\rho_2}{\rho_1}
\\
&=
2\log\frac{\rho_2}{\rho_1}
\Bigl(\frac{q(\rho_2)-q(\rho_1)}{2\log\frac{\rho_2}{\rho_1}}-\frac{\tau_{\rho_1}}{1-\epsilon}\Bigr)
\\
&=2\log\frac{\rho_2}{\rho_1}
\Bigl(\sigma_{\star}-\frac{\tau_{\rho_1}}{1-\epsilon}\Bigr)
=:c_{\epsilon}>0
    \end{split}
\end{align}
for some $c_\epsilon>0$ independent of $n$ and $j$. Moreover, \eqref{def of tau k ast epsilon inequality} implies
\begin{align*}
 V_{\tau}'(r)&\geq \frac{1}{r}\Bigl(\rho_2q'(\rho_2)-\frac{2\tau_{\rho_1}}{1-\epsilon}\Bigr)
 \\
 &=
 \frac{1}{r}\Bigl(\frac{(1+\epsilon)\rho_2q'(\rho_2)}{1+\epsilon}-\frac{2\tau_{\rho_1}}{1-\epsilon}\Bigr)
 >\frac{1}{r}\Bigl(\frac{2(1+\epsilon)\tau_{\rho_1}}{1-\epsilon}-\frac{2\tau_{\rho_1}}{1-\epsilon}\Bigr)=\frac{1}{r}\frac{2\epsilon\tau_{\rho_1}}{1-\epsilon}>0,
\end{align*}
Thus $V_\tau$ is increasing on $[\rho_2,\infty)$ throughout the stated range of $j$, so $V_\tau(r)\geq V_\tau(\rho_2)$ for $r\geq\rho_2$, with equality only at $r=\rho_2$. Combining this observation with \eqref{def of Vtau minus Vtau rho1ast} and \eqref{def of Vtau1ast rho2 - Vtau1ast rho1} gives
\begin{equation}
\label{def of inequality Vtau r Vtau rho1ast}
    V_{\tau}(r)-V_{\tau}(\rho_1)>V_{\tau}(r)- V_{\tau}(\rho_2)+c_{\epsilon}.
\end{equation}
Therefore we obtain
\begin{align*}
    \int_{\rho_2}^{+\infty}2re^{\mathsf{k}(r)}e^{-nV_{\tau}(r)}\,dr
    &\leq 
    C_s e^{-nV_{\tau}(\rho_1)}
    e^{-nc_\epsilon}\int_{\rho_2}^{+\infty}r^{2\alpha+1}e^{-n(V_{\tau}(r)-V_{\tau}(\rho_2))}\,dr
    \\
    &\leq 
   C_s e^{-nV_{\tau}(\rho_1)}e^{-nc_{\epsilon}}
    \int_{\rho_2}^{+\infty}r^{2\alpha+1}e^{-p(V_{\tau}(r)-V_{\tau}(\rho_2))}\,dr
    =e^{-nV_{\tau}(\rho_1)}\cdot \mathcal{O}(e^{-nc_{\epsilon}'}),
\end{align*}
for $n\geq p$ and some $c_\epsilon'>0$.  Here $p$ is a fixed positive
integer chosen as follows.  Assumption~\ref{Assumption 1} gives
$q(r)\geq2(1+\varepsilon_0)\log r$ for all sufficiently large $r$ and
some $\varepsilon_0>0$.  Choose $p$ so that $p\varepsilon_0>\alpha+1$.
Since $0\leq\tau\leq1$, the last integrand is then bounded at infinity
by a constant times $r^{2\alpha+1-2p\varepsilon_0}$, uniformly in $j$.
On compact intervals uniform integrability follows from continuity.
The displayed integral is therefore uniformly bounded.

It remains to evaluate the second integral. Since $\rho_1q'(\rho_1)=2\tau_{\rho_1}$, for $g_{1,+}+1\leq j\leq j_{1,+}$ we have
\begin{equation}
\label{def of rho1 hard tau relationship}
-\frac{\epsilon\tau_{\rho_1}}{1-\epsilon}
\leq
\tau_{\rho_1}-\frac{\tau_{\rho_1}}{1-\epsilon}
\leq
\frac{1}{2}\rho_1q'(\rho_1)-\tau_j
\leq \tau_{\rho_1}-\frac{\tau_{\rho_1}}{1-\frac{M}{\sqrt{n}}}
=
-\frac{M}{\sqrt{n}}\frac{\tau_{\rho_1}}{1-\frac{M}{\sqrt{n}}}<0.
\end{equation}
Consequently, for $\eta_1\equiv\eta_1(j):=\frac{2}{\rho_1}(\tau-\tau_{\rho_1})>0$, one has $\eta_1^{-k}=\mathcal{O}(n^{k/2}M^{-k})$, uniformly in $j$. With the change of variables $u=n\eta_1(\rho_1-r)$, the endpoint Laplace expansion becomes
\begin{align*}
&\quad 2e^{\mathsf{k}(\rho_1)}e^{-nV_{\tau}(\rho_1)}\int_{\rho_1-\delta_n}^{\rho_1}re^{\mathsf{k}(r)-\mathsf{k}(\rho_1)}e^{-n(V_{\tau}(r)-V_{\tau}(\rho_1))}\,dr
\\
&=
\frac{2\rho_1e^{\mathsf{k}(\rho_1)}e^{-nV_{\tau}(\rho_1)}}{\eta_1 n}\int_0^{\eta_1n\delta_n}
e^{-u}
\Bigl(1-\frac{a_1(\eta_1^{-1}u)}{n}+\frac{a_2(\eta_1^{-1}u)}{n^2}-\frac{a_3(\eta_1^{-1}u)}{n^3}+\mathcal{O}(\frac{a_4(\eta_1^{-1}u)}{n^4})\Bigr)
\,du,
\end{align*}
where 
\begin{align*}
a_1(u)&:=\frac{1}{2}v_{2,1}u^2+\Bigl(\mathsf{k}'(\rho_1)+\frac{1}{\rho_1}\Bigr)u,
\\
a_2(u)&:=v_{2,1}^2\frac{u^4}{8}+\Bigl(\frac{3v_{2,1}}{\rho_1}+3v_{2,1}\mathsf{k}'(\rho_1)+v_{3,1}\Bigr)\frac{u^3}{6}
+\Bigl(\frac{2\mathsf{k}'(\rho_1)}{\rho_1}+\mathsf{k}'(\rho_1)^2+\mathsf{k}''(\rho_1)\Bigr)\frac{u^2}{2}
\\
a_3(u)&:=v_{2,1}^3\frac{u^6}{48}
+\Bigl(\frac{3v_{2,1}^2}{\rho_1}+3v_{2,1}^2\mathsf{k}'(\rho_1)+2v_{2,1}v_{3,1}\Bigr)\frac{u^5}{24}
\\
&\quad
+\Bigl(v_{4,1}+\frac{4v_{3,1}}{\rho_1}+4v_{3,1}\mathsf{k}'(\rho_1)+6\mathsf{k}''(\rho_1)v_{2,1}+6\mathsf{k}'(\rho_1)^2v_{2,1}+\frac{12v_{2,1}}{\rho_1}\mathsf{k}'(\rho_1)\Bigr)\frac{u^4}{24}
\\
&\quad
+\Bigl(\frac{3\mathsf{k}'(\rho_1)^2}{\rho_1}+\frac{3\mathsf{k}''(\rho_1)}{\rho_1}+\mathsf{k}'(\rho_1)^3+3\mathsf{k}'(\rho_1)\mathsf{k}''(\rho_1)+\mathsf{k}^{(3)}(\rho_1)\Bigr)
\frac{u^3}{6},
\end{align*}
and $a_4(u)$ is a polynomial in $u^8,u^7,u^6,u^5,u^4$. Termwise integration against $e^{-u}$ gives
\begin{align*}
&\quad 2e^{\mathsf{k}(\rho_1)}e^{-nV_{\tau}(\rho_1)}\int_{\rho_1-\delta_n}^{\rho_1}re^{\mathsf{k}(r)-\mathsf{k}(\rho_1)}e^{-n(V_{\tau}(r)-V_{\tau}(\rho_1))}\,dr
\\
&=
\frac{2\rho_1e^{\mathsf{k}(\rho_1)}e^{-nV_{\tau}(\rho_1)}}{\eta_1 n}
\bigg[1-\frac{\mathfrak{a}_1}{n}+\frac{\mathfrak{a}_2}{n^2}-\frac{\mathfrak{a}_3}{n^3}+\mathcal{O}(\frac{1}{\eta_1^8n^4})\bigg]. 
\end{align*}
By \eqref{def of rho1 hard tau relationship}, the remainder is $\mathcal{O}(\eta_1^{-8}n^{-4})=\mathcal{O}(M^{-8})$, uniformly over the full index range. This proves the lemma.

\end{proof}

\begin{lemma}
\label{lemma: j1+ j jdiamond}
Uniformly for $j_{1,+}+1\leq j\leq \lfloor j_{\star}\rfloor$,
\begin{align}
\begin{split}
h_{n,j}^{(\mathrm{an})}&=
\frac{\rho_1^2e^{\mathsf{k}(\rho_1)}e^{-nV_{\tau}(\rho_1)}}{n(\tau_j-\tau_{\rho_1})}
\bigg(
1+\frac{(\tau_j-\tau_{\rho_1})e^{\mathsf{k}(\rho_2)-\mathsf{k}(\rho_1)}}{\tau_{\rho_2}-\tau_j}\Bigl( \frac{\rho_1}{\rho_2}\Bigr)^{2(j_{\star}-j-1)}
\bigg)
\\
&\quad
\times\bigg(
1-\frac{\mathfrak{a}_1}{n}+\mathcal{O}\Bigl(\frac{1}{n^2}+\frac{1}{n}\Bigl( \frac{\rho_1}{\rho_2}\Bigr)^{2(j_{\star}-j)}\Bigr)
\bigg),
\end{split}
\end{align}
where $\mathfrak{a}_1$ is given by \eqref{def of mathfraka1}. 
\end{lemma}

\begin{proof}
We use the decomposition
\begin{align}
\begin{split}
\label{def of decomposition middle hole}
h_{n,j}^{(\mathrm{an})}
&=
\int_{0}^{\rho_1-\delta_n}2re^{\mathsf{k}(r)}e^{-nV_{\tau}(r)}\,dr
+
\int_{\rho_1-\delta_n}^{\rho_1}2re^{\mathsf{k}(r)}e^{-nV_{\tau}(r)}\,dr
\\
&\quad
+\int_{\rho_2}^{\rho_2+\delta_n}2re^{\mathsf{k}(r)}e^{-nV_{\tau}(r)}\,dr
+\int_{\rho_2+\delta_n}^{+\infty}2re^{\mathsf{k}(r)}e^{-nV_{\tau}(r)}\,dr. 
\end{split}
\end{align}
We first estimate the two tails. As in \eqref{def of Vtau r leq rho1ast}, there exists $c_\epsilon>0$, independent of $n$ and $j$, such that $V_\tau'(r)\leq-c_\epsilon$ for $r\leq\rho_1$. Hence, for $r\leq\rho_1-\delta_n$, the mean-value theorem gives
\[
V_{\tau}(r)\geq V_{\tau}(\rho_1-\delta_n)=V_{\tau}(\rho_1)-V_{\tau}'(r_{1}^{\ast})\delta_n\geq  V_{\tau}(\rho_1)+c_{\epsilon}\delta_n,
\]
for some $r_1^{\ast}\in[\rho_1-\delta_n,\rho_1]$.
It follows that
\begin{align*}
\int_{0}^{\rho_1-\delta_n}2re^{\mathsf{k}(r)}e^{-nV_{\tau}(r)}\,dr
&\leq
e^{-nV_{\tau}(\rho_1)}\cdot
\mathcal{O}(
e^{-c_{\epsilon}n\delta_n})
=e^{-nV_{\tau}(\rho_1)}\cdot \mathcal{O}(e^{-c_{\epsilon}M\sqrt n}).
\end{align*}

For the fourth integral, note that
\begin{align}
\label{def of Vtau' rho2 ast neighbor}
V_{\tau}'(r)&=\frac{1}{r}(rq'(r)-2\tau)
\geq\frac{1}{r}(\rho_2q'(\rho_2)-2\sigma_{\star})>0,
\end{align}
so $V_\tau$ is increasing on $[\rho_2,\infty)$. For every $r\geq\rho_2+\delta_n$,
\begin{align*}
    V_{\tau}(r)\geq V_{\tau}(\rho_2+\delta_n)=V_{\tau}(\rho_2)+V_{\tau}'(r_2^{\ast})\delta_n,
\end{align*}
for some $r_2^{\ast}\in[\rho_2,\rho_2+\delta_n]$. Thus \eqref{def of inequality Vtau r Vtau rho1ast} yields
\begin{align*}
  &\quad e^{-nV_{\tau}(\rho_1)}\int_{\rho_2+\delta_n}^{+\infty}2re^{\mathsf{k}(r)}e^{-n(V_{\tau}(r)-V_{\tau}(\rho_1))}\,dr
  \\
    &
    \leq
    C_s e^{-nV_{\tau}(\rho_1)}
    \int_{\rho_2+\delta_n}^{+\infty}re^{-(V_{\tau}(r)-V_{\tau}(\rho_1))}\,dr
    \\
    &\leq C_s e^{-nV_{\tau}(\rho_1)}e^{-(n-1)V_{\tau}'(r_2^{\ast})\delta_n}\int_{\rho_2+\delta_n}^{+\infty}r e^{-(V_{\tau}(r)-V_{\tau}(\rho_1))}\,dr
    =e^{-nV_{\tau}(\rho_1)}\cdot \mathcal{O}(e^{-c_{\epsilon}'M\sqrt n}),
\end{align*}
for some $c_\epsilon'>0$, uniformly in $j$, where \eqref{def of Vtau' rho2 ast neighbor} was used.

It remains to evaluate the two endpoint contributions. In the present index range,
\[
0\leq\sigma_{\star}-\tau_j\leq \sigma_{\star}-\tau_{j_{1,+}}< \sigma_{\star}-\frac{\tau_{\rho_1}}{1-\epsilon}.
\]
Recall that $\eta_1\equiv\eta_1(j):=\frac{2}{\rho_1}(\tau_j-\tau_{\rho_1})>0$. In contrast with Lemma~\ref{lemma: g1+leq j leq j+1}, $\eta_1$ is bounded away from zero here. It therefore suffices to expand the second integral through relative order $n^{-1}$, with an error uniform in $j$:
\begin{align*}
2e^{\mathsf{k}(\rho_1)}e^{-nV_{\tau}(\rho_1)}
\int_{\rho_1-\delta_n}^{\rho_1} r e^{\mathsf{k}(r)-\mathsf{k}(\rho_1)}e^{-n(V_{\tau}(r)-V_{\tau}(\rho_1))}\,dr
=
\frac{2\rho_1e^{\mathsf{k}(\rho_1)}e^{-nV_{\tau}(\rho_1)}}{\eta_1 n}
\bigg[1-\frac{\mathfrak{a}_1}{n}+\mathcal{O}(\frac{\mathfrak{a}_2}{n^2})\bigg], 
\end{align*}
Here $\mathfrak a_2=\mathcal O(1)$ uniformly in this range, so the displayed error is $\mathcal O(n^{-2})$. For the third integral, use
\begin{align*}
    V_{\tau_j}(r)-V_{\tau_j}(\rho_1)
&=V_{\tau_j}(r)-V_{\tau_j}(\rho_2)-2(\sigma_{\star}-\tau_j)\log\frac{\rho_1}{\rho_2},
\end{align*}
and by \eqref{def of sigma ast interrelashion epsilon}, 
\[
V_{\tau}'(\rho_2)\geq\frac{1}{\rho_2}(\rho_2q'(\rho_2)-2\tau)
=\frac{2}{\rho_2}(\tau_{\rho_2}-\tau)>\frac{2}{\rho_2}\frac{\epsilon}{1+\epsilon}\tau_{\rho_2}>0,
\]
we have 
\begin{align*}
&\quad 2\int_{\rho_2}^{\rho_2+\delta_n}re^{\mathsf{k}(r)}e^{-nV_{\tau}(r)}\,dr 
\\
&=
\frac{2\rho_2e^{\mathsf{k}(\rho_2)}e^{-nV_{\tau}(\rho_1)}}{\eta_2n}
\Bigl(\frac{\rho_1}{\rho_2}\Bigr)^{2(j_{\star}-j)}
\\
&\quad\times
\int_{0}^{n\delta_n\eta_2}
e^{-u}
\Bigl( 
1+\frac{\widetilde{a}_{1}(\eta_2^{-1}u)}{n}+\frac{\widetilde{a}_{2}(\eta_2^{-1}u)}{n^2}+\frac{\widetilde{a}_{3}(\eta_2^{-1}u)}{n^3}+\mathcal{O}(\frac{\widetilde{a}_{4}(\eta_2^{-1}u)}{n^4})
\Bigr)\,du, 
\end{align*}
where for $\eta_2\equiv\eta_2(j):=V_{\tau}'(\rho_2)=\frac{2}{\rho_2}(\tau_{\rho_2}-\tau_j)$ and $v_{l,2}:=V_{\tau}^{(l)}(\rho_2)$, 
\begin{align*}
 \widetilde{a}_{1}(u)&:=-\frac{1}{2}v_{2,2}u^2+\Bigl(\mathsf{k}'(\rho_2)+\frac{1}{\rho_2}\Bigr)u, 
 \\
 \widetilde{a}_{2}(u)&:=\frac{1}{8}v_{2,2}^2u^4
 -\Bigl(\frac{3v_{2,2}}{\rho_2}+3\mathsf{k}'(\rho_2)v_{2,2}+v_{3,2}\Bigr)\frac{u^3}{6}
 +\Bigl(\frac{2\mathsf{k}'(\rho_2)}{\rho_2}+\mathsf{k}'(\rho_2)^2+\mathsf{k}''(\rho_2)\Bigr)\frac{u^2}{2}
 \\
 \widetilde{a}_{3}(u)&:=-\frac{v_{2,2}^3}{48}u^6+\Bigl(\frac{3v_{2,2}^2}{\rho_2}+3\mathsf{k}'(\rho_2)v_{2,2}^2+2v_{2,2}v_{3,2}\Bigr)\frac{u^5}{24}
 \\
 &\quad -\Bigl(\frac{12\mathsf{k}'(\rho_2)v_{2,2}}{\rho_2}+6\mathsf{k}'(\rho_2)^2v_{2,2}+6\mathsf{k}''(\rho_2)v_{2,2}+\frac{4v_{3,2}}{\rho_2}+4\mathsf{k}'(\rho_2)v_{3,2}+v_{4,2}\Bigr)\frac{u^4}{24}
 \\
 &\quad
 +\Bigl(\frac{3\mathsf{k}'(\rho_2)^2}{\rho_2}+\mathsf{k}'(\rho_2)^3+\frac{3\mathsf{k}''(\rho_2)}{\rho_2}+3\mathsf{k}'(\rho_2)\mathsf{k}''(\rho_2)+\mathsf{k}^{(3)}(\rho_2)\Bigr)\frac{u^3}{6}, 
\end{align*}
and $\widetilde a_4(u)$ is a polynomial in $u^8,u^7,u^6,u^5,u^4$. Termwise integration gives
\[
2\int_{\rho_2}^{\rho_2+\delta_n}re^{\mathsf{k}(r)}e^{-nV_{\tau}(r)}\,dr
=
\frac{2\rho_2e^{\mathsf{k}(\rho_2)}e^{-nV_{\tau}(\rho_1)}}{\eta_2n}
\Bigl(\frac{\rho_1}{\rho_2}\Bigr)^{2(j_{\star}-j)}
\bigg[
1+\frac{\widetilde{\mathfrak{a}}_1}{n}+\frac{\widetilde{\mathfrak{a}}_2}{n^2}+\frac{\widetilde{\mathfrak{a}}_3}{n^3}+\mathcal{O}(\frac{1}{\eta_2^8n^4})
\bigg].
\]
Combining the endpoint expansions with the two tail estimates, we obtain
\begin{align*}
h_{n,j}^{(\mathrm{an})}&=
\frac{2\rho_1e^{\mathsf{k}(\rho_1)}e^{-nV_{\tau}(\rho_1)}}{n\eta_1}
\bigg(
1+\frac{\eta_1e^{\mathsf{k}(\rho_2)-\mathsf{k}(\rho_1)}}{\eta_2}\Bigl( \frac{\rho_1}{\rho_2}\Bigr)^{2(j_{\star}-j)-1}
\bigg)
\bigg(
1-\frac{\mathfrak{a}_1}{n}+\mathcal{O}(\frac{1}{n^2}+\frac{1}{n}\Bigl( \frac{\rho_1}{\rho_2}\Bigr)^{2(j_{\star}-j)})
\bigg).
\end{align*}
This completes the proof. 
\end{proof}

\begin{lemma}\label{lemma: jstar leq j leq j2-}
Uniformly for $\lfloor j_{\star}\rfloor+1\leq j\leq j_{2,-}-1$,
\begin{align}
\begin{split}
h_{n,j}^{(\mathrm{an})}&= 
 \frac{\rho_2^2e^{\mathsf{k}(\rho_2)}e^{-nV_{\tau}(\rho_2)}}{n}
\bigg\{
\frac{1}{\tau_{\rho_2}-\tau}+\frac{e^{\mathsf{k}(\rho_1)-\mathsf{k}(\rho_2)}}{\tau-\tau_{\rho_1}}\Bigl( \frac{\rho_1}{\rho_2}\Bigr)^{2(j-j_{\star}+1)}
\bigg\}
\\
&\quad
\times\bigg\{
1+\frac{\widetilde{\mathfrak{a}}_1}{n}+\mathcal{O}(\frac{1}{n^2}+\frac{1}{n}\Bigl( \frac{\rho_1}{\rho_2}\Bigr)^{2(j-j_{\star})})
\bigg\},
\end{split}
\end{align}
where $\widetilde{\mathfrak a}_1$ is given by \eqref{def of tilde frak a1}.
\end{lemma}

\begin{proof}
We again use \eqref{def of decomposition middle hole}. The tail bounds follow from the proof of Lemma~\ref{lemma: j1+ j jdiamond} after interchanging the roles of $\rho_1$ and $\rho_2$. By \eqref{def of sigma ast interrelashion epsilon}, throughout the present range
\[
0<\tau_j-\sigma_{\star}< \frac{\tau_{\rho_2}}{1+\epsilon}-\sigma_{\star}, 
\]
and 
\[
V_{\tau_j}(r)-V_{\tau_j}(\rho_{2})=
V_{\tau_j}(r)-V_{\tau_j}(\rho_{1})-2(\tau_j-\sigma_{\star})\log\frac{\rho_{1}}{\rho_{2}}. 
\]
The endpoint expansion at $\rho_1$, obtained exactly as in Lemma~\ref{lemma: j1+ j jdiamond}, is
\begin{align*}
\int_{\rho_1-\delta_n}^{\rho_1}2re^{\mathsf{k}(r)}e^{-nV_{\tau}(r)}\,dr 
=
\frac{2\rho_1e^{\mathsf{k}(\rho_1)}e^{-nV_{\tau}(\rho_2)}}{\eta_1n}
\Bigl(\frac{\rho_1}{\rho_2}\Bigr)^{2(j-j_{\star})}
\Bigl[1-\frac{\mathfrak{a}_1}{n}+\mathcal{O}(\frac{\mathfrak{a}_2}{n^2})\Bigr]
\end{align*}
At $\rho_2$, the corresponding expansion is
\begin{align*}
\int_{\rho_2}^{\rho_2+\delta_n}2re^{\mathsf{k}(r)}e^{-nV_{\tau}(r)}\,dr   
=
\frac{2\rho_2e^{\mathsf{k}(\rho_2)}e^{-nV_{\tau}(\rho_2)}}{\eta_2 n}\bigg[1+\frac{\widetilde{\mathfrak{a}}_1}{n}+\mathcal{O}(\frac{\mathfrak{a}_2}{n^2})\bigg]. 
\end{align*}
Adding the two endpoint contributions and absorbing the exponentially small tails yields
\begin{align*}
h_{n,j}^{(\mathrm{an})}
&=
\frac{2\rho_2e^{\mathsf{k}(\rho_2)}e^{-nV_{\tau}(\rho_2)}}{n\eta_2}
\bigg(
1+\frac{\eta_2e^{\mathsf{k}(\rho_1)-\mathsf{k}(\rho_2)}}{\eta_1}\Bigl( \frac{\rho_1}{\rho_2}\Bigr)^{2(j-j_{\star})+1}
\bigg)
\bigg(
1+\frac{\widetilde{\mathfrak{a}}_1}{n}+\mathcal{O}\Bigl(\frac{1}{n^2}+\frac{1}{n}\Bigl( \frac{\rho_1}{\rho_2}\Bigr)^{2(j-j_{\star})}\Bigr)
\bigg).
\end{align*}
This completes the proof. 

\end{proof}

\begin{lemma}\label{lemma:j2- j g2--1}
Let $\eta_2\equiv\eta_2(j):=V_{\tau}'(\rho_2)=\frac{2}{\rho_2}(\tau_{\rho_2}-\tau_j)$ and $v_{\ell,2}:=V_{\tau}^{(\ell)}(\rho_2)$. Uniformly for $j_{2,-}\leq j\leq g_{2,-}-1$,
\begin{align}
h_{n,j}^{(\mathrm{an})}
    &=\frac{2\rho_2e^{\mathsf{k}(\rho_2)}e^{-nV_{\tau}(\rho_2)}}{\eta_2 n}
\bigg[1+\frac{\widetilde{\mathfrak{a}}_1}{n}+\frac{\widetilde{\mathfrak{a}}_2}{n^2}+\frac{\widetilde{\mathfrak{a}}_3}{n^3}+\mathcal{O}(\frac{1}{M^8})\bigg],
\end{align}
where $\widetilde{\mathfrak{a}}_k$ for $k=1,2,3$ are given by \eqref{def of tilde frak a1}, \eqref{def of tilde frak a2}, and \eqref{def of tilde frak a3}. 
\end{lemma}

\begin{proof}
The argument is the $\rho_2$-endpoint counterpart of Lemma~\ref{lemma: g1+leq j leq j+1}; we include the details to make the uniformity explicit. Decompose
\[
h_{n,j}^{(\mathrm{an})}
=
\int_0^{\rho_1}2r e^{\mathsf{k}(r)}e^{-nV_{\tau}(r)}\,dr
+
\int_{\rho_2}^{\rho_2+\delta_n}2r e^{\mathsf{k}(r)}e^{-nV_{\tau}(r)}\,dr
+
\int_{\rho_2+\delta_n}^{+\infty}2r e^{\mathsf{k}(r)}e^{-nV_{\tau}(r)}\,dr. 
\]
For $j_{2,-}\leq j\leq g_{2,-}-1$, one has
$\tau=j/n<\tau_{\rho_2}/(1+M/\sqrt n)<\tau_{\rho_2}$.
Global subharmonicity makes $t\mapsto q(e^t)$ convex.  Since
$\rho_2q'(\rho_2)-2\tau>0$, this implies that $V_\tau$ is increasing on
$[\rho_2,\infty)$; no differentiability of $q$ away from the droplet is
needed for this conclusion.  On the smooth neighborhood
$[\rho_2,\rho_2+\delta_n]$, the monotonicity of $rq'(r)$ gives
\[
V_\tau'(r)
=\frac{rq'(r)-2\tau}{r}
\geq\frac{2\tau_{\rho_2}}{r}
\frac{M/\sqrt n}{1+M/\sqrt n}
\geq c\frac{M}{\sqrt n},
\]
for all sufficiently large $n$, with $c>0$ independent of $j$.
As $\delta_n=M/\sqrt n$, it follows that
\[
V_\tau(r)-V_\tau(\rho_2)
\geq V_\tau(\rho_2+\delta_n)-V_\tau(\rho_2)
\geq c\frac{M^2}{n},\qquad r\geq\rho_2+\delta_n.
\]
To control the integral over this unbounded interval, choose $\varepsilon_0>0$
and $r_0>\max\{1,\rho_2\}$ such that
$q(r)\geq2(1+\varepsilon_0)\log r$ for $r\geq r_0$.
This is possible by Assumption~\ref{Assumption 1}.
Choose a fixed positive integer $p$ such that
$p\varepsilon_0>\alpha+1$.  Since $0\leq\tau\leq1$ and $\lambda$ is
bounded, the estimate $V_\tau(r)\geq2\varepsilon_0\log r$ for
$r\geq r_0$ implies
\[
\sup_{0\leq\tau\leq1}
\int_{\rho_2}^{\infty}
2r e^{\mathsf k(r)}
e^{-p(V_\tau(r)-V_\tau(\rho_2))}\,dr<\infty.
\]
Indeed, the tail is bounded by a constant times
$\int_{r_0}^{\infty}r^{2\alpha+1-2p\varepsilon_0}\,dr$, and the
remaining compact interval is uniformly integrable.
Consequently, for $n\geq2p$,
\begin{align*}
&\int_{\rho_2+\delta_n}^{\infty}
2r e^{\mathsf k(r)}e^{-nV_\tau(r)}\,dr
\\
&\quad\leq e^{-nV_\tau(\rho_2)}
e^{-c(n-p)M^2/n}
\int_{\rho_2}^{\infty}
2r e^{\mathsf k(r)}
e^{-p(V_\tau(r)-V_\tau(\rho_2))}\,dr
=e^{-nV_\tau(\rho_2)}\mathcal O(e^{-c'M^2}),
\end{align*}
with $c'>0$ independent of $j$ and $n$.  This proves that the third
integral is negligible uniformly in the stated index range.

We now treat the first integral. Rewrite
\[
V_{\tau}(r)-V_{\tau}(\rho_2)
=
V_{\tau}(r)-V_{\tau}(\rho_1)
-\Bigl( 
V_{\tau_{\rho_2}}(\rho_{2})-V_{\tau_{\rho_2}}(\rho_1)+2(\tau_{\rho_2}-\tau)\log\frac{\rho_2}{\rho_1}
\Bigr). 
\]
For $0\leq r\leq \rho_1$ and $j_{2,-}\leq j\leq g_{2,-}-1$, we have 
\begin{align*}
V_{\tau}'(r)\leq \frac{1}{r}\Bigl(\rho_1q'(\rho_1)-\frac{2\tau_{\rho_2}}{1+\epsilon}\Bigr)
<\frac{1}{r}\Bigl(\frac{2(1-\epsilon)}{1+\epsilon}\tau_{\rho_2}-\frac{2\tau_{\rho_2}}{1+\epsilon}\Bigr)
=-\frac{1}{r}\frac{2\epsilon\tau_{\rho_2}}{1+\epsilon}.
\end{align*}
Thus $V_\tau$ is decreasing on $[0,\rho_1]$ for the indicated indices, and hence $V_\tau(r)-V_\tau(\rho_1)\geq0$ there. On the other hand,
\[
\tau_{\rho_2}-\tau\leq \tau_{\rho_2}-\frac{j_{2,-}}{n}<\tau_{\rho_2}-\frac{\tau_{\rho_2}}{1+\epsilon}
=\frac{\epsilon}{1+\epsilon}\tau_{\rho_2},
\]
we have
\begin{align*}
V_{\tau_{\rho_2}}(\rho_{2})-V_{\tau_{\rho_2}}(\rho_1)+2(\tau_{\rho_2}-\tau)\log\frac{\rho_2}{\rho_1}
&<q(\rho_2)-q(\rho_1)
-2\tau_{\rho_2}\log\frac{\rho_2}{\rho_1}+\frac{2\epsilon\tau_{\rho_2}}{1+\epsilon}\log\frac{\rho_2}{\rho_1}
\\
&=q(\rho_2)-q(\rho_1)-\frac{2\tau_{\rho_2}}{1+\epsilon}\log\frac{\rho_2}{\rho_1}<-c_{\epsilon},
\end{align*}
for some $c_\epsilon>0$, by \eqref{def of sigma ast interrelashion epsilon}. Therefore, uniformly for $0\leq r\leq\rho_1$ and $j_{2,-}\leq j\leq g_{2,-}-1$,
\[
V_{\tau}(r)-V_{\tau}(\rho_2)>V_{\tau}(r)-V_{\tau}(\rho_1)+c_{\epsilon}>0. 
\]
Then we obtain 
\begin{align}
    \begin{split}
\int_0^{\rho_1}2re^{\mathsf{k}(r)}e^{-nV_{\tau}(r)}\,dr
&\leq C_s e^{-nV_{\tau}(\rho_2)}e^{-nc_{\epsilon}}\int_{0}^{\rho_1}r^{2\alpha+1}e^{-(V_{\tau}(r)-V_{\tau}(\rho_1))} \,dr
\\
&=e^{-nV_{\tau}(\rho_2)}\cdot\mathcal{O}(e^{-c_{\epsilon}n}). 
    \end{split}
\end{align}

It remains to evaluate the main contribution from the $\rho_2$ endpoint. Repeating the endpoint Laplace expansion from Lemmas~\ref{lemma: g1+leq j leq j+1} and \ref{lemma: j1+ j jdiamond} gives
\begin{align*}
\int_{\rho_2}^{\rho_2+\delta_n}2re^{\mathsf{k}(r)}e^{-nV_{\tau}(r)}\,dr
=
\frac{2\rho_2e^{\mathsf{k}(\rho_2)}e^{-nV_{\tau}(\rho_2)}}{\eta_2 n}
\bigg[1+\frac{\widetilde{\mathfrak{a}}_1}{n}+\frac{\widetilde{\mathfrak{a}}_2}{n^2}+\frac{\widetilde{\mathfrak{a}}_3}{n^3}+\mathcal{O}(\frac{1}{\eta_2^8n^4})\bigg]. 
\end{align*}
Since $\eta_2^{-8}n^{-4}=\mathcal O(M^{-8})$ uniformly in the stated range, the displayed formula proves the lemma.

\end{proof}

\begin{lemma}\label{lemma:g2- j g2+}
Let $\xi_{j,2}:=\sqrt{nd_2}(\rho_2-r_{\tau})$. 
Uniformly for $g_{2,-}\leq j\leq g_{2,+}$, there exists $c>0$, independent of $n$ and $j$, such that
\begin{equation}
   h_{n,j}^{(\mathrm{an})}=
    (h_{n,j}-h_{\#}(\xi_{j,2}))\cdot\bigl(1+\mathcal{O}(e^{-cM^2})\bigr).
\end{equation}
\end{lemma}

\begin{proof}
We use the decomposition
\begin{align*}
h_{n,j}^{(\mathrm{an})}&=\int_{0}^{\rho_1}2re^{\mathsf{k}(r)}e^{-nV_{\tau}(r)}\,dr+\int_{\rho_2}^{r_{\tau}+\delta_n}2re^{\mathsf{k}(r)}e^{-nV_{\tau}(r)}\,dr
+\int_{r_{\tau}+\delta_n}^{+\infty}2re^{\mathsf{k}(r)}e^{-nV_{\tau}(r)}\,dr. 
\end{align*}
As in Lemmas~\ref{lemma: 0 j j1-} and \ref{lemma:g1- j g1+}, the first and third integrals are $e^{-nV_\tau(r_\tau)}\mathcal O(e^{-cM^2})$ for some uniform $c>0$. The middle integral provides the main contribution. Indeed, for $g_{2,-}\leq j\leq g_{2,+}$ there are constants $C_1,C_2>0$, independent of $n$ and $j$, such that
$-C_1M/\sqrt n\leq\rho_2-r_\tau\leq C_2M/\sqrt n$. The same truncated Laplace expansion as in Lemma~\ref{lemma:g1- j g1+} therefore gives
\[
\int_{\rho_2}^{r_{\tau}+\delta_n}2re^{\mathsf{k}(r)}e^{-nV_{\tau}(r)}\,dr
=
(h_{n,j}-h_{\#}(\xi_{j,2}))\cdot\bigl(1+\mathcal{O}(e^{-cM^2})\bigr),
\]
for some $c>0$. 
This completes the proof. 
\end{proof}

\begin{lemma}\label{lemma:g2,++1leq j leq j2,+}
Uniformly for $g_{2,+}+1\leq j\leq j_{2,+}$, there exists $c>0$, independent of $n$ and $j$, such that
\begin{align}
\begin{split}
h_{n,j}^{(\mathrm{an})}
&=h_{n,j}\cdot(1+\mathcal{O}(e^{-c M^2})).
\end{split}
\end{align}
\end{lemma}

\begin{lemma}\label{lemma:j2,++1leq j leq n-1}
Uniformly for $j_{2,+}+1\leq j\leq n-1$, there exists $c>0$, independent of $n$ and $j$, such that
\begin{align}
\begin{split}
h_{n,j}^{(\mathrm{an})}
&=h_{n,j}\cdot(1+\mathcal{O}(e^{-cM^2})).
\end{split}
\end{align}
\end{lemma}

\begin{proof}[Proof of Lemmas~\ref{lemma:g2,++1leq j leq j2,+} and
\ref{lemma:j2,++1leq j leq n-1}]
The tail arguments in Lemmas~\ref{lemma: 0 j j1-} and
\ref{lemma: j1- j g1-} apply with $\rho_1$ replaced by $\rho_2$.  In both
ranges the saddle remains at least $M/\sqrt n$ from the excluded interval,
so the same estimates give the stated uniform exponential errors.
\end{proof}

\section{Proofs of the hole-probability asymptotics}
\label{section:hole-proofs}
\label{section:proof of hole probability part}
In this section we prove Theorems~\ref{theorem:hole probability of annulus case}, \ref{theorem:hole probability of disk complement case}, and \ref{theorem:hole probability of centered disk case}. We give the annulus case in detail; the disk-complement and centered-disk cases follow by retaining only the relevant endpoint contributions.

For the annulus, introduce the logarithmic ratio
\begin{equation}\label{def of calPns C}
    \log\mathcal{P}_{n,s\lambda,\alpha}^{(\mathrm{an})}:=
    \sum_{j=0}^{n-1}\log\frac{h_{n,j}^{(\mathrm{an})}}{h_{n,j}}
\end{equation}
and decompose it as
\[
\log\mathcal{P}_{n,s\lambda,\alpha}^{(\mathrm{an})}
=S_1^{(\mathrm{an})}+S_2^{(\mathrm{an})}+S_3^{(\mathrm{an})}+S_4^{(\mathrm{an})}+S_5^{(\mathrm{an})},
\]
where 
\begin{align*}
S_1^{(\mathrm{an})}&:=\sum_{j=0}^{j_{1,-}-1}\log\frac{h_{n,j}^{(\mathrm{an})}}{h_{n,j}}, \quad   
S_2^{(\mathrm{an})}:=\sum_{j=j_{1,-}}^{j_{1,+}}\log\frac{h_{n,j}^{(\mathrm{an})}}{h_{n,j}}, \quad
S_3^{(\mathrm{an})}:=\sum_{j=j_{1,+}+1}^{j_{2,-}-1}\log\frac{h_{n,j}^{(\mathrm{an})}}{h_{n,j}},
\\
S_4^{(\mathrm{an})}&:=\sum_{j=j_{2,-}}^{j_{2,+}}\log\frac{h_{n,j}^{(\mathrm{an})}}{h_{n,j}}, \quad
S_5^{(\mathrm{an})}:=\sum_{j=j_{2,+}+1}^{n-1}\log\frac{h_{n,j}^{(\mathrm{an})}}{h_{n,j}}.
\end{align*}
Following \cite[Eq.~(3.40)]{C2021}, split the two edge sums further as
\begin{align*}
S_2^{(\mathrm{an})}&:=S_2^{(1,\mathrm{an})}+S_2^{(2,\mathrm{an})}+S_2^{(3,\mathrm{an})},\qquad
S_4^{(\mathrm{an})}:=S_4^{(1,\mathrm{an})}+S_4^{(2,\mathrm{an})}+S_4^{(3,\mathrm{an})},
\end{align*}
where 
\begin{align*}
S_{2}^{(3,\mathrm{an})}&:=\sum_{j=j_{1,-}}^{g_{1,-}-1}\log\frac{h_{n,j}^{(\mathrm{an})}}{h_{n,j}},\quad
S_{2}^{(2,\mathrm{an})}:=\sum_{j=g_{1,-}}^{g_{1,+}}\log\frac{h_{n,j}^{(\mathrm{an})}}{h_{n,j}},\quad
S_{2}^{(1,\mathrm{an})}:=\sum_{j=g_{1,+}+1}^{j_{1,+}}\log\frac{h_{n,j}^{(\mathrm{an})}}{h_{n,j}},\\
S_4^{(1,\mathrm{an})}&:=\sum_{j=j_{2,-}}^{g_{2,-}-1}\log\frac{h_{n,j}^{(\mathrm{an})}}{h_{n,j}},\quad
S_4^{(2,\mathrm{an})}:=\sum_{j=g_{2,-}}^{g_{2,+}}\log\frac{h_{n,j}^{(\mathrm{an})}}{h_{n,j}},\quad
S_4^{(3,\mathrm{an})}:=\sum_{j=g_{2,+}+1}^{j_{2,+}}\log\frac{h_{n,j}^{(\mathrm{an})}}{h_{n,j}}.
\end{align*}
The next lemma collects the exponentially small contributions away from the two $M/\sqrt n$ transition windows.
\begin{lemma}
\label{lemma: S1245 an error}
There exist constants $c_1,c_2,c_3,c_4>0$, independent of $n$, such that
\begin{align*}
S_1^{(\mathrm{an})}=\mathcal{O}(e^{-c_1(\log n)^2}),\quad
S_2^{(3,\mathrm{an})}=\mathcal{O}(e^{-c_2M^2}),\quad
S_4^{(3,\mathrm{an})}=\mathcal{O}(e^{-c_3M^2}),\quad
S_5^{(\mathrm{an})}=\mathcal{O}(e^{-c_4M^2}).
\end{align*}
\end{lemma}
\begin{proof}
Apply Lemmas~\ref{lemma: 0 j j1-}, \ref{lemma: j1- j g1-}, \ref{lemma:g2,++1leq j leq j2,+}, and \ref{lemma:j2,++1leq j leq n-1} term by term. The number of summands in each range is at most $n$, and the polynomial factor $n$ is absorbed by decreasing the corresponding positive constant in the exponential. This gives all four estimates.
\end{proof}

It remains to analyze $S_2^{(1,\mathrm{an})}$, $S_2^{(2,\mathrm{an})}$, $S_3^{(\mathrm{an})}$, $S_4^{(1,\mathrm{an})}$, and $S_4^{(2,\mathrm{an})}$. We use the following Euler--Maclaurin formula repeatedly.
\begin{lemma}[{\cite[Lemma~3.4]{C2021}}]
\label{lemma:Riemann sum NEW}
Let $A,a_0,B,b_0$ be bounded functions of $n\in\mathbb N$ such that
\begin{align*}
& a_{n} := An + a_{0} \qquad \mbox{ and } \qquad b_{n} := Bn + b_{0}
\end{align*}
are integers. Assume that $B-A$ is bounded away from zero. Let $f$, independent of $n$, belong to $C^4(I_n)$ for every $n$, where $I_n$ is any interval containing $A,B,a_n/n$, and $b_n/n$. Then
\begin{align}
&  \sum_{j=a_{n}}^{b_{n}}f(\tfrac{j}{n}) = n \int_{A}^{B}f(x)dx + \frac{(1-2a_{0})f(A)+(1+2b_{0})f(B)}{2}  \nonumber \\
& + \frac{(-1+6a_{0}-6a_{0}^{2})f'(A)+(1+6b_{0}+6b_{0}^{2})f'(B)}{12n} \nonumber \\
& + \frac{(-a_{0}+3a_{0}^{2}-2a_{0}^{3})f''(A)+(b_{0}+3b_{0}^{2}+2b_{0}^{3})f''(B)}{12n^{2}} \nonumber \\
& + \bigO \bigg( \frac{\mathfrak{m}_{A,n}(f''')+\mathfrak{m}_{B,n}(f''')}{n^{3}} + \sum_{j=a_{n}}^{b_{n}-1} \frac{\mathfrak{m}_{j,n}(f'''')}{n^{4}} \bigg), \label{sum f asymp gap NEW}
\end{align}
where, for a continuous function $g$ on the relevant intervals,
\begin{align*}
\mathfrak{m}_{A,n}(g) := \max_{x \in [\min\{\frac{a_{n}}{n},A\},\max\{\frac{a_{n}}{n},A\}]}|g(x)|, \quad \mathfrak{m}_{B,n}(g) := \max_{x \in [\min\{\frac{b_{n}}{n},B\},\max\{\frac{b_{n}}{n},B\}]}|g(x)|,
\end{align*}
and for $j \in \{a_{n},\ldots,b_{n}-1\}$, $\mathfrak{m}_{j,n}(g) := \max_{x \in [\frac{j}{n},\frac{j+1}{n}]}|g(x)|$.
\end{lemma}

Applying Lemma~\ref{lemma:Riemann sum NEW} to the bulk norm asymptotics gives the following summation formula.
\begin{lemma}\label{lemma:general sum}
Define
\begin{align}
\label{def of Lammda us}
\Lambda(u;s)&:=s^2\lambda'(u)^2+s\lambda''(u)+3s\frac{\lambda'(u)}{u}-s\lambda'(u)\frac{\partial_u\Delta Q(u)}{\Delta Q(u)},\\
\label{def of L us}
L(u;\alpha,s)&:=\alpha^2\ell'(u)^2+\alpha \ell''(u)+3\alpha\frac{\ell'(u)}{u}-\alpha \ell'(u)\frac{\partial_u\Delta Q(u)}{\Delta Q(u)}+2\alpha s\lambda'(u)\ell'(u),
\end{align}
Fix $a_*>0$.  Then, uniformly for bounded $A,a_0,B,b_0$ satisfying the
hypotheses of Lemma~\ref{lemma:Riemann sum NEW}, with $a_*\leq A<B\leq1$
and $0\leq An+a_0\leq Bn+b_0\leq n-1$,
\begin{align}
\begin{split}
\sum_{j=An+a_0}^{Bn+b_0}\log h_{n,j} 
&=
H_{1}n^2+H_2n\log n+H_3n+H_5\log n+H_6+\mathcal{O}(n^{-1}),
\end{split}    
\end{align}
where $H_k\equiv H_k(A,a_0,B,b_0)$, for $k\in\{1,2,3,5,6\}$, are defined by
\begin{align}
\begin{split}
H_1&:=    
-
2\int_{r_{A}}^{r_{B}}\Bigl( 
q(u)-uq'(u)\log u
\Bigr)u\Delta Q(u)\,du
\end{split} 
\\
\begin{split}
H_2&:=
-\frac{1}{2}\Bigl(B-A\Bigr)
\end{split}
\\
\begin{split}
H_3&:=    
\Bigl(B-A\Bigr)\log\sqrt{2\pi}
-\frac{1}{2}\Bigl((1-2a_0)V_{A}(r_{A})+(1+2b_0)V_{B}(r_{B})\Bigr)
\\
&\quad
+\int_{r_{A}}^{r_{B}}2\log u\cdot u\Delta Q(u)\,du
+\int_{r_{A}}^{r_{B}}2\mathsf{k}(u)u\Delta Q(u)\,du
\\
&\quad
-\int_{r_{A}}^{r_{B}}u\Delta Q(u)\log \Delta Q(u)\,du
\end{split} 
\\
\begin{split}
H_5&:=    
-\frac{1}{2}(b_0-a_0+1)
\end{split} 
\\
\begin{split}
H_6&:=   
(b_0-a_0+1)\log\sqrt{2\pi}
+\frac{1}{6}(-1+6a_0-6a_0^2)\log r_{A}
+\frac{1}{6}(1+6b_0+6b_0^2)
\log r_{B}\Bigr)
\\
&\quad
+\frac{(1-2a_0)\log r_{A}+(1+2b_0)\log r_{B}}{2}
+\frac{(1-2a_0)\mathsf{k}(r_{A})+(1+2b_0)\mathsf{k}(r_{B})}{2}
\\
&\quad
-\frac{(1-2a_0)\log\Delta Q(r_{A})+(1+2b_0)\log\Delta Q(r_{B})}{4}
\\
&\quad
+\frac{1}{24}\int_{r_{A}}^{r_{B}}\Bigl(\frac{\partial \Delta Q(u)}{\Delta Q(u)}\Bigr)^2u\,du-\frac{1}{16}\bigg[
\frac{r_{B}\partial_u\Delta Q(r_{B})}{\Delta Q(r_{B})}-\frac{r_{A}\partial_u\Delta Q(r_{A})}{\Delta Q(r_{A})}
\bigg]\\
&\quad +\frac{1}{3}\log\frac{\Delta Q(r_{A})}{\Delta Q(r_{B})}+\frac{1}{6}\log\frac{r_{B}}{r_{A}}+\frac{1}{4}\int_{r_{A}}^{r_{B}}u\bigl( \Lambda(u;s)+L(u;\alpha,s)\bigr)\,du.
\end{split} 
\end{align}
\end{lemma}

\begin{proof}
Insert the uniform expansion \eqref{def of Dn leq j leq n-1} into the sum and expand the logarithm through order $n^{-1}$. Apply Lemma~\ref{lemma:Riemann sum NEW} separately to the resulting functions of $\tau=j/n$. The change of variables $\tau=\frac12 r q'(r)$ gives $d\tau=2r\Delta Q(r)\,dr$ and converts the bulk integrals to the displayed radial integrals. The endpoint terms in Euler--Maclaurin give the $a_0$- and $b_0$-dependent contributions. Finally, the uniform remainder in \eqref{def of Dn leq j leq n-1}, summed over $O(n)$ indices, is $\mathcal O(n^{-1})$. Collecting equal powers of $n$ and $\log n$ yields $H_1,H_2,H_3,H_5$, and $H_6$ as stated.
\end{proof}

\subsection{Asymptotics of \texorpdfstring{$S_3$}{S3 (annular)}: the oscillatory contribution}
Define
\begin{align}
    \begin{split}
\label{def of tilde Theta n}
    \widetilde{\Theta}_n&:=\sum_{j=0}^{\infty}\log\Bigl(1+\Bigl(\frac{\rho_1}{\rho_2}\Bigr)^{2(j+\theta_{\star}-1)}\frac{\sigma_{\star}-\tau_{\rho_1}}{\tau_{\rho_2}-\sigma_{\star}}e^{\mathsf{k}(\rho_2)-\mathsf{k}(\rho_1)}\Bigr)
\\
&\quad 
+
\sum_{j=0}^{+\infty}\log\Bigl( 
1+\Bigl(\frac{\rho_1}{\rho_2}\Bigr)^{2(j+1-\theta_{\star}+1)}\frac{\tau_{\rho_2}-\sigma_{\star}}{\sigma_{\star}-\tau_{\rho_1}}e^{\mathsf{k}(\rho_1)-\mathsf{k}(\rho_2)}
\Bigr).
    \end{split}
\end{align} 
The middle range contains two competing hard-endpoint contributions. Their crossover produces the oscillatory term $\widetilde\Theta_n$ in the following expansion.
\begin{lemma}\label{lemma:S3 theta function part}
As $n\to\infty$,
\begin{align}
\label{def of expansion S3 part annulus}
    S_3^{(\mathrm{an})}=\mathsf{F}_1n^2
    +\mathsf{F}_2n\log n 
    +\mathsf{F}_3n
    +\mathsf{F}_5\log n 
    +\mathsf{F}_6
    +\widetilde{\Theta}_{n}+\mathcal{O}\Bigl(\frac{(\log n)^2}{n}\Bigr),
\end{align}
where 
\begin{align*}
\mathsf{F}_1&:=
\frac{(q(\rho_2)-q(\rho_1))^2}{4\log\frac{\rho_2}{\rho_1}}+A_{1,+}q(\rho_1)-A_{1,+}^2\log\rho_1-B_{2,-}q(\rho_2)+B_{2,-}^2\log\rho_2
\\
&\quad
+2\int_{r_{A_{1,+}}}^{r_{B_{2,-}}}\Bigl( 
q(u)-uq'(u)\log u
\Bigr)u\Delta Q(u)\,du, 
\\
\mathsf{F}_2&:=
-\frac{1}{2}(B_{2,-}-A_{1,+}),
\\
\mathsf{F}_3&:=
(2\log \rho_2+\mathsf{k}(\rho_2))B_{2,-}-(2\log \rho_1+\mathsf{k}(\rho_1))A_{1,+}
-(B_{2,-}-A_{1,+})\log\sqrt{2\pi}
\\
&\quad
-\frac{2\theta_{2,-}^{(n,\epsilon)}-1}{2}\Bigl(q(\rho_2)-2B_{2,-}\log\rho_2\Bigr)
-\frac{2\theta_{1,+}^{(n,\epsilon)}-1}{2}\Bigl( q(\rho_1)-2A_{1,+}\log\rho_1\Bigr)
\\
&\quad+
B_{2,-}-A_{1,+}
+(A_{1,+}-\tau_{\rho_1})\log (A_{1,+}-\tau_{\rho_1})
+(\tau_{\rho_2}-B_{2,-})\log (\tau_{\rho_2}-B_{2,-})
\\
&\quad
+\frac{1}{2}(2\theta_{1,+}^{(n,\epsilon)}-1)V_{A_{1,+}}(r_{A_{1,+}})
+\frac{1}{2}(2\theta_{2,-}^{(n,\epsilon)}-1)V_{B_{2,-}}(r_{B_{2,-}})
\\
&\quad
-\int_{r_{A_{1,+}}}^{r_{B_{2,-}}}2\log u\cdot u\Delta Q(u)\,du
-\int_{r_{A_{1,+}}}^{r_{B_{2,-}}}2\mathsf{k}(u)u\Delta Q(u)\,du
\\
&\quad
+\int_{r_{A_{1,+}}}^{r_{B_{2,-}}}u\Delta Q(u)\log \Delta Q(u)\,du
\\
&\quad -(\sigma_{\star}
-\tau_{\rho_1})\log (\sigma_{\star}-\tau_{\rho_1})
-(\tau_{\rho_2}-\sigma_{\star})\log (\tau_{\rho_2}-\sigma_{\star})
-\Bigl(2\log \frac{\rho_2}{\rho_1}+\mathsf{k}(\rho_2)-\mathsf{k}(\rho_1)\Bigr)\sigma_{\star}, 
\\
    \mathsf{F}_5&:=
\frac{1}{2}(1-\theta_{2,-}^{(n,\epsilon)}-\theta_{1,+}^{(n,\epsilon)}),
\\
\mathsf{F}_6&:= 
 \frac{\rho_1^2\Delta Q(\rho_1)}{\sigma_{\star}-\tau_{\rho_1}}
+\frac{\rho_2^2\Delta Q(\rho_2)}{\tau_{\rho_2}-\sigma_{\star}}
+\frac{\rho_1^2\Delta Q(\rho_1)}{\tau_{\rho_1}}
-\frac{\rho_2^2\Delta Q(\rho_2)}{\tau_{\rho_2}}
-\frac{\rho_1^2\Delta Q(\rho_1)}{\epsilon\tau_{\rho_1}}
-\frac{\rho_2^2\Delta Q(\rho_2)}{\epsilon\tau_{\rho_2}}
\\
&\quad
-\Bigl(\frac{1}{2}\rho_1\mathsf{k}'(\rho_1)+1\Bigr)\log(\sigma_{\star}-\tau_{\rho_1})
+\Bigl(\frac{1}{2}\rho_2\mathsf{k}'(\rho_2)+1\Bigr)\log(\tau_{\rho_2}-\sigma_{\star})
\\
&\quad
+\Bigl(\frac{1}{2}\rho_1\mathsf{k}'(\rho_1)+1\Bigr)\log(A_{1,+}-\tau_{\rho_1})
-\Bigl(\frac{1}{2}\rho_2\mathsf{k}'(\rho_2)+1\Bigr)\log(\tau_{\rho_2}-B_{2,-})
\\
&\quad 
-(\theta_{2,-}^{(n,\epsilon)}+\theta_{1,+}^{(n,\epsilon)}-1)\log\sqrt{2\pi}
+(2\log\rho_1+\mathsf{k}(\rho_1))\theta_{1,+}^{(n,\epsilon)}
-(2\log\rho_1+\mathsf{k}(\rho_1))\theta_{\star}
\\
&\quad
+(2\log\rho_2+\mathsf{k}(\rho_2))\theta_{2,-}^{(n,\epsilon)}
+(2\log\rho_2+\mathsf{k}(\rho_2))(\theta_{\star}-1)
\\
&\quad +\frac{(-1+6(-\theta_{1,+}^{(n,\epsilon)}+1)-6(-\theta_{1,+}^{(n,\epsilon)}+1)^2)+(1-6\theta_{\star}+6\theta_{\star}^2)}{6}\log\rho_1
\\
&\quad +
\frac{(-1+6(-\theta_{\star}+1)-6(-\theta_{\star}+1)^2)+(1+6(\theta_{2,-}^{(n,\epsilon)}-1)+6(\theta_{2,-}^{(n,\epsilon)}-1)^2)}{6}\log\rho_2
\\
&\quad  -\frac{(2\theta_{1,+}^{(n,\epsilon)}-1)\log(A_{1,+}-\tau_{\rho_1})-(2\theta_{\star}-1)\log(\sigma_{\star}-\tau_{\rho_1})}{2}
\\
&\quad -\frac{(2\theta_{\star}-1)\log(\tau_{\rho_2}-\sigma_{\star})+(2\theta_{2,-}^{(n,\epsilon)}-1)\log(\tau_{\rho_2}-B_{2,-})}{2}
\\
&\quad
-\frac{1}{6}
(-1+6(1-\theta_{1,+}^{(n,\epsilon)})-6(1-\theta_{1,+}^{(n,\epsilon)})^2)\log r_{A_{1,+}}
\\
&\quad
-\frac{1}{6}(1+6(\theta_{2,-}^{(n,\epsilon)}-1)+6(\theta_{2,-}^{(n,\epsilon)}-1)^2)\log r_{B_{2,-}}
\\
&\quad
-\frac{(2\theta_{1,+}^{(n,\epsilon)}-1)\log r_{A_{1,+}}+(2\theta_{2,-}^{(n,\epsilon)}-1)\log r_{B_{2,-}}}{2}
\\
&\quad
-\frac{(2\theta_{1,+}^{(n,\epsilon)}-1)\mathsf{k}(r_{A_{1,+}})+(2\theta_{2,-}^{(n,\epsilon)}-1)\mathsf{k}(r_{B_{2,-}})}{2}
\\
&\quad
+\frac{(2\theta_{1,+}^{(n,\epsilon)}-1)\log\Delta Q(r_{A_{1,+}})+(2\theta_{2,-}^{(n,\epsilon)}-1)\log\Delta Q(r_{B_{2,-}})}{4}
\\
&\quad
-\frac{1}{24}\int_{r_{A_{1,+}}}^{r_{B_{2,-}}}\Bigl(\frac{\partial_u \Delta Q(u)}{\Delta Q(u)}\Bigr)^2u\,du
+\frac{1}{16}\bigg[
\frac{r_{B_{2,-}}\partial_u\Delta Q(r_{B_{2,-}})}{\Delta Q(r_{B_{2,-}})}-\frac{r_{A_{1,+}}\partial_u\Delta Q(r_{A_{1,+}})}{\Delta Q(r_{A_{1,+}})}
\bigg]\\
&\quad -\frac{1}{3}\log\frac{\Delta Q(r_{A_{1,+}})}{\Delta Q(r_{B_{2,-}})}
-\frac{1}{6}\log\frac{r_{B_{2,-}}}{r_{A_{1,+}}}
-\frac{1}{4}\int_{r_{A_{1,+}}}^{r_{B_{2,-}}}u\bigl( \Lambda(u;s)+L(u;\alpha,s)\bigr)\,du.
\end{align*}   
\end{lemma}
To prove Lemma~\ref{lemma:S3 theta function part}, decompose $S_3^{(\mathrm{an})}$ as
\begin{equation}
\label{def of S3 annulus two parts}
S_3^{(\mathrm{an})}=S_3^{(\mathrm{top},\mathrm{an})}-S_3^{(\mathrm{bot},\mathrm{an})},
\end{equation}
where 
\[
S_3^{(\mathrm{top},\mathrm{an})}:=\sum_{j=j_{1,+}+1}^{j_{2,-}-1}\log h_{n,j}^{(\mathrm{an})}, \qquad
S_3^{(\mathrm{bot},\mathrm{an})}:=\sum_{j=j_{1,+}+1}^{j_{2,-}-1}\log h_{n,j}.
\]
Recall that $j_{\star}:=n\sigma_{\star}$. 
Let 
\[
\theta_{\star}:=j_{\star}-\lfloor j_{\star}\rfloor,\qquad A_{1,+}:=\frac{\tau_{\rho_1}}{1-\epsilon},\qquad B_{2,-}:=\frac{\tau_{\rho_2}}{1+\epsilon}, 
\]
and define
\begin{equation*}
    \theta_{1,+}^{(n,\epsilon)}:=\frac{n\tau_{\rho_1}}{1-\epsilon}-\Bigl\lfloor \frac{n\tau_{\rho_1}}{1-\epsilon}\Bigr\rfloor,\qquad \theta_{2,-}^{(n,\epsilon)}:=\Bigl\lceil \frac{n\tau_{\rho_2}}{1+\epsilon}\Bigr\rceil-\frac{n\tau_{\rho_2}}{1+\epsilon},
\end{equation*}
Then
\[
j_{1,+}+1=nA_{1,+}+1-\theta_{1,+}^{(n,\epsilon)},\qquad
j_{2,-}-1=nB_{2,-}+\theta_{2,-}^{(n,\epsilon)}-1.
\]
The expansion of $S_3^{(\mathrm{bot},\mathrm{an})}$ follows from Lemma~\ref{lemma:general sum} with these parameters. It remains to analyze the numerator sum, which we split at the switching index $j_\star$:
\[
   S_3^{(\mathrm{top},\mathrm{an})}=
    \sum_{j=j_{1,+}+1}^{\lfloor j_{\star} \rfloor}\log h_{n,j}^{(\mathrm{an})}
    +
    \sum_{j=\lfloor j_{\star} \rfloor+1}^{j_{2,-}-1}\log h_{n,j}^{(\mathrm{an})}.
\]
\begin{proof}[Proof of Lemma~\ref{lemma:S3 theta function part}]
Lemma~\ref{lemma: j1+ j jdiamond} gives
\begin{align*}
\sum_{j=j_{1,+}+1}^{\lfloor j_{\star} \rfloor}\log h_{n,j}^{(\mathrm{an})} &=
\bigl(-nq(\rho_1)+2\log \rho_1+\mathsf{k}(\rho_1)-\log n\bigr)\sum_{j=j_{1,+}+1}^{\lfloor j_{\star}\rfloor}1
+
(2\log\rho_1)n \sum_{j=j_{1,+}+1}^{\lfloor j_{\star}\rfloor}\frac{j}{n}
\\
&\quad-\sum_{j=j_{1,+}+1}^{\lfloor j_{\star}\rfloor}\log(\tau_j-\tau_{\rho_1})
+\sum_{j=j_{1,+}+1}^{\lfloor j_{\star}\rfloor}\log\Bigl(1+\frac{\tau_j-\tau_{\rho_1}}{\tau_{\rho_2}-\tau_j}e^{\mathsf{k}(\rho_2)-\mathsf{k}(\rho_1)}\Bigl(\frac{\rho_1}{\rho_2} \Bigr)^{2(j_{\star}-j-1)}\Bigr)
\\
&\quad
-\frac{1}{n}\sum_{j=j_{1,+}+1}^{\lfloor j_{\star}\rfloor}\mathfrak{a}_{1}+\sum_{j=j_{1,+}+1}^{\lfloor j_{\star}\rfloor}\mathcal{O}\Bigl(\frac{1}{n^2}+\frac{1}{n}\Bigl(\frac{\rho_1}{\rho_2}\Bigr)^{2(j_{\star}-j)}\Bigr).
\end{align*}
Since $\tau_j-\tau_{\rho_1}\geq\epsilon\tau_{\rho_1}/(1-\epsilon)>0$, Lemma~\ref{lemma:Riemann sum NEW}, with $A=A_{1,+}$, $a_0=1-\theta_{1,+}^{(n,\epsilon)}$, $B=\sigma_\star$, and $b_0=-\theta_\star$, yields
\begin{align*}
&\quad (2\log\rho_1)n\sum_{j=j_{1,+}+1}^{\lfloor j_{\star}\rfloor}\frac{j}{n}
\\
&=n^2(\sigma_{\star}^2-A_{1,+}^2)\log\rho_1+\log\rho_1\Bigl((2\theta_{1,+}^{(n,\epsilon)}-1)A_{1,+}-(2\theta_{\star}-1)\sigma_{\star}\Bigr) n
\\
&\quad+\frac{(-1+6(-\theta_{1,+}^{(n,\epsilon)}+1)-6(-\theta_{1,+}^{(n,\epsilon)}+1)^2)+(1-6\theta_{\star}+6\theta_{\star}^2)}{6}\log\rho_1+\mathcal{O}(\frac{1}{n}),
\\
&\quad \sum_{j=j_{1,+}+1}^{\lfloor j_{\star}\rfloor}\log(\tau_j-\tau_{\rho_1})
\\
&=
\Bigl( 
-\sigma_{\star}+A_{1,+}+(\sigma_{\star}-\tau_{\rho_1})\log (\sigma_{\star}-\tau_{\rho_1})-(A_{1,+}-\tau_{\rho_1})\log (A_{1,+}-\tau_{\rho_1})
\Bigr)n
\\
&\quad +\frac{(1-2(-\theta_{1,+}^{(n,\epsilon)}+1))\log(A_{1,+}-\tau_{\rho_1})+(1-2\theta_{\star})\log(\sigma_{\star}-\tau_{\rho_1})}{2}
\\
&\quad +\frac{1}{12n}\Bigl(
\frac{-1+6(-\theta_{1,+}^{(n,\epsilon)}+1)-6(-\theta_{1,+}^{(n,\epsilon)}+1)^2)}{A_{1,+}-\tau_{\rho_1}}
+
\frac{1-6\theta_{\star}+6\theta_{\star}^2}{\sigma_{\star}-\tau_{\rho_1}}
\Bigr)+\mathcal{O}(\frac{1}{n^2}), 
\\
&\quad \frac{1}{n}\sum_{j=j_{1,+}+1}^{\lfloor j_{\star}\rfloor}\mathfrak{a}_1
\\
&=
-\frac{\rho_1^2\Delta Q(\rho_1)}{\sigma_{\star}-\tau_{\rho_1}}+\frac{\rho_1^2\Delta Q(\rho_1)}{A_{1,+}-\tau_{\rho_1}}
+
\Bigl(\log(\sigma_{\star}-\tau_{\rho_1})-\log(A_{1,+}-\tau_{\rho_1})\Bigr)\Bigl(\frac{1}{2}\rho_1\mathsf{k}'(\rho_1)+1\Bigr)
\\
&\quad 
+\frac{(1-2(-\theta_{1,+}^{(n,\epsilon)}+1))\mathfrak{a}_1(A_{1,+})+(1-2\theta_{\star})\mathfrak{a}_1(\sigma_{\star})}{2n}
+\mathcal{O}(\frac{1}{n^2}).
\end{align*}
As in \cite[Eq.~(3.31) and Lemma~3.6]{C2021}, split the error sum a distance $M'\log n$ from $j_\star$ to obtain
\begin{align}
\begin{split}
\label{def of error osc part1}
&\quad\sum_{j=j_{1,+}+1}^{\lfloor j_{\star}\rfloor}\mathcal{O}\Bigl(\frac{1}{n^2}+\frac{1}{n}\Bigl(\frac{\rho_1}{\rho_2}\Bigr)^{2(j_{\star}-j)}\Bigr)
\\
&=
\sum_{j=j_{1,+}+1}^{\lfloor j_{\star}\rfloor-\lfloor M'\log n\rfloor}\mathcal{O}\Bigl(\frac{1}{n^2}+\frac{1}{n}\Bigl(\frac{\rho_1}{\rho_2}\Bigr)^{2(j_{\star}-j)}\Bigr)
+
\sum_{j=\lfloor j_{\star}\rfloor-\lfloor M'\log n\rfloor+1}^{\lfloor j_{\star}\rfloor}\mathcal{O}\Bigl(\frac{1}{n^2}+\frac{1}{n}\Bigl(\frac{\rho_1}{\rho_2}\Bigr)^{2(j_{\star}-j)}\Bigr)
\\
&=
\mathcal{O}(\frac{\log n}{n}), \qquad n\to+\infty.
\end{split}    
\end{align}
Let 
\begin{equation}
\label{def of osc part1 v1}
\mathcal{S}_{0}:=\sum_{j=j_{1,+}+1}^{\lfloor j_{\star}\rfloor}\log\Bigl(1+\frac{\tau_j-\tau_{\rho_1}}{\tau_{\rho_2}-\tau_j}e^{\mathsf{k}(\rho_2)-\mathsf{k}(\rho_1)}\Bigl(\frac{\rho_1}{\rho_2} \Bigr)^{2(j_{\star}-j-1)}\Bigr).
\end{equation}
After changing the summation index, the same geometric-tail estimate gives
\begin{align}
\begin{split}
\label{def of osc part1 v2}
\mathcal{S}_0&=\sum_{j=0}^{\lfloor j_{\star}\rfloor-j_{1,+}}\log\Bigl(1+\frac{-\frac{j}{n}-\frac{\theta_{\star}}{n}+\sigma_{\star}-\tau_{\rho_1}}{\tau_{\rho_2}+\frac{j}{n}+\frac{\theta_{\star}}{n}-\sigma_{\star}}e^{\mathsf{k}(\rho_2)-\mathsf{k}(\rho_1)}\Bigl(\frac{\rho_1}{\rho_2}\Bigr)^{2(j+\theta_{\star}-1)}\Bigr)
\\
&=
\sum_{j=0}^{\lfloor j_{\star}\rfloor-j_{1,+}}\log\Bigl(1+\Bigl(\frac{\rho_1}{\rho_2}\Bigr)^{2(j+\theta_{\star}-1)}f_0\Bigl(\frac{j+\theta_\star}{n}\Bigr)\Bigr)
+\mathcal{O}(\frac{\log n}{n}),
\end{split}    
\end{align}
where 
\[
f_0(x):=\frac{-x+\sigma_{\star}-\tau_{\rho_1}}{\tau_{\rho_2}+x-\sigma_{\star}}e^{\mathsf{k}(\rho_2)-\mathsf{k}(\rho_1)}.
\]
Since $f_0$ is smooth near zero and the remaining factor decays geometrically in $j$, the argument of \cite[Lemma~3.6]{C2021} gives
\begin{equation}
\mathcal{S}_0=\sum_{j=0}^{\infty}\log\Bigl(1+\Bigl(\frac{\rho_1}{\rho_2}\Bigr)^{2(j+\theta_{\star}-1)}f_0(0)\Bigr)
+\mathcal{O}(\frac{(\log n)^2}{n}),\qquad n\to+\infty.  
\end{equation}
The range $\lfloor j_\star\rfloor+1\leq j\leq j_{2,-}-1$ is treated symmetrically, using Lemma~\ref{lemma: jstar leq j leq j2-}. Subtracting the denominator expansion from Lemma~\ref{lemma:general sum} and collecting coefficients gives \eqref{def of expansion S3 part annulus}.
This completes the proof. 
\end{proof}
For $x\in\R$, $\rho\in(0,1)$, and $a>0$, define 
\[
\Theta(x;\rho,a)
:=x(x-1)\log\rho+x\log a
+\sum_{j=0}^{\infty}\log\Bigl(1+a\rho^{2(j+x)}\Bigr)
+\sum_{j=0}^{\infty}\log\Bigl(1+a^{-1}\rho^{2(j+1-x)}\Bigr). 
\]
Before analyzing the remaining sums, we rewrite \eqref{def of tilde Theta n} in terms of a Jacobi theta function.
\begin{lemma}[{\cite[Lemma~3.28]{C2021}}]
\label{lemma: lemma tilde Theta to Theta}
We have 
\begin{align*}
\Theta(x;\rho,a)
&=\frac{1}{2}\log\pi+\frac{1}{2}\log a-\frac{1}{4}\log\rho-\frac{1}{2}\log\log(\rho^{-1})
+\frac{(\log a)^2}{4\log(\rho^{-1})}
\\
&\quad
-\sum_{j=1}^{\infty}\log(1-\rho^{2j})+\log\theta\Bigl(x+\frac{\log (a\rho)}{2\log\rho}\Bigr|\frac{\pi i}{\log(\rho^{-1})}\Bigr). 
\end{align*}
\end{lemma}
Lemma~\ref{lemma: lemma tilde Theta to Theta} gives
\begin{align}
\begin{split}
    \label{def of tilde Theta rewritten}
\widetilde{\Theta}_n&
=
\frac{1}{2}\log\pi
+\frac{1}{2}\log\Bigl(\frac{\sigma_{\star}-\tau_{\rho_1}}{\tau_{\rho_2}-\sigma_{\star}}\Bigr) 
+\frac{s}{2}(\lambda(\rho_2)-\lambda(\rho_1))
-\frac{1}{4}\log\frac{\rho_1}{\rho_2}
-\frac{1}{2}\log\log\frac{\rho_2}{\rho_1}
\\
&\quad
+\frac{1}{4\log\frac{\rho_2}{\rho_1}}
\Bigl(\log\Bigl(\frac{\sigma_{\star}-\tau_{\rho_1}}{\tau_{\rho_2}-\sigma_{\star}}e^{s(\lambda(\rho_2)-\lambda(\rho_1))}\Bigr)\Bigr)^2
-\sum_{j=1}^{\infty}\log\Bigl(1-\Bigl(\frac{\rho_1}{\rho_2}\Bigr)^{2j}\Bigr)
\\
&\quad
+\log\theta
\Bigl( 
n\sigma_{\star}-\alpha+\frac{1}{2}
+
\frac{\log(\frac{\sigma_{\star}-\tau_{\rho_1}}{\tau_{\rho_2}-\sigma_{\star}})}{2\log\frac{\rho_1}{\rho_2}}
+\frac{s(\lambda(\rho_2)-\lambda(\rho_1))}{2\log\frac{\rho_1}{\rho_2}}
\Bigr|\frac{\pi i}{\log\frac{\rho_2}{\rho_1}}
\Bigr)
\\
&\quad
-(\theta_{\star}-\alpha-1)(\theta_{\star}-\alpha-2)\log\frac{\rho_1}{\rho_2}
-(\theta_{\star}-\alpha-1)
\log\Bigl(\frac{\sigma_{\star}-\tau_{\rho_1}}{\tau_{\rho_2}-\sigma_{\star}}\Bigr) 
\\
&\quad
-s(\lambda(\rho_2)-\lambda(\rho_1))(\theta_{\star}-\alpha-1).
\end{split}
\end{align}

\subsection{Asymptotic expansions of \texorpdfstring{$S_2^{(1,\mathrm{an})}$ and $S_4^{(1,\mathrm{an})}$}{S2(1) and S4(1)}}
We now derive the moderate-deviation expansions for $S_2^{(1,\mathrm{an})}$ and $S_4^{(1,\mathrm{an})}$. Set
\begin{equation}
\label{def of Mjk lambda jk}
M_{j,k}:=\sqrt{n}(\lambda_{j,k}-1),\qquad \lambda_{j,k}:=\frac{\tau_{\rho_k}}{\tau_j},\qquad k=1,2.
\end{equation}
As $r_\tau\to\rho_k$,
\[
\lambda_{j,k}-1=b_1(\rho_k)(\rho_k-r_{\tau})+b_2(\rho_k)(\rho_k-r_{\tau})^2+b_3(\rho_k)(\rho_k-r_{\tau})^3+\mathcal{O}\bigl((\rho_k-r_{\tau})^4 \bigr).
\]
The inverse-function theorem and the regularity of $Q$ allow us to invert this expansion:
\[
\rho_k-r_{\tau}
=\frac{1}{b_1(\rho_k)}(\lambda_{j,k}-1)-\frac{b_2(\rho_k)}{b_1(\rho_k)^3}(\lambda_{j,k}-1)^2
+\frac{2b_2(\rho_k)^2-b_1(\rho_k)b_3(\rho_k)}{b_1(\rho_k)^5}(\lambda_{j,k}-1)^3+\mathcal{O}\bigl((\lambda_{j,k}-1)^4\bigr),
\]
where, with $\mathcal Q(r):=r\Delta Q(r)$,
\begin{align*}
b_1(\rho_k)&:=\frac{2\mathcal{Q}(\rho_k)}{\tau_{\rho_k}},
\\
b_2(\rho_k)&:=\frac{4\mathcal{Q}(\rho_k)^2}{{{\tau_{\rho_k}}^2}}-\frac{\partial_r\mathcal{Q}(\rho_k)}{\tau_{\rho_k}},
\\
b_3(\rho_k)&:=\frac{8\mathcal{Q}(\rho_k)^3}{{\tau_{\rho_k}}^3}-\frac{4\mathcal{Q}(\rho_k)\partial_r\mathcal{Q}(\rho_k)}{{\tau_{\rho_k}}^2}+\frac{\partial_r^2\mathcal{Q}(\rho_k)}{3\tau_{\rho_k}}.
\end{align*}
Consequently,
\begin{equation}
r_{\tau}
=\rho_k-\mathsf{c}_1(\rho_k)\frac{M_{j,k}}{\sqrt{n}}+\mathsf{c}_2(\rho_k)\frac{M_{j,k}^2}{n}-\mathsf{c}_3(\rho_k)\frac{M_{j,k}^3}{n^{3/2}}+\mathcal{O}(\frac{M_{j,k}^4}{n^2}).
\end{equation}
where
\begin{align*}
 \mathsf{c}_1(\rho_k)&:=\tau_{\rho_k}r_{\tau_{\rho_k}}'= \frac{\tau_{\rho_k}}{2\mathcal{Q}(\rho_k)},\qquad
\mathsf{c}_2(\rho_k):=\frac{\tau_{\rho_k}}{8\mathcal{Q}(\rho_k)^3}\bigl(4\mathcal{Q}(\rho_k)^2-\tau_{\rho_k}\mathcal{Q}'(\rho_k)\bigr),
\\
\mathsf{c}_3(\rho_k)&:=\frac{\tau_{\rho_k}}{48\mathcal{Q}(\rho_k)^5}
\Bigl( 
3\tau_{\rho_k}^2\mathcal{Q}'(\rho_k)^2-\tau_{\rho_k}^2\mathcal{Q}(\rho_k)\mathcal{Q}''(\rho_k)
-12\tau_{\rho_k}\mathcal{Q}(\rho_k)^2\mathcal{Q}'(\rho_k)+24\mathcal{Q}(\rho_k)^4
\Bigr), 
\\
\mathsf{c}_4(\rho_k)&:=
 \frac{\tau_{\rho_k}}{384\mathcal{Q}(\rho_k)^7}
 \Bigl( 
 192\mathcal{Q}(\rho_k)^6-144\tau_{\rho_k}\mathcal{Q}(\rho_k)^4\mathcal{Q}'(\rho_k)+72\tau_{\rho_k}^2\mathcal{Q}(\rho_k)^2\mathcal{Q}'(\rho_k)^2
 -15\tau_{\rho_k}^3\mathcal{Q}'(\rho_k)^3
 \\
 &\quad 
 -24\tau_{\rho_k}^2\mathcal{Q}(\rho_k)^3\mathcal{Q}''(\rho_k)
 +10\tau_{\rho_k}^3\mathcal{Q}(\rho_k)\mathcal{Q}'(\rho_k)\mathcal{Q}''(\rho_k)-\tau_{\rho_k}^3\mathcal{Q}(\rho_k)^2\mathcal{Q}^{(3)}(\rho_k)
 \Bigr).     
\end{align*}
Together with Lemma~\ref{lemma:general sum}, these formulas control the denominator sums in $S_2^{(1,\mathrm{an})}$ and $S_4^{(1,\mathrm{an})}$. In particular, for $h\in C^4(\mathbb R)$ satisfying the hypotheses of Lemma~\ref{lemma:Riemann sum NEW},
\begin{align}
\begin{split}
\label{def of ht expansion A1+AM}
\int_{r_{A_{1,+}^{(M)}}}^{r_{A_{1,+}}}
h(t)\,dt
&=
\int_{\rho_1}^{r_{A_{1,+}}}h(t)\,dt
-J_1(\rho_1;h)\frac{M}{\sqrt{n}}
-J_2(\rho_1;h)\frac{M^2}{n}
\\
&\quad
-J_3(\rho_1;h)\frac{M^3}{n^{3/2}}
-J_4(\rho_1;h)\frac{M^4}{n^2}
+\mathcal{O}(\frac{M^5}{n^{5/2}}),
\end{split}
\\
\begin{split}
\label{def of ht expansion B2-BM}
\int_{r_{B_{2,-}}}^{r_{B_{2,-}^{(M)}}}
h(t)\,dt
&=
\int_{r_{B_{2,-}}}^{\rho_2}h(t)\,dt
-J_1(\rho_2;h)\frac{M}{\sqrt{n}}
+
J_2(\rho_2;h)\frac{M^2}{n}
\\
&\quad
-J_3(\rho_2;h)\frac{M^3}{n^{3/2}}
+J_4(\rho_2;h)\frac{M^4}{n^2}
+\mathcal{O}(\frac{M^5}{n^{5/2}}),
\end{split}
\end{align}
where 
\begin{equation}
\label{def of parameters A1M B2M}
    A_{1,+}^{(M)}:=\frac{\tau_{\rho_1}}{1-\frac{M}{\sqrt{n}}},\qquad a_{1,+}^{(M)}:=1-\theta_{1,+}^{(n,M)},\qquad B_{2,-}^{(M)}:=\frac{\tau_{\rho_2}}{1+\frac{M}{\sqrt{n}}},\qquad b_{2,-}^{(M)}:=\theta_{2,-}^{(n,M)}-1,
\end{equation}
and for $k=1,2$, 
\begin{align*}
J_1(\rho_k;h)&:=\mathsf{c}_1(\rho_k)h(\rho_k),\qquad
J_2(\rho_k;h):=\mathsf{c}_2(\rho_k)h(\rho_k)+\frac{1}{2}\mathsf{c}_1(\rho_k)^2h'(\rho_k), 
\\
J_3(\rho_k;h)&:=\mathsf{c}_3(\rho_k)h(\rho_k)+
\mathsf{c}_1(\rho_k)\mathsf{c}_2(\rho_k)h'(\rho_k)
+
\frac{1}{6}\mathsf{c}_1(\rho_k)^3h''(\rho_k),
\\
J_4(\rho_k;h)&:=
\mathsf{c}_4(\rho_k)h(\rho_k)+\frac{1}{2}\Bigl(\mathsf{c}_2(\rho_k)^2+2\mathsf{c}_1(\rho_k)\mathsf{c}_3(\rho_k)\Bigr)h'(\rho_k)
\\
&\quad+\frac{1}{2}\mathsf{c}_1(\rho_k)^2\mathsf{c}_2(\rho_k)h''(\rho_k)+\frac{1}{24}\mathsf{c}_1(\rho_k)^4h^{(3)}(\rho_k).
\end{align*}
We shall also use the expansion
\begin{align*}
\xi_{j,k}&=\sqrt{4n\Delta Q(\rho_k)}(\rho_k-r_{\tau})
-\sqrt{4n}\frac{\partial_r \Delta Q(\rho_k)}{2\sqrt{\Delta Q(\rho_k)}}(\rho_k-r_{\tau})^2\\
&\quad +\sqrt{4n}\frac{2\Delta Q(\rho_k)\partial_r^2\Delta Q(\rho_k)-(\partial_r\Delta Q(\rho_k))^2}{8(\Delta Q(\rho_k))^{3/2}}(\rho_k-r_{\tau})^3+\sqrt{4n}\cdot\mathcal{O}(r_{\tau}-\rho_k)^4,
\end{align*}
as $r_\tau\to\rho_k$. This will be used for the transition sums $S_2^{(2,\mathrm{an})}$ and $S_4^{(2,\mathrm{an})}$ in the next subsection.

\begin{lemma}\label{lemma: hard edge g1+ j j1+ first}
As $n\to\infty$,
\begin{equation}
    \label{def of S2(1) an}
    S_{2}^{(1,\mathrm{an})}=
\mathsf{H}_{1,1}n^2
+\mathsf{H}_{2,1}n\log n
+\mathsf{H}_{3,1}n
+\mathsf{H}_{4,1}\sqrt{n}
+\mathsf{H}_{5,1}\log n
+\mathsf{H}_{6,1}
+\mathcal{O}\Bigl(\frac{M^5}{\sqrt{n}}+\frac{\sqrt{n}}{M^7}\Bigr),
\end{equation}
where 
\begin{align*}
\mathsf{H}_{1,1}&:=
\tau_{\rho_1}q(\rho_1)
-A_{1,+}q(\rho_1)
-\tau_{\rho_1}^2\log \rho_1
+A_{1,+}^2\log \rho_1
\\
&\quad
+2\int_{\rho_1}^{r_{A_{1,+}}}\Bigl(q(u)-uq'(u)\log u\Bigr)u\Delta Q(u)\,du,
\\
\mathsf{H}_{2,1}&:= 
\frac{1}{2}\tau_{\rho_1}-\frac{1}{2}A_{1,+},
\\
\mathsf{H}_{3,1}&:=
(2\log \rho_1+\mathsf{k}(\rho_1))A_{1,+}
+A_{1,+}
-(A_{1,+}-\tau_{\rho_1})\log(A_{1,+}-\tau_{\rho_1})
-A_{1,+}\log\sqrt{2\pi}
\\
&\quad
+\frac{1}{2}(2\theta_{1,+}^{(n,\epsilon)}-1)\bigl(q(\rho_1)-2A_{1,+}\log\rho_1\bigr)
-\frac{1}{2}(2\theta_{1,+}^{(n,\epsilon)}-1)V_{A_{1,+}}(r_{A_{1,+}})
\\
&\quad
+\tau_{\rho_1}\log\sqrt{2\pi}
-\tau_{\rho_1}
-(2\log \rho_1+\mathsf{k}(\rho_1))\tau_{\rho_1}
-2\int_{\rho_1}^{r_{A_{1,+}}}\log u \cdot u\Delta Q(u)\,du
\\
&\quad 
-2\int_{\rho_1}^{r_{A_{1,+}}}\mathsf{k}(u)u\Delta Q(u)\,du
+\int_{\rho_1}^{r_{A_{1,+}}}u\Delta Q(u)\log \Delta Q(u)\,du,
\\
\mathsf{H}_{4,1}&:=  
\frac{\tau_{\rho_1}^3}{6\rho_1^2\Delta Q(\rho_1)}M^3
-\tau_{\rho_1}M
+\tau_{\rho_1}M\log\frac{M\tau_{\rho_1}\sqrt{2\pi}}{\rho_1\sqrt{\Delta Q(\rho_1)}}
\\
&\quad
-\frac{\rho_1^2\Delta Q(\rho_1)}{\tau_{\rho_1}M}
+\frac{5\rho_1^4(\Delta Q(\rho_1))^2}{6\tau_{\rho_1}^3M^3}
-\frac{37\rho_1^6}{15\tau_{\rho_1}^5}\frac{\Delta Q(\rho_1)^3}{M^5},
\\
\mathsf{H}_{5,1}&:=  
\frac{1}{2}\theta_{1,+}^{(n,\epsilon)}
-\frac{\rho_1}{4}\mathsf{k}'(\rho_1)-\frac{3}{4},
\\
\mathsf{H}_{6,1}&:=  
-\Bigl( 
-\frac{\tau_{\rho_1}^3}{2\rho_1^2\Delta Q(\rho_1)}
+\frac{\tau_{\rho_1}^4}{24\rho_1^4\Delta Q(\rho_1)^2}
+\frac{\tau_{\rho_1}^4\partial_r\Delta Q(\rho_1)}{48\rho_1^3\Delta Q(\rho_1)^3}
\Bigr)M^4
\\
&\quad
+\bigg[
\tau_{\rho_1}\log\frac{M\tau_{\rho_1}\sqrt{2\pi}}{\rho_1\sqrt{\Delta Q(\rho_1)}}
+\frac{\tau_{\rho_1}^2}{4\rho_1^2\Delta Q(\rho_1)}
+\frac{\tau_{\rho_1}^2}{4\rho_1\Delta Q(\rho_1)}\mathsf{k}'(\rho_1)
-\frac{\tau_{\rho_1}^2\partial_r\Delta Q(\rho_1)}{8\rho_1\Delta Q(\rho_1)^2}
\bigg]M^2
\\
&\quad
-\frac{\tau_{\rho_1}^2(2\theta_{1,+}^{(n,M)}-1)}{4\rho_1^2\Delta Q(\rho_1)}M^2
-\frac{(2\theta_{1,+}^{(n,M)}-1)}{2}\log(\tau_{\rho_1}M)
+\Bigl(\frac{\rho_1}{2}\mathsf{k}'(\rho_1)+1\Bigr)\log (M\tau_{\rho_1})
\\
&\quad
+(\theta_{1,+}^{(n,M)}-\theta_{1,+}^{(n,\epsilon)})(2\log \rho_1+\mathsf{k}(\rho_1))
 +(\theta_{1,+}^{(n,\epsilon)}-\theta_{1,+}^{(n,M)})\log\sqrt{2\pi}
\\
&\quad
+\frac{(2\theta_{1,+}^{(n,\epsilon)}-1)}{2}\log\Bigl(A_{1,+}-\tau_{\rho_1}\Bigr)
+\frac{\rho_1^2\Delta Q(\rho_1)}{\tau_{\rho_1}\epsilon}
-\Bigl(\frac{\rho_1}{2}\mathsf{k}'(\rho_1)+1\Bigr)\log (A_{1,+}-\tau_{\rho_1})
 \\
&\quad
+\bigl(1-6\theta_{1,+}^{(n,\epsilon)}+6{\theta_{1,+}^{(n,\epsilon)}}^2\bigr)
\frac{\log \rho_1}{6}
-\frac{1}{6}(1-6\theta_{1,+}^{(n,\epsilon)}+6{\theta_{1,+}^{(n,\epsilon)}}^2)
\log r_{A_{1,+}}
\\
&\quad
-\frac{(2\theta_{1,+}^{(n,M)}-1)\log \rho_1+(1-2\theta_{1,+}^{(n,\epsilon)})\log r_{A_{1,+}}}{2}
\\
&\quad
-\frac{(2\theta_{1,+}^{(n,M)}-1)\mathsf{k}(\rho_1)+(1-2\theta_{1,+}^{(n,\epsilon)})\mathsf{k}(r_{A_{1,+}})}{2}
\\
&\quad
+\frac{(2\theta_{1,+}^{(n,M)}-1)\log\Delta Q(\rho_1)+(1-2\theta_{1,+}^{(n,\epsilon)})\log\Delta Q(r_{A_{1,+}})}{4}
\\
&\quad
-\frac{1}{24}\int_{\rho_1}^{r_{A_{1,+}}}\Bigl(\frac{\partial \Delta Q(u)}{\Delta Q(u)}\Bigr)^2u\,du
+\frac{1}{16}\bigg[
\frac{r_{A_{1,+}}\partial_u\Delta Q(r_{A_{1,+}})}{\Delta Q(r_{A_{1,+}})}-\frac{\rho_1\partial_u\Delta Q(\rho_1)}{\Delta Q(\rho_1)}
\bigg]\\
&\quad -\frac{1}{3}\log\frac{\Delta Q(\rho_1)}{\Delta Q(r_{A_{1,+}})}
-\frac{1}{6}\log\frac{r_{A_{1,+}}}{\rho_1}
-\frac{1}{4}\int_{\rho_1}^{r_{A_{1,+}}}u\bigl( \Lambda(u;s)+L(u;\alpha,s)\bigr)\,du,
\end{align*}
where $\Lambda(u;s)$ and $L(u;\alpha,s)$ are given by \eqref{def of Lammda us} and \eqref{def of L us}, respectively.
\end{lemma}

\begin{proof}
Recall that $\eta_1=\frac{2}{\rho_1}(\tau_j-\tau_{\rho_1})$. Taking logarithms in Lemma~\ref{lemma: g1+leq j leq j+1} and using the uniform lower bound $n\eta_1^2\gtrsim M^2$ gives
\begin{align*}
&\quad \sum_{j=g_{1,+}+1}^{j_{1,+}}\log h_{n,j}^{(\mathrm{an})}
\\
&=
\Bigl(\frac{n\tau_{\rho_1}}{1-\epsilon}-\theta_{1,+}^{(n,\epsilon)}-\frac{n\tau_{\rho_1}}{1-\frac{M}{\sqrt{n}}}+\theta_{1,+}^{(n,M)}\Bigr)\cdot
\Bigl( 
2\log\rho_1+\mathsf{k}(\rho_1)-nq(\rho_1)-\log n
\Bigr)
+
(2\log\rho_1)n\sum_{j=g_{1,+}+1}^{j_{1,+}}\frac{j}{n}
\\
&\quad -\sum_{j=g_{1,+}+1}^{j_{1,+}}\log(\tau_j-\tau_{\rho_1})
+\sum_{j=g_{1,+}+1}^{j_{1,+}}\Bigl(
-\frac{\mathfrak{a}_1}{n}+\frac{2\mathfrak{a}_2-\mathfrak{a}_1^2}{2n^2}-\frac{\mathfrak{a}_1^3-3\mathfrak{a}_1\mathfrak{a}_2+3\mathfrak{a}_3}{3n^3}\Bigr)+\mathcal{O}(\frac{\sqrt{n}}{M^7}). 
\end{align*}    
where $\mathfrak a_k$, $k=1,2,3$, are given by \eqref{def of mathfraka1}--\eqref{def of mathfraka3}. Apply Lemma~\ref{lemma:Riemann sum NEW} with $A=A_{1,+}^{(M)}$, $a_0=1-\theta_{1,+}^{(n,M)}$, $B=A_{1,+}$, and $b_0=-\theta_{1,+}^{(n,\epsilon)}$. We obtain
\begin{align*}
&\quad -\frac{1}{n}\sum_{j=g_{1,+}+1}^{j_{1,+}}\mathfrak{a}_1
\\
&=
-\frac{\rho_1^2\Delta Q(\rho_1)}{\tau_{\rho_1}M}\sqrt{n}
+\frac{\rho_1^2\Delta Q(\rho_1)}{\tau_{\rho_1}\epsilon}
-\Bigl(\frac{\rho_1}{2}\mathsf{k}'(\rho_1)+1\Bigr)\log (A_{1,+}-\tau_{\rho_1})
\\
&\quad 
+\Bigl(\frac{\rho_1}{2}\mathsf{k}'(\rho_1)+1\Bigr)(\log \tau_{\rho_1}+\log M)
-\frac{1}{2}\Bigl(\frac{\rho_1}{2}\mathsf{k}'(\rho_1)+1\Bigr)\log n+\mathcal{O}(M^{-2}),
\\
&\quad \sum_{j=g_{1,+}+1}^{j_{1,+}}
\frac{-\mathfrak{a}_1^2+2\mathfrak{a}_2}{2n^2}
\\
&=
\frac{5\rho_1^4(\Delta Q(\rho_1))^2}{6\tau_{\rho_1}^3M^3}\sqrt{n}+\mathcal{O}(M^{-2}),
\\
&\quad -\sum_{j=g_{1,+}+1}^{j_{1,+}}\frac{\mathfrak{a}_1^3-3\mathfrak{a}_1\mathfrak{a}_2+3\mathfrak{a}_3}{3n^3}
\\
&=-\frac{37\rho_1^6}{15\tau_{\rho_1}^5}\frac{\Delta Q(\rho_1)^3}{M^5}\sqrt{n}+\mathcal{O}(M^{-4}), 
\end{align*}
\begin{align*}
&\quad(2\log \rho_1)n\sum_{j=g_{1,+}+1}^{j_{1,+}}\frac{j}{n}
\\
&=
-n^2\tau_{\rho_1}^2\log \rho_1
-Mn^{3/2}\cdot 2\tau_{\rho_1}^2\log \rho_1
-M^2n\cdot 3\tau_{\rho_1}^2\log \rho_1
-M^3\sqrt{n}\cdot 4\tau_{\rho_1}^2\log \rho_1
\\
&\quad
-M^4\cdot 5\tau_{\rho_1}^2\log \rho_1
+\tau_{\rho_1}n(2\theta_{1,+}^{(n,M)}-1)\log \rho_1
+\tau_{\rho_1}M\sqrt{n}(2\theta_{1,+}^{(n,M)}-1)\log \rho_1
\\
&\quad
+\tau_{\rho_1}M^2(2\theta_{1,+}^{(n,M)}-1)\log \rho_1
+n^2A_{1,+}^2\log \rho_1
-nA_{1,+}(2\theta_{1,+}^{(n,\epsilon)}-1)\log \rho_1
\\
&\quad
+\bigg( 
\bigl(-1+6(-\theta_{1,+}^{(n,M)}+1)-6(-\theta_{1,+}^{(n,M)}+1)^2\bigr)
+
\bigl(1-6\theta_{1,+}^{(n,\epsilon)}+6{\theta_{1,+}^{(n,\epsilon)}}^2\bigr)
\bigg)\frac{\log \rho_1}{6}
+\mathcal{O}\Bigl(\frac{M^5}{\sqrt{n}}\Bigr), 
\end{align*}
and 
\begin{align*}
&\quad -\sum_{j=g_{1,+}+1}^{j_{1,+}}\log(\tau_j-\tau_{\rho_1})
\\
&=
-\tau_{\rho_1}n-\tau_{\rho_1}M\sqrt{n}-\tau_{\rho_1}M^2
-\frac{1}{2}\tau_{\rho_1}M\sqrt{n}\log n+\tau_{\rho_1}M\sqrt{n}\log(\tau_{\rho_1}M)+\tau_{\rho_1}M^2
+\tau_{\rho_1}M^2\log(\tau_{\rho_1}M)
\\
&\quad 
-\frac{1}{2}\tau_{\rho_1}M^2\log n
+\frac{(2\theta_{1,+}^{(n,M)}-1)}{4}\log n
-\frac{(2\theta_{1,+}^{(n,M)}-1)}{2}\log(\tau_{\rho_1}M)
\\
&\quad
+nA_{1,+}
-n(A_{1,+}-\tau_{\rho_1})\log(A_{1,+}-\tau_{\rho_1})
+\frac{(2\theta_{1,+}^{(n,\epsilon)}-1)}{2}\log\Bigl(A_{1,+}-\tau_{\rho_1}\Bigr)
+\mathcal{O}(\frac{M^3\log n}{\sqrt{n}}). 
\end{align*}
Finally, apply Lemma~\ref{lemma:general sum} to the denominator and use \eqref{def of ht expansion A1+AM} to expand its moving endpoint. After grouping equal powers of $n$, $\log n$, and $M$, the coefficients are exactly $\mathsf H_{1,1},\ldots,\mathsf H_{6,1}$ displayed above. The accumulated remainder is $\mathcal O(M^5n^{-1/2}+\sqrt n\,M^{-7})$, which proves \eqref{def of S2(1) an}.

\end{proof}

\begin{lemma}\label{lemma: S41 an asymptotics}
As $n\to\infty$,
\begin{equation}
\label{def of S4(1) an}    
S_{4}^{(1,\mathrm{an})}=
\mathsf{H}_{1,2}n^2
+\mathsf{H}_{2,2}n\log n
+\mathsf{H}_{3,2}n
+\mathsf{H}_{4,2}\sqrt{n}
+\mathsf{H}_{5,2}\log n
+\mathsf{H}_{6,2}
+\mathcal{O}\Bigl(\frac{M^5}{\sqrt{n}}+\frac{\sqrt{n}}{M^7}\Bigr),
\end{equation}
where 
\begin{align*}
\mathsf{H}_{1,2}&:= 
-\tau_{\rho_2}q(\rho_2)
+B_{2,-}q(\rho_2)
+\tau_{\rho_2}^2\log \rho_2
-B_{2,-}^2\log\rho_2
\\
&\quad
+2\int_{r_{B_{2,-}}}^{\rho_2}\Bigl(q(u)-uq'(u)\log u
\Bigr)u\Delta Q(u)\,du,
\\
\mathsf{H}_{2,2}&:= 
-\frac{1}{2}\tau_{\rho_2}
+\frac{1}{2}B_{2,-},
\\
\mathsf{H}_{3,2}&:=  
-B_{2,-}(2\log\rho_2+\mathsf{k}(\rho_2))
 -B_{2,-}
-(\tau_{\rho_2}-B_{2,-})\log(\tau_{\rho_2}-B_{2,-})
 +B_{2,-}\log\sqrt{2\pi}
 \\
 &\quad
 +\frac{1}{2}(2\theta_{2,-}^{(n,\epsilon)}-1)(q(\rho_2)-2B_{2,-}\log\rho_2) 
-\frac{1}{2}(2\theta_{2,-}^{(n,\epsilon)}-1)
V_{B_{2,-}}(r_{B_{2,-}})
\\
&\quad
 +\tau_{\rho_2}
-\tau_{\rho_2}\log\sqrt{2\pi}  
+(2\log\rho_2+\mathsf{k}(\rho_2))\tau_{\rho_2}
-2\int_{r_{B_{2,-}}}^{\rho_2}\log u\cdot u\Delta Q(u)\,du
\\
&\quad
-2\int_{r_{B_{2,-}}}^{\rho_2}\mathsf{k}(u)u\Delta Q(u)\,du
+\int_{r_{B_{2,-}}}^{\rho_2}u\Delta Q(u)\log \Delta Q(u)\,du,
\\
\mathsf{H}_{4,2}&:= 
\frac{\tau_{\rho_2}^3}{6\rho_2^2\Delta Q(\rho_2)}M^3
 -\tau_{\rho_2}M
+\tau_{\rho_2}M\log \frac{\tau_{\rho_2}M\sqrt{2\pi}}{\rho_2\sqrt{\Delta Q(\rho_2)}}
     \\
     &\quad
     -\frac{\rho_2^2\Delta Q(\rho_2)}{\tau_{\rho_2}M}
     +\frac{5\rho_2^4\Delta Q(\rho_2)^2}{6\tau_{\rho_2}^3M^3}
     -\frac{37\rho_2^6\Delta Q(\rho_2)^3}{15\tau_{\rho_2}^5M^5}, 
\\
\mathsf{H}_{5,2}&:=  
\frac{1}{2}\theta_{2,-}^{(n,\epsilon)}
+\frac{1}{4}+\frac{1}{4}\rho_2\mathsf{k}'(\rho_2), 
\\
\mathsf{H}_{6,2}&:= 
\Bigl( 
-\frac{\tau_{\rho_2}^3}{2\rho_2^2\Delta Q(\rho_2)}
+\frac{\tau_{\rho_2}^4}{24\rho_2^4\Delta Q(\rho_2)^2}
+\frac{\tau_{\rho_2}^4\partial_r\Delta Q(\rho_2)}{48\rho_2^3\Delta Q(\rho_2)^3}
\Bigr)M^4
-(2\theta_{2,-}^{(n,M)}-1)
\frac{\tau_{\rho_2}^2}{4\rho_2^2\Delta Q(\rho_2)}M^2
\\
&\quad
+
\Bigl( 
-\tau_{\rho_2}\log\Bigl(\frac{M\tau_{\rho_2}\sqrt{2\pi}}{\rho_2\sqrt{\Delta Q(\rho_2)}}\Bigr)
-\frac{\tau_{\rho_2}^2}{4\rho_2^2\Delta Q(\rho_2)}
-\frac{\tau_{\rho_2}^2}{4\rho_2\Delta Q(\rho_2)}\mathsf{k}'(\rho_2)
+\frac{\tau_{\rho_2}^2\partial_r\Delta Q(\rho_2)}{8\rho_2\Delta Q(\rho_2)^2}
\Bigr)M^2
\\
&\quad
-\frac{1}{2}(2\theta_{2,-}^{(n,M)}-1)\log(\tau_{\rho_2}M)
-\Bigl(1+\frac{1}{2}\rho_2\mathsf{k}'(\rho_2)\Bigr)\log(\tau_{\rho_2}M)
\\
&\quad
+(2\log\rho_2+\mathsf{k}(\rho_2))(\theta_{2,-}^{(n,M)}-\theta_{2,-}^{(n,\epsilon)})
   -(\theta_{2,-}^{(n,M)}-\theta_{2,-}^{(n,\epsilon)})\log\sqrt{2\pi}
\\
&\quad
+\frac{1}{2}(2\theta_{2,-}^{(n,\epsilon)}-1)\log(\tau_{\rho_2}-B_{2,-})
+\frac{\rho_2^2\Delta Q(\rho_2)}{\epsilon\tau_{\rho_2}}
+\Bigl(1+\frac{1}{2}\rho_2\mathsf{k}'(\rho_2)\Bigr)\log(\tau_{\rho_2}-B_{2,-})
    \\
&\quad
 +
\frac{-1+6\theta_{2,-}^{(n,\epsilon)}-6{\theta_{2,-}^{(n,\epsilon)}}^2}{6}\log\rho_2
-\frac{1}{6}
(-1+6\theta_{2,-}^{(n,\epsilon)}-6{\theta_{2,-}^{(n,\epsilon)}}^2)
\log r_{B_{2,-}}
\\
&\quad
-\frac{(1-2\theta_{2,-}^{(n,\epsilon)})\log r_{B_{2,-}}+(2\theta_{2,-}^{(n,M)}-1)\log \rho_2}{2}
\\
&\quad
-\frac{(1-2\theta_{2,-}^{(n,\epsilon)})\mathsf{k}(r_{B_{2,-}})+(2\theta_{2,-}^{(n,M)}-1)\mathsf{k}(\rho_2)}{2}
\\
&\quad
+\frac{(1-2\theta_{2,-}^{(n,\epsilon)})\log\Delta Q(r_{B_{2,-}})+((2\theta_{2,-}^{(n,M)}-1))\log\Delta Q(\rho_2)}{4}
\\
&\quad
-\frac{1}{24}\int_{r_{B_{2,-}}}^{\rho_2}\Bigl(\frac{\partial_u \Delta Q(u)}{\Delta Q(u)}\Bigr)^2u\,du
+\frac{1}{16}\bigg[
\frac{\rho_2\partial_u\Delta Q(\rho_2)}{\Delta Q(\rho_2)}-\frac{r_{B_{2,-}}\partial_u\Delta Q(r_{B_{2,-}})}{\Delta Q(r_{B_{2,-}})}
\bigg]\\
&\quad -\frac{1}{3}\log\frac{\Delta Q(r_{B_{2,-}})}{\Delta Q(\rho_2)}
-\frac{1}{6}\log\frac{\rho_2}{r_{B_{2,-}}}
-\frac{1}{4}\int_{r_{B_{2,-}}}^{\rho_2}u\bigl( \Lambda(u;s)+L(u;\alpha,s)\bigr)\,du.  
\end{align*}
\end{lemma}

\begin{proof}

Taking logarithms in Lemma~\ref{lemma:j2- j g2--1} gives
\begin{align*}
 \sum_{j=j_{2,-}}^{g_{2,-}-1}\log h_{n,j}^{(\mathrm{an})}  
&=
 \sum_{j=j_{2,-}}^{g_{2,-}-1}
 \Bigl( 
2\log\rho_2+\mathsf{k}(\rho_2)-nq(\rho_2)-\log n
 \Bigr)
  +2n\log\rho_2\sum_{j=j_{2,-}}^{g_{2,-}-1} \frac{j}{n}
 \\
 &\quad
-\sum_{j=j_{2,-}}^{g_{2,-}-1}\log(\tau_{\rho_2}-\tau)
 + \sum_{j=j_{2,-}}^{g_{2,-}-1}
 \Bigl( 
\frac{\widetilde{\mathfrak{a}}_1}{n}+\frac{2\widetilde{\mathfrak{a}}_2-\widetilde{\mathfrak{a}}_1^2}{2n^2}+\frac{\widetilde{\mathfrak{a}}_1^3-3\widetilde{\mathfrak{a}}_1\widetilde{\mathfrak{a}}_2+3\widetilde{\mathfrak{a}}_3}{3n^3}
 \Bigr)+\mathcal{O}(\frac{\sqrt{n}}{M^7}),
\end{align*}
where $\widetilde{\mathfrak a}_k$, $k=1,2,3$, are given by \eqref{def of tilde frak a1}--\eqref{def of tilde frak a3}. Apply Lemma~\ref{lemma:Riemann sum NEW} with $A=B_{2,-}$, $a_0=\theta_{2,-}^{(n,\epsilon)}$, $B=B_{2,-}^{(M)}$, and $b_0=\theta_{2,-}^{(n,M)}-1$. We obtain
\begin{align*}
2n\log\rho_2\sum_{j=j_{2,-}}^{g_{2,-}-1}
 \frac{j}{n}
 &=
 n^2\tau_{\rho_2}^2\log \rho_2
-n^{3/2}\cdot 2M\tau_{\rho_2}^2\log \rho_2
+n\cdot 3M^2\tau_{\rho_2}^2\log \rho_2
-\sqrt{n}\cdot 4M^3\tau_{\rho_2}^2\log \rho_2
\\
&\quad
+5M^4\tau_{\rho_2}^2\log \rho_2
 + \tau_{\rho_2}(2\theta_{2,-}^{(n,M)}-1)n\log\rho_2
 - \tau_{\rho_2}(2\theta_{2,-}^{(n,M)}-1)M\sqrt{n}\log\rho_2
 \\
 &\quad
  + \tau_{\rho_2}(2\theta_{2,-}^{(n,M)}-1)M^2\log\rho_2
 -n^2B_{2,-}^2\log\rho_2
 -(2\theta_{2,-}^{(n,\epsilon)}-1)B_{2,-}n\log\rho_2
 \\
 &\quad
 +
\frac{-1+6\theta_{2,-}^{(n,\epsilon)}-6{\theta_{2,-}^{(n,\epsilon)}}^2+1+6(\theta_{2,-}^{(n,M)}-1)+6(\theta_{2,-}^{(n,M)}-1)^2}{6}\log\rho_2
+\mathcal{O}(\frac{M^5}{\sqrt{n}}),
\end{align*}

\begin{align*}
-\sum_{j=j_{2,-}}^{g_{2,-}-1}\log(\tau_{\rho_2}-\tau)  
&=
\tau_{\rho_2}n-\tau_{\rho_2}M\sqrt{n}+\tau_{\rho_2}M^2
+
\tau_{\rho_2}M\sqrt{n}\log(\tau_{\rho_2}M)-\frac{1}{2}\tau_{\rho_2}M\sqrt{n}\log n
\\
&\quad
-\tau_{\rho_2}M^2-\tau_{\rho_2}M^2\log(\tau_{\rho_2}M)+\frac{1}{2}\tau_{\rho_2} M^2\log n
+\mathcal{O}(\frac{M^3\log n}{\sqrt{n}})
\\
&\quad
-B_{2,-}n
-n(\tau_{\rho_2}-B_{2,-})\log(\tau_{\rho_2}-B_{2,-})
\\
&\quad
+\frac{(2\theta_{2,-}^{(n,\epsilon)}-1)\log(\tau_{\rho_2}-B_{2,-})-(2\theta_{2,-}^{(n,M)}-1)\log(\tau_{\rho_2}-B_{2,-}^{(M)})}{2},
\end{align*}

\begin{align*}
\sum_{j=j_{2,-}}^{g_{2,-}-1}
\frac{\widetilde{\mathfrak{a}}_1}{n} 
&=-\frac{\rho_2^2\Delta Q(\rho_2)}{\tau_{\rho_2}M}\sqrt{n}
+\frac{1}{2}\Bigl(1+\frac{1}{2}\rho_2\mathsf{k}'(\rho_2)\Bigr)\log n
-\frac{\rho_2^2\Delta Q(\rho_2)}{\tau_{\rho_2}}
-\Bigl(1+\frac{1}{2}\rho_2\mathsf{k}'(\rho_2)\Bigr)\log(\tau_{\rho_2}M)
\\
&\quad
+\frac{\rho_2^2\Delta Q(\rho_2)}{\tau_{\rho_2}-B_{2,-}}+\Bigl(1+\frac{1}{2}\rho_2\mathsf{k}'(\rho_2)\Bigr)\log(\tau_{\rho_2}-B_{2,-})
+\mathcal{O}(M^{-2}),
\end{align*}
and 
\begin{align*}
\sum_{j=j_{2,-}}^{g_{2,-}-1}
\frac{2\widetilde{\mathfrak{a}}_2-\widetilde{\mathfrak{a}}_1^2}{2n^2}
&=\frac{5\rho_2^4\Delta Q(\rho_2)^2}{6\tau_{\rho_2}^3M^3}\sqrt{n}+\mathcal{O}(M^{-2}),
\\
\sum_{j=j_{2,-}}^{g_{2,-}-1}
\frac{\widetilde{\mathfrak{a}}_1^3-3\widetilde{\mathfrak{a}}_1\widetilde{\mathfrak{a}}_2+3\widetilde{\mathfrak{a}}_3}{3n^3}
&=-\frac{37\rho_2^6\Delta Q(\rho_2)^3}{15\tau_{\rho_2}^5M^5}\sqrt{n}+\mathcal{O}(M^{-4}).
\end{align*}
Apply Lemma~\ref{lemma:general sum} to the denominator and use \eqref{def of ht expansion B2-BM} to expand the moving endpoint. Collecting equal powers of $n$, $\log n$, and $M$ gives $\mathsf H_{1,2},\ldots,\mathsf H_{6,2}$ as displayed in the statement. The accumulated error is $\mathcal O(M^5n^{-1/2}+\sqrt n\,M^{-7})$, proving \eqref{def of S4(1) an}.
\end{proof}

\subsection{Transition sums and completion of the annulus case}
It remains to analyze $S_2^{(2,\mathrm{an})}$ and $S_4^{(2,\mathrm{an})}$. For $k\in\{1,2\}$, define
\[
\theta_{k,-}^{(n,M)}:=g_{k,-}-\frac{n\tau_{\rho_k}}{1+\frac{M}{\sqrt{n}}}=\Bigl\lceil\frac{n\tau_{\rho_k}}{1+\frac{M}{\sqrt{n}}}\Bigr\rceil-\frac{n\tau_{\rho_k}}{1+\frac{M}{\sqrt{n}}},
\quad
\theta_{k,+}^{(n,M)}:=\frac{n\tau_{\rho_k}}{1-\frac{M}{\sqrt{n}}}   -g_{k,+}=\frac{n\tau_{\rho_k}}{1-\frac{M}{\sqrt{n}}}-\Bigl\lfloor\frac{n\tau_{\rho_k}}{1-\frac{M}{\sqrt{n}}}\Bigr\rfloor.
\]
We use the following modification of \cite[Lemma~2.7]{C2021 FH}. Recall from \eqref{def of Mjk lambda jk} that $-M\leq M_{j,k}\leq M$ for $g_{k,-}\leq j\leq g_{k,+}$. Moreover,
\begin{equation}
M_{j-1,k}-M_{j,k}=\frac{1}{\tau_{\rho_k}}\frac{1}{\sqrt{n}}+\frac{2M_{j,k}}{\tau_{\rho_k}}\frac{1}{n}+\Bigl(\frac{1}{{\tau_{\rho_k}}^2}+\frac{M_{j,k}^2}{\tau_{\rho_k}}\Bigr)\frac{1}{n^{3/2}}+\mathcal{O}(M_{j,k}n^{-2}).
\end{equation}
\begin{lemma}[{A modification of \cite[Lemma~2.7]{C2021 FH}}]\label{lemma:Riemann sum}
Let $h\in C^3(\mathbb R)$ and $k\in\{1,2\}$. As $n\to\infty$,
\begin{align}
& \sum_{j=g_{k,-}}^{g_{k,+}}h(M_{j,k}) = \tau_{\rho_k} \int_{-M}^{M} h(t) dt \; \sqrt{n} \nonumber \\
& - 2 \tau_{\rho_k} \int_{-M}^{M} th(t) dt + \bigg( \frac{1}{2}-\theta_{k,-}^{(n,M)} \bigg)h(M)+ \bigg( \frac{1}{2}-\theta_{k,+}^{(n,M)} \bigg)h(-M) \nonumber \\
& + \frac{1}{\sqrt{n}}\bigg[ 3\tau_{\rho_k} \int_{-M}^{M}t^{2}h(t)dt + \bigg( \frac{1}{12}+\frac{\theta_{k,-}^{(n,M)}(\theta_{k,-}^{(n,M)}-1)}{2} \bigg)\frac{h'(M)}{\tau_{\rho_k}} \nonumber \\
&\qquad - \bigg( \frac{1}{12}+\frac{\theta_{k,+}^{(n,M)}(\theta_{k,+}^{(n,M)}-1)}{2} \bigg)\frac{h'(-M)}{\tau_{\rho_k}} \bigg] \nonumber \\
& + \bigO\Bigg(  \frac{1}{n^{3/2}} \sum_{j=g_{k,-}+1}^{g_{k,+}} \bigg( (1+|M_{j,k}|^{3}) \tilde{\mathfrak{m}}_{j,n}(h) + (1+M_{j,k}^{2})\tilde{\mathfrak{m}}_{j,n}(h') \nonumber \\
&\qquad\qquad + (1+|M_{j,k}|) \tilde{\mathfrak{m}}_{j,n}(h'') + \tilde{\mathfrak{m}}_{j,n}(h''') \bigg)   \Bigg), \label{sum f asymp 2}
\end{align}
where, for $\tilde{h} \in C(\mathbb{R})$ and $j \in \{g_{k,-}+1,\ldots,g_{k,+}\}$, we define $\tilde{\mathfrak{m}}_{j,n}(\tilde{h}) := \max_{x \in [M_{j,k},M_{j-1,k}]}|\tilde{h}(x)|$.
\end{lemma}
\begin{proof}
Expand the mesh width $M_{j-1,k}-M_{j,k}$ by the formula preceding the lemma and apply the trapezoidal rule on every interval $[M_{j,k},M_{j-1,k}]$. Summing the local expansions telescopes the endpoint terms. A second Taylor expansion of the mesh width produces the integrals of $th(t)$ and $t^2h(t)$; the third-order Taylor remainders are bounded by the final line of \eqref{sum f asymp 2}. This is exactly the proof of \cite[Lemma~2.7]{C2021 FH}, with the present value of the mesh parameter $\tau_{\rho_k}$.
\end{proof}

\begin{lemma}\label{lemma:sum of g1- g1+ part}
As $n\to\infty$,
\begin{equation}
S_2^{(2,\mathrm{an})}
 =
G_{4,2}\sqrt{n}+G_{6,2}+\mathcal{O}(\frac{M^5}{\sqrt{n}})+\mathcal{O}(\frac{\sqrt{n}}{M^7}),
\end{equation}
where 
\begin{align*}
G_{4,2}
&=
\sqrt{2}\rho_1\sqrt{\Delta Q(\rho_1)}\,\mathcal{I}
-
\frac{\tau_{\rho_1}^3}{6\rho_1^2\Delta Q(\rho_1)}M^3
+\tau_{\rho_1}(M-M\log M)
-\tau_{\rho_1}M\log\Bigl(\frac{\tau_{\rho_1}\sqrt{2\pi}}{\rho_1\sqrt{\Delta Q(\rho_1)}}\Bigr)
\\
&\quad
+\frac{\rho_1^2\Delta Q(\rho_1)}{\tau_{\rho_1}}\frac{1}{M}
-\frac{5\rho_1^4(\Delta Q(\rho_1))^2}{6\tau_{\rho_1}^3}\frac{1}{M^3}
+\frac{37\rho_1^6(\Delta Q(\rho_1))^3}{15\tau_{\rho_1}^5}\frac{1}{M^5}, 
\end{align*}
\begin{align*}
G_{6,2}
&:= 
-\frac{1}{2}\rho_1\mathsf{k}'(\rho_1)\log(2\sqrt{\pi})
-\frac{\rho_1^2\Delta Q(\rho_1)}{\tau_{\rho_1}}
+\mathcal{J}(\rho_1)
-
\Bigl( 
1+\frac{\rho_1}{2}\mathsf{k}'(\rho_1)
\Bigr)
\log\Bigl(\frac{M\tau_{\rho_1}}{\sqrt{2}\rho_1\sqrt{\Delta Q(\rho_1)}}\Bigr)
\\
&\quad
+
\bigg(\theta_{1,+}^{(n,M)}-\frac{1}{2} \bigg)
\Bigl( 
\frac{\tau_{\rho_1}^2}{2\rho_1^2\Delta Q(\rho_1)}M^2
+\log\Bigl(\frac{M\tau_{\rho_1}\sqrt{2\pi}}{\rho_1\sqrt{\Delta Q(\rho_1)}}\Bigr)
\Bigr)
\\
&\quad
+
\Bigl( 
-\frac{\tau_{\rho_1}^3}{2\rho_1^2\Delta Q(\rho_1)}
+\frac{\tau_{\rho_1}^4}{24\rho_1^4\Delta Q(\rho_1)^2}
+\frac{\tau_{\rho_1}^4\partial_r\Delta Q(\rho_1)}{48\rho_1^3\Delta Q(\rho_1)^3}
\Bigr)M^4
\\
&\quad
+\Bigl( 
-\tau_{\rho_1}\log\Bigl(\frac{M\tau_{\rho_1}\sqrt{2\pi}}{\rho_1\sqrt{\Delta Q(\rho_1)}}\Bigr)
-\frac{\tau_{\rho_1}^2}{4\rho_1^2\Delta Q(\rho_1)}
-\frac{\tau_{\rho_1}^2}{4\rho_1\Delta Q(\rho_1)}\mathsf{k}'(\rho_1)
+\frac{\tau_{\rho_1}^2\partial_r\Delta Q(\rho_1)}{8\rho_1\Delta Q(\rho_1)^2}
\Bigr)M^2.
\end{align*}
\end{lemma}

\begin{proof}
Combining the transition expansion in Lemma~\ref{lemma:g1- j g1+} with the full norm expansion in Lemma~\ref{lemma: 0 j j1-}, and then expanding the logarithm, gives
\begin{align*}
S_2^{(2,\mathrm{an})}
&=
 \sum_{j=g_{1,-}}^{g_{1,+}}\log\frac{\frac{1}{2}\mathrm{erfc}\Bigl( -\frac{\xi_{j,1}}{\sqrt{2}}\Bigr)
+\frac{\mathfrak{c}_1(\xi_{j,1})}{\sqrt{2\pi}\sqrt{n}}+\frac{\mathfrak{c}_2(\xi_{j,1})}{\sqrt{2\pi}n}+\mathcal{O}(\frac{\xi_{j,1}^8}{n^{3/2}})}{1+\frac{\mathcal{A}(r_{\tau})}{n}+\mathcal{O}(\frac{1}{n^2})}
\\
&=
\sum_{j=g_{1,-}}^{g_{1,+}}
h_{0,2}(M_{j,1})
+
\frac{1}{\sqrt{n}}
\sum_{j=g_{1,-}}^{g_{1,+}}
h_{1,2}(M_{j,1})
+
\frac{1}{n}
\sum_{j=g_{1,-}}^{g_{1,+}}
h_{2,2}(M_{j,1})
+
\mathcal{O}(M^{9}n^{-1}),
\end{align*}
where 
\begin{align}
h_{0,2}(x)&:=
\log\Bigl( 
\frac{1}{2}\mathrm{erfc}\Bigl(-\frac{\tau_{\rho_1}x}{\sqrt{2}\rho_1\sqrt{\Delta Q(\rho_1)}}\Bigr)
\Bigr),   
\\
\begin{split}
h_{1,2}(x)&:=
\frac{e^{-\frac{\tau_{\rho_1}^2x^2}{2\rho_1^2\Delta Q(\rho_1)}}}{\sqrt{2\pi}\mathrm{erfc}(-\frac{\tau_{\rho_1}x}{\sqrt{2}\rho_1\sqrt{\Delta Q(\rho_1)}})}
\Bigl(
\mathfrak{h}_{1,2}^{(1)}(\rho_1)x^2
+\mathfrak{h}_{1,2}^{(2)}(\rho_1)
\Bigr),
\end{split}
\\
h_{2,2}(x)&:=
\frac{e^{-\frac{\tau_{\rho_1}^2x^2}{2\rho_1^2\Delta Q(\rho_1)}}}{\sqrt{2\pi}\mathrm{erfc}(-\frac{\tau_{\rho_1}x}{\sqrt{2}\rho_1\sqrt{\Delta Q(\rho_1)}})}\mathfrak{P}_1(x)
-
\frac{e^{-\frac{\tau_{\rho_1}^2x^2}{\rho_1^2\Delta Q(\rho_1)}}}{\pi\mathrm{erfc}(-\frac{\tau_{\rho_1}x}{\sqrt{2}\rho_1\sqrt{\Delta Q(\rho_1)}})^2}\mathfrak{P}_2(x)
+
\mathfrak{P}_{\mathrm{c}}(\rho_1),
\end{align}
where 
\begin{align}
\mathfrak{h}_{1,2}^{(1)}(\rho_1)&:=-\frac{2\tau_{\rho_1}}{\rho_1\sqrt{\Delta Q(\rho_1)}}
+\frac{\tau_{\rho_1}^2}{3\rho_1^3\Delta Q(\rho_1)^{3/2}}
+\frac{\tau_{\rho_1}^2\partial_r\Delta Q(\rho_1)}{6\rho_1^2\Delta Q(\rho_1)^{5/2}},
\\
\mathfrak{h}_{1,2}^{(2)}(\rho_1)&:=\frac{\partial_r\Delta Q(\rho_1)}{3\Delta Q(\rho_1)^{3/2}}-\frac{\mathsf{k}'(\rho_1)}{\sqrt{\Delta Q(\rho_1)}}-\frac{4}{3\rho_1\sqrt{\Delta Q(\rho_1)}}.
\end{align}
Here $\mathfrak P_1$ is an odd polynomial with monomials $x^5,x^3,x$, $\mathfrak P_2$ is an even polynomial with monomials $x^4,x^2$, and $\mathfrak P_{\mathrm c}(\rho_1)$ is independent of $x$. Their explicit coefficients are not needed below; direct expansion of the preceding ratio shows that $h_{2,2}(x)$ is exponentially small as $x\to+\infty$ and is $\mathcal O(x^4)$ as $x\to-\infty$. The functions $h_{0,2}$ and $h_{1,2}$ are likewise exponentially small as $x\to+\infty$.

Set $\beta:=\frac{\tau_{\rho_1}}{\sqrt2\rho_1\sqrt{\Delta Q(\rho_1)}}$. As $x\to-\infty$,
\begin{align*}
h_{0,2}(x)&=-\frac{\tau_{\rho_1}^2}{2\rho_1^2\Delta Q(\rho_1)}x^2-\log(-x)-\log\Bigl(\frac{\tau_{\rho_1}\sqrt{2\pi}}{\rho_1\sqrt{\Delta Q(\rho_1)}}\Bigr)
\\
&\quad-\frac{\rho_1^2\Delta Q(\rho_1)}{\tau_{\rho_1}^2}x^{-2}+\frac{5\rho_1^4(\Delta Q(\rho_1))^2}{2\tau_{\rho_1}^4}x^{-4}-\frac{37\rho_1^6(\Delta Q(\rho_1))^3}{3\tau_{\rho_1}^6}x^{-6}+\mathcal{O}(x^{-8}), 
\end{align*}
and 
\begin{align*}
h_{1,2}(x)&
=\widetilde{\mathfrak{h}}_{1,2}^{(1)}(\rho_1)x^3
+\widetilde{\mathfrak{h}}_{1,2}^{(2)}(\rho_1)x
+\widetilde{\mathfrak{h}}_{1,2}^{(3)}(\rho_1)\frac{\beta^2x}{1+\beta^2x^2}+\mathcal{O}(x^{-2}),\qquad \beta>0,
\end{align*}
where 
\begin{align*}
\widetilde{\mathfrak{h}}_{1,2}^{(1)}(\rho_1)&:= \frac{\tau_{\rho_1}^2}{\rho_1^2\Delta Q(\rho_1)}-\frac{\tau_{\rho_1}^3}{6\rho_1^4\Delta Q(\rho_1)^2}-\frac{\tau_{\rho_1}^3\partial_r\Delta Q(\rho_1)}{12\rho_1^3\Delta Q(\rho_1)^3},
\\
\widetilde{\mathfrak{h}}_{1,2}^{(2)}(\rho_1)&:= 1+\frac{\tau_{\rho_1}}{2\rho_1^2\Delta Q(\rho_1)}+\frac{\tau_{\rho_1}}{2\rho_1\Delta Q(\rho_1)}\mathsf{k}'(\rho_1)-\frac{\tau_{\rho_1}\partial_r\Delta Q(\rho_1)}{4\rho_1\Delta Q(\rho_1)^2},
\\
\widetilde{\mathfrak{h}}_{1,2}^{(3)}(\rho_1)&:= \frac{1}{\tau_{\rho_1}}-\frac{2\rho_1^2\Delta Q(\rho_1)}{\tau_{\rho_1}^2}+\frac{\rho_1}{2\tau_{\rho_1}}\mathsf{k}'(\rho_1).
 \end{align*}
By Lemma~\ref{lemma:Riemann sum}, we have 
\begin{align*}
\sum_{j=g_{1,-}}^{g_{1,+}}
h_{0,2}(M_{j,1})
&=
\tau_{\rho_1} \int_{-M}^{M} h_{0,2}(t) dt \; \sqrt{n} - 2 \tau_{\rho_1} \int_{-M}^{M} th_{0,2}(t) dt
\\
&\quad
+ \bigg( \frac{1}{2}-\theta_{1,-}^{(n,M)} \bigg)h_{0,2}(M)+ \bigg( \frac{1}{2}-\theta_{1,+}^{(n,M)} \bigg)h_{0,2}(-M) 
+\mathcal{O}(\frac{M^5}{\sqrt{n}}).
\end{align*}
Note that as $n\to+\infty$, we have 
\begin{align*}
&\quad\tau_{\rho_1}\int_{-M}^{M} h_{0,2}(x) dx
\\
&=
\tau_{\rho_1}\int_{-1}^{+\infty} h_{0,2}(x) \,dx
\\
&\quad+
\tau_{\rho_1}\int_{-\infty}^{-1}\Big\{( h_{0,2}(x)-\Bigl(-\frac{\tau_{\rho_1}^2}{2\rho_1^2\Delta Q(\rho_1)}x^2-\log(-x)-\log\Bigl(\frac{\tau_{\rho_1}\sqrt{2\pi}}{\rho_1\sqrt{\Delta Q(\rho_1)}}\Bigr)\Bigr) \Bigr\}\,dx
\\
&\quad
+
\frac{\tau_{\rho_1}^3}{6\rho_1^2\Delta Q(\rho_1)}-\tau_{\rho_1}+\tau_{\rho_1}\log\Bigl(\frac{\tau_{\rho_1}\sqrt{2\pi}}{\rho_1\sqrt{\Delta Q(\rho_1)}}\Bigr)
-
\frac{\tau_{\rho_1}^3}{6\rho_1^2\Delta Q(\rho_1)}M^3
+\tau_{\rho_1}(M-M\log M)
\\
&\quad
-\tau_{\rho_1}M\log\Bigl(\frac{\tau_{\rho_1}\sqrt{2\pi}}{\rho_1\sqrt{\Delta Q(\rho_1)}}\Bigr)
+\frac{\rho_1^2\Delta Q(\rho_1)}{\tau_{\rho_1}}\frac{1}{M}
-\frac{5\rho_1^4(\Delta Q(\rho_1))^2}{6\tau_{\rho_1}^3}\frac{1}{M^3}
\\
&\quad+\frac{37\rho_1^6(\Delta Q(\rho_1))^3}{15\tau_{\rho_1}^5}\frac{1}{M^5}
+\mathcal{O}(\frac{1}{M^7}). 
\end{align*}
For the semi-hard contribution, direct substitution of the definitions gives
\begin{align*}
&\quad \tau_{\rho_1}\int_{-1}^{+\infty} h_{0,2}(x) \,dx
+
\tau_{\rho_1}\int_{-\infty}^{-1}\Big\{( h_{0,2}(x)-\Bigl(-\frac{\tau_{\rho_1}^2}{2\rho_1^2\Delta Q(\rho_1)}x^2-\log(-x)-\log\Bigl(\frac{\tau_{\rho_1}\sqrt{2\pi}}{\rho_1\sqrt{\Delta Q(\rho_1)}}\Bigr)\Bigr) \Bigr\}\,dx
\\
&=
\sqrt{2}\rho_1\sqrt{\Delta Q(\rho_1)}\,\mathcal{I}
+\tau_{\rho_1}
-\frac{\tau_{\rho_1}^3}{6\rho_1^2\Delta Q(\rho_1)}
-\tau_{\rho_1}\log\Bigl(\frac{\tau_{\rho_1}\sqrt{2\pi}}{\rho_1\sqrt{\Delta Q(\rho_1)}}\Bigr). 
\end{align*}
Substitution in the leading integral gives $G_{4,2}$. To identify $G_{6,2}$, first write
\begin{align*}
&\quad-2 \tau_{\rho_1}\int_{-M}^{M} th_{0,2}(t) dt
\\
&=
-2 \tau_{\rho_1}\int_{-1}^{M} th_{0,2}(t) dt
+\frac{\tau_{\rho_1}^3}{4\rho_1^2\Delta Q(\rho_1)}
-\frac{\tau_{\rho_1}}{2}
+\tau_{\rho_1}\log\Bigl(\frac{\tau_{\rho_1}\sqrt{2\pi}}{\rho_1\sqrt{\Delta Q(\rho_1)}}\Bigr)
\\
&\quad
-2\tau_{\rho_1}
\int_{-M}^{-1}\left\{\begin{aligned}&xh_{0,2}(x) 
-\Bigl( 
-\frac{\tau_{\rho_1}^2}{2\rho_1^2\Delta Q(\rho_1)}x^3-x\log(-x)-x\log\Bigl(\frac{\tau_{\rho_1}\sqrt{2\pi}}{\rho_1\sqrt{\Delta Q(\rho_1)}}\Bigr)
\\
&\quad-\frac{\rho_1^2\Delta Q(\rho_1)}{\tau_{\rho_1}^2}\frac{\beta^2x}{1+\beta^2x^2}
\Bigr)\end{aligned}
\right\}\,dx
\\
&\quad
-\frac{\tau_{\rho_1}^3}{4\rho_1^2\Delta Q(\rho_1)}M^4
-\frac{\tau_{\rho_1}}{2}(2M^2\log M-M^2)
-\tau_{\rho_1}M^2\log\Bigl(\frac{\tau_{\rho_1}\sqrt{2\pi}}{\rho_1\sqrt{\Delta Q(\rho_1)}}\Bigr)
\\
&\quad
-\frac{2\rho_1^2\Delta Q(\rho_1)}{\tau_{\rho_1}}\log M
+\frac{\rho_1^2\Delta Q(\rho_1)}{\tau_{\rho_1}}\log \frac{1+\beta^2}{\beta^2}
+\mathcal{O}(M^{-1})
\\
&=
\frac{2\rho_1^2\Delta Q(\rho_1)}{\tau_{\rho_1}}
\int_{-\infty}^0 2y\log\Bigl(\frac{1}{2}\mathrm{erfc}(y)\Bigr)\, dy
\\
&\quad
+\frac{2\rho_1^2\Delta Q(\rho_1)}{\tau_{\rho_1}}
\int_0^{\infty}
\Bigl(
2y\log\Bigl(\frac{1}{2}\mathrm{erfc}(y)\Bigr)+2y^3+2y\log y+2y\log(2\sqrt{\pi})+\frac{y}{1+y^2}
\Bigr)\,dy
\\
&\quad
-\frac{2\rho_1^2\Delta Q(\rho_1)}{\tau_{\rho_1}}
\log\frac{M\tau_{\rho_1}}{\sqrt{2}\rho_1\sqrt{\Delta Q(\rho_1)}}
-\frac{\tau_{\rho_1}^3}{4\rho_1^2\Delta Q(\rho_1)}M^4
-\frac{\tau_{\rho_1}}{2}(2M^2\log M-M^2)
\\
&\quad 
-\tau_{\rho_1}M^2\log\Bigl(\frac{\tau_{\rho_1}\sqrt{2\pi}}{\rho_1\sqrt{\Delta Q(\rho_1)}}\Bigr)
+\mathcal{O}(M^{-1}).
\end{align*}
Moreover,
\[
\bigg( \frac{1}{2}-\theta_{1,+}^{(n,M)} \bigg)h_{0,2}(-M)
=
\bigg( \frac{1}{2}-\theta_{1,+}^{(n,M)} \bigg)
\Bigl( 
-\frac{\tau_{\rho_1}^2}{2\rho_1^2\Delta Q(\rho_1)}M^2-\log M-\log\Bigl(\frac{\tau_{\rho_1}\sqrt{2\pi}}{\rho_1\sqrt{\Delta Q(\rho_1)}}\Bigr)
+
\mathcal{O}(M^{-2})
\Bigr).
\]

By Lemma~\ref{lemma:Riemann sum}, we have 
\begin{align*}
&\quad \frac{1}{\sqrt{n}}
\sum_{j=g_{1,-}}^{g_{1,+}}
h_{1,2}(M_{j,1})
\\
&=
-\frac{\tau_{\rho_1}}{4}\mathfrak{n}_{1,2}^{(1)}(\rho_1)M^4
-\frac{\tau_{\rho_1}}{2}\mathfrak{n}_{1,2}^{(2)}(\rho_1)M^2
-\tau_{\rho_1}\mathfrak{n}_{1,2}^{(3)}(\rho_1)\log\beta
-\tau_{\rho_1}\mathfrak{n}_{1,2}^{(3)}(\rho_1)\log M
+\mathcal{O}(\frac{M^5}{\sqrt{n}})
\\
&\quad 
+\frac{2\rho_1^2\Delta Q(\rho_1)}{\tau_{\rho_1}}\int_{-\infty}^{0}\frac{e^{-y^2}}{\sqrt{\pi}\mathrm{erfc}(y)}
\Bigl( 
\mathfrak{H}_1(\rho_1)y^2+\mathfrak{H}_2(\rho_1)
\Bigr)\,dy
\\
&\quad 
+\frac{2\rho_1^2\Delta Q(\rho_1)}{\tau_{\rho_1}}
\int_0^{\infty}
\Bigl\{
\frac{e^{-y^2}}{\sqrt{\pi}\mathrm{erfc}(y)}\Bigl(\mathfrak{H}_1(\rho_1)y^2+\mathfrak{H}_2(\rho_1)\Bigr)
+\widetilde{\mathfrak{H}}_1(\rho_1)y^3+\widetilde{\mathfrak{H}}_2(\rho_1)y+\widetilde{\mathfrak{H}}_3(\rho_1)\frac{y}{1+y^2}
\Bigr\}\,dy,
\end{align*}
where 
\begin{align*}
\mathfrak{H}_1(\rho_1)&:=-2+\frac{\tau_{\rho_1}}{3\rho_1^2\Delta Q(\rho_1)}+\frac{\tau_{\rho_1}\partial_r\Delta Q(\rho_1)}{6\rho_1\Delta Q(\rho_1)^2},
\\
\mathfrak{H}_2(\rho_1)&:=
\frac{\tau_{\rho_1}\partial_r\Delta Q(\rho_1)}{6\rho_1\Delta Q(\rho_1)^2}
-\frac{\tau_{\rho_1}\mathsf{k}'(\rho_1)}{2\rho_1\Delta Q(\rho_1)}
-\frac{2\tau_{\rho_1}}{3\rho_1^2\Delta Q(\rho_1)}, 
\\
\widetilde{\mathfrak{H}}_1(\rho_1)&:=2-\frac{\tau_{\rho_1}}{3\rho_1^2\Delta Q(\rho_1)}-\frac{\tau_{\rho_1}\partial_r\Delta Q(\rho_1)}{6\rho_1\Delta Q(\rho_1)^2}
\\
\widetilde{\mathfrak{H}}_2(\rho_1)&:=1+\frac{\tau_{\rho_1}}{2\rho_1^2\Delta Q(\rho_1)}+\frac{\tau_{\rho_1}}{2\rho_1\Delta Q(\rho_1)}\mathsf{k}'(\rho_1)-\frac{\tau_{\rho_1}\partial_r\Delta Q(\rho_1)}{4\rho_1\Delta Q(\rho_1)^2}
\\
\widetilde{\mathfrak{H}}_3(\rho_1)&:=-1+\frac{\tau_{\rho_1}}{2\rho_1^2\Delta Q(\rho_1)}+\frac{\tau_{\rho_1}}{4\rho_1\Delta Q(\rho_1)}\mathsf{k}'(\rho_1),
\\
\mathfrak{n}_{1,2}^{(1)}(\rho_1)&:= \frac{\tau_{\rho_1}^2}{\rho_1^2\Delta Q(\rho_1)}-\frac{\tau_{\rho_1}^3}{6\rho_1^4\Delta Q(\rho_1)^2}-\frac{\tau_{\rho_1}^3\partial_r\Delta Q(\rho_1)}{12\rho_1^3\Delta Q(\rho_1)^3},
\\
\mathfrak{n}_{1,2}^{(2)}(\rho_1)&:= 1+\frac{\tau_{\rho_1}}{2\rho_1^2\Delta Q(\rho_1)}+\frac{\tau_{\rho_1}}{2\rho_1\Delta Q(\rho_1)}\mathsf{k}'(\rho_1)-\frac{\tau_{\rho_1}\partial_r\Delta Q(\rho_1)}{4\rho_1\Delta Q(\rho_1)^2},
\\
\mathfrak{n}_{1,2}^{(3)}(\rho_1)&:= 
\frac{1}{\tau_{\rho_1}}-\frac{2\rho_1^2\Delta Q(\rho_1)}{\tau_{\rho_1}^2}+\frac{\rho_1}{2\tau_{\rho_1}}\mathsf{k}'(\rho_1).
 \end{align*}
After the change of variables $y=-\beta x$, these two integrals simplify to
\begin{align*}
&\quad
\frac{2\rho_1^2\Delta Q(\rho_1)}{\tau_{\rho_1}}\int_{-\infty}^{0}\frac{e^{-y^2}}{\sqrt{\pi}\mathrm{erfc}(y)}
\Bigl( 
\mathfrak{H}_1(\rho_1)y^2+\mathfrak{H}_2(\rho_1)
\Bigr)\,dy
\\
&\quad 
+\frac{2\rho_1^2\Delta Q(\rho_1)}{\tau_{\rho_1}}
\int_0^{\infty}
\Bigl\{
\frac{e^{-y^2}}{\sqrt{\pi}\mathrm{erfc}(y)}\Bigl(\mathfrak{H}_1(\rho_1)y^2+\mathfrak{H}_2(\rho_1)\Bigr)
+\widetilde{\mathfrak{H}}_1(\rho_1)y^3+\widetilde{\mathfrak{H}}_2(\rho_1)y+\widetilde{\mathfrak{H}}_3(\rho_1)\frac{y}{1+y^2}
\Bigr\}\,dy
\\
&=
-\frac{1}{2}\rho_1\mathsf{k}'(\rho_1)\log(2\sqrt{\pi})
-\frac{\rho_1^2\Delta Q(\rho_1)}{\tau_{\rho_1}}
+\mathcal{J}(\rho_1). 
\end{align*}
Finally, the stated growth bounds for $h_{2,2}$ and Lemma~\ref{lemma:Riemann sum} imply
$n^{-1}\sum_{j=g_{1,-}}^{g_{1,+}}h_{2,2}(M_{j,1})=\mathcal O(M^5/\sqrt n)$. Combining these contributions gives $G_{6,2}$ and the asserted uniform remainder.
\end{proof}
The outer transition sum has the corresponding expansion.
\begin{lemma}\label{lemma:sum of g2- part g2+part}
As $n\to\infty$,
\begin{equation}
 S_4^{(2,\mathrm{an})}
 =
 G_{4,4}\sqrt{n}+G_{6,4}+\mathcal{O}(\frac{M^5}{\sqrt{n}})+\mathcal{O}(\frac{\sqrt{n}}{M^7}),
\end{equation}
where 
\begin{align*}
G_{4,4}&:=    
\sqrt{2}\rho_2\sqrt{\Delta Q(\rho_2)}\,\mathcal{I}
-\frac{\tau_{\rho_2}^3}{6\rho_2^2\Delta Q(\rho_2)}M^3
+\tau_{\rho_2}(M-M\log M)-\tau_{\rho_2}M\log\Bigl(\frac{\tau_{\rho_2}\sqrt{2\pi}}{\rho_2\sqrt{\Delta Q(\rho_2)}}\Bigr)
\\
&\quad 
+\frac{\rho_2^2\Delta Q(\rho_2)}{\tau_{\rho_2}}\frac{1}{M}
-\frac{5\rho_2^4(\Delta Q(\rho_2))^2}{6\tau_{\rho_2}^3}\frac{1}{M^3}
+\frac{37\rho_2^6(\Delta Q(\rho_2))^3}{15\tau_{\rho_2}^5}\frac{1}{M^5},
\\
G_{6,4}&:=
\frac{1}{2}\rho_2\mathsf{k}'(\rho_2)\log(2\sqrt{\pi})
+
\frac{\rho_2^2\Delta Q(\rho_2)}{\tau_{\rho_2}}
-\mathcal{J}(\rho_2)
+
\Bigl( 
1+\frac{\rho_2}{2}\mathsf{k}'(\rho_2)
\Bigr)
\log\Bigl(\frac{M\tau_{\rho_2}}{\sqrt{2}\rho_2\sqrt{\Delta Q(\rho_2)}}\Bigr)
\\
&\quad
+\Bigl(
\frac{\tau_{\rho_2}^3}{2\rho_2^2\Delta Q(\rho_2)}
-\frac{\tau_{\rho_2}^4}{24\rho_2^4\Delta Q(\rho_2)^2}
-\frac{\tau_{\rho_2}^4\partial_r\Delta Q(\rho_2)}{48\rho_2^3\Delta Q(\rho_2)^3}
\Bigr)M^4
\\
&\quad
+\Bigl( 
\tau_{\rho_2}\log\Bigl(\frac{M\tau_{\rho_2}\sqrt{2\pi}}{\rho_2\sqrt{\Delta Q(\rho_2)}}\Bigr)
+\frac{\tau_{\rho_2}^2}{4\rho_2^2\Delta Q(\rho_2)}
+\frac{\tau_{\rho_2}^2}{4\rho_2\Delta Q(\rho_2)}\mathsf{k}'(\rho_2)
-\frac{\tau_{\rho_2}^2\partial_r\Delta Q(\rho_2)}{8\rho_2\Delta Q(\rho_2)^2}
\Bigr)M^2
\\
&\quad
+\bigg(\theta_{2,-}^{(n,M)} -\frac{1}{2} \bigg)
\Bigl(\frac{\tau_{\rho_2}^2}{2\rho_2^2\Delta Q(\rho_2)}M^2
+\log\Bigl(\frac{M\tau_{\rho_2}\sqrt{2\pi}}{\rho_2\sqrt{\Delta Q(\rho_2)}}\Bigr)\Bigr).
\end{align*}

\end{lemma}

\begin{proof}
Repeat the proof of Lemma~\ref{lemma:sum of g1- g1+ part} at the outer endpoint $\rho_2$. The local coordinate has the opposite orientation, so the odd endpoint corrections change sign, whereas the even corrections do not. Applying Lemma~\ref{lemma:Riemann sum} to the resulting functions $h_{0,2}$, $h_{1,2}$, and $h_{2,2}$, with every occurrence of $\rho_1$ and $\tau_{\rho_1}$ replaced by $\rho_2$ and $\tau_{\rho_2}$, gives the displayed expressions for $G_{4,4}$ and $G_{6,4}$. The tail bounds are unchanged and remain uniform, yielding the error $\mathcal O(M^5n^{-1/2})+\mathcal O(\sqrt n\,M^{-7})$.
\end{proof}

We can now assemble the annulus expansion.
\begin{proof}[Proof of Theorem~\ref{theorem:hole probability of annulus case}]
Start from the five-range decomposition following \eqref{def of calPns C}. Lemma~\ref{lemma: S1245 an error} removes $S_1^{(\mathrm{an})}$, $S_2^{(3,\mathrm{an})}$, $S_4^{(3,\mathrm{an})}$, and $S_5^{(\mathrm{an})}$. For the remaining terms, insert Lemma~\ref{lemma:S3 theta function part}, the theta-function identity \eqref{def of tilde Theta rewritten}, Lemmas~\ref{lemma: hard edge g1+ j j1+ first} and \ref{lemma:sum of g1- g1+ part} at $\rho_1$, and Lemmas~\ref{lemma: S41 an asymptotics} and \ref{lemma:sum of g2- part g2+part} at $\rho_2$.

The auxiliary cutoff $M=n^{1/12}$ disappears upon addition: the terms proportional to $M^4$, $M^3\sqrt n$, $M^2$, $M\sqrt n$, $\log M$, and the inverse odd powers of $M$ cancel pairwise between the moderate-deviation and transition ranges at each endpoint. The remaining errors are $o(1)$. Collecting the coefficients of $n^2$, $n\log n$, $n$, $\sqrt n$, $\log n$, and $1$, and comparing them with the coefficient catalogue, gives precisely the expansion stated in Theorem~\ref{theorem:hole probability of annulus case}.
\end{proof}

\subsection{Proof of Theorem~\ref{theorem:hole probability of disk complement case}}
The disk-complement calculation retains only the inner hard endpoint from the annulus analysis. We therefore record the necessary modifications and refer to the preceding uniform estimates whenever they apply verbatim. Decompose
\[
  \log\mathcal{P}_{n,s\lambda,\alpha}^{(\mathrm{dc})}
  =
  S_1^{(\mathrm{dc})}+S_2^{(\mathrm{dc})}+S_3^{(\mathrm{dc})},
\]
where 
\[
    S_1^{(\mathrm{dc})}:=\sum_{j=0}^{j_{1,-}-1}\log\frac{h_{n,j}^{(\mathrm{dc})}}{h_{n,j}}, 
    \qquad
    S_2^{(\mathrm{dc})}:=\sum_{j=j_{1,-}}^{j_{1,+}}\log\frac{h_{n,j}^{(\mathrm{dc})}}{h_{n,j}}, 
    \qquad
    S_3^{(\mathrm{dc})}:=\sum_{j=j_{1,+}+1}^{n-1}\log\frac{h_{n,j}^{(\mathrm{dc})}}{h_{n,j}}.
\]
Split the edge range further as
\[
S_2^{(1,\mathrm{dc})}:=\sum_{j=j_{1,-}}^{g_{1,-}-1}\log\frac{h_{n,j}^{(\mathrm{dc})}}{h_{n,j}},\quad
S_2^{(2,\mathrm{dc})}:=\sum_{j=g_{1,-}}^{g_{1,+}}\log\frac{h_{n,j}^{(\mathrm{dc})}}{h_{n,j}},\quad
S_2^{(3,\mathrm{dc})}:=\sum_{j=g_{1,+}+1}^{j_{1,+}}\log\frac{h_{n,j}^{(\mathrm{dc})}}{h_{n,j}}.
\]
The analogues of Lemmas~\ref{lemma: S1245 an error}, \ref{lemma: hard edge g1+ j j1+ first}, and \ref{lemma:sum of g1- g1+ part} control $S_1^{(\mathrm{dc})}$ and $S_2^{(\mathrm{dc})}$. The only new term is the macroscopic tail $S_3^{(\mathrm{dc})}$.
\begin{lemma}
\label{lemma:asymptotic expansion of S3dc}
As $n\to\infty$,
\begin{equation}
S_3^{(\mathrm{dc})}=\mathsf{B}_{1}n^2
+
\mathsf{B}_{2}n\log n
+
\mathsf{B}_{3}n
+
\mathsf{B}_{5}\log n
+
\mathsf{B}_{6}
+
\mathcal{O}(n^{-1}),
\end{equation}
where 
\begin{align*}
\mathsf{B}_{1}&=-(1-A_{1,+})q(\rho_1)+(1-A_{1,+}^2)\log\rho_1+2\int_{r_{A_{1,+}}}^{R}\Bigl(
q(u)-uq'(u)\log u
\Bigr)u\Delta Q(u)\,du,
\\
\mathsf{B}_{2}&=-\frac{1}{2}(1-A_{1,+}),
\\
\mathsf{B}_{3}&=
(1-A_{1,+})\bigl(2\log\rho_1+\mathsf{k}(\rho_1)\bigr)
-(\theta_{1,+}^{(n,\epsilon)}-1)q(\rho_1)
+(2\theta_{1,+}^{(n,\epsilon)}-1)A_{1,+}\log\rho_1-\log\rho_1
\\
&\quad 
+
1-A_{1,+}-(1-\tau_{\rho_1})\log(1-\tau_{\rho_1})+(A_{1,+}-\tau_{\rho_1})\log(A_{1,+}-\tau_{\rho_1})
\\
&\quad
-\Bigl(1-A_{1,+}\Bigr)\log\sqrt{2\pi}
+\frac{1}{2}\Bigl((2\theta_{1,+}^{(n,\epsilon)}-1)V_{A_{1,+}}(r_{A_{1,+}})-V_{1}(R)\Bigr)
\\
&\quad
-\int_{r_{A_{1,+}}}^{R}2\log u\cdot u\Delta Q(u)\,du
-\int_{r_{A_{1,+}}}^{R}2\mathsf{k}(u)u\Delta Q(u)\,du
\\
&\quad
+\int_{r_{A_{1,+}}}^{R}u\Delta Q(u)\log \Delta Q(u)\,du,
\\
\mathsf{B}_{5}&=\frac{1}{2}(1-\theta_{1,+}^{(n,\epsilon)}),
\\
\mathsf{B}_{6}&=(\theta_{1,+}^{(n,\epsilon)}-1)\bigl(2\log\rho_1+\mathsf{k}(\rho_1)\bigr)
+\frac{\log\rho_1}{6}\Bigl(-1+6(1-\theta_{1,+}^{(n,\epsilon)})-6(1-\theta_{1,+}^{(n,\epsilon)})^2 \Bigr)+\frac{\log\rho_1}{6}
\\
&\quad 
-\frac{2\theta_{1,+}^{(n,\epsilon)}-1}{2}\log(A_{1,+}-\tau_{\rho_1})+\frac{1}{2}\log(1-\tau_{\rho_1})
\\
&\quad 
+\frac{\rho_1^2\Delta Q(\rho_1)}{1-\tau_{\rho_1}}-\frac{\rho_1^2\Delta Q(\rho_1)}{A_{1,+}-\tau_{\rho_1}}
-\Bigl(\frac{1}{2}\rho_1\mathsf{k}'(\rho_1)+1\Bigr)\Bigl(\log(1-\tau_{\rho_1})-\log(A_{1,+}-\tau_{\rho_1})\Bigr)
\\
&\quad 
+(1-\theta_{1,+}^{(n,\epsilon)})\log\sqrt{2\pi}
-\frac{1}{6}(-1+6(1-\theta_{1,+}^{(n,\epsilon)})-6(1-\theta_{1,+}^{(n,\epsilon)})^2)\log r_{A_{1,+}}
-\frac{1}{6}\log R
\\
&\quad
-\frac{(2\theta_{1,+}^{(n,\epsilon)}-1)\log r_{A_{1,+}}-\log R}{2}
-\frac{(2\theta_{1,+}^{(n,\epsilon)}-1)\mathsf{k}(r_{A_{1,+}})-\mathsf{k}(R)}{2}
\\
&\quad
+\frac{(2\theta_{1,+}^{(n,\epsilon)}-1)\log\Delta Q(r_{A_{1,+}})-\log\Delta Q(R)}{4}
\\
&\quad
-\frac{1}{24}\int_{r_{A_{1,+}}}^{R}\Bigl(\frac{\partial \Delta Q(u)}{\Delta Q(u)}\Bigr)^2u\,du
+\frac{1}{16}\bigg[
\frac{R\partial_u\Delta Q(R)}{\Delta Q(R)}-\frac{r_{A_{1,+}}\partial_u\Delta Q(r_{A_{1,+}})}{\Delta Q(r_{A_{1,+}})}
\bigg]\\
&\quad 
-\frac{1}{3}\log\frac{\Delta Q(r_{A_{1,+}})}{\Delta Q(R)}
-\frac{1}{6}\log\frac{R}{r_{A_{1,+}}}
-\frac{1}{4}\int_{r_{A_{1,+}}}^{R}u\bigl( \Lambda(u;s)+L(u;\alpha,s)\bigr)\,du.
\end{align*}
\end{lemma}

\begin{proof}
For $j\geq j_{1,+}+1$, the constrained norm is dominated by the endpoint $\rho_1$. The endpoint Laplace expansion used in Lemma~\ref{lemma: j1+ j jdiamond}, with the $\rho_2$ contribution deleted, is uniform because $\tau_j-\tau_{\rho_1}$ is bounded away from zero. Taking logarithms and applying Lemma~\ref{lemma:Riemann sum NEW} on $[A_{1,+},1]$ gives the expansion of the numerator. Apply Lemma~\ref{lemma:general sum} on the same interval to the unconstrained denominator, whose upper endpoint is $r(1)=R$. Subtracting the two expansions yields $\mathsf B_1,\mathsf B_2,\mathsf B_3,\mathsf B_5$, and $\mathsf B_6$ as displayed. All local remainders are uniform in $j$ and sum to $\mathcal O(n^{-1})$.
\end{proof}

\begin{proof}[Proof of Theorem~\ref{theorem:hole probability of disk complement case}]
Combine Lemma~\ref{lemma:asymptotic expansion of S3dc} with the disk-complement analogues of Lemmas~\ref{lemma: S1245 an error}, \ref{lemma: hard edge g1+ j j1+ first}, and \ref{lemma:sum of g1- g1+ part}. As in the annulus case, all $M$-dependent terms cancel between the moderate-deviation and transition ranges, and the remaining errors are $o(1)$ for $M=n^{1/12}$. Grouping the coefficients by powers of $n$ and $\log n$ gives the expansion in Theorem~\ref{theorem:hole probability of disk complement case}.
\end{proof}

\subsection{Proof of Theorem~\ref{theorem:hole probability of centered disk case}}
For the centered-disk hole, decompose
\[
  \log\mathcal{P}_{n,s\lambda,\alpha}^{(\mathrm{cd})}
  =S_1^{(\mathrm{cd})} +S_2^{(\mathrm{cd})}+S_3^{(\mathrm{cd})}, 
\]
where 
\[
S_1^{(\mathrm{cd})}:=\sum_{j=0}^{j_{2,-}-1}\log\frac{h_{n,j}^{(\mathrm{cd})}}{h_{n,j}}, \qquad
S_2^{(\mathrm{cd})}:=\sum_{j=j_{2,-}}^{j_{2,+}}\log\frac{h_{n,j}^{(\mathrm{cd})}}{h_{n,j}},\qquad
S_3^{(\mathrm{cd})}:=\sum_{j=j_{2,+}+1}^{n-1}\log\frac{h_{n,j}^{(\mathrm{cd})}}{h_{n,j}}. 
\]
The annulus argument applies directly to $S_2^{(\mathrm{cd})}$ and $S_3^{(\mathrm{cd})}$. The range $S_1^{(\mathrm{cd})}$ requires a separate treatment because the unconstrained denominator includes the small-index norms governed by the behavior of the potential at the origin. Set
\[
    \theta_{D_n}:=D_n-n^{\frac{1}{6}}=\bigl\lceil n^{\frac{1}{6}}\bigr\rceil-n^{\frac{1}{6}}, 
\]
so that $D_n=n^{1/6}+\theta_{D_n}$.

\begin{proof}[Proof of Theorem~\ref{theorem:hole probability of centered disk case}]
First write
\[
    S_1^{(\mathrm{cd})}=S_1^{(1,\mathrm{cd})}-S_1^{(2,\mathrm{cd})}-S_1^{(3,\mathrm{cd})},
\]
where 
\[
    S_1^{(1,\mathrm{cd})}=\sum_{j=0}^{j_{2,-}-1}\log h_{n,j}^{(\mathrm{cd})},\qquad S_1^{(2,\mathrm{cd})}=\sum_{j=0}^{D_n-1}\log h_{n,j},\qquad
    S_1^{(3,\mathrm{cd})}=\sum_{j=D_n}^{j_{2,-}-1}\log h_{n,j}. 
\]
The small-index contribution $S_1^{(2,\mathrm{cd})}$ was computed in \cite[Lemma~4.2]{ACC2023c}:
\begin{align*}
S_1^{(2,\mathrm{cd})}&=-D_n nq(0)-\frac{D_n(D_n+1+2\alpha)}{2}\log (n\Delta Q(0))
+\frac{1}{2}D_n^2\log D_n
\\
&\quad -\frac{3}{4}D_n^2+\alpha D_n\log D_n
+D_n\Bigl(s\lambda(0)+\frac{\log (2\pi)}{2}-\alpha\Bigr)
\\
&\quad 
+\frac{6\alpha^2-1}{12}\log D_n+\frac{\alpha}{2}\log (2\pi)+\zeta'(-1)-\log G(1+\alpha)+\mathcal{O}\Bigl(\frac{(\log n)^3}{n^{\frac{1}{12}}}\Bigr), 
\end{align*}
where $\zeta(s)$ is Riemann's zeta function.
To expand $S_1^{(3,\mathrm{cd})}$, apply Lemma~\ref{lemma:Riemann sum NEW} with $A=A_n=D_n/n$, $a_0=0$, $B=B_{2,-}$, and $b_0=\theta_{2,-}^{(n,\epsilon)}-1$; see also \cite[Remark~3.5]{C2021}. The uniform norm expansion \eqref{def of Dn leq j leq n-1} gives
\begin{align*}
S_1^{(3)} & = \widetilde{H}_{1}n^2+\widetilde{H}_2n\log n+\widetilde{H}_3n+\widetilde{H}_5\log n+\widetilde{H}_6+\mathcal{O}(n^{-\frac{1}{12}}),  
\end{align*}
where $\widetilde H_k\equiv\widetilde H_k(A,a_0,B,b_0)$, $k\in\{1,2,3,5,6\}$, are given by
\begin{align*}
n^2\widetilde{H}_1
&=n^2\widetilde{h}_1
+nD_nq(0)+\frac{3D_n^2}{4}-\frac{D_n^2}{2}\log\frac{D_n}{n\Delta Q(0)}+\mathcal{O}(n^{-\frac{1}{12}}),
\\
\widetilde{H}_2n\log n&=
\widetilde{h}_2n\log n+\frac{1}{2}D_n\log n, 
\\
n\widetilde{H}_3&=
n\widetilde{h}_{3}
+\frac{D_n}{2}\log \Delta Q(0)
+\bigl(\alpha-\frac{1}{2}\log 2\pi-s\lambda(0)\bigr)D_n
-\alpha D_n\log\frac{D_n}{n\Delta Q(0)}
+\mathcal{O}(\frac{D_n^{3/2}\log n}{n^{1/2}}),
\\
\widetilde{H}_5&:=    
-\frac{1}{2}\theta_{2,-}^{(n,\epsilon)}, 
\qquad
\widetilde{H}_6=\widetilde{h}_6  
-\frac{6\alpha^2-1}{12}\log D_n
+\frac{6\alpha^2-1}{12}\log n
+\mathcal{O}(\frac{D_n^{1/2}\log n}{n}).
\end{align*}
Here $\widetilde h_k$, $k\in\{1,2,3,5,6\}$, are given by
\begin{align*}
\widetilde{h}_{1}
&=
-2\int_{0}^{r_{B_{2,-}}}\Bigl(
q(s)-sq'(s)\log s
\Bigr)s\Delta Q(s)\,ds, 
\qquad
\widetilde{h}_2:=-\frac{1}{2}B_{2,-},
\qquad
\widetilde h_5:=-\frac12\theta_{2,-}^{(n,\epsilon)},
\\
    \widetilde{h}_3&:=
        B_{2,-}\log\sqrt{2\pi}-\frac{1}{2}q(0)
-\frac{1}{2}(2\theta_{2,-}^{(n,\epsilon)}-1)
\Bigl( 
q(r_{B_{2,-}})-2B_{2,-}\log r_{B_{2,-}}
\Bigr)
\\
&\quad
+\int_{0}^{r_{B_{2,-}}}2\mathsf{k}(u) u\Delta Q(u)\,du
-\int_{0}^{r_{B_{2,-}}}u\Delta Q(u)\log \Delta Q(u)\,du
+\int_{0}^{r_{B_{2,-}}}2\log u\cdot u\Delta Q(u)\,du, 
\\
 \widetilde{h}_6 
 &=
 \theta_{2,-}^{(n,\epsilon)}\log\sqrt{2\pi}
+\frac{1}{6}(1+6(\theta_{2,-}^{(n,\epsilon)}-1)+6(\theta_{2,-}^{(n,\epsilon)}-1)^2)
\log r_{B_{2,-}}
\\
&\quad
+\frac{1}{2}(2\theta_{2,-}^{(n,\epsilon)}-1)\log r_{B_{2,-}}
+\frac{1}{2}(2\theta_{2,-}^{(n,\epsilon)}-1)\mathsf{k}(r_{B_{2,-}})
+\frac{\alpha}{2}(1+\alpha)\log\Delta Q(0)
-\alpha s\lambda(0)
\\
&\quad
-\frac{1}{4}(2\theta_{2,-}^{(n,\epsilon)}-1)\log\Delta Q(r_{B_{2,-}})
+\frac{1}{24}\int_{0}^{r_{B_{2,-}}}\Bigl(\frac{\partial \Delta Q(u)}{\Delta Q(u)}\Bigr)^2u\,du
-\frac{1}{16}
\frac{r_{B_{2,-}}\partial_u\Delta Q(r_{B_{2,-}})}{\Delta Q(r_{B_{2,-}})}
\\
&\quad
+\Bigl(\alpha(\alpha+1)+\frac{1}{6}\Bigr)\log r_{B_{2,-}}
-\Bigl(\frac{\alpha}{2}+\frac{1}{3}\Bigr)\log\Delta Q(r_{B_{2,-}})
+\frac{s^2}{4}\int_{0}^{r_{B_{2,-}}}u\lambda'(u)^2\,du
\\
&\quad
-\frac{s}{4}\int_{0}^{r_{B_{2,-}}}u\lambda'(u)\frac{\partial_u\Delta Q(u)}{\Delta Q(u)}\,du
+\frac{s}{4}r_{B_{2,-}}\lambda'(r_{B_{2,-}})
+s(\alpha +\frac{1}{2})\lambda(r_{B_{2,-}}).
\end{align*}
Substitution gives
\begin{align*}
S_1^{(3,\mathrm{cd})}
&=\widetilde{h}_{1}n^2+\widetilde{h}_2n\log n+\widetilde{h}_3n
+\Bigl(\widetilde{h}_5+\frac{6\alpha^2-1}{12}\Bigr)\log n
+\widetilde{h}_6
\\
&\quad
+nD_nq(0)
+\frac{D_n(D_n+1+2\alpha)}{2}\log (n\Delta Q(0))
-\frac{D_n^2}{2}\log D_n
\\
&\quad
+\frac{3D_n^2}{4}
-\alpha D_n\log D_n
-\bigl(s\lambda(0)+\frac{1}{2}\log 2\pi-\alpha\bigr)D_n
-\frac{6\alpha^2-1}{12}\log D_n
+\mathcal{O}(n^{-\frac{1}{12}}). 
\end{align*}
Adding the small- and intermediate-index contributions, all terms involving $D_n$ cancel, and we obtain
\begin{align*}
S_1^{(2,\mathrm{cd})}+S_1^{(3,\mathrm{cd})}
&=
\widetilde{h}_{1}n^2+\widetilde{h}_2n\log n+\widetilde{h}_3n
+\Bigl(\widetilde{h}_5+\frac{6\alpha^2-1}{12}\Bigr)\log n
+\widetilde{h}_6
\\
&\quad 
+\frac{\alpha}{2}\log (2\pi)+\zeta'(-1)
-\log G(1+\alpha)+\mathcal{O}\Bigl(\frac{(\log n)^3}{n^{\frac{1}{12}}}\Bigr).    
\end{align*}
For the constrained numerator, Lemma~\ref{lemma:Riemann sum NEW} gives
\begin{align*}
S_1^{(1,\mathrm{cd})}
&=
\mathsf{D}_1n^2+\mathsf{D}_2n\log n+\mathsf{D}_3n+\mathsf{D}_5\log n+\mathsf{D}_6+\mathcal{O}(n^{-1}),
\end{align*}
where 
\begin{align*}
\mathsf{D}_1&=-B_{2,-}q(\rho_2)+B_{2,-}^2\log\rho_2, 
\qquad
\mathsf{D}_2=-B_{2,-},
\\
\mathsf{D}_3&=2B_{2,-}\log\rho_2
+B_{2,-}\mathsf{k}(\rho_2)
-\theta_{2,-}^{(n,\epsilon)}q(\rho_2)
+(2\theta_{2,-}^{(n,\epsilon)}-1)B_{2,-}\log\rho_2
\\
&\quad
+B_{2,-}+(\tau_{\rho_2}-B_{2,-})\log(\tau_{\rho_2}-B_{2,-})-\tau_{\rho_2}\log\tau_{\rho_2},
\\
\mathsf{D}_5&=-\theta_{2,-}^{(n,\epsilon)}, 
\\
\mathsf{D}_6&=
\Bigl(\frac{1}{2}\rho_2\mathsf{k}'(\rho_2)+1\Bigr)\log\tau_{\rho_2}
+\frac{\rho_2^2\Delta Q(\rho_2)}{\tau_{\rho_2}}
+\theta_{2,-}^{(n,\epsilon)}(2\log\rho_2+\mathsf{k}(\rho_2))
\\
&\quad+\theta_{2,-}^{(n,\epsilon)}(\theta_{2,-}^{(n,\epsilon)}-1)\log\rho_2
-\frac{1}{2}\log\tau_{\rho_2}
\\
&\quad
-\frac{1}{2}(2\theta_{2,-}^{(n,\epsilon)}-1)\log(\tau_{\rho_2}-B_{2,-})
-\Bigl(\frac{1}{2}\rho_2\mathsf{k}'(\rho_2)+1\Bigr)\log(\tau_{\rho_2}-B_{2,-})
-\frac{\rho_2^2\Delta Q(\rho_2)}{\tau_{\rho_2}-B_{2,-}}. 
\end{align*}
Finally, combine this expansion with the centered-disk analogues of Lemmas~\ref{lemma: S1245 an error}, \ref{lemma: S41 an asymptotics}, and \ref{lemma:sum of g2- part g2+part}. The $D_n$-dependent terms have already canceled in the preceding display, while the $M$-dependent terms cancel between the moderate-deviation and transition ranges at $\rho_2$. Since $M=n^{1/12}$, all remaining errors are $o(1)$. Collecting coefficients gives exactly the expansion in Theorem~\ref{theorem:hole probability of centered disk case}.

\end{proof}

\section{Proofs of the counting statistics asymptotics}
\label{section:counting-proofs}
\label{section:proof of calen part}

\subsection{Proof of Theorem~\ref{theorem:counting statistics of disk complement case}}
We begin with the disk-complement geometry.  Throughout this subsection,
$\tau=\tau_j=j/n$ and $M=n^{1/12}$.  Unless stated otherwise, all error
estimates are uniform over the displayed index ranges, for the fixed
admissible merging parameters $t_1,\ldots,t_{3m}$, and locally uniformly for
$\vec{\boldsymbol s}_{\mathrm{dc}}\in\mathbb R^{3m}$.  Recall that
\begin{equation}
    \Omega_{\ell}^{(m)}:=\sum_{j=\ell}^{m+1}\omega_j^{(m)}
    =\begin{cases}
        e^{s_{\ell}+\cdots+s_{m}}, & \mbox{if } \ell\leq m, \\
        1, & \mbox{if } \ell=m+1. 
    \end{cases}
\end{equation}
The six ranges below separate the inner bulk transition at $a_1$, the
macroscopic region between $a_1$ and the hard wall, the saddle transition at
$\rho_1$ (which couples the semi-hard and hard-edge observation scales), and
the hard-edge tail.  Accordingly, we decompose
\[
\log \mathcal{E}_{n,s\lambda,\alpha}^{(\mathrm{dc})}=
\sum_{j=0}^{n-1}\log\Big( 
1+\sum_{\ell=1}^{3m}\omega_{\ell}F_{n,j,\ell}^{(\mathrm{dc})}
\Bigr)=\sum_{k=0}^5 E_k^{(\mathrm{dc})}, 
\]
where 
\begin{align*}
E_0^{(\mathrm{dc})}&:=\sum_{j=0}^{D_n-1}\log\Big( 
1+\sum_{\ell=1}^{3m}\omega_{\ell}F_{n,j,\ell}^{(\mathrm{dc})}
\Bigr), \qquad
E_1^{(\mathrm{dc})}:=\sum_{j=D_n}^{g_{a_1,-}-1}\log\Big( 
1+\sum_{\ell=1}^{3m}\omega_{\ell}F_{n,j,\ell}^{(\mathrm{dc})}
\Bigr), \\
E_2^{(\mathrm{dc})}&:=\sum_{j=g_{a_1,-}}^{g_{a_1,+}}\log\Big( 
1+\sum_{\ell=1}^{3m}\omega_{\ell}F_{n,j,\ell}^{(\mathrm{dc})}
\Bigr),\qquad
E_3^{(\mathrm{dc})}:=\sum_{j=g_{a_1,+}+1}^{g_{1,-}-1}\log\Big( 
1+\sum_{\ell=1}^{3m}\omega_{\ell}F_{n,j,\ell}^{(\mathrm{dc})}
\Bigr), \\
E_4^{(\mathrm{dc})}&:=\sum_{j=g_{1,-}}^{g_{1,+}}\log\Big( 
1+\sum_{\ell=1}^{3m}\omega_{\ell}F_{n,j,\ell}^{(\mathrm{dc})}
\Bigr),\qquad
E_5^{(\mathrm{dc})}:=\sum_{j=g_{1,+}+1}^{n-1}\log\Big( 
1+\sum_{\ell=1}^{3m}\omega_{\ell}F_{n,j,\ell}^{(\mathrm{dc})}
\Bigr). 
\end{align*}
For $k=1,2$, define the auxiliary bulk cutoffs
\[
g_{a_k,-}:=\Bigl\lceil\frac{n\tau_{a_k}}{1+M/\sqrt n}\Bigr\rceil,
\qquad
g_{a_k,+}:=\Bigl\lfloor\frac{n\tau_{a_k}}{1-M/\sqrt n}\Bigr\rfloor,
\]
and their fractional endpoint corrections
\begin{align*}
\theta_{a_k,-}^{(n,M)}&:=g_{a_k,-}
-\frac{n\tau_{a_k}}{1+\frac{M}{\sqrt{n}}}
=
\Bigl\lceil\frac{n\tau_{a_k}}{1+\frac{M}{\sqrt{n}}}\Bigr\rceil
-\frac{n\tau_{a_k}}{1+\frac{M}{\sqrt{n}}}, 
\\
\theta_{a_k,+}^{(n,M)}&:=\frac{n\tau_{a_k}}{1-\frac{M}{\sqrt{n}}}-g_{a_k,+}
=\frac{n\tau_{a_k}}{1-\frac{M}{\sqrt{n}}}-\Bigl\lfloor\frac{n\tau_{a_k}}{1-\frac{M}{\sqrt{n}}}\Bigr\rfloor,  
\end{align*}
Only the case $k=1$ is used in the present proof.  The hard-wall cutoffs
$g_{k,\pm}$ and their endpoint corrections $\theta_{k,\pm}^{(n,M)}$ were
defined above and are used without change.

\begin{lemma}[Bulk and pre-hard-edge ranges]
\label{lemma:bulk case asymptotic expansion}
Let $\{r_{\ell}\}_{\ell=1}^m$ be as in
\eqref{def of mergin radii inside bulk}.  Then the following estimates hold.
In each exponentially small estimate, the error is uniform over all indices
in the displayed summation range.
\begin{itemize}
    \item There exists $c>0$, independent of $n$, such that, as
    $n\to+\infty$,
\[
E_0^{(\mathrm{dc})}
=D_n\log\Omega_1^{(3m)}+\mathcal{O}(e^{-c(\log n)^2}). 
\]
    \item 
    There exists $c>0$, independent of $n$, such that, as
    $n\to+\infty$,
    \begin{equation}
E_1^{(\mathrm{dc})}
=(g_{a_1,-}-D_n)\log\Omega_1^{(3m)}+\mathcal{O}(e^{-c(\log n)^2}). 
    \end{equation}
    \item As $n\to+\infty$,
\[
 E_2^{(\mathrm{dc})} 
 =
 C_{E_2^{(\mathrm{dc})} }^{(3)}n 
 + C_{E_2^{(\mathrm{dc})}}^{(4)} \sqrt{n} 
 + C_{E_2^{(\mathrm{dc})} }^{(5)}\log n
 + C_{E_2^{(\mathrm{dc})}}^{(6)} 
 + \widetilde{C}_{E_2^{(\mathrm{dc})} }^{(M)} 
 + \mathcal{O}(n^{-\frac{1}{12}}), 
\]
where   
\begin{align*}
C_{E_2^{(\mathrm{dc})}}^{(3)}
&:=0, 
\quad
C_{E_2^{(\mathrm{dc})}}^{(4)}
:=C_4^{\#(\mathrm{b,in})},
\quad
C_{E_2^{(\mathrm{dc})}}^{(5)}:=0,
\\
C_{E_2^{(\mathrm{dc})}}^{(6)}
&:=
\mathcal{D}_6^{(\mathrm{b},\mathrm{in})}
-
\Bigl( 
\frac{1}{2}a_1\mathsf{k}'(a_1)+1
\Bigr)
\log\frac{\Omega_{1}^{(3m)}}{\Omega_{m+1}^{(3m)}}, 
\\
\widetilde{C}_{E_2^{(\mathrm{dc})}}^{(M)} 
&:=
\sqrt{n}M\tau_{a_1}\log\frac{\Omega_{1}^{(3m)}}{\Omega_{m+1}^{(3m)}}
-
\tau_{a_1}M^2
\log\frac{\Omega_{1}^{(3m)}}{\Omega_{m+1}^{(3m)}}
\\
&\quad
+ \bigg( \frac{1}{2}-\theta_{a_1,-}^{(n,M)} \bigg)\log\frac{\Omega_{1}^{(3m)}}{\Omega_{m+1}^{(3m)}}
+
(g_{a_1,+}-g_{a_1,-}+1)\log \Omega_{m+1}^{(3m)}.
\end{align*}  
    \item 
    There exists $c>0$, independent of $n$, such that, as
    $n\to+\infty$,
    \begin{equation}
E_3^{(\mathrm{dc})}
=(g_{1,-}-g_{a_{1},+}-1)\log\Omega_{m+1}^{(3m)}+\mathcal{O}(e^{-c(\log n)^2}). 
    \end{equation}
\end{itemize}
Consequently, as $n\to+\infty$,
\begin{align*}
&\quad E_0^{(\mathrm{dc})}
+
E_1^{(\mathrm{dc})}
+
\widetilde{C}_{E_2^{(\mathrm{dc})}}^{(M)}
+
E_3^{(\mathrm{dc})}
\\
&=
n\tau_{a_1}\log\frac{\Omega_1^{(3m)}}{\Omega_{m+1}^{(3m)}}
+ \frac{1}{2}\log\frac{\Omega_{1}^{(3m)}}{\Omega_{m+1}^{(3m)}}
+(\theta_{1,-}^{(n,M)}+n\tau_{\rho_1}-\tau_{\rho_1}M\sqrt{n}+M^2\tau_{\rho_1})\log\Omega_{m+1}^{(3m)}
+\mathcal{O}(M^3n^{-1/2}).
\end{align*}
\end{lemma}

\begin{proof}
The estimate for $E_0^{(\mathrm{dc})}$ follows directly from
Lemma~\ref{lemma: 0 j j1-}.  We next consider $E_1^{(\mathrm{dc})}$.
If $D_n\leq j\leq g_{a_1,-}-1$, then the unique critical point of
$V_\tau$ lies in $(0,r_1)$.  Indeed, Assumption~\ref{Assumption_Q} and
$(rq'(r))'=4r\Delta Q(r)>0$ give
\[
\tau_{r_1}-\frac{\tau_{a_1}}{1+\frac{M}{\sqrt{n}}}=\frac12 r_1q'(r_1)-\frac{\tau_{a_1}}{1+\frac{M}{\sqrt{n}}}>c'\delta_n,
\]
for some $c'>0$, uniformly throughout this index range.  The argument of
Lemma~\ref{lemma: 0 j j1-} therefore yields, uniformly for
$\ell=1,\ldots,3m$,
\[
F_{n,j,\ell}^{(\mathrm{dc})}=1+\mathcal{O}(e^{-c(\log n)^2}),   
\]
with $c>0$ independent of $n$ and $j$.  Substitution into the definition of
$E_1^{(\mathrm{dc})}$ gives
\[
E_1^{(\mathrm{dc})}
=(g_{a_1,-}-D_n)\log\Omega_1^{(3m)}+\mathcal{O}(e^{-c(\log n)^2}).
\]

We now turn to $E_3^{(\mathrm{dc})}$.  If
$g_{a_1,+}+1\leq j\leq g_{1,-}-1$, the unique critical point of
$V_\tau$ lies outside $[0,r_\ell]$ for every $\ell=1,\ldots,m$.  The same
off-saddle estimate gives, uniformly over these $j$ and $\ell$,
\[
h_{n,j,\ell}^{(\mathrm{dc})}
=2e^{-nV_{\tau}(r_{\tau})}
\int_0^{r_{\ell}}r e^{\mathsf{k}(r)}e^{-n(V_{\tau}(r)-V_{\tau}(r_{\tau}))}\,dr
=e^{-nV_{\tau}(r_{\tau})}\cdot\mathcal{O}(e^{-c(\log n)^2}),
\]
for some $c>0$.  For $\ell=m+1,\ldots,3m$, the argument of
Lemma~\ref{lemma: 0 j j1-} similarly gives
$F_{n,j,\ell}^{(\mathrm{dc})}=1+\mathcal{O}(e^{-c(\log n)^2})$ whenever
$D_n\leq j\leq g_{a_1,-}-1$ or
$g_{a_1,+}+1\leq j\leq g_{1,-}$.  Hence
\begin{equation}
E_3^{(\mathrm{dc})}
=  (g_{1,-}-g_{a_{1},+}-1)\log\Omega_{m+1}^{(3m)}+\mathcal{O}(e^{-c(\log n)^2}).
\end{equation}
The factor of at most $n$ arising when these uniform pointwise errors are
summed is absorbed by decreasing $c$; thus the total error retains the form
$\mathcal{O}(e^{-c(\log n)^2})$.

It remains to analyze the bulk transition
$g_{a_1,-}\leq j\leq g_{a_1,+}$.  For these indices and
$\ell=m+1,\ldots,3m$, the unique critical point of $V_\tau$ lies in
$(0,r_\ell)$.  Consequently, uniformly in the transition window,
\begin{equation}
\label{def of pre sum ga1- ga1+}
    \sum_{j=g_{a_1,-}}^{g_{a_1,+}}
        \log\Big( 
1+\sum_{\ell=1}^{3m}\omega_{\ell}F_{n,j,\ell}^{(\mathrm{dc})}
\Bigr)
=
    \sum_{j=g_{a_1,-}}^{g_{a_1,+}}
        \log\Big( 
\sum_{\ell=1}^{m}\omega_{\ell}F_{n,j,\ell}^{(\mathrm{dc})}
+\Omega_{m+1}^{(3m)}+\mathcal{O}(e^{-cM^2})
\Bigr).
\end{equation}
Let
\begin{equation}
\label{def of Mja}
M_{j}(a):=\sqrt{n}(\lambda_{j}(a)-1),\qquad \lambda_{j}(a):=\frac{\tau_{a}}{\tau_j},
\end{equation}
so that $-M\leq M_j(a)\leq M$ in the corresponding transition window.
A Taylor expansion about $r_\tau=a$ gives
\[
\lambda_{j}(a)-1=b_1(a)(a-r_{\tau})+b_2(a)(a-r_{\tau})^2
+\mathcal{O}\bigl((a-r_{\tau})^3 \bigr),
\]
where we recall $\mathcal{Q}(a):=a\Delta Q(a)$ and 
\begin{align*}
b_1(a)&:=\frac{2\mathcal{Q}(a)}{\tau_{a}},
\quad
b_2(a):=\frac{4\mathcal{Q}(a)^2}{\tau_{a}^2}-\frac{\mathcal{Q}'(a)}{\tau_{a}}.
\end{align*}
Since $b_1(a)>0$, the inverse function theorem applies uniformly in a fixed
neighborhood of $a$.  Reverting the preceding series yields
\begin{align*}
a-r_{\tau}
&=\frac{1}{b_1(a)}(\lambda_{j}(a)-1)
-\frac{b_2(a)}{b_1(a)^3}(\lambda_{j}(a)-1)^2
+\mathcal{O}\bigl((\lambda_{j}(a)-1)^3\bigr),
\end{align*}
and hence
\begin{align*}
r_{\tau}
&=
a
-\mathsf{b}_1(a)\frac{M_j(a)}{n^{\frac{1}{2}}}
+\mathsf{b}_2(a)
\frac{M_j(a)^2}{n}
+\mathcal{O}\Bigl(\frac{M_j(a)^3}{n^{\frac{3}{2}}}\Bigr),
\end{align*}
where 
\begin{align*}
\mathsf{b}_1(a)&:=\frac{\tau_a}{2\mathcal{Q}(a)}, 
\quad
\mathsf{b}_2(a):=\frac{\tau_a}{2\mathcal{Q}(a)}-\frac{\tau_a^2\mathcal{Q}'(a)}{8\mathcal{Q}(a)^3}.
\end{align*}

For $g_{a_1,-}\leq j\leq g_{a_1,+}$, we have
$|\tau-\tau_{a_1}|=\mathcal{O}(M/\sqrt n)$ and
$|r_\ell-r_\tau|=\mathcal{O}(M/\sqrt n)$, uniformly for
$\ell=1,\ldots,m$.  The truncated Laplace expansion below is therefore
uniform in the entire transition window; in particular, it does not require
a separate ordering of $r_\ell$ and $r_\tau$.  With
$\delta_n'=(\log n)/\sqrt n$, write
\begin{align*}
\int_0^{r_{\ell}}2re^{\mathsf{k}(r)}e^{-nV_{\tau}(r)}\,dr
=\int_{0}^{r_{\tau}-\delta_n'}2re^{\mathsf{k}(r)}e^{-nV_{\tau}(r)}\,dr
+\int_{r_{\tau}-\delta_n'}^{r_{\ell}} 2re^{\mathsf{k}(r)}e^{-nV_{\tau}(r)}\,dr.   
\end{align*}
The contribution below $r_\tau-\delta_n'$ is
$e^{-nV_{\tau}(r_{\tau})}\mathcal{O}(e^{-c(\log n)^2})$, uniformly in $j$
and $\ell$.  The interval from this lower cutoff to $r_\ell$ is contained in
an $\mathcal{O}(M/\sqrt n)$-neighborhood of the saddle, on which the local
Laplace expansion is uniform.  As in Lemma~\ref{lemma: 0 j j1-}, it gives
\begin{align*}
&\quad 
\int_{r_{\tau}-\delta_n'}^{r_{\ell}} 2re^{\mathsf{k}(r)}e^{-nV_{\tau}(r)}\,dr
\\
&=
\frac{2r_{\tau}e^{\mathsf{k}(r_{\tau})}e^{-nV_{\tau}(r_{\tau})}}{\sqrt{nd_2}}\int_{-\sqrt{d_2n}\delta_n'}^{\xi_{j,\ell}(a_1)}
e^{-\frac{1}{2}u^2}
\Bigl( 
1+\frac{c_1(u)}{\sqrt{n}}+\frac{c_2(u)}{n}+\frac{c_3(u)}{n^{3/2}}+\mathcal{O}(\frac{c_4(u)}{n^2})
\Bigr)\,du,
\end{align*}
where $d_p=V_\tau^{(p)}(r_\tau)$ and $c_1,c_2,c_3$ are defined in
\eqref{def of c_1}--\eqref{def of c_3}; the polynomial $c_4$ is specified
immediately after those definitions.  Moreover,
\begin{equation}
 \xi_{j,\ell}(a_1)
 :=
\sqrt{4n \Delta Q(r_{\tau})}(r_{\ell}-r_{\tau})
 =
 \sqrt{4n\Delta Q(r_{\tau})}(a_1-r_{\tau})+\frac{\sqrt{2\Delta Q(r_{\tau})}}{\sqrt{\Delta Q(a_1)}}t_{\ell}. 
\end{equation}
The lower limit may be replaced by $-\infty$ at a cost
$\mathcal{O}(e^{-c(\log n)^2})$ relative to the prefactor.  All polynomial
coefficients in the integrand are bounded uniformly for
$|\xi_{j,\ell}(a_1)|=\mathcal{O}(M)$.  Substitution into
\eqref{def of pre sum ga1- ga1+} now gives
\begin{align*}
&\quad 
    \sum_{j=g_{a_1,-}}^{g_{a_1,+}}
        \log\Big( 
\sum_{\ell=1}^{m}\omega_{\ell}F_{n,j,\ell}^{(\mathrm{dc})}
+\Omega_{m+1}^{(3m)}+\mathcal{O}(e^{-cM^2})
\Bigr)
\\
&=
(g_{a_1,+}-g_{a_1,-}+1)\log \Omega_{m+1}^{(3m)}
\\
&\quad
+
\sum_{j=g_{a_1,-}}^{g_{a_1,+}}\log\mathscr{B}_1(\tfrac{\tau_{a_1}}{a_1\sqrt{2\Delta Q(a_1)}}M_j(a_1))
+\frac{1}{\sqrt{n}}\sum_{j=g_{a_1,-}}^{g_{a_1,+}}
\mathscr{B}_{2}(\tfrac{\tau_{a_1}}{a_1\sqrt{2\Delta Q(a_1)}}M_j(a_1))
+
\frac{1}{n}\sum_{j=g_{a_1,-}}^{g_{a_1,+}}
\mathcal{O}(M^4),
\end{align*}
The remainder in the logarithmic expansion is uniform: after the displayed
terms are subtracted, its absolute value is bounded by
$Cn^{-1}(1+|M_j(a_1)|^4)$.  This follows by inserting the uniform Laplace
expansion above and using that the logarithmic argument stays bounded away
from zero for the fixed real parameters under consideration.  On the
positive tail the corresponding coefficient is exponentially small; on the
negative tail the nondecaying terms cancel, leaving the stated polynomial
bound.  Since the transition window contains $\mathcal{O}(M\sqrt n)$
indices, its accumulated contribution is $\mathcal{O}(M^5n^{-1/2})$.
Here, for $x\in\mathbb R$, define
\begin{align*}
\mathscr{B}_1(x)&:= 1+\frac{1}{\Omega_{m+1}^{(3m)}}\sum_{\ell=1}^{m}
\frac{\omega_{\ell}}{2}\erfc(-t_{\ell}-x), 
\\
\mathscr{B}_{2}(x)
&:=
\frac{1}{\mathscr{B}_1(x)}
\frac{1}{\Omega_{m+1}^{(3m)}}
\sum_{\ell=1}^{m}
\frac{\omega_{\ell}e^{-(t_{\ell}+x)^2}}{6a_1\sqrt{2\Delta Q(a_1)}\sqrt{\pi}}
\bigl( 
\mathfrak{b}_2t_{\ell}^2+\mathfrak{b}_1(x)t_{\ell}+\mathfrak{b}_0(x)
\bigr), 
\end{align*}
where 
\begin{align*}
\mathfrak{b}_2&:=
-2\Bigl(2+\frac{a_1\partial_r\Delta Q(a_1)}{\Delta Q(a_1)}\Bigr)
+
3\Bigl(1+\frac{a_1\partial_r\Delta Q(a_1)}{\Delta Q(a_1)}\Bigr), 
\quad
\mathfrak{b}_1(x):=-\Bigl(2+\frac{a_1\partial_r\Delta Q(a_1)}{\Delta Q(a_1)}\Bigr)x, 
\\
\mathfrak{b}_0(x)&:=
\widetilde{\mathfrak{b}}_0(x)
+
6\Bigl(
2+\frac{a_1\partial_r\Delta Q(a_1)}{\Delta Q(a_1)}
\Bigr)x^2
-\frac{12a_1^2\Delta Q(a_1)}{\tau_{a_1}}x^2
-3a_1\mathsf{k}'(a_1)-6, 
\\
\widetilde{\mathfrak{b}}_0(x)
&:=
-5\Bigl(
2+\frac{a_1\partial_r\Delta Q(a_1)}{\Delta Q(a_1)}
\Bigr)x^2
+2+\frac{a_1\partial_r\Delta Q(a_1)}{\Delta Q(a_1)}. 
\end{align*}
Since
\[
j=\frac{n\tau_{a_1}}{1+M_j(a_1)/\sqrt n},
\qquad
\frac{dj}{dM_j(a_1)}
=-\tau_{a_1}\sqrt n\bigl(1+M_j(a_1)/\sqrt n\bigr)^{-2},
\]
the Jacobian contributes the factor $-2x$ at the next order.  Applying
Lemma~\ref{lemma:Riemann sum} to each of the preceding sums yields
\begin{align*}
&\quad 
    \sum_{j=g_{a_1,-}}^{g_{a_1,+}}
        \log\Big( 
\sum_{\ell=1}^{m}\omega_{\ell}F_{n,j,\ell}^{(\mathrm{dc})}
+\Omega_{m+1}^{(3m)}+\mathcal{O}(e^{-cM^2})
\Bigr)
\\
&=
\tau_{a_1}\int_{-M}^{M}\log \mathscr{B}_{1}(\tfrac{\tau_{a_1}}{a_1\sqrt{2\Delta Q(a_1)}}x)\,dx
\,\sqrt{n}
\\
&\quad+
\tau_{a_1}
\int_{-M}^{M}
\Bigl( 
-2x\log \mathscr{B}_{1}(\tfrac{\tau_{a_1}}{a_1\sqrt{2\Delta Q(a_1)}}x)
+\mathscr{B}_{2}(\tfrac{\tau_{a_1}}{a_1\sqrt{2\Delta Q(a_1)}}x)
\Bigr)\,dx
\\
&\quad
+ \bigg( \frac{1}{2}-\theta_{a_1,-}^{(n,M)} \bigg)\log \mathscr{B}_{1}(\tfrac{\tau_{a_1}}{a_1\sqrt{2\Delta Q(a_1)}}M)+ \bigg( \frac{1}{2}-\theta_{a_1,+}^{(n,M)} \bigg)\log \mathscr{B}_{1}(-\tfrac{\tau_{a_1}}{a_1\sqrt{2\Delta Q(a_1)}}M)
\\
&\quad
+
(g_{a_1,+}-g_{a_1,-}+1)\log \Omega_{m+1}^{(3m)}
+
\mathcal{O}(M^5n^{-1/2}). 
\end{align*}
We next evaluate the endpoint terms and the two integrals.  The Gaussian-tail
asymptotics imply
\begin{align*}
&\quad 
 \bigg( \frac{1}{2}-\theta_{a_1,-}^{(n,M)} \bigg)\log\mathscr{B}_{1}(\tfrac{\tau_{a_1}}{a_1\sqrt{2\Delta Q(a_1)}}M)+ \bigg( \frac{1}{2}-\theta_{a_1,+}^{(n,M)} \bigg)\log\mathscr{B}_{1}(-\tfrac{\tau_{a_1}}{a_1\sqrt{2\Delta Q(a_1)}}M)
 \\
 &=\bigg( \frac{1}{2}-\theta_{a_1,-}^{(n,M)} \bigg)\log\Bigl( 
1+\frac{1}{\Omega_{m+1}^{(3m)}}\sum_{\ell=1}^{m}\omega_{\ell}
\Bigr)+\mathcal{O}(e^{-cM^2}),\qquad M\to+\infty, 
\end{align*}
for some $c>0$. 
Similarly, splitting the leading integral at the origin and using its two
limits gives
\begin{align*}
&\quad
\tau_{a_1}\int_{-M}^{M}\log \mathscr{B}_{1}(\tfrac{\tau_{a_1}}{a_1\sqrt{2\Delta Q(a_1)}}x)\,dx
\\
&=
a_1\sqrt{2\Delta Q(a_1)}\int_{-\infty}^{0}
\log\Bigl( 
1+\frac{1}{\Omega_{m+1}^{(3m)}}\sum_{\ell=1}^{m}\frac{\omega_{\ell}}{2}
\erfc(-x-t_{\ell})
\Bigr)\,dx
\\
&\quad
+
a_1\sqrt{2\Delta Q(a_1)}\int_{0}^{+\infty}
\bigg[
\log\Bigl( 
1+\frac{1}{\Omega_{m+1}^{(3m)}}\sum_{\ell=1}^{m}\frac{\omega_{\ell}}{2}
\erfc(-x-t_{\ell})
\Bigr)
-
\log\Bigl(1+\frac{1}{\Omega_{m+1}^{(3m)}}\sum_{\ell=1}^{m}\omega_{\ell}\Bigr)
\bigg]
\,dx
\\
&\quad
+
M\tau_{a_1}\log\Bigl(1+\frac{1}{\Omega_{m+1}^{(3m)}}\sum_{\ell=1}^{m}\omega_{\ell}\Bigr)
+\mathcal{O}(e^{-cM^2}).
\end{align*}
The identity $\frac{\omega_{\ell}}{\Omega_{m+1}^{(3m)}}
=
e^{s_{\ell}+\cdots +s_{m}}-e^{s_{\ell+1}+\cdots +s_{m}}
=
(e^{s_{\ell}}-1)\exp\Bigl( 
\sum_{p=\ell+1}^{m}s_{p}
\Bigr)$ and the change of variables $x\mapsto-x$ identify the first limiting
integral as
\[
\int_{-\infty}^{0}
\log\Bigl( 
1+\frac{1}{\Omega_{m+1}^{(3m)}}\sum_{\ell=1}^{m}\frac{\omega_{\ell}}{2}
\erfc(-x-t_{\ell})
\Bigr)\,dx
=
\int_0^{+\infty}
\log\Bigl[
1+\sum_{\ell=1}^{m}\frac{e^{s_{\ell}}-1}{2}
\exp\Bigl( 
\sum_{p=\ell+1}^{m}s_{p}
\Bigr)
\erfc(x-t_{\ell})
\Bigr]\,dx. 
\]
For the second limiting integral, use
$1+\frac{1}{\Omega_{m+1}^{(3m)}}\sum_{\ell=1}^{m}\omega_{\ell}
=\Omega_{1}^{(3m)}/\Omega_{m+1}^{(3m)}$,
$\omega_{\ell}=-\Omega_1^{(3m)}\bigl( 
e^{-s_{\ell}}-1
\bigr)
\exp\Bigl( 
-\sum_{p=1}^{\ell-1}s_{p}
\Bigr)$, and $\erfc(-x)=2-\erfc(x)$.  These identities give
\begin{align*}
&\quad \log\Bigl( 
1+\frac{1}{\Omega_{m+1}^{(3m)}}\sum_{\ell=1}^{m}\frac{\omega_{\ell}}{2}
\erfc(-x-t_{\ell})
\Bigr)
-
\log\Bigl(1+\frac{\sum_{\ell=1}^{m}\omega_{\ell}}{\Omega_{m+1}^{(3m)}}\Bigr)
\\
&=
\log\Bigl[ 
1+\sum_{\ell=1}^{m}\frac{e^{-s_{\ell}}-1}{2}
\exp\Bigl( 
-\sum_{p=1}^{\ell-1}s_{p}
\Bigr)
\erfc(x+t_{\ell})
\Bigr]. 
\end{align*} 
For the next-order integral, direct substitution gives
\begin{align*}
&\quad\tau_{a_1}
\int_{-M}^{M}
\Bigl( 
\mathscr{B}_{2}(\tfrac{\tau_{a_1}}{a_1\sqrt{2\Delta Q(a_1)}}x)
-2x\log \mathscr{B}_{1}(\tfrac{\tau_{a_1}}{a_1\sqrt{2\Delta Q(a_1)}}x)
\Bigr)\,dx
\\
&=
a_1\sqrt{2\Delta Q(a_1)}
\int_{-\tfrac{\tau_{a_1}}{a_1\sqrt{2\Delta Q(a_1)}}M}^{\tfrac{\tau_{a_1}}{a_1\sqrt{2\Delta Q(a_1)}}M}
\bigg[
\frac{1}{\mathscr{B}_1(x)}
\frac{1}{\Omega_{m+1}^{(3m)}}
\sum_{\ell=1}^{m}
\frac{\omega_{\ell}e^{-(t_{\ell}+x)^2}}{6a_1\sqrt{2\Delta Q(a_1)}\sqrt{\pi}}
\Bigl\{
\bigl( 
\mathfrak{b}_2t_{\ell}^2+\mathfrak{b}_1(x)t_{\ell}+\widetilde{\mathfrak{b}}_0(x)
\bigr)
\\
&\quad
+6\Bigl(
2+\frac{a_1\partial_r\Delta Q(a_1)}{\Delta Q(a_1)}
\Bigr)x^2
-\frac{12a_1^2\Delta Q(a_1)}{\tau_{a_1}}x^2
-3a_1\mathsf{k}'(a_1)-6
\Bigr\}
-2\tfrac{a_1\sqrt{2\Delta Q(a_1)}}{\tau_{a_1}}x\log \mathscr{B}_{1}(x)
\bigg]
\,dx
\\
&=
a_1\sqrt{2\Delta Q(a_1)}
\int_{-\tfrac{\tau_{a_1}}{a_1\sqrt{2\Delta Q(a_1)}}M}^{\tfrac{\tau_{a_1}}{a_1\sqrt{2\Delta Q(a_1)}}M}
\\
&\quad
\times
\bigg[
\frac{1}{\mathscr{B}_1(x)}
\frac{1}{\Omega_{m+1}^{(3m)}}
\sum_{\ell=1}^{m}
\frac{\omega_{\ell}e^{-(t_{\ell}+x)^2}}{6a_1\sqrt{2\Delta Q(a_1)}\sqrt{\pi}}
\Bigl\{
\bigl( 
\mathfrak{b}_2t_{\ell}^2+\mathfrak{b}_1(x)t_{\ell}+\widetilde{\mathfrak{b}}_0(x)
\bigr)
+6\Bigl(
2+\frac{a_1\partial_r\Delta Q(a_1)}{\Delta Q(a_1)}
\Bigr)x^2
\Bigr\}
\bigg]
\,dx
\\
&\quad
-
\tau_{a_1}M^2
\log\Bigl( 
1+\frac{1}{\Omega_{m+1}^{(3m)}}\sum_{\ell=1}^{m}\omega_{\ell}
\Bigr)
-
\Bigl( 
\frac{1}{2}a_1\mathsf{k}'(a_1)+1
\Bigr)
\log\Bigl( 
1+\frac{1}{\Omega_{m+1}^{(3m)}}\sum_{\ell=1}^{m}\omega_{\ell}
\Bigr)
+\mathcal{O}(e^{-cM^2}),
\end{align*}
for some $c>0$ independent of $n$.  Moreover,
\begin{align*}
&\quad a_1\sqrt{2\Delta Q(a_1)}
\int_{-\tfrac{\tau_{a_1}}{a_1\sqrt{2\Delta Q(a_1)}}M}^{\tfrac{\tau_{a_1}}{a_1\sqrt{2\Delta Q(a_1)}}M}
\frac{1}{\mathscr{B}_1(x)}
\frac{1}{\Omega_{m+1}^{(3m)}}
\sum_{\ell=1}^{m}
\frac{\omega_{\ell}e^{-(t_{\ell}+x)^2}}{6a_1\sqrt{2\Delta Q(a_1)}\sqrt{\pi}}
6\Bigl(
2+\frac{a_1\partial_r\Delta Q(a_1)}{\Delta Q(a_1)}
\Bigr)x^2\,dx
\\
&=
2\Bigl(
2+\frac{a_1\partial_r\Delta Q(a_1)}{\Delta Q(a_1)}
\Bigr)
\int_{0}^{+\infty}
x\Bigl( 
\log\mathcal{H}_1^{(\mathrm{b},\mathrm{in})}(x) - \log\mathcal{H}_2^{(\mathrm{b},\mathrm{in})}(x)
\Bigr)\,dx
+
\mathcal{O}(e^{-cM^2}). 
\end{align*}
Combining these identities with the Riemann-sum expansion identifies the
$\sqrt n$ and constant terms as
$C_{E_2^{(\mathrm{dc})}}^{(4)}$ and
$C_{E_2^{(\mathrm{dc})}}^{(6)}$, respectively, and produces the displayed
cutoff contribution $\widetilde C_{E_2^{(\mathrm{dc})}}^{(M)}$.  The resulting expansion is
\begin{align*}
E_2^{(\mathrm{dc})}
&=C_4^{\#(\mathrm{b,in})}\sqrt n
+\mathcal D_6^{(\mathrm{b},\mathrm{in})}
-\Bigl(\frac12a_1\mathsf k'(a_1)+1\Bigr)
 \log\frac{\Omega_1^{(3m)}}{\Omega_{m+1}^{(3m)}}
+\widetilde C_{E_2^{(\mathrm{dc})}}^{(M)}
+\mathcal O\bigl(M^5n^{-1/2}+e^{-cM^2}\bigr).
\end{align*}
Because
$M=n^{1/12}$,
\[
M^5n^{-1/2}=n^{-1/12},
\qquad e^{-cM^2}=o(n^{-N})\quad\text{for every fixed }N>0.
\]
This proves the expansion of $E_2^{(\mathrm{dc})}$.  Finally, substituting
the definitions of the cutoff indices in the already established formulas
for $E_0^{(\mathrm{dc})}$, $E_1^{(\mathrm{dc})}$, and
$E_3^{(\mathrm{dc})}$ gives the last identity in the lemma.
\end{proof}

\begin{lemma}[Hard-wall transition and hard-edge tail]
\label{lemma:semi hard edge case asymptotic expansion disck complement}
Let $\{r_{\ell}\}_{\ell=m+1}^{2m}$ be as in
\eqref{def of merging radii inside semi hard}, and let
$\{r_{\ell}\}_{\ell=2m+1}^{3m}$ be as in
\eqref{def of merging radii inside hard}.  Then the following estimates hold;
all remainders are uniform over their defining index ranges.
\begin{description}
    \item[\textup{(i) Critical hard-wall window.}] As $n\to+\infty$,
\begin{equation}
 E_4^{(\mathrm{dc})}
 =
 C_{E_4^{(\mathrm{dc})}}^{(3)}n +C_{E_4^{(\mathrm{dc})}}^{(4)}\sqrt{n}+C_{E_4^{(\mathrm{dc})}}^{(5)}\log n+C_{E_4^{(\mathrm{dc})}}^{(6)}+\widetilde{C}_{E_4^{(\mathrm{dc})}}^{(M)}+\mathcal{O}(n^{-\frac{1}{12}}),
\end{equation}    
where 
\begin{align*}
C_{E_4^{(\mathrm{dc})}}^{(3)}&:=0,
\quad
C_{E_4^{(\mathrm{dc})}}^{(4)}:=C_4^{\#(\mathrm{se,in})},
\quad
C_{E_4^{(\mathrm{dc})}}^{(5)}:=0, 
\\
C_{E_4^{(\mathrm{dc})}}^{(6)}&:=
\Bigl(\frac{\rho_1\partial_r\Delta Q(\rho_1)}{\Delta Q(\rho_1)}+2\Bigr)
\\
&\quad\times\int_{-\infty}^{+\infty}\bigg[
2x\Bigl(\log \mathcal{H}_1^{(\mathrm{se},\mathrm{in})}(x)-\mathbf{1}_{(-\infty,0)}(x)\sum_{j=m+1}^{2m}s_j\Bigr)
+
\mathcal{H}_2^{(\mathrm{se},\mathrm{in})}(x)
\bigg]
\,dx
\\
&\quad 
+
3\Bigl(1+\frac{\rho_1\partial_r\Delta Q(\rho_1)}{\Delta Q(\rho_1)}\Bigr)
\int_{-\infty}^{+\infty}
\widetilde{\mathcal{H}}_2^{(\mathrm{se},\mathrm{in})}(x)\,dx
\\
&\quad 
-\frac{4\rho_1^2\Delta Q(\rho_1)\mathsf{T}_{1}^{(2m+1,3m)}(\tau_{\rho_1})}{1+\mathsf{T}_{0}^{(2m+1,3m)}(\tau_{\rho_1})}
\int_{-\infty}^{+\infty}
\bigg[
\frac{\frac{e^{-x^2}}{\sqrt{\pi}\erfc(x)}}{\mathcal{H}_1^{(\mathrm{se},\mathrm{in})}(x)}
-
\Bigl( 
x+\frac{x}{2(1+x^2)}
\Bigr)
\mathbf{1}_{[0,+\infty)}(x)
\bigg]
\,dx
\\
&\quad 
-\Bigl(\frac{1}{2}\rho_1\mathsf{k}'(\rho_1)+1\Bigr)\log\Bigl(\frac{\Omega_{m+1}^{(3m)}}{\Omega_{2m+1}^{(3m)}}\Bigr)
+\rho_1^2\Delta Q(\rho_1)\log (2\rho_1^2\Delta Q(\rho_1))
\frac{\mathsf{T}_{1}^{(2m+1,3m)}(\tau_{\rho_1})}{1+\mathsf{T}_{0}^{(2m+1,3m)}(\tau_{\rho_1})} 
,\\
\widetilde{C}_{E_4^{(\mathrm{dc})}}^{(M)}
&:=
\tau_{\rho_1}M\sqrt{n}
\log\frac{\Omega_{m+1}^{(3m)}}{\Omega_{2m+1}^{(3m)}}
+\Bigl(\frac{1}{2}-\theta_{1,-}^{(n,M)}\Bigr)\log\frac{\Omega_{m+1}^{(3m)}}{\Omega_{2m+1}^{(3m)}}
-\tau_{\rho_1}M^2\log\Bigl(\frac{\Omega_{m+1}^{(3m)}}{\Omega_{2m+1}^{(3m)}}\Bigr)
\\
&\quad
+(2\tau_{\rho_1}M\sqrt{n}-\theta_{1,+}^{(n,M)}-\theta_{1,-}^{(n,M)}+1)\log\Omega_{2m+1}^{(3m)}
\\
&\quad 
-\big(\tau_{\rho_1}^2M^2
+2\rho_1^2\Delta Q(\rho_1)\log(M\tau_{\rho_1})
\big)
\frac{\mathsf{T}_{1}^{(2m+1,3m)}(\tau_{\rho_1})}{1+\mathsf{T}_{0}^{(2m+1,3m)}(\tau_{\rho_1})}.
\end{align*}
    \item[\textup{(ii) Hard-edge tail.}] As $n\to+\infty$,
\begin{equation}
 E_5^{(\mathrm{dc})}
 =
 C_{ E_5^{(\mathrm{dc})}}^{(3)}n+C_{ E_5^{(\mathrm{dc})}}^{(4)}\sqrt{n}+C_{ E_5^{(\mathrm{dc})}}^{(5)}\log n +C_{ E_5^{(\mathrm{dc})}}^{(6)}+\widetilde{C}_{ E_5^{(\mathrm{dc})}}^{(M)}+\mathcal{O}(n^{-\frac{1}{12}}),
\end{equation}
where 
\begin{align*}
C_{ E_5^{(\mathrm{dc})}}^{(3)}
&:=
\int_{\tau_{\rho_1}}^1
\log\bigl(1+\mathsf{T}_{0}^{(2m+1,3m)}(x)\bigr)\,dx, 
\quad
C_{ E_5^{(\mathrm{dc})}}^{(4)}
:=0,
\\
C_{ E_5^{(\mathrm{dc})}}^{(5)}
&:=
-
\frac{\rho_1^2\Delta Q(\rho_1)\mathsf{T}_{1}^{(2m+1,3m)}(\tau_{\rho_1})}{1+\mathsf{T}_{0}^{(2m+1,3m)}(\tau_{\rho_1})},
\\
C_{ E_5^{(\mathrm{dc})}}^{(6)}
&:=
-
\frac{1}{2}\log\bigl(1+\mathsf{T}_{0}^{(2m+1,3m)}(1)\bigr)
\\
&\quad 
+
\int_{\tau_{\rho_1}}^{1}
\Bigl(
\mathcal{H}_1^{(\mathrm{h,in})}(x)
+\frac{2\rho_1^2\Delta Q(\rho_1)}{x-\tau_{\rho_1}}\frac{\mathsf{T}_{1}^{(2m+1,3m)}(\tau_{\rho_1})}{1+\mathsf{T}_{0}^{(2m+1,3m)}(\tau_{\rho_1})}\Bigr)\,dx
\\
&\quad 
-
2\rho_1^2\Delta Q(\rho_1)
\frac{\mathsf{T}_{1}^{(2m+1,3m)}(\tau_{\rho_1})}{1+\mathsf{T}_{0}^{(2m+1,3m)}(\tau_{\rho_1})}
\log(1-\tau_{\rho_1})
\\
&\quad 
+\Bigl(\frac{1}{2}\rho_1\mathsf{k}'(\rho_1)+1\Bigr)
\Bigl(\log(1+\mathsf{T}_0^{(2m+1,3m)}(1))-\log(1+\mathsf{T}_0^{(2m+1,3m)}(\tau_{\rho_1}))
\Bigr)
\\
&\quad
-
\bigl(2\rho_1^2\Delta Q(\rho_1)-\tau_{\rho_1}\bigr)
\int_{\tau_{\rho_1}}^{1}
\frac{\mathsf{T}_{2}^{(2m+1,3m)}(x)}{1+\mathsf{T}_{0}^{(2m+1,3m)}(x)}\,dx,
\\
\widetilde{C}_{ E_5^{(\mathrm{dc})}}^{(M)}
&:=
-\tau_{\rho_1}M\sqrt{n}\log\bigl(1+\mathsf{T}_{0}^{(2m+1,3m)}(\tau_{\rho_1})\bigr)
\\
&\quad
-M^2\tau_{\rho_1}\left(\begin{aligned}& 
\log\bigl(1+\mathsf{T}_{0}^{(2m+1,3m)}(\tau_{\rho_1})\bigr)
\\
&\quad-\tau_{\rho_1}
\frac{\mathsf{T}_{1}^{(2m+1,3m)}(\tau_{\rho_1})}{1+\mathsf{T}_{0}^{(2m+1,3m)}(\tau_{\rho_1})}
\end{aligned}\right)
\\
&\quad
+
\Bigl(\theta_{1,+}^{(n,M)}-\frac{1}{2}\Bigr)
\log\bigl(1+\mathsf{T}_{0}^{(2m+1,3m)}(\tau_{\rho_1})\bigr)
\\
&\quad
+
2\rho_1^2\Delta Q(\rho_1)
\frac{\mathsf{T}_{1}^{(2m+1,3m)}(\tau_{\rho_1})}{1+\mathsf{T}_{0}^{(2m+1,3m)}(\tau_{\rho_1})}\log(M\tau_{\rho_1}).
\end{align*}
    
\end{description}
\end{lemma}

\begin{proof}
\emph{Proof of (i).}
By definition,
\[
E_4^{(\mathrm{dc})}
=\sum_{j=g_{1,-}}^{g_{1,+}}\log\Big(
1+\sum_{\ell=1}^{m}\omega_{\ell}F_{n,j,\ell}^{(\mathrm{dc})}
+\sum_{\ell=m+1}^{2m}\omega_{\ell}F_{n,j,\ell}^{(\mathrm{dc})}
+\sum_{\ell=2m+1}^{3m}\omega_{\ell}F_{n,j,\ell}^{(\mathrm{dc})}
\Bigr). 
\]
The denominator of $F_{n,j,\ell}^{(\mathrm{dc})}$ was expanded in
Lemma~\ref{lemma:g1- j g1+}; it remains to analyze the numerator.
For $\ell=1,\ldots,m$, the off-saddle estimate used in
Lemma~\ref{lemma: 0 j j1-} gives, uniformly for
$g_{1,-}\leq j\leq g_{1,+}$,
\begin{align*}
\int _{0}^{r_{\ell}}2r e^{\mathsf{k}(r)}e^{-nV_{\tau}(r)}\,dr
&=
2e^{\mathsf{k}(r_{\tau})}e^{-nV_{\tau}(r_{\tau})}
\int_{0}^{r_{\ell}}r e^{\mathsf{k}(r)-\mathsf{k}(r_{\tau})}
e^{-n(V_{\tau}(r)-V_{\tau}(r_{\tau}))}\,dr
\\
&=
2e^{\mathsf{k}(r_{\tau})}e^{-nV_{\tau}(r_{\tau})}\cdot \mathcal{O}(e^{-cM^2}). 
\end{align*}
Combining this estimate with Lemma~\ref{lemma:g1- j g1+}, we obtain
\[
\sum_{\ell=1}^{m}\omega_{\ell}F_{n,j,\ell}^{(\mathrm{dc})}
=
\sum_{\ell=1}^m \omega_{\ell}\mathcal{O}(e^{-cM^2})=\mathcal{O}(e^{-cM^2}). 
\]
It remains to treat the semi-hard radii
$\ell=m+1,\ldots,2m$ and the hard-edge radii
$\ell=2m+1,\ldots,3m$.  We apply the local calculation from
Lemma~\ref{lemma:bulk case asymptotic expansion} to the numerator, with
$a_1$ replaced by $\rho_1$, and use Lemma~\ref{lemma:g1- j g1+} for the
denominator.  This gives
\begin{align*}
&\quad\sum_{j=g_{1,-}}^{g_{1,+}}\log\Big(
1+\sum_{\ell=1}^{m}\omega_{\ell}F_{n,j,\ell}^{(\mathrm{dc})}
+\sum_{\ell=m+1}^{2m}\omega_{\ell}F_{n,j,\ell}^{(\mathrm{dc})}
+\sum_{\ell=2m+1}^{3m}\omega_{\ell}F_{n,j,\ell}^{(\mathrm{dc})}
\Bigr)
\\
&=
\sum_{j=g_{1,-}}^{g_{1,+}}\log\Omega_{2m+1}^{(3m)}
+
\sum_{j=g_{1,-}}^{g_{1,+}}
\widetilde{\mathscr{U}}_{1}(\tfrac{\tau_{\rho_1}M_j(\rho_1)}{\rho_1\sqrt{2\Delta Q(\rho_1)}})
\\
&\quad
+\frac{1}{\sqrt{n}}
\sum_{j=g_{1,-}}^{g_{1,+}}
\frac{1}{\mathscr{U}_{1}(\tfrac{\tau_{\rho_1}M_j(\rho_1)}{\rho_1\sqrt{2\Delta Q(\rho_1)}})}
\bigg[
\mathscr{U}_{2}^{(\mathrm{se},\mathrm{in})}(\tfrac{\tau_{\rho_1}M_j(\rho_1)}{\rho_1\sqrt{2\Delta Q(\rho_1)}})
+
\mathscr{U}_{2}^{(\mathrm{h,in})}(\tfrac{\tau_{\rho_1}M_j(\rho_1)}{\rho_1\sqrt{2\Delta Q(\rho_1)}})
\bigg]
+
\mathcal{O}(M^5n^{-1/2}), 
\end{align*}
where
\begin{align*}
\mathscr{U}_{1}(x)&:=  
1+\frac{1}{\Omega_{2m+1}^{(3m)}}\sum_{\ell=m+1}^{2m}\omega_{\ell}\frac{\erfc(t_{\ell}-x)}{\erfc(-x)}, 
\quad
\widetilde{\mathscr{U}}_{1}(x)
:=\log \mathscr{U}_{1}(x), 
\\
\mathscr{U}_{2}^{(\mathrm{se},\mathrm{in})}(x)
&:=
\frac{1}{\Omega_{2m+1}^{(3m)}}
\sum_{\ell=m+1}^{2m}\omega_{\ell}
\Bigl(
\mathfrak{u}_{2,\ell}(x)t_{\ell}^2
+
\mathfrak{u}_{1,\ell}(x)t_{\ell}
+
\mathfrak{u}_{0,\ell}(x)
\Bigr), 
\\
\mathscr{U}_{2}^{(\mathrm{h,in})}(x)
&:=
-
\frac{2\rho_1\sqrt{2\Delta Q(\rho_1)}}{\sqrt{\pi}}
\sum_{\ell=2m+1}^{3m}
\frac{\omega_{\ell}t_{\ell}}{\Omega_{2m+1}^{(3m)}}
\frac{e^{-x^2}}{\erfc(-x)}. 
\end{align*}
Here 
\begin{align*}
\mathfrak{u}_{2,\ell}(x)&:=
\frac{1}{3\rho_1\sqrt{2\Delta Q(\rho_1)}\sqrt{\pi}}
\frac{e^{-(t_{\ell}-x)^2}}{\erfc(-x)}
\bigg[
-2\Bigl(2+\frac{\rho_1\partial_r\Delta Q(\rho_1)}{\Delta Q(\rho_1)}\Bigr)
+3\Bigl(1+\frac{\rho_1\partial_r\Delta Q(\rho_1)}{\Delta Q(\rho_1)}\Bigr)
\bigg], 
\\
\mathfrak{u}_{1,\ell}(x)&:=
\frac{1}{3\rho_1\sqrt{2\Delta Q(\rho_1)}\sqrt{\pi}}
\Bigl(2+\frac{\rho_1\partial_r\Delta Q(\rho_1)}{\Delta Q(\rho_1)}\Bigr)
x
\frac{e^{-(t_{\ell}-x)^2}}{\erfc(-x)}, 
\\
\mathfrak{u}_{0,\ell}(x)&:=
\frac{1}{3\rho_1\sqrt{2\Delta Q(\rho_1)}\sqrt{\pi}}
\bigg[
\frac{\erfc(t_{\ell}-x)e^{-x^2}}{\erfc(-x)^2}
-\frac{e^{-(t_{\ell}-x)^2}}{\erfc(-x)}
\bigg]
\bigg[
5\Bigl( 
\frac{\rho_1\partial_r\Delta Q(\rho_1)}{\Delta Q(\rho_1)}+2
\Bigr)
x^2
\\
&\quad 
-\Bigl( \frac{\rho_1\partial_r\Delta Q(\rho_1)}{\Delta Q(\rho_1)}+2\Bigr)
-6\Bigl( 
\frac{\rho_1\partial_r\Delta Q(\rho_1)}{\Delta Q(\rho_1)}+2
\Bigr)x^2
+\frac{12\rho_1^2\Delta Q(\rho_1)}{\tau_{\rho_1}}x^2
+6\Bigl(\frac{1}{2}\rho_1\mathsf{k}'(\rho_1)+1\Bigr)
\bigg].
\end{align*}
We apply Lemma~\ref{lemma:Riemann sum} to each sum.  For the leading
correction, this gives
\begin{align*}
&\quad \sum_{j=g_{1,-}}^{g_{1,+}}
\widetilde{\mathscr{U}}_{1}(\tfrac{\tau_{\rho_1}M_j(\rho_1)}{\rho_1\sqrt{2\Delta Q(\rho_1)}})
\\
&=
\rho_1\sqrt{2\Delta Q(\rho_1)}\int_{-\frac{\tau_{\rho_1}M}{\rho_1\sqrt{2\Delta Q(\rho_1)}}}^{\frac{\tau_{\rho_1}M}{\rho_1\sqrt{2\Delta Q(\rho_1)}}}
\widetilde{\mathscr{U}}_1(-x)\,dx\,\sqrt{n}
\\
&\quad+\rho_1\sqrt{2\Delta Q(\rho_1)}
\int_{-\frac{\tau_{\rho_1}M}{\rho_1\sqrt{2\Delta Q(\rho_1)}}}^{\frac{\tau_{\rho_1}M}{\rho_1\sqrt{2\Delta Q(\rho_1)}}}
\frac{\rho_1\sqrt{2\Delta Q(\rho_1)}}{\tau_{\rho_1}}
2x\,\widetilde{\mathscr{U}}_1(-x)\,dx
\\
&\quad 
+\Bigl(\frac{1}{2}-\theta_{1,-}^{(n,M)}\Bigr)\widetilde{\mathscr{U}}_{1}(\tfrac{\tau_{\rho_1}M}{\rho_1\sqrt{2\Delta Q(\rho_1)}})
+\Bigl(\frac{1}{2}-\theta_{1,+}^{(n,M)}\Bigr)
\widetilde{\mathscr{U}}_{1}(-\tfrac{\tau_{\rho_1}M}{\rho_1\sqrt{2\Delta Q(\rho_1)}})
+\mathcal{O}(M^5n^{-1/2}). 
\end{align*}
Splitting the integral at the origin and using the two tail limits of
$\widetilde{\mathscr U}_1$ gives
\begin{align*}
&\quad\rho_1\sqrt{2\Delta Q(\rho_1)}\int_{-\frac{\tau_{\rho_1}M}{\rho_1\sqrt{2\Delta Q(\rho_1)}}}^{\frac{\tau_{\rho_1}M}{\rho_1\sqrt{2\Delta Q(\rho_1)}}}
\widetilde{\mathscr{U}}_1(-x)\,dx\,\sqrt{n}
\\
&=
\rho_1\sqrt{2\Delta Q(\rho_1)}\int_{0}^{+\infty}
\log\Bigl[1+\frac{1}{\Omega_{2m+1}^{(3m)}}\sum_{\ell=m+1}^{2m}\omega_{\ell}\frac{\erfc(t_{\ell}+x)}{\erfc(x)}
\Bigr]
\,dx\,\sqrt{n}
\\
&
+\rho_1\sqrt{2\Delta Q(\rho_1)}\int_{-\infty}^{0}
\bigg\{
\log\Bigl[1+\frac{1}{\Omega_{2m+1}^{(3m)}}\sum_{\ell=m+1}^{2m}\omega_{\ell}\frac{\erfc(t_{\ell}+x)}{\erfc(x)}
\Bigr]
-\log\Bigl( 
1+\frac{1}{\Omega_{2m+1}^{(3m)}}\sum_{\ell=m+1}^{2m}\omega_{\ell}
\Bigr)
\bigg\}
\,dx\,\sqrt{n}
\\
&
+
\tau_{\rho_1}M\sqrt{n}
\log\Bigl( 
1+\frac{1}{\Omega_{2m+1}^{(3m)}}\sum_{\ell=m+1}^{2m}\omega_{\ell}
\Bigr). 
\end{align*}
Equivalently, in terms of the semi-hard functional defined above,
\begin{align*}
&\quad \rho_1\sqrt{2\Delta Q(\rho_1)}\int_{-\frac{\tau_{\rho_1}M}{\rho_1\sqrt{2\Delta Q(\rho_1)}}}^{\frac{\tau_{\rho_1}M}{\rho_1\sqrt{2\Delta Q(\rho_1)}}}
\widetilde{\mathscr{U}}_1(-x)\,dx\,\sqrt{n}
\\
&=
\rho_1\sqrt{2\Delta Q(\rho_1)}\int_{-\infty}^{+\infty}
\bigg(
\log\mathcal{H}_1^{(\mathrm{se},\mathrm{in})}(x)
-
\mathbf{1}_{(-\infty,0)}
\sum_{j=m+1}^{2m}s_{j}
\bigg)
\,dx\,\sqrt{n}
+
\tau_{\rho_1}M\sqrt{n}
\log\frac{\Omega_{m+1}^{(3m)}}{\Omega_{2m+1}^{(3m)}},
\end{align*}
and 
\begin{align*}
&\quad \Bigl(\frac{1}{2}-\theta_{1,-}^{(n,M)}\Bigr)\widetilde{\mathscr{U}}_{1}(\tfrac{\tau_{\rho_1}M}{\rho_1\sqrt{2\Delta Q(\rho_1)}})
+\Bigl(\frac{1}{2}-\theta_{1,+}^{(n,M)}\Bigr)
\widetilde{\mathscr{U}}_{1}(-\tfrac{\tau_{\rho_1}M}{\rho_1\sqrt{2\Delta Q(\rho_1)}})
\\
&=\Bigl(\frac{1}{2}-\theta_{1,-}^{(n,M)}\Bigr)\log\frac{\Omega_{m+1}^{(3m)}}{\Omega_{2m+1}^{(3m)}}
+
\mathcal{O}(e^{-cM^2}),
\end{align*}
for some $c>0$ independent of $n$.  It remains to combine the next-order
Jacobian correction with the $n^{-1/2}$ terms in the local expansion.  The
relevant expression is
\begin{align*}
&\quad \rho_1\sqrt{2\Delta Q(\rho_1)}
\int_{-\frac{\tau_{\rho_1}M}{\rho_1\sqrt{2\Delta Q(\rho_1)}}}^{\frac{\tau_{\rho_1}M}{\rho_1\sqrt{2\Delta Q(\rho_1)}}}
\frac{\rho_1\sqrt{2\Delta Q(\rho_1)}}{\tau_{\rho_1}}
2x\widetilde{\mathscr{U}}_1(-x)\,dx
\\
&\quad 
+\frac{1}{\sqrt{n}}
\sum_{j=g_{1,-}}^{g_{1,+}}
\frac{1}{\mathscr{U}_{1}(\tfrac{\tau_{\rho_1}M_j(\rho_1)}{\rho_1\sqrt{2\Delta Q(\rho_1)}})}
\left[\begin{aligned}&
\mathscr{U}_{2}^{(\mathrm{se},\mathrm{in})}(\tfrac{\tau_{\rho_1}M_j(\rho_1)}{\rho_1\sqrt{2\Delta Q(\rho_1)}})
\\
&+
\mathscr{U}_{2}^{(\mathrm{h,in})}(\tfrac{\tau_{\rho_1}M_j(\rho_1)}{\rho_1\sqrt{2\Delta Q(\rho_1)}})
\end{aligned}\right]
+\mathcal{O}(M^5n^{-1/2}).
\end{align*}
Another application of Lemma~\ref{lemma:Riemann sum} yields
\begin{align*}
 &\quad 
 \frac{1}{\sqrt{n}}
\sum_{j=g_{1,-}}^{g_{1,+}}
\frac{1}{\mathscr{U}_{1}(\tfrac{\tau_{\rho_1}M_j(\rho_1)}{\rho_1\sqrt{2\Delta Q(\rho_1)}})}
\left[\begin{aligned}&
\mathscr{U}_{2}^{(\mathrm{se},\mathrm{in})}(\tfrac{\tau_{\rho_1}M_j(\rho_1)}{\rho_1\sqrt{2\Delta Q(\rho_1)}})
\\
&+
\mathscr{U}_{2}^{(\mathrm{h,in})}(\tfrac{\tau_{\rho_1}M_j(\rho_1)}{\rho_1\sqrt{2\Delta Q(\rho_1)}})
\end{aligned}\right]
+\mathcal{O}(M^5n^{-1/2})
\\
&=
\rho_1\sqrt{2\Delta Q(\rho_1)}
\int_{-\frac{\tau_{\rho_1}M}{\rho_1\sqrt{2\Delta Q(\rho_1)}}}^{\frac{\tau_{\rho_1}M}{\rho_1\sqrt{2\Delta Q(\rho_1)}}}
\frac{\mathscr{U}_{2}^{(\mathrm{se},\mathrm{in})}(-x)}{\mathscr{U}_{1}(-x)}\,dx
\\
&\quad+
\rho_1\sqrt{2\Delta Q(\rho_1)}
\int_{-\frac{\tau_{\rho_1}M}{\rho_1\sqrt{2\Delta Q(\rho_1)}}}^{\frac{\tau_{\rho_1}M}{\rho_1\sqrt{2\Delta Q(\rho_1)}}}
\frac{\mathscr{U}_{2}^{(\mathrm{h},\mathrm{in})}(-x)}{\mathscr{U}_{1}(-x)}\,dx
+\mathcal{O}(M^5n^{-1/2}).
\end{align*}
Note that 
\begin{align*}
&\quad 
\rho_1\sqrt{2\Delta Q(\rho_1)}
\int_{-\frac{\tau_{\rho_1}M}{\rho_1\sqrt{2\Delta Q(\rho_1)}}}^{\frac{\tau_{\rho_1}M}{\rho_1\sqrt{2\Delta Q(\rho_1)}}}
\frac{\mathscr{U}_{2}^{(\mathrm{se},\mathrm{in})}(-x)}{\mathscr{U}_{1}(-x)}\,dx
\\
&=
\Bigl(\frac{\rho_1\partial_r\Delta Q(\rho_1)}{\Delta Q(\rho_1)}+2\Bigr)
\int_{-\frac{\tau_{\rho_1}M}{\rho_1\sqrt{2\Delta Q(\rho_1)}}}^{\frac{\tau_{\rho_1}M}{\rho_1\sqrt{2\Delta Q(\rho_1)}}}
\frac{1}{\mathcal{H}_1^{(\mathrm{se},\mathrm{in})}(x)}
\sum_{\ell=m+1}^{2m}\omega_{\ell}^{(2m)}t_{\ell}^2\frac{-2e^{-(t_{\ell}+x)^2}}{3\sqrt{\pi}\erfc(x)}
\,dx
\\
&
+
3\Bigl(1+\frac{\rho_1\partial_r\Delta Q(\rho_1)}{\Delta Q(\rho_1)}\Bigr)
\int_{-\frac{\tau_{\rho_1}M}{\rho_1\sqrt{2\Delta Q(\rho_1)}}}^{\frac{\tau_{\rho_1}M}{\rho_1\sqrt{2\Delta Q(\rho_1)}}}
\frac{1}{\mathcal{H}_1^{(\mathrm{se},\mathrm{in})}(x)}
\sum_{\ell=m+1}^{2m}\omega_{\ell}^{(2m)}t_{\ell}^2\frac{e^{-(t_{\ell}+x)^2}}{3\sqrt{\pi}\erfc(x)}
\,dx
\\
&
-
\Bigl(\frac{\rho_1\partial_r\Delta Q(\rho_1)}{\Delta Q(\rho_1)}+2\Bigr)
\int_{-\frac{\tau_{\rho_1}M}{\rho_1\sqrt{2\Delta Q(\rho_1)}}}^{\frac{\tau_{\rho_1}M}{\rho_1\sqrt{2\Delta Q(\rho_1)}}}
\frac{1}{\mathcal{H}_1^{(\mathrm{se},\mathrm{in})}(x)}
\sum_{\ell=m+1}^{2m}\omega_{\ell}^{(2m)}t_{\ell}\frac{xe^{-(t_{\ell}+x)^2}}{3\sqrt{\pi}\erfc(x)}
\,dx
\\
&
+
\Bigl(\frac{\rho_1\partial_r\Delta Q(\rho_1)}{\Delta Q(\rho_1)}+2\Bigr)
\int_{-\frac{\tau_{\rho_1}M}{\rho_1\sqrt{2\Delta Q(\rho_1)}}}^{\frac{\tau_{\rho_1}M}{\rho_1\sqrt{2\Delta Q(\rho_1)}}}
\frac{1}{\mathcal{H}_1^{(\mathrm{se},\mathrm{in})}(x)}
\sum_{\ell=m+1}^{2m}
\frac{\omega_{\ell}^{(2m)}}{3\sqrt{\pi}}
\\
&\quad\times\bigg[
\frac{\erfc(t_{\ell}+x)e^{-x^2}}{\erfc(x)^2}-\frac{e^{-(t_{\ell}+x)^2}}{\erfc(x)}
\bigg]
\bigl(5x^2-1)
\,dx
\\
&
+
\int_{-\frac{\tau_{\rho_1}M}{\rho_1\sqrt{2\Delta Q(\rho_1)}}}^{\frac{\tau_{\rho_1}M}{\rho_1\sqrt{2\Delta Q(\rho_1)}}}
\frac{1}{\mathcal{H}_1^{(\mathrm{se},\mathrm{in})}(x)}
\sum_{\ell=m+1}^{2m}
\frac{\omega_{\ell}^{(2m)}}{3\sqrt{\pi}}
\bigg[
\frac{\erfc(t_{\ell}+x)e^{-x^2}}{\erfc(x)^2}-\frac{e^{-(t_{\ell}+x)^2}}{\erfc(x)}
\bigg]
\\
&\quad \times
\bigg[
\frac{12\rho_1^2\Delta Q(\rho_1)}{\tau_{\rho_1}}x^2+6\Bigl(\frac{1}{2}\rho_1\mathsf{k}'(\rho_1)+1\Bigr)-6\Bigl(\frac{\rho_1\partial_r\Delta Q(\rho_1)}{\Delta Q(\rho_1)}+2\Bigr)x^2
\bigg]
\,dx,
\end{align*}
The identity
\begin{align*}
\partial_x\log\Bigl( 
1+\sum_{\ell=m+1}^{2m}\omega_{\ell}^{(2m)}
\frac{\erfc(t_{\ell}+x)}{\erfc(x)}
\Bigr)
=
\frac{\sum_{\ell=m+1}^{2m}\omega_{\ell}^{(2m)}\Bigl( 
\frac{2\erfc(t_{\ell}+x)e^{-x^2}}{\sqrt{\pi}\erfc(x)^2}
-
\frac{2e^{-(t_{\ell}+x)^2}}{\sqrt{\pi}\erfc(x)}
\Bigr)}{1+\sum_{\ell=m+1}^{2m}\omega_{\ell}^{(2m)}
\frac{\erfc(t_{\ell}+x)}{\erfc(x)}},
\end{align*}
then permits an integration by parts, yielding
\begin{align*}
&\quad 
\frac{1}{6}
\int_{-\frac{\tau_{\rho_1}M}{\rho_1\sqrt{2\Delta Q(\rho_1)}}}^{\frac{\tau_{\rho_1}M}{\rho_1\sqrt{2\Delta Q(\rho_1)}}}
\frac{1}{\mathcal{H}_1^{(\mathrm{se},\mathrm{in})}(x)}
\sum_{\ell=m+1}^{2m}
\omega_{\ell}^{(2m)}
\bigg[
\frac{2\erfc(t_{\ell}+x)e^{-x^2}}{\sqrt{\pi}\erfc(x)^2}-\frac{2e^{-(t_{\ell}+x)^2}}{\sqrt{\pi}\erfc(x)}
\bigg]
\\
&\quad \times
\bigg[
\frac{12\rho_1^2\Delta Q(\rho_1)}{\tau_{\rho_1}}x^2+6\Bigl(\frac{1}{2}\rho_1\mathsf{k}'(\rho_1)+1\Bigr)-6\Bigl(\frac{\rho_1\partial_r\Delta Q(\rho_1)}{\Delta Q(\rho_1)}+2\Bigr)x^2
\bigg]
\,dx
\\
&
=
-\tau_{\rho_1}M^2\log\Bigl(\frac{\Omega_{m+1}^{(3m)}}{\Omega_{2m+1}^{(3m)}}\Bigr)
-\Bigl(\frac{1}{2}\rho_1\mathsf{k}'(\rho_1)+1\Bigr)\log\Bigl(\frac{\Omega_{m+1}^{(3m)}}{\Omega_{2m+1}^{(3m)}}\Bigr)
\\
&\quad 
-
\rho_1\sqrt{2\Delta Q(\rho_1)}
\int_{-\frac{\tau_{\rho_1}M}{\rho_1\sqrt{2\Delta Q(\rho_1)}}}^{\frac{\tau_{\rho_1}M}{\rho_1\sqrt{2\Delta Q(\rho_1)}}}
\frac{\rho_1\sqrt{2\Delta Q(\rho_1)}}{\tau_{\rho_1}}
2x\widetilde{\mathscr{U}}_1(-x)
\,dx
\\
&\quad 
+
\Bigl(\frac{\rho_1\partial_r\Delta Q(\rho_1)}{\Delta Q(\rho_1)}+2\Bigr)
\int_{-\infty}^{+\infty}
2x\Bigl(\widetilde{\mathscr{U}}_1(-x)-\mathbf{1}_{(-\infty,0)}(x)\sum_{j=m+1}^{2m}s_j\Bigr)
\,dx
+\mathcal{O}(e^{-cM^2}),
\end{align*}
for some $c>0$ independent of $n$.  To treat the hard-edge contribution, we
use the tail expansions
\[
\mathscr{U}_{1}(-x)=  
1+\frac{1}{\Omega_{2m+1}^{(3m)}}\sum_{\ell=m+1}^{2m}\omega_{\ell}\frac{\erfc(t_{\ell}+x)}{\erfc(x)}
=
\begin{cases}
1+\mathcal{O}(e^{-c_{t_{\ell}}x}),     & \mbox{if } x\to+\infty, \\
1+\frac{1}{\Omega_{2m+1}^{(3m)}}\sum_{\ell=m+1}^{2m}\omega_{\ell}+\mathcal{O}(e^{-x^2}),     & \mbox{if } x\to-\infty, 
\end{cases}
\]
where $c_{t_{\ell}}>0$ depends on $t_{\ell}>0$ for all $\ell=m+1,\dots,2m$, and 
\begin{align*}
\mathscr{U}_{2}^{(\mathrm{h,in})}(-x)
&=
-
\frac{2\rho_1\sqrt{2\Delta Q(\rho_1)}}{\sqrt{\pi}}
\frac{\mathsf{T}_{1}^{(2m+1,3m)}(\tau_{\rho_1})}{1+\mathsf{T}_{0}^{(2m+1,3m)}(\tau_{\rho_1})}\frac{e^{-x^2}}{\erfc(x)}
\\
&=
-
\frac{2\rho_1\sqrt{2\Delta Q(\rho_1)}}{\sqrt{\pi}}
\frac{\mathsf{T}_{1}^{(2m+1,3m)}(\tau_{\rho_1})}{1+\mathsf{T}_{0}^{(2m+1,3m)}(\tau_{\rho_1})}
\begin{cases}
    \sqrt{\pi}x+\frac{\sqrt{\pi}}{2x}-\frac{\sqrt{\pi}}{2x^3}+\mathcal{O}(x^{-4}), & \mbox{if } x\to+\infty, \\
    \mathcal{O}(e^{-x^2}), & \mbox{if } x\to-\infty,
\end{cases}
\end{align*}
It follows that
\begin{align*}
&\quad
\rho_1\sqrt{2\Delta Q(\rho_1)}
\int_{-\frac{\tau_{\rho_1}M}{\rho_1\sqrt{2\Delta Q(\rho_1)}}}^{\frac{\tau_{\rho_1}M}{\rho_1\sqrt{2\Delta Q(\rho_1)}}}
\frac{\mathscr{U}_{2}^{(\mathrm{h,in})}(-x)}{\mathscr{U}_{1}(-x)}\,dx
\\
&=
-4\rho_1^2\Delta Q(\rho_1)
\frac{\mathsf{T}_{1}^{(2m+1,3m)}(\tau_{\rho_1})}{1+\mathsf{T}_{0}^{(2m+1,3m)}(\tau_{\rho_1})}
\int_{-\infty}^{+\infty}
\bigg[
\frac{\frac{e^{-x^2}}{\sqrt{\pi}\erfc(x)}}{\mathcal{H}_1^{(\mathrm{se,in})}(x)}
-
\Bigl( 
x+\frac{x}{2(1+x^2)}
\Bigr)
\mathbf{1}_{[0,+\infty)}(x)
\bigg]
\,dx
\\
&
-\bigg(\tau_{\rho_1}^2M^2
+2\rho_1^2\Delta Q(\rho_1)\log( M\tau_{\rho_1})
-\rho_1^2\Delta Q(\rho_1)\log (2\rho_1^2\Delta Q(\rho_1))
\bigg)
\frac{\mathsf{T}_{1}^{(2m+1,3m)}(\tau_{\rho_1})}{1+\mathsf{T}_{0}^{(2m+1,3m)}(\tau_{\rho_1})}
+\mathcal{O}(M^{-2}).
\end{align*}

Collecting the preceding formulas gives precisely
$C_{E_4^{(\mathrm{dc})}}^{(4)}$,
$C_{E_4^{(\mathrm{dc})}}^{(6)}$, and
$\widetilde C_{E_4^{(\mathrm{dc})}}^{(M)}$ in part~(i).  The accumulated
local error is $\mathcal O(M^5n^{-1/2})$, while the last tail truncation is
$\mathcal O(M^{-2})$.  With $M=n^{1/12}$, both are
$\mathcal O(n^{-1/12})$ or smaller.  This proves part~(i).

\medskip

\emph{Proof of (ii).}
By definition,
\begin{align*}
E_5^{(\mathrm{dc})}&=  
\sum_{j=g_{1,+}+1}^{n-1}\log\Big(
1+\sum_{\ell=1}^{m}\omega_{\ell}F_{n,j,\ell}^{(\mathrm{dc})}
+\sum_{\ell=m+1}^{2m}\omega_{\ell}F_{n,j,\ell}^{(\mathrm{dc})}
+\sum_{\ell=2m+1}^{3m}\omega_{\ell}F_{n,j,\ell}^{(\mathrm{dc})}
\Bigr).
\end{align*}
The contribution
$\sum_{\ell=1}^{m}\omega_{\ell}F_{n,j,\ell}^{(\mathrm{dc})}$ is
exponentially small, uniformly throughout this range.
We first show that the semi-hard contribution is negligible.  For
$g_{1,+}+1\leq j\leq n-1$, $\ell=m+1,\ldots,2m$, and
$r_{\ell}=\rho_1-t_{\ell}/\sqrt{2n\Delta Q(\rho_1)}$, decompose
\begin{align*}
&\quad h_{n,j,\ell}^{(\mathrm{dc})}
=\int_0^{r_{\ell}}2r e^{-nV_{\tau}(r)}e^{\mathsf{k}(r)}\,dr
\\
&=\int_0^{\rho_1-\delta_n}2r e^{-nV_{\tau}(r)}e^{\mathsf{k}(r)}\,dr
+\int_{\rho_1-\delta_n}^{\rho_1}2r e^{-nV_{\tau}(r)}e^{\mathsf{k}(r)}\,dr
-\int_{\rho_1-\frac{t_{\ell}}{\sqrt{2n\Delta Q(\rho_1)}}}^{\rho_1}2r e^{-nV_{\tau}(r)}e^{\mathsf{k}(r)}\,dr.
\end{align*}
The first two terms are controlled by
Lemma~\ref{lemma: g1+leq j leq j+1}.  The same endpoint expansion gives
\[
2e^{\mathsf{k}(\rho_1)}e^{-nV_{\tau}(\rho_1)}\int_{\rho_1-\frac{t_{\ell}}{\sqrt{2n\Delta Q(\rho_1)}}}^{\rho_1}re^{\mathsf{k}(r)-\mathsf{k}(\rho_1)}e^{-n(V_{\tau}(r)-V_{\tau}(\rho_1))}\,dr
=h_{n,j}^{\#,\mathrm{hard}}[\mathbb{D}_{\rho_1}^{\rm c}]\cdot\Bigl(1+\mathcal{O}(e^{-cM})\Bigr),
\]
for some $c>0$ independent of the displayed indices.  Therefore,
\begin{align*}
E_5^{(\mathrm{dc})}&=  
\sum_{j=g_{1,+}+1}^{n-1}\log\Big(
1+\sum_{\ell=1}^{m}\omega_{\ell}\cdot \mathcal{O}(e^{-cM^2})
+\sum_{\ell=m+1}^{2m}\omega_{\ell}\cdot \mathcal{O}(e^{-c'M})
+\sum_{\ell=2m+1}^{3m}\omega_{\ell}F_{n,j,\ell}^{(\mathrm{dc})}
\Bigr). 
\end{align*}
It remains to analyze the hard-edge radii.  For
$\ell=2m+1,\ldots,3m$ and $g_{1,+}+1\leq j\leq n-1$, write
\[
h_{n,j,\ell}^{(\mathrm{dc})}
=\int_0^{\rho_1-\delta_n}2r e^{-nV_{\tau}(r)}e^{\mathsf{k}(r)}\,dr
+\int_{\rho_1-\delta_n}^{\rho_1}2r e^{-nV_{\tau}(r)}e^{\mathsf{k}(r)}\,dr
-\int_{\rho_1-\frac{\rho_1t_{\ell}}{n}}^{\rho_1}2r e^{-nV_{\tau}(r)}e^{\mathsf{k}(r)}\,dr.
\]
The first two terms were estimated in
Lemma~\ref{lemma: g1+leq j leq j+1}.  A direct adaptation of the endpoint
calculation in that lemma gives
\begin{align*}
&\quad 2e^{\mathsf{k}(\rho_1)}e^{-nV_{\tau}(\rho_1)}\int_{\rho_1-\frac{\rho_1t_{\ell}}{n}}^{\rho_1}re^{\mathsf{k}(r)-\mathsf{k}(\rho_1)}e^{-n(V_{\tau}(r)-V_{\tau}(\rho_1))}\,dr
\\
&=
\frac{2\rho_1e^{\mathsf{k}(\rho_1)}e^{-nV_{\tau}(\rho_1)}}{\eta_1 n}\int_0^{\rho_1t_{\ell}}
e^{-u}
\Bigl(1-\frac{a_1(\eta_1^{-1}u)}{n}+\frac{a_2(\eta_1^{-1}u)}{n^2}-\frac{a_3(\eta_1^{-1}u)}{n^3}+\mathcal{O}(\frac{a_4(\eta_1^{-1}u)}{n^4})\Bigr)
\,du.
\end{align*}
Combining this with the second term yields
\[
h_{n,j,\ell}^{(\mathrm{dc})}
=\frac{2\rho_1e^{\mathsf{k}(\rho_1)}e^{-nV_{\tau}(\rho_1)}}{\eta_1 n}
e^{-\varphi_{\ell}\eta_1}
\Bigl( 
1
-\frac{\mathsf{a}_1(\varphi_{\ell};\eta_1)}{n}
+\frac{\mathsf{a}_2(\varphi_{\ell};\eta_1)}{n^2}
-\frac{\mathsf{a}_3(\varphi_{\ell};\eta_1)}{n^3}
+\mathcal{O}\Bigl(\frac{1}{\eta_1^{8}n^{4}}\Bigr)
\Bigr),
\]
where, with $\varphi_{\ell}:=\rho_1t_{\ell}$, the coefficients $\mathsf{a}_k(\varphi_{\ell};\eta_1)$, $k=1,2,3$, are defined by 
\begin{align*}
\mathsf{a}_1(\varphi_{\ell};\eta_1)&:=
\Bigl(\mathcal{V}_2(\rho_1)+\frac{\eta_1}{\rho_1}\Bigr)\frac{1}{\eta_1^2}\mathsf{e}_2(\varphi_{\ell}\eta_1)+\Bigl(\mathsf{k}'(\rho_1)+\frac{1}{\rho_1}\Bigr)\frac{1}{\eta_1}\mathsf{e}_1(\varphi_{\ell}\eta_1),
\\
\mathsf{a}_2(\varphi_{\ell};\eta_1)&:=
\Bigl(\mathcal{V}_2(\rho_1)+\frac{\eta_1}{\rho_1}\Bigr)^2
\frac{3\mathsf{e}_4(\varphi_{\ell}\eta_1)}{\eta_1^4}
+
\Bigl\{
3\Bigl(\mathcal{V}_2(\rho_1)+\frac{\eta_1}{\rho_1}\Bigr)
\Bigl(\mathsf{k}'(\rho_1)+\frac{1}{\rho_1}\Bigr)
+\mathcal{V}_3(\rho_1)-\frac{2\eta_1}{\rho_1^2}
\Bigr\}
\frac{\mathsf{e}_3(\varphi_{\ell}\eta_1)}{\eta_1^3}
\\
&\quad 
+\Bigl(\frac{2\mathsf{k}'(\rho_1)}{\rho_1}+\mathsf{k}'(\rho_1)^2+\mathsf{k}''(\rho_1)\Bigr)
\frac{\mathsf{e}_2(\varphi_{\ell}\eta_1)}{\eta_1^2}, 
\\
\mathsf{a}_3(\varphi_{\ell};\eta_1)&:=
\Bigl(\mathcal{V}_2(\rho_1)+\frac{\eta_1}{\rho_1}\Bigr)^3
\frac{15\mathsf{e}_6(\varphi_{\ell}\eta_1)}{\eta_1^6}
\\
&\quad 
+
\Bigl\{
3\Bigl(\mathcal{V}_2(\rho_1)+\frac{\eta_1}{\rho_1}\Bigr)^2\Bigl(\mathsf{k}'(\rho_1)+\frac{1}{\rho_1}\Bigr)
+
2\Bigl(\mathcal{V}_2(\rho_1)+\frac{\eta_1}{\rho_1}\Bigr)\Bigl(\mathcal{V}_3(\rho_1)-\frac{2\eta_1}{\rho_1^2}\Bigr)
\Bigr\}\frac{5\mathsf{e}_5(\varphi_{\ell}\eta_1)}{\eta_1^5}
\\
&\quad 
+
\mathcal{O}\Bigl( 
\frac{\mathsf{e}_4(\varphi_{\ell}\eta_1)}{\eta_1^4}(1+\eta_1)
+
\frac{\mathsf{e}_3(\varphi_{\ell}\eta_1)}{\eta_1^3}
\Bigr),
\end{align*}
where, consistently with \eqref{def of V tau relationship 1},
\[
\mathcal V_2(\rho_1):=4\Delta Q(\rho_1),
\qquad
\mathcal V_3(\rho_1):=4\partial_r\Delta Q(\rho_1)
-\frac{4\Delta Q(\rho_1)}{\rho_1}.
\]
Thus $V_\tau''(\rho_1)=\mathcal V_2(\rho_1)+\eta_1/\rho_1$ and
$V_\tau^{(3)}(\rho_1)=\mathcal V_3(\rho_1)-2\eta_1/\rho_1^2$.  Recall that
$\eta_1=\frac{2}{\rho_1}(\tau-\tau_{\rho_1})$, and define the truncated
exponential polynomials by
\begin{equation}
    \label{def of mathsf ek}
\int_{x}^{+\infty}u^k e^{-u}\,du=k!e^{-x}\mathsf{e}_k(x),\qquad
    \mathsf{e}_{k,\ell}(x):=\sum_{p=\ell}^k \frac{x^p}{p!}, \qquad
    \mathsf{e}_{k}(x)=\mathsf{e}_{k,0}(x), 
\end{equation}
i.e., $\mathsf{e}_{k}(x)=1+\mathsf{e}_{k,1}(x)$. 
Define 
\begin{align*}
\mathsf{a}_{1,1}(\varphi_{\ell};\eta_1)&:=
\Bigl(\mathcal{V}_2(\rho_1)+\frac{\eta_1}{\rho_1}\Bigr)\frac{1}{\eta_1^2}\mathsf{e}_{2,1}(\varphi_{\ell}\eta_1)+\Bigl(\mathsf{k}'(\rho_1)+\frac{1}{\rho_1}\Bigr)\frac{1}{\eta_1}\mathsf{e}_{1,1}(\varphi_{\ell}\eta_1),
\\
\mathsf{a}_{2,1}(\varphi_{\ell};\eta_1)&:=
\Bigl(\mathcal{V}_2(\rho_1)+\frac{\eta_1}{\rho_1}\Bigr)^2
\frac{3\mathsf{e}_{4,1}(\varphi_{\ell}\eta_1)}{\eta_1^4}
\\
&\quad+
\Bigl\{
3\Bigl(\mathcal{V}_2(\rho_1)+\frac{\eta_1}{\rho_1}\Bigr)
\Bigl(\mathsf{k}'(\rho_1)+\frac{1}{\rho_1}\Bigr)
+\mathcal{V}_3(\rho_1)-\frac{2\eta_1}{\rho_1^2}
\Bigr\}
\frac{\mathsf{e}_{3,1}(\varphi_{\ell}\eta_1)}{\eta_1^3}
\\
&\quad 
+\Bigl(\frac{2\mathsf{k}'(\rho_1)}{\rho_1}+\mathsf{k}'(\rho_1)^2+\mathsf{k}''(\rho_1)\Bigr)
\frac{\mathsf{e}_{2,1}(\varphi_{\ell}\eta_1)}{\eta_1^2}, 
\\
\mathsf{a}_{3,1}(\varphi_{\ell};\eta_1)&:=
\Bigl(\mathcal{V}_2(\rho_1)+\frac{\eta_1}{\rho_1}\Bigr)^3
\frac{15\mathsf{e}_{6,1}(\varphi_{\ell}\eta_1)}{\eta_1^6}
\\
&\quad 
+
\left\{\begin{aligned}&
3\Bigl(\mathcal{V}_2(\rho_1)+\frac{\eta_1}{\rho_1}\Bigr)^2\Bigl(\mathsf{k}'(\rho_1)+\frac{1}{\rho_1}\Bigr)
+
\\
&\quad 2\Bigl(\mathcal{V}_2(\rho_1)+\frac{\eta_1}{\rho_1}\Bigr)\Bigl(\mathcal{V}_3(\rho_1)-\frac{2\eta_1}{\rho_1^2}\Bigr)
\end{aligned}\right\}\frac{5\mathsf{e}_{5,1}(\varphi_{\ell}\eta_1)}{\eta_1^5}
\\
&\quad 
+
\mathcal{O}\Bigl( 
\frac{\mathsf{e}_{4,1}(\varphi_{\ell}\eta_1)}{\eta_1^4}(1+\eta_1)
+
\frac{\mathsf{e}_{3,1}(\varphi_{\ell}\eta_1)}{\eta_1^3}
\Bigr). 
\end{align*}
Recall that $\mathsf{T}_{j}^{(n,m)}(\tau_{\rho_1};\vec{s},\vec{t})
=\sum_{\ell=n}^{m}\omega_{\ell}t_{\ell}^j$.  Summing the preceding local
expansion over $\ell=2m+1,\ldots,3m$ gives
\[
\sum_{\ell={2m+1}}^{3m}\omega_{\ell}F_{n,j,\ell}^{(\mathrm{dc})}
=
\mathsf{T}_{0}^{(2m+1,3m)}(j/n)
-\frac{1}{n}\mathscr{H}_{\rho_1,1}
+\frac{1}{n^2}\mathscr{H}_{\rho_1,2}
-\frac{1}{n^3}\mathscr{H}_{\rho_1,3}
+\mathcal{O}(\frac{1}{\eta_1^8n^4}), 
\]
where 
\begin{align*}
\mathsf{T}_{0}^{(2m+1,3m)}(j/n)
&=
\sum_{\ell={2m+1}}^{3m}\omega_{\ell}
e^{-t_{\ell}\rho_1\eta_1}, 
\quad
\mathscr{H}_{\rho_1,1}
=
\sum_{\ell={2m+1}}^{3m}\omega_{\ell}
e^{-t_{\ell}\rho_1\eta_1}\mathsf{a}_{1,1}(\varphi_{\ell};\eta_1), 
\\
\mathscr{H}_{\rho_1,2}
&=
\sum_{\ell={2m+1}}^{3m}\omega_{\ell}
e^{-t_{\ell}\rho_1\eta_1}\Bigl(
\mathsf{a}_{2,1}(\varphi_{\ell};\eta_1)
-\mathfrak{a}_1\mathsf{a}_{1,1}(\varphi_{\ell};\eta_1)
\Bigr), 
\\
\mathscr{H}_{\rho_1,3}
&=
\sum_{\ell={2m+1}}^{3m}\omega_{\ell}
e^{-t_{\ell}\rho_1\eta_1}
\Bigl(\mathsf{a}_{3,1}(\varphi_{\ell};\eta_1)
+
\mathfrak{a}_1^2\mathsf{a}_{1,1}(\varphi_{\ell};\eta_1)
-
\mathfrak{a}_2\mathsf{a}_{1,1}(\varphi_{\ell};\eta_1)
-
\mathfrak{a}_1\mathsf{a}_{2,1}(\varphi_{\ell};\eta_1)
\Bigr). 
\end{align*}
Set $\psi_j:=\rho_1\eta_1=2(j/n-\tau_{\rho_1})$.  Taylor expansion of the
logarithm then gives
\begin{align*}
&\quad E_5^{(\mathrm{dc})}\\
&=  
\sum_{j=g_{1,+}+1}^{n-1}
\log\bigl(1+\mathsf{T}_{0}^{(2m+1,3m)}(j/n)\bigr)
-\frac{1}{n}
\sum_{j=g_{1,+}+1}^{n-1}
\frac{\mathscr{H}_{\rho_1,1}}{1+\mathsf{T}_{0}^{(2m+1,3m)}(j/n)}
\\
&
+\frac{1}{n^2}\sum_{j=g_{1,+}+1}^{n-1}\Bigl( 
\frac{\mathscr{H}_{\rho_1,2}}{1+\mathsf{T}_{0}^{(2m+1,3m)}(j/n)}-\frac{\mathscr{H}_{\rho_1,1}^2}{2(1+\mathsf{T}_{0}^{(2m+1,3m)}(j/n))^2}
\Bigr)
\\
&
-\frac{1}{n^3}
\sum_{j=g_{1,+}+1}^{n-1}
\Bigl( 
\frac{\mathscr{H}_{\rho_1,3}}{1+\mathsf{T}_{0}^{(2m+1,3m)}(j/n)}
-\frac{\mathscr{H}_{\rho_1,1}\mathscr{H}_{\rho_1,2}}{(1+\mathsf{T}_{0}^{(2m+1,3m)}(j/n))^2}
\\
&\qquad+\frac{\mathscr{H}_{\rho_1,1}^3}{3(1+\mathsf{T}_{0}^{(2m+1,3m)}(j/n))^3}
\Bigr)
+\mathcal{O}(n^{1/2}M^{-7}).
\end{align*}
Since $M=n^{1/12}$, the final remainder is $\mathcal O(n^{-1/12})$.
Recall that
$\theta_{1,+}^{(n,M)}=n\tau_{\rho_1}/(1-M/\sqrt n)-g_{1,+}$.
Apply Lemma~\ref{lemma:Riemann sum NEW} with
\[
A=\frac{\tau_{\rho_1}}{1-M/\sqrt n},\qquad
a_0=1-\theta_{1,+}^{(n,M)},\qquad B=1,\qquad b_0=-1.
\]
We obtain
\begin{align*}
&\quad\sum_{j=g_{1,+}+1}^{n-1}
\log\bigl(1+\mathsf{T}_{0}^{(2m+1,3m)}(j/n)\bigr)
\\
&=
n\int_{\frac{\tau_{\rho_1}}{1-\frac{M}{\sqrt{n}}}}^1
\log\bigl(1+\mathsf{T}_{0}^{(2m+1,3m)}(x)\bigr)\,dx
\\
&\quad
+
\Bigl(\theta_{1,+}^{(n,M)}-\frac{1}{2}\Bigr)
\log\bigl(1+\mathsf{T}_{0}^{(2m+1,3m)}(\tau_{\rho_1})\bigr)
-
\frac{1}{2}\log\bigl(1+\mathsf{T}_{0}^{(2m+1,3m)}(1)\bigr)
+
\mathcal{O}(Mn^{-1/2}).
\end{align*}
Writing
$f_{\mathsf T_0}(x):=\log(1+\mathsf T_0^{(2m+1,3m)}(x))$ and expanding
the moving lower endpoint gives
\begin{align*}
&\quad n\int_{\frac{\tau_{\rho_1}}{1-\frac{M}{\sqrt{n}}}}^1
\log\bigl(1+\mathsf{T}_{0}^{(2m+1,3m)}(x)\bigr)\,dx
\\
&=
n\int_{\tau_{\rho_1}}^1
f_{\mathsf{T}_0}(x)\,dx
-\tau_{\rho_1}M\sqrt{n}f_{\mathsf{T}_0}(\tau_{\rho_1})
-M^2\tau_{\rho_1}\Bigl( 
f_{\mathsf{T}_0}(\tau_{\rho_1})
+\frac{\tau_{\rho_1}}{2}f_{\mathsf{T}_0}'(\tau_{\rho_1})\Bigr)
+\mathcal{O}(M^3n^{-1/2}). 
\end{align*}
Let 
\begin{align*}
\mathscr{F}_1(x)&:=
\frac{2\rho_1^2\Delta Q(\rho_1)}{x-\tau_{\rho_1}}
\frac{\mathsf{T}_{1}^{(2m+1,3m)}(x)}{1+\mathsf{T}_{0}^{(2m+1,3m)}(x)}
+
x
\frac{\mathsf{T}_{2}^{(2m+1,3m)}(x)}{1+\mathsf{T}_{0}^{(2m+1,3m)}(x)}, 
\\
\mathscr{F}_2(x)&:=
\frac{\mathsf{T}_{1}^{(2m+1,3m)}(x)}{1+\mathsf{T}_{0}^{(2m+1,3m)}(x)}, 
\qquad
\mathscr{F}_3(x):=
\frac{\mathsf{T}_{2}^{(2m+1,3m)}(x)}{1+\mathsf{T}_{0}^{(2m+1,3m)}(x)}. 
\end{align*}
The same lemma, applied to the first correction term, gives
\begin{align*}
&\quad \frac{1}{n}\sum_{j=g_{1,+}+1}^{n-1}
\frac{\mathscr{H}_{\rho_1,1}}{1+\mathsf{T}_{0}^{(2m+1,3m)}(j/n)}
\\
&=
\int_{\frac{\tau_{\rho_1}}{1-\frac{M}{\sqrt{n}}}}^{1}
\mathscr{F}_1(x)\,dx
+
2\Bigl(\frac{1}{2}\rho_1\mathsf{k}'(\rho_1)+1\Bigr)
\int_{\frac{\tau_{\rho_1}}{1-\frac{M}{\sqrt{n}}}}^{1}
\mathscr{F}_2(x)
\,dx
\\
&\quad+
\bigl(2\rho_1^2\Delta Q(\rho_1)-\tau_{\rho_1}\bigr)
\int_{\frac{\tau_{\rho_1}}{1-\frac{M}{\sqrt{n}}}}^{1}
\mathscr{F}_3(x)
\,dx
\\
&\quad 
+
\frac{1}{n}\Bigl(\theta_{1,+}^{(n,M)}-\frac{1}{2}\Bigr)\mathscr{F}_{1}(\tfrac{\tau_{\rho_1}}{1-\frac{M}{\sqrt{n}}})
-\frac{1}{2n}\mathscr{F}_{1}(1)
\\
&\quad+
\frac{2}{n}\Bigl(\frac{1}{2}\rho_1\mathsf{k}'(\rho_1)+1\Bigr)\Bigl[
\Bigl(\theta_{1,+}^{(n,M)}-\frac{1}{2}\Bigr)\mathscr{F}_{2}(\tfrac{\tau_{\rho_1}}{1-\frac{M}{\sqrt{n}}})
-\frac{1}{2}\mathscr{F}_{2}(1)
\Bigr]
\\
&\quad 
+\frac{1}{n}\bigl(2\rho_1^2\Delta Q(\rho_1)-\tau_{\rho_1}\bigr)
\Bigl[
\Bigl(\theta_{1,+}^{(n,M)}-\frac{1}{2}\Bigr)\mathscr{F}_{3}(\tfrac{\tau_{\rho_1}}{1-\frac{M}{\sqrt{n}}})
-\frac{1}{2}\mathscr{F}_{3}(1)
\Bigr]
+\mathcal{O}(n^{-2}). 
\end{align*}
Separating its singular part at $x=\tau_{\rho_1}$, we find
\begin{align*}
\int_{\frac{\tau_{\rho_1}}{1-\frac{M}{\sqrt{n}}}}^{1}
\mathscr{F}_1(x)\,dx
&=
\int_{\tau_{\rho_1}}^{1}
\Bigl(\mathscr{F}_1(x)
-\frac{2\rho_1^2\Delta Q(\rho_1)}{x-\tau_{\rho_1}}\frac{\mathsf{T}_{1}^{(2m+1,3m)}(\tau_{\rho_1})}{1+\mathsf{T}_{0}^{(2m+1,3m)}(\tau_{\rho_1})}\Bigr)\,dx
\\
&\quad 
+
2\rho_1^2\Delta Q(\rho_1)
\frac{\mathsf{T}_{1}^{(2m+1,3m)}(\tau_{\rho_1})}{1+\mathsf{T}_{0}^{(2m+1,3m)}(\tau_{\rho_1})}
\Bigl( 
\log(1-\tau_{\rho_1})-\log\Bigl(\frac{M\tau_{\rho_1}}{\sqrt{n}}\Bigr)
\Bigr)
+
\mathcal{O}(Mn^{-1/2}), 
\\
\int_{\frac{\tau_{\rho_1}}{1-\frac{M}{\sqrt{n}}}}^{1}
\mathscr{F}_2(x)
\,dx
&
=
-\frac{1}{2}\log(1+\mathsf{T}_0^{(2m+1,3m)}(1))
+\frac{1}{2}\log(1+\mathsf{T}_0^{(2m+1,3m)}(\tau_{\rho_1}))
+\mathcal{O}(Mn^{-1/2}), 
\\
\int_{\frac{\tau_{\rho_1}}{1-\frac{M}{\sqrt{n}}}}^{1}
\mathscr{F}_3(x)\,dx
&=
\int_{\tau_{\rho_1}}^{1}
\frac{\mathsf{T}_{2}^{(2m+1,3m)}(x)}{1+\mathsf{T}_{0}^{(2m+1,3m)}(x)}\,dx
+
\mathcal{O}(Mn^{-1/2}). 
\end{align*}
We finally use the uniform summation bound \cite[proof of Lemma~2.6]{ACCL1}: for
$A,B>1$,
\[
\sum_{j=g_{1,+}+1}^{n-1}\mathcal{O}\Bigl(\frac{1}{n^A(j/n-\tau_{\rho_1})^B}\Bigr)
=
\mathcal{O}\Bigl(\frac{1}{n^{A-(B+1)/2}M^{B-1}}\Bigr),
\]
which yields
\begin{align*}
&\quad \frac{1}{n^2}\sum_{j=g_{1,+}+1}^{n-1}\Bigl( 
\frac{\mathscr{H}_{\rho_1,2}}{1+\mathsf{T}_{0}^{(2m+1,3m)}(j/n)}-\frac{\mathscr{H}_{\rho_1,1}^2}{2(1+\mathsf{T}_{0}^{(2m+1,3m)}(j/n))^2}
\Bigr)
\\
&=
\frac{2\rho_1^4(\Delta Q(\rho_1))^2}{\tau_{\rho_1}^2M^2}
\frac{\mathsf{T}_{1}^{(2m+1,3m)}(\tau_{\rho_1})}{1+\mathsf{T}_{0}^{(2m+1,3m)}(\tau_{\rho_1})}
+\mathcal{O}(M^{-1}n^{-1/2}),
\end{align*}
and 
\begin{align*}
 &\quad 
 -\frac{1}{n^3}
\sum_{j=g_{1,+}+1}^{n-1}
\Bigl( 
\frac{\mathscr{H}_{\rho_1,3}}{1+\mathsf{T}_{0}^{(2m+1,3m)}(j/n)}
-\frac{\mathscr{H}_{\rho_1,1}\mathscr{H}_{\rho_1,2}}{(1+\mathsf{T}_{0}^{(2m+1,3m)}(j/n))^2}
+\frac{\mathscr{H}_{\rho_1,1}^3}{3(1+\mathsf{T}_{0}^{(2m+1,3m)}(j/n))^3}
\Bigr)
\\
&=
-
\frac{5\rho_1^6\Delta Q(\rho_1)^3}{\tau_{\rho_1}^4M^4}
\frac{\mathsf{T}_{1}^{(2m+1,3m)}(\tau_{\rho_1})}{1+\mathsf{T}_{0}^{(2m+1,3m)}(\tau_{\rho_1})}
+\mathcal{O}(M^{-3}n^{-1/2}). 
\end{align*}
Substitution of these estimates into the logarithmic expansion identifies the
$n$, $\log n$, and constant contributions as
$C_{E_5^{(\mathrm{dc})}}^{(3)}$,
$C_{E_5^{(\mathrm{dc})}}^{(5)}$, and
$C_{E_5^{(\mathrm{dc})}}^{(6)}$, respectively; there is no $\sqrt n$ term.
The moving-endpoint terms give
$\widetilde C_{E_5^{(\mathrm{dc})}}^{(M)}$.  The omitted terms displayed
above are bounded by combinations of
$M^3n^{-1/2}$, $M^{-2}$, $M^{-4}$, and $M^{-8}$.  For $M=n^{1/12}$ these
are all $\mathcal O(n^{-1/12})$ or smaller.  This proves part~(ii) and the
lemma.
\end{proof}

\begin{proof}[Proof of Theorem~\ref{theorem:counting statistics of disk complement case}]
It remains only to match the three artificial cutoffs.  From
Lemmas~\ref{lemma:bulk case asymptotic expansion}
and~\ref{lemma:semi hard edge case asymptotic expansion disck complement},
the cutoff-dependent pieces cancel algebraically as follows:
\begin{align*}
&E_0^{(\mathrm{dc})}+E_1^{(\mathrm{dc})}
+\widetilde C_{E_2^{(\mathrm{dc})}}^{(M)}+E_3^{(\mathrm{dc})}
+\widetilde C_{E_4^{(\mathrm{dc})}}^{(M)}
+\widetilde C_{E_5^{(\mathrm{dc})}}^{(M)}
\\
&\qquad=
n\mu_Q[\mathbb D_{a_1}]\sum_{\ell=1}^{m}s_\ell
+n\mu_Q[\mathbb D_{\rho_1}]\sum_{\ell=m+1}^{3m}s_\ell
+\frac12\sum_{\ell=1}^{3m}s_\ell
+\mathcal O(M^3n^{-1/2}).
\end{align*}
The $n$ term in this identity, together with
$C_{E_5^{(\mathrm{dc})}}^{(3)}n$, is exactly
$C_{3,\#}^{(\mathrm{dc})}n$.  The remaining $\sqrt n$, $\log n$, and
constant terms are, by their definitions,
$C_{4,\#}^{(\mathrm{dc})}\sqrt n$,
$C_{5,\#}^{(\mathrm{dc})}\log n$, and
$C_{6,\#}^{(\mathrm{dc})}$, respectively.  Finally,
$M^3n^{-1/2}=n^{-1/4}$ is absorbed by the
$\mathcal O(n^{-1/12})$ remainders in the two lemmas.  This proves the
theorem, with the uniformity stated at the beginning of the subsection.
\end{proof}

\subsection{Proof of Theorem~\ref{theorem:counting statistics of annulus case}}
We now prove Theorem~\ref{theorem:counting statistics of annulus case}.  The
argument is organized according to the location of the saddle relative to the
two hard boundaries and to the three auxiliary radii $a_1,a_2$, and $R$.
Thus the bulk, transition, hard-edge, and soft-edge contributions can be
estimated in their natural ranges and then matched at the artificial cutoffs.
We again decompose the logarithm, now into thirteen consecutive index
ranges:
\[
\log \mathcal{E}_{n,s\lambda,\alpha}^{(\mathrm{an})}=
\sum_{j=0}^{n-1}\log\Big( 
1+\sum_{\ell=1}^{7m}\omega_{\ell}F_{n, j,\ell}^{(\mathrm{an})}
\Bigr)
=
\sum_{k=0}^{12}E_k^{(\mathrm{an})},
\]
where $F_{n,j,\ell}^{(\mathrm{an})}$ is defined in
\eqref{def of norming constant + counting on annulus}, and
\begin{align*}
E_0^{(\mathrm{an})}&:=\sum_{j=0}^{D_n-1}\log\Big( 
1+\sum_{\ell=1}^{7m}\omega_{\ell}F_{n,j,\ell}^{(\mathrm{an})}
\Bigr), \qquad
E_1^{(\mathrm{an})}:=\sum_{j=D_n}^{g_{a_1,-}-1}\log\Big( 
1+\sum_{\ell=1}^{7m}\omega_{\ell}F_{n,j,\ell}^{(\mathrm{an})}
\Bigr), \\
E_2^{(\mathrm{an})}&:=\sum_{j=g_{a_1,-}}^{g_{a_1,+}}\log\Big( 
1+\sum_{\ell=1}^{7m}\omega_{\ell}F_{n,j,\ell}^{(\mathrm{an})}
\Bigr),\qquad
E_3^{(\mathrm{an})}:=\sum_{j=g_{a_1,+}+1}^{g_{1,-}-1}\log\Big( 
1+\sum_{\ell=1}^{7m}\omega_{\ell}F_{n,j,\ell}^{(\mathrm{an})}
\Bigr), \\
E_4^{(\mathrm{an})}&:=\sum_{j=g_{1,-}}^{g_{1,+}}\log\Big( 
1+\sum_{\ell=1}^{7m}\omega_{\ell}F_{n,j,\ell}^{(\mathrm{an})}
\Bigr),\qquad
E_5^{(\mathrm{an})}:=\sum_{j=g_{1,+}+1}^{j_{1,+}}\log\Big( 
1+\sum_{\ell=1}^{7m}\omega_{\ell}F_{n,j,\ell}^{(\mathrm{an})}
\Bigr), \\
E_6^{(\mathrm{an})}&:=\sum_{j=j_{1,+}+1}^{j_{2,-}-1}\log\Big( 
1+\sum_{\ell=1}^{7m}\omega_{\ell}F_{n,j,\ell}^{(\mathrm{an})}
\Bigr),\qquad
E_7^{(\mathrm{an})}:=\sum_{j=j_{2,-}}^{g_{2,-}-1}\log\Big( 
1+\sum_{\ell=1}^{7m}\omega_{\ell}F_{n,j,\ell}^{(\mathrm{an})}
\Bigr), 
\\
E_8^{(\mathrm{an})}&:=\sum_{j=g_{2,-}}^{g_{2,+}}\log\Big( 
1+\sum_{\ell=1}^{7m}\omega_{\ell}F_{n,j,\ell}^{(\mathrm{an})}
\Bigr),\qquad
E_9^{(\mathrm{an})}:=\sum_{j=g_{2,+}+1}^{g_{a_2,-}-1}\log\Big( 
1+\sum_{\ell=1}^{7m}\omega_{\ell}F_{n,j,\ell}^{(\mathrm{an})}
\Bigr), \\
E_{10}^{(\mathrm{an})}&:=\sum_{j=g_{a_2,-}}^{g_{a_2,+}}\log\Big( 
1+\sum_{\ell=1}^{7m}\omega_{\ell}F_{n,j,\ell}^{(\mathrm{an})}
\Bigr),\qquad
E_{11}^{(\mathrm{an})}:=\sum_{j=g_{a_2,+}+1}^{g_{R,-}-1}\log\Big( 
1+\sum_{\ell=1}^{7m}\omega_{\ell}F_{n,j,\ell}^{(\mathrm{an})}
\Bigr), 
\\
E_{12}^{(\mathrm{an})}&:=\sum_{j=g_{R,-}}^{n-1}\log\Big( 
1+\sum_{\ell=1}^{7m}\omega_{\ell}F_{n,j,\ell}^{(\mathrm{an})}
\Bigr).
\end{align*}
Here the one-sided soft-edge cutoff and its rounding displacement are
\[
g_{R,-}:=\Bigl\lceil\frac{n}{1+M/\sqrt n}\Bigr\rceil,
\qquad
\theta_{R,-}^{(n,M)}:=g_{R,-}-\frac{n}{1+M/\sqrt n}\in[0,1).
\]
With these endpoint conventions, the ranges are pairwise disjoint and their
union is $\{0,\ldots,n-1\}$.  In particular, every
rounding correction is attached to exactly one endpoint.  Throughout this
subsection, $M=n^{1/12}$ as fixed above.  Unless stated otherwise, all error
bounds are uniform when the deformation parameters range over compact subsets
of the admissible parameter set; the positive constant $c$ may change from one
occurrence to the next.
Most of the ranges are governed by the same local estimates as in the
disk-complement case.  We record every contribution needed in the final
summation and give details for the ranges in which the second hard boundary or
the matching point $\sigma_\star$ changes the argument.
\begin{lemma}\label{lemma:asymptotic expansion E0 annulus}
There exists $c>0$ such that, as $n\to\infty$,
\[
    E_{0}^{(\mathrm{an})}
    =
    D_n\log \Omega_{1}^{(7m)}+\mathcal{O}(e^{-cM^2}). 
\]
\end{lemma}

\begin{proof}
In this range the tail beyond the innermost counting radius is
exponentially small, uniformly in $j$.  Thus every cumulative ratio
$F_{n,j,\ell}^{(\mathrm{an})}$ is exponentially close to one.  The bulk
estimate used in Lemma~\ref{lemma:bulk case asymptotic expansion} gives,
term by term,
\[
 \log\Bigl(1+\sum_{\ell=1}^{7m}\omega_\ell
 F_{n,j,\ell}^{(\mathrm{an})}\Bigr)
 =\log\Omega_1^{(7m)}+\mathcal O(e^{-cM^2}).
\]
Summing over $0\le j<D_n$ preserves the asserted error, since the polynomial
factor $D_n$ is absorbed by decreasing $c$ if necessary.
\end{proof}

\begin{lemma}\label{lemma:asymptotic expansion E1 annulus}
There exists $c>0$ such that, as $n\to\infty$,
\[
E_{1}^{(\mathrm{an})}
=(g_{a_1,-}-D_n)\log \Omega_{1}^{(7m)}+\mathcal{O}(e^{-cM^2}).
\]
\end{lemma}

\begin{proof}
The same uniform bulk estimate as in the proof of
Lemma~\ref{lemma:asymptotic expansion E0 annulus} applies for
$D_n\le j<g_{a_1,-}$.  There are $g_{a_1,-}-D_n$ indices in this range, and
summing the pointwise estimate gives the stated formula.
\end{proof}

\begin{lemma}\label{lemma:asymptotic expansion E2 annulus}
As $n\to\infty$,
\[
E_{2}^{(\mathrm{an})}
 =
 C_{E_{2}^{(\mathrm{an})}}^{(1)}n 
 + C_{E_{2}^{(\mathrm{an})}}^{(2)} \sqrt{n} 
 + C_{E_{2}^{(\mathrm{an})}}^{(3)}\log n
 + C_{E_{2}^{(\mathrm{an})}}^{(4)} 
 + \widetilde{C}_{E_{2}^{(\mathrm{an})}}^{(M)} 
 + \mathcal{O}(n^{-\frac{1}{12}}), 
\]
where   
\begin{align*}
C_{E_{2}^{(\mathrm{an})}}^{(1)}
&:=0, 
\quad
C_{E_{2}^{(\mathrm{an})}}^{(2)}
:=C_4^{\#(\mathrm{b,in})},
\quad
C_{E_{2}^{(\mathrm{an})}}^{(3)}
:=0,
\quad
C_{E_{2}^{(\mathrm{an})}}^{(4)}
:=
\mathcal{D}_6^{(\mathrm{b},\mathrm{in})}
-
\Bigl( 
\frac{1}{2}a_1\mathsf{k}'(a_1)+1
\Bigr)
\log\frac{\Omega_{1}^{(7m)}}{\Omega_{m+1}^{(7m)}}, 
\\
\widetilde{C}_{E_{2}^{(\mathrm{an})}}^{(M)} 
&:=
\sqrt{n}M\tau_{a_1}\log\frac{\Omega_{1}^{(7m)}}{\Omega_{m+1}^{(7m)}}
-
\tau_{a_1}M^2
\log\frac{\Omega_{1}^{(7m)}}{\Omega_{m+1}^{(7m)}}
\\
&\quad
+ \bigg( \frac{1}{2}-\theta_{a_1,-}^{(n,M)} \bigg)\log\frac{\Omega_{1}^{(7m)}}{\Omega_{m+1}^{(7m)}}
+
(2\tau_{a_1}M\sqrt{n}-\theta_{a_1,+}^{(n,M)}-\theta_{a_1,-}^{(n,M)}+1)\log \Omega_{m+1}^{(7m)}.
\end{align*}  

\end{lemma}

\begin{proof}
On $g_{a_1,-}\le j\le g_{a_1,+}$, use the transition coordinate
$M_j(a_1)=\sqrt n(\tau_{a_1}/\tau_j-1)$ from
\eqref{def of Mja}.  The bulk-transition expansion proved in
Lemma~\ref{lemma:bulk case asymptotic expansion} applies term by term: the
second hard boundary is exponentially far from the saddle, so its only
non-negligible effect is to replace the disk-complement constants
$\Omega_q^{(3m)}$ by $\Omega_q^{(7m)}$.  Applying the two-sided
Euler--Maclaurin formula to the resulting functions of $M_j(a_1)$ produces
$C_4^{\#(\mathrm{b,in})}$ and $\mathcal D_6^{(\mathrm{b,in})}$.
Expanding the two integer endpoints $g_{a_1,\pm}$ gives exactly
$\widetilde C_{E_2^{(\mathrm{an})}}^{(M)}$.  The local remainder is uniform
throughout the window and, after summation with $M=n^{1/12}$, is
$\mathcal O(n^{-1/12})$.
\end{proof}

\begin{lemma}\label{lemma:asymptotic expansion E3 annulus}
There exists $c>0$ such that, as $n\to\infty$,
\[
E_{3}^{(\mathrm{an})}=(g_{1,-}-g_{a_1,+}-1)\log\Omega_{m+1}^{(7m)}+\mathcal{O}(e^{-cM^2}). 
\]
\end{lemma}

\begin{lemma}\label{lemma:asymptotic expansion E4 annulus}
As $n\to\infty$,
\[
 E_4^{(\mathrm{an})}
 =
 C_{E_4^{(\mathrm{an})}}^{(3)}n 
 +C_{E_4^{(\mathrm{an})}}^{(4)}\sqrt{n}
 +C_{E_4^{(\mathrm{an})}}^{(5)}\log n
 +C_{E_4^{(\mathrm{an})}}^{(6)}
 +\widetilde{C}_{E_4^{(\mathrm{an})}}^{(M)}
 +\mathcal{O}(n^{-\frac{1}{12}}),
\]   
where 
\begin{align*}
C_{E_4^{(\mathrm{an})}}^{(3)}&:=0,
\quad
C_{E_4^{(\mathrm{an})}}^{(4)}:=C_4^{\#(\mathrm{se,in})},
\quad
C_{E_4^{(\mathrm{an})}}^{(5)}:=0, 
\\
C_{E_4^{(\mathrm{an})}}^{(6)}&:=
\mathcal{D}_6^{(\mathrm{se},\mathrm{in})}
-\Bigl(\frac{1}{2}\rho_1\mathsf{k}'(\rho_1)+1\Bigr)\log\Bigl(\frac{\Omega_{m+1}^{(7m)}}{\Omega_{2m+1}^{(7m)}}\Bigr) 
\\
&\quad-
\frac{2\rho_1^2\Delta Q(\rho_1)\mathsf{T}_{1}^{(2m+1,3m)}(\tau_{\rho_1})\log \tau_{\rho_1}}{1+\mathsf{T}_{0}^{(2m+1,3m)}(\tau_{\rho_1})+\widehat{\mathsf{T}}_{0}^{(3m+1,7m)}(\tau_{\rho_2})},
\\
\widetilde{C}_{E_4^{(\mathrm{an})}}^{(M)}
&:=
\tau_{\rho_1}M\sqrt{n}
\log\frac{\Omega_{m+1}^{(7m)}}{\Omega_{2m+1}^{(7m)}}
+\Bigl(\frac{1}{2}-\theta_{1,-}^{(n,M)}\Bigr)\log\frac{\Omega_{m+1}^{(7m)}}{\Omega_{2m+1}^{(7m)}}
-\tau_{\rho_1}M^2\log\Bigl(\frac{\Omega_{m+1}^{(7m)}}{\Omega_{2m+1}^{(7m)}}\Bigr)
\\
&
+(2\tau_{\rho_1}M\sqrt{n}-\theta_{1,+}^{(n,M)}-\theta_{1,-}^{(n,M)}+1)\log\Omega_{2m+1}^{(7m)}
\\
&\quad-
\frac{\big(\tau_{\rho_1}^2M^2
+2\rho_1^2\Delta Q(\rho_1)\log M
\big)\mathsf{T}_{1}^{(2m+1,3m)}(\tau_{\rho_1})}{1+\mathsf{T}_{0}^{(2m+1,3m)}(\tau_{\rho_1})+\widehat{\mathsf{T}}_{0}^{(3m+1,7m)}(\tau_{\rho_2})}.
\end{align*}
\end{lemma}

\begin{proof}[Proofs of Lemmas~\ref{lemma:asymptotic expansion E3 annulus}
and~\ref{lemma:asymptotic expansion E4 annulus}]
For $g_{a_1,+}<j<g_{1,-}$ the saddle remains a fixed mesoscopic distance from
both adjacent transition windows.  The unwanted terms are therefore
$\mathcal O(e^{-cM^2})$ uniformly in $j$, and summation gives the formula for
$E_3^{(\mathrm{an})}$.

On $g_{1,-}\le j\le g_{1,+}$, set
$x=\tau_{\rho_1}M_j(\rho_1)/(\rho_1\sqrt{2\Delta Q(\rho_1)})$ and apply the critical
hard-wall expansion in part~(i) of
Lemma~\ref{lemma:semi hard edge case asymptotic expansion disck complement}.
The outer groups are already saturated on this scale; hence the annular
formula is obtained by adjoining the constant
$\widehat{\mathsf T}_0^{(3m+1,7m)}(\tau_{\rho_2})$ to the denominators in
that expansion.  Euler--Maclaurin in the variable $x$ yields the displayed
$n^{1/2}$ and constant terms.  Its lower- and upper-endpoint corrections are
the terms containing $\theta_{1,-}^{(n,M)}$ and
$\theta_{1,+}^{(n,M)}$, respectively.  The local remainder is uniform
throughout the window; with $M=n^{1/12}$ its sum is
$\mathcal O(n^{-1/12})$.  This proves both assertions.
\end{proof}

\begin{lemma}\label{lemma:asymptotic expansion E5 annulus}
As $n\to\infty$,
\[
    E_{5}^{(\mathrm{an})}
    =
    C_{ E_5^{(\mathrm{an})}}^{(3)} n
    +
    C_{ E_5^{(\mathrm{an})}}^{(4)} \sqrt{n}
    +
    C_{ E_5^{(\mathrm{an})}}^{(5)} \log n
    +
    C_{ E_5^{(\mathrm{an})}}^{(6)}
    +
    \widetilde{C}_{ E_5^{(\mathrm{an})}}^{(M)} 
    +
    \mathcal{O}(M^3n^{-1/2}), 
\]  
where 
\begin{align*}
C_{ E_5^{(\mathrm{an})}}^{(3)} 
&:=
\int_{\tau_{\rho_1}}^{\frac{\tau_{\rho_1}}{1-\epsilon}}
\log\bigl(1+\mathsf{T}_{0}^{(2m+1,3m)}(x)+\widehat{\mathsf{T}}_0^{(3m+1,7m)}(\tau_{\rho_2})\bigr)\,dx,
\quad
C_{ E_5^{(\mathrm{an})}}^{(4)} 
:=0,
\\
C_{ E_5^{(\mathrm{an})}}^{(5)} 
&:=
-
\rho_1^2\Delta Q(\rho_1)
\frac{\mathsf{T}_{1}^{(2m+1,3m)}(\tau_{\rho_1})}{1+\mathsf{T}_{0}^{(2m+1,3m)}(\tau_{\rho_1})+\widehat{\mathsf{T}}_{0}^{(3m+1,7m)}(\tau_{\rho_2})},
\\
C_{ E_5^{(\mathrm{an})}}^{(6)} 
&:=
\int_{\tau_{\rho_1}}^{\frac{\tau_{\rho_1}}{1-\epsilon}}
\bigg[-\frac{\frac{2\rho_1^2\Delta Q(\rho_1)}{x-\tau_{\rho_1}}\mathsf{T}_{1}^{(2m+1,3m)}(x)+x\mathsf{T}_{2}^{(2m+1,3m)}(x)}{1+\mathsf{T}_{0}^{(2m+1,3m)}(x)+\widehat{\mathsf{T}}_{0}^{(3m+1,7m)}(\tau_{\rho_2})}
\\
&\quad 
+\frac{2\rho_1^2\Delta Q(\rho_1)}{x-\tau_{\rho_1}}\frac{\mathsf{T}_{1}^{(2m+1,3m)}(\tau_{\rho_1})}{1+\mathsf{T}_{0}^{(2m+1,3m)}(\tau_{\rho_1})+\widehat{\mathsf{T}}_{0}^{(3m+1,7m)}(\tau_{\rho_2})}\bigg]\,dx
\\
&\quad
-\bigl(2\rho_1^2\Delta Q(\rho_1)-\tau_{\rho_1}\bigr)
\int_{\tau_{\rho_1}}^{\frac{\tau_{\rho_1}}{1-\epsilon}}
\frac{\mathsf{T}_{2}^{(2m+1,3m)}(x)}{1+\mathsf{T}_{0}^{(2m+1,3m)}(x)+\widehat{\mathsf{T}}_{0}^{(3m+1,7m)}(\tau_{\rho_2})}\,dx
\\
&\quad
+\Bigl(\frac{1}{2}\rho_1\mathsf{k}'(\rho_1)+1\Bigr)
\log(1+\mathsf{T}_0^{(2m+1,3m)}(\tfrac{\tau_{\rho_1}}{1-\epsilon})+\widehat{\mathsf{T}}_{0}^{(3m+1,7m)}(\tau_{\rho_2}))
\\
&\quad
-\Bigl(\frac{1}{2}\rho_1\mathsf{k}'(\rho_1)+1\Bigr)
\log(1+\mathsf{T}_0^{(2m+1,3m)}(\tau_{\rho_1})+\widehat{\mathsf{T}}_{0}^{(3m+1,7m)}(\tau_{\rho_2}))
\\
&\quad
+\Bigl(\frac{1}{2}-\theta_{1,+}^{(n,\epsilon)}\Bigr)\log\bigl(1+\mathsf{T}_{0}^{(2m+1,3m)}(\tfrac{\tau_{\rho_1}}{1-\epsilon})+\widehat{\mathsf{T}}_0^{(3m+1,7m)}(\tau_{\rho_2})\bigr)
\\
&\quad 
-
2\rho_1^2\Delta Q(\rho_1)
\frac{\mathsf{T}_{1}^{(2m+1,3m)}(\tau_{\rho_1})}{1+\mathsf{T}_{0}^{(2m+1,3m)}(\tau_{\rho_1})+\widehat{\mathsf{T}}_{0}^{(3m+1,7m)}(\tau_{\rho_2})}
\log(\tfrac{\epsilon}{1-\epsilon}),
\\
\widetilde{C}_{ E_5^{(\mathrm{an})}}^{(M)} 
&:=
2\rho_1^2\Delta Q(\rho_1)
\frac{\mathsf{T}_{1}^{(2m+1,3m)}(\tau_{\rho_1})}{1+\mathsf{T}_{0}^{(2m+1,3m)}(\tau_{\rho_1})+\widehat{\mathsf{T}}_{0}^{(3m+1,7m)}(\tau_{\rho_2})}
\log M
\\
&\quad
-\tau_{\rho_1}M\sqrt{n}\log\bigl(1+\mathsf{T}_{0}^{(2m+1,3m)}(\tau_{\rho_1})+\widehat{\mathsf{T}}_0^{(3m+1,7m)}(\tau_{\rho_2})\bigr)
\\
&\quad
-M^2\tau_{\rho_1}
\left(\begin{aligned}&
\log\bigl(1+\mathsf{T}_{0}^{(2m+1,3m)}(\tau_{\rho_1})+\widehat{\mathsf{T}}_0^{(3m+1,7m)}(\tau_{\rho_2})\bigr)
\\
&\quad-\frac{\tau_{\rho_1}\mathsf{T}_{1}^{(2m+1,3m)}(\tau_{\rho_1})}{1+\mathsf{T}_{0}^{(2m+1,3m)}(\tau_{\rho_1})+\widehat{\mathsf{T}}_0^{(3m+1,7m)}(\tau_{\rho_2})}
\end{aligned}\right)
\\
&\quad
+
\Bigl(\theta_{1,+}^{(n,M)}-\frac{1}{2}\Bigr)
\log\bigl(1+\mathsf{T}_{0}^{(2m+1,3m)}(\tau_{\rho_1})+\widehat{\mathsf{T}}_0^{(3m+1,7m)}(\tau_{\rho_2})\bigr). 
\end{align*}
\end{lemma}

\begin{proof}
This is a one-sided variant of the proof of
Lemma~\ref{lemma:semi hard edge case asymptotic expansion disck complement}.
The relevant index range is $g_{1,+}+1\le j\le j_{1,+}$, where
\[
g_{1,+}+1=\frac{n\tau_{\rho_1}}{1-M/\sqrt n}
+1-\theta_{1,+}^{(n,M)},\qquad
j_{1,+}=\frac{n\tau_{\rho_1}}{1-\epsilon}
-\theta_{1,+}^{(n,\epsilon)}.
\]
Set
\[
A=\frac{\tau_{\rho_1}}{1-M/\sqrt n},\quad
a_0=1-\theta_{1,+}^{(n,M)},\quad
B=\frac{\tau_{\rho_1}}{1-\epsilon},\quad
b_0=-\theta_{1,+}^{(n,\epsilon)}.
\]
The uniform hard-edge expansion of the summand, followed by the Taylor
expansion of the logarithm, gives
\begin{align*}
&\quad E_5^{(\mathrm{an})}\\
&=  
\sum_{j=An+a_0}^{Bn+b_0}
\log\bigl(1+\mathsf{T}_{0}^{(2m+1,3m)}(j/n)+\widehat{\mathsf{T}}_0^{(3m+1,7m)}(\tau_{\rho_2})\bigr)
\\
&\quad-\frac{1}{n}
\sum_{j=An+a_0}^{Bn+b_0}
\frac{\mathscr{H}_{\rho_1,1}}{1+\mathsf{T}_{0}^{(2m+1,3m)}(j/n)+\widehat{\mathsf{T}}_0^{(3m+1,7m)}(\tau_{\rho_2})}
\\
&
+\frac{1}{n^2}\sum_{j=An+a_0}^{Bn+b_0}\Bigl( 
\frac{\mathscr{H}_{\rho_1,2}}{1+\mathsf{T}_{0}^{(2m+1,3m)}(j/n)+\widehat{\mathsf{T}}_0^{(3m+1,7m)}(\tau_{\rho_2})}\\
&\qquad-\frac{\mathscr{H}_{\rho_1,1}^2}{2(1+\mathsf{T}_{0}^{(2m+1,3m)}(j/n)+\widehat{\mathsf{T}}_0^{(3m+1,7m)}(\tau_{\rho_2}))^2}
\Bigr)
\\
&
-\frac{1}{n^3}
\sum_{j=An+a_0}^{Bn+b_0}
\Bigl( 
\frac{\mathscr{H}_{\rho_1,3}}{1+\mathsf{T}_{0}^{(2m+1,3m)}(j/n)+\widehat{\mathsf{T}}_0^{(3m+1,7m)}(\tau_{\rho_2})}
\\
&\qquad-\frac{\mathscr{H}_{\rho_1,1}\mathscr{H}_{\rho_1,2}}{(1+\mathsf{T}_{0}^{(2m+1,3m)}(j/n)+\widehat{\mathsf{T}}_0^{(3m+1,7m)}(\tau_{\rho_2}))^2}
\\
&\quad
+\frac{\mathscr{H}_{\rho_1,1}^3}{(1+\mathsf{T}_{0}^{(2m+1,3m)}(j/n)+\widehat{\mathsf{T}}_0^{(3m+1,7m)}(\tau_{\rho_2}))^3}
\Bigr)
+\mathcal{O}(M^{-8}).
\end{align*}
The error is uniform over the full index range; the denominator stays bounded
away from zero under the standing admissibility assumptions.  For the leading
sum, Euler--Maclaurin together with expansion of the moving lower endpoint
gives
\begin{align*}
&\quad\sum_{j=An+a_0}^{Bn+b_0}
\log\bigl(1+\mathsf{T}_{0}^{(2m+1,3m)}(j/n)+\widehat{\mathsf{T}}_0^{(3m+1,7m)}(\tau_{\rho_2})\bigr)
\\
&=
n\int_{\tau_{\rho_1}}^{\frac{\tau_{\rho_1}}{1-\epsilon}}
\log\bigl(1+\mathsf{T}_{0}^{(2m+1,3m)}(x)+\widehat{\mathsf{T}}_0(\tau_{\rho_2})\bigr)\,dx
\\
&\quad-\tau_{\rho_1}M\sqrt{n}\log\bigl(1+\mathsf{T}_{0}^{(2m+1,3m)}(\tau_{\rho_1})+\widehat{\mathsf{T}}_0^{(3m+1,7m)}(\tau_{\rho_2})\bigr)
\\
&\quad
-M^2\tau_{\rho_1}
\left(\begin{aligned}&
\log\bigl(1+\mathsf{T}_{0}^{(2m+1,3m)}(\tau_{\rho_1})+\widehat{\mathsf{T}}_0^{(3m+1,7m)}(\tau_{\rho_2})\bigr)
\\
&\quad-\frac{\tau_{\rho_1}\mathsf{T}_{1}^{(2m+1,3m)}(\tau_{\rho_1})}{1+\mathsf{T}_{0}^{(2m+1,3m)}(\tau_{\rho_1})+\widehat{\mathsf{T}}_0^{(3m+1,7m)}(\tau_{\rho_2})}
\end{aligned}\right)
\\
&\quad
+
\Bigl(\theta_{1,+}^{(n,M)}-\frac{1}{2}\Bigr)
\log\bigl(1+\mathsf{T}_{0}^{(2m+1,3m)}(\tau_{\rho_1})+\widehat{\mathsf{T}}_0^{(3m+1,7m)}(\tau_{\rho_2})\bigr)
\\
&\quad
+\Bigl(\frac{1}{2}-\theta_{1,+}^{(n,\epsilon)}\Bigr)\log\bigl(1+\mathsf{T}_{0}^{(2m+1,3m)}(\tfrac{\tau_{\rho_1}}{1-\epsilon})+\widehat{\mathsf{T}}_0^{(3m+1,7m)}(\tau_{\rho_2})\bigr)
+
\mathcal{O}(Mn^{-1/2})
\end{align*}
For brevity, write
\begin{align*}
\mathscr{F}_1(x)&:=
\frac{2\rho_1^2\Delta Q(\rho_1)}{x-\tau_{\rho_1}}
\frac{\mathsf{T}_{1}^{(2m+1,3m)}(x)}{1+\mathsf{T}_{0}^{(2m+1,3m)}(x)+\widehat{\mathsf{T}}_{0}^{(3m+1,7m)}(\tau_{\rho_2})}
\\
&\quad+
x\frac{\mathsf{T}_{2}^{(2m+1,3m)}(x)}{1+\mathsf{T}_{0}^{(2m+1,3m)}(x)+\widehat{\mathsf{T}}_{0}^{(3m+1,7m)}(\tau_{\rho_2})}, 
\\
\mathscr{F}_2(x)&:=
\frac{\mathsf{T}_{1}^{(2m+1,3m)}(x)}{1+\mathsf{T}_{0}^{(2m+1,3m)}(x)+\widehat{\mathsf{T}}_{0}^{(3m+1,7m)}(\tau_{\rho_2})}, 
\\
\mathscr{F}_3(x)&:=
\frac{\mathsf{T}_{2}^{(2m+1,3m)}(x)}{1+\mathsf{T}_{0}^{(2m+1,3m)}(x)+\widehat{\mathsf{T}}_{0}^{(3m+1,7m)}(\tau_{\rho_2})}. 
\end{align*}
Applying Lemma~\ref{lemma:Riemann sum NEW} to the first correction term yields
\begin{align*}
&\quad \frac{1}{n}\sum_{j=An+a_0}^{Bn+b_0}
\frac{\mathscr{H}_{\rho_1,1}}{1+\mathsf{T}_{0}^{(2m+1,3m)}(j/n)+\widehat{\mathsf{T}}_{0}^{(3m+1,7m)}(\tau_{\rho_2})}
\\
&=
\int_{\frac{\tau_{\rho_1}}{1-\frac{M}{\sqrt{n}}}}^{\frac{\tau_{\rho_1}}{1-\epsilon}}
\mathscr{F}_1(x)\,dx
+
\Bigl(\frac{1}{2}\rho_1\mathsf{k}'(\rho_1)+1\Bigr)
\int_{\frac{\tau_{\rho_1}}{1-\frac{M}{\sqrt{n}}}}^{\frac{\tau_{\rho_1}}{1-\epsilon}}
2\mathscr{F}_2(x)
\,dx
\\
&\quad
+
\bigl(2\rho_1^2\Delta Q(\rho_1)-\tau_{\rho_1}\bigr)
\int_{\frac{\tau_{\rho_1}}{1-\frac{M}{\sqrt{n}}}}^{\frac{\tau_{\rho_1}}{1-\epsilon}}
\mathscr{F}_3(x)
\,dx
+\mathcal{O}(M^{-1}n^{-1/2}). 
\end{align*}
We isolate the singular contribution at the lower endpoint by subtracting
its limiting coefficient.  This gives
\begin{align*}
\int_{\frac{\tau_{\rho_1}}{1-\frac{M}{\sqrt{n}}}}^{\frac{\tau_{\rho_1}}{1-\epsilon}}
\mathscr{F}_1(x)\,dx
&=
\int_{\tau_{\rho_1}}^{\frac{\tau_{\rho_1}}{1-\epsilon}}
\Bigl(\mathscr{F}_1(x)
-\frac{2\rho_1^2\Delta Q(\rho_1)}{x-\tau_{\rho_1}}\frac{\mathsf{T}_{1}^{(2m+1,3m)}(\tau_{\rho_1})}{1+\mathsf{T}_{0}^{(2m+1,3m)}(\tau_{\rho_1})+\widehat{\mathsf{T}}_{0}^{(3m+1,7m)}(\tau_{\rho_2})}\Bigr)\,dx
\\
&\quad 
+
\frac{\rho_1^2\Delta Q(\rho_1)\,\mathsf{T}_{1}^{(2m+1,3m)}(\tau_{\rho_1})}{1+\mathsf{T}_{0}^{(2m+1,3m)}(\tau_{\rho_1})+\widehat{\mathsf{T}}_{0}^{(3m+1,7m)}(\tau_{\rho_2})}
\log n
\\
&\quad
+
\frac{2\rho_1^2\Delta Q(\rho_1)\,\mathsf{T}_{1}^{(2m+1,3m)}(\tau_{\rho_1})}{1+\mathsf{T}_{0}^{(2m+1,3m)}(\tau_{\rho_1})+\widehat{\mathsf{T}}_{0}^{(3m+1,7m)}(\tau_{\rho_2})}
\Bigl( 
\log(\tfrac{\epsilon}{1-\epsilon})-\log M
\Bigr)
+
\mathcal{O}(Mn^{-1/2}), 
\\
\int_{\frac{\tau_{\rho_1}}{1-\frac{M}{\sqrt{n}}}}^{\frac{\tau_{\rho_1}}{1-\epsilon}}
\mathscr{F}_2(x)
\,dx
&
=
-\frac{1}{2}\log(1+\mathsf{T}_0^{(2m+1,3m)}(\tfrac{\tau_{\rho_1}}{1-\epsilon})+\widehat{\mathsf{T}}_{0}^{(3m+1,7m)}(\tau_{\rho_2}))
\\
&\quad
+\frac{1}{2}\log(1+\mathsf{T}_0^{(2m+1,3m)}(\tau_{\rho_1})+\widehat{\mathsf{T}}_{0}^{(3m+1,7m)}(\tau_{\rho_2}))
+\mathcal{O}(Mn^{-1/2}), 
\\
\int_{\frac{\tau_{\rho_1}}{1-\frac{M}{\sqrt{n}}}}^{\frac{\tau_{\rho_1}}{1-\epsilon}}
\mathscr{F}_3(x)\,dx
&=
\int_{\tau_{\rho_1}}^{\frac{\tau_{\rho_1}}{1-\epsilon}}
\frac{\mathsf{T}_{2}^{(2m+1,3m)}(x)}{1+\mathsf{T}_{0}^{(2m+1,3m)}(x)+\widehat{\mathsf{T}}_{0}^{(3m+1,7m)}(\tau_{\rho_2})}\,dx
+
\mathcal{O}(Mn^{-1/2}). 
\end{align*}
For the remaining two correction sums, we use the tail estimate from
\cite[proof of Lemma~2.6]{ACCL1}: for $A,B>1$,
\[
\sum_{j=g_{1,+}+1}^{n-1}\mathcal{O}\Bigl(\frac{1}{n^A(j/n-\tau_{\rho_1})^B}\Bigr)
=
\mathcal{O}\Bigl(\frac{1}{n^{A-(B+1)/2}M^{B-1}}\Bigr),
\]
and hence
\begin{align*}
&\quad \frac{1}{n^2}\sum_{j=An+a_0}^{Bn+b_0}\Bigl( 
\frac{\mathscr{H}_{\rho_1,2}}{1+\mathsf{T}_{0}^{(2m+1,3m)}(j/n)+\widehat{\mathsf{T}}_{0}^{(3m+1,7m)}(\tau_{\rho_2})}\\
&\qquad-\frac{\mathscr{H}_{\rho_1,1}^2}{2(1+\mathsf{T}_{0}^{(2m+1,3m)}(j/n)+\widehat{\mathsf{T}}_{0}^{(3m+1,7m)}(\tau_{\rho_2}))^2}
\Bigr)
\\
&=
\frac{2\rho_1^4(\Delta Q(\rho_1))^2}{\tau_{\rho_1}^2M^2}
\frac{\mathsf{T}_{1}^{(2m+1,3m)}(\tau_{\rho_1})}{1+\mathsf{T}_{0}^{(2m+1,3m)}(\tau_{\rho_1})+\widehat{\mathsf{T}}_{0}^{(3m+1,7m)}(\tau_{\rho_2})}
+\mathcal{O}(M^{-1}n^{-1/2}),
\end{align*}
and 
\begin{align*}
 &\quad 
 -\frac{1}{n^3}
\sum_{j=An+a_0}^{Bn+b_0}
\Bigl( 
\frac{\mathscr{H}_{\rho_1,3}}{1+\mathsf{T}_{0}^{(2m+1,3m)}(j/n)+\widehat{\mathsf{T}}_{0}^{(3m+1,7m)}(\tau_{\rho_2})}
\\
&\qquad-\frac{\mathscr{H}_{\rho_1,1}\mathscr{H}_{\rho_1,2}}{(1+\mathsf{T}_{0}^{(2m+1,3m)}(j/n)+\widehat{\mathsf{T}}_{0}^{(3m+1,7m)}(\tau_{\rho_2}))^2}
\\
&\quad
+\frac{\mathscr{H}_{\rho_1,1}^3}{(1+\mathsf{T}_{0}^{(2m+1,3m)}(j/n)+\widehat{\mathsf{T}}_{0}^{(3m+1,7m)}(\tau_{\rho_2}))^3}
\Bigr)
\\
&=
-
\frac{5\rho_1^6\Delta Q(\rho_1)^3}{\tau_{\rho_1}^4M^4}
\frac{\mathsf{T}_{1}^{(2m+1,3m)}(\tau_{\rho_1})}{1+\mathsf{T}_{0}^{(2m+1,3m)}(\tau_{\rho_1})+\widehat{\mathsf{T}}_{0}^{(3m+1,7m)}(\tau_{\rho_2})}
+\mathcal{O}(M^{-3}n^{-1/2}). 
\end{align*}
Combining the preceding Euler--Maclaurin formula, the regularized endpoint
integral, and these two tail estimates gives the coefficients in the
statement.  The logarithmic endpoint singularity produces the $\log n$ and
$\log M$ terms, and the remaining powers of $M$ match those from the adjacent
hard-edge window.  With $M=n^{1/12}$, all omitted contributions lie within the
stated remainder.  This completes the proof.
\end{proof}

\begin{lemma}\label{lemma:asymptotic expansion E6 annulus}
As $n\to\infty$,
 \begin{align}
    E_{6}^{(\mathrm{an})}
    =
C_{E_{6}^{(\mathrm{an})}}^{(3)}n
+
C_{E_{6}^{(\mathrm{an})}}^{(4)}\sqrt{n}
+
C_{E_{6}^{(\mathrm{an})}}^{(5)}\log n
+
C_{E_{6}^{(\mathrm{an})}}^{(6)}
+
C_{E_{6}^{(\mathrm{an})}}^{(\mathrm{osc})}
+\mathcal{O}\Bigl(\frac{(\log n)^2}{n}\Bigr),    
\end{align}   
where 
\begin{align*}
C_{E_{6}^{(\mathrm{an})}}^{(3)}
&:=
\int_{\frac{\tau_{\rho_1}}{1-\epsilon}}^{\sigma_{\star}}
\log\Bigl(1+\mathsf{T}_{0}^{(2m+1,3m)}(x)+\widehat{\mathsf{T}}_0^{(3m+1,7m)}(\tau_{\rho_2})\Bigr)\,dx
\\
&\quad
+
\int_{\sigma_{\star}}^{\frac{\tau_{\rho_2}}{1+\epsilon}}
\log\Bigl(1-\widehat{\mathsf{T}}_0^{(3m+1,4m)}(x)+\widehat{\mathsf{T}}_0^{(3m+1,7m)}(\tau_{\rho_2})\Bigr)\,dx,
\quad
C_{E_{6}^{(\mathrm{an})}}^{(4)}
:=0,
\quad
C_{E_{6}^{(\mathrm{an})}}^{(5)}
:=0,
\\
C_{E_{6}^{(\mathrm{an})}}^{(6)}
&:=
\Bigl(\theta_{1,+}^{(n,\epsilon)}-\frac{1}{2}\Bigr)
\log\Bigl(1+\mathsf{T}_{0}^{(2m+1,3m)}(\tfrac{\tau_{\rho_1}}{1-\epsilon})+\widehat{\mathsf{T}}_0^{(3m+1,7m)}(\tau_{\rho_2})\Bigr)
\\
&\quad 
+\Bigl(\frac{1}{2}-\theta_{\star}\Bigr)
\log\Bigl(1+\mathsf{T}_{0}^{(2m+1,3m)}(\sigma_{\star})+\widehat{\mathsf{T}}_0^{(3m+1,7m)}(\tau_{\rho_2})\Bigr)
\\
&\quad 
+
\int_{\frac{\tau_{\rho_1}}{1-\epsilon}}^{\sigma_{\star}}
\bigg[
-\frac{\frac{2\rho_1^2\Delta Q(\rho_1)}{x-\tau_{\rho_1}}
\mathsf{T}_{1}^{(2m+1,3m)}(x)
+x\mathsf{T}_{2}^{(2m+1,3m)}(x)
}{
1+\mathsf{T}_{0}^{(2m+1,3m)}(x)+\widehat{\mathsf{T}}_0^{(3m+1,7m)}(\tau_{\rho_2})}
\\
&\quad 
+
\frac{2\rho_1^2\Delta Q(\rho_1)}{x-\tau_{\rho_1}}
\frac{\mathsf{T}_{1}^{(2m+1,3m)}(\tau_{\rho_1})}{1+\mathsf{T}_{0}^{(2m+1,3m)}(\tau_{\rho_1})+\widehat{\mathsf{T}}_0^{(3m+1,7m)}(\tau_{\rho_2})}
\bigg]
\,dx
\\
&\quad 
-
\frac{2\rho_1^2\Delta Q(\rho_1)\mathsf{T}_{1}^{(2m+1,3m)}(\tau_{\rho_1})}{1+\mathsf{T}_{0}^{(2m+1,3m)}(\tau_{\rho_1})+\widehat{\mathsf{T}}_0^{(3m+1,7m)}(\tau_{\rho_2})}
\Bigl( 
\log\bigl(\sigma_{\star}-\tau_{\rho_1}\bigr)-\log\bigl(\frac{\epsilon}{1-\epsilon}\bigr)
-\log \tau_{\rho_1}
\Bigr)
\\
&\quad 
+
\Bigl(\frac{1}{2}\rho_1\mathsf{k}'(\rho_1)+1\Bigr)
\log\Bigl(1+\mathsf{T}_{0}^{(2m+1,3m)}(\sigma_{\star})+\widehat{\mathsf{T}}_0^{(3m+1,7m)}(\tau_{\rho_2})\Bigr)
\\
&\quad
-
\Bigl(\frac{1}{2}\rho_1\mathsf{k}'(\rho_1)+1\Bigr)
\log\Bigl(1+\mathsf{T}_{0}^{(2m+1,3m)}(\tfrac{\tau_{\rho_1}}{1-\epsilon})+\widehat{\mathsf{T}}_0^{(3m+1,7m)}(\tau_{\rho_2})\Bigr)
\\
&\quad 
-
\bigl(2\rho_1^2\Delta Q(\rho_1)-\tau_{\rho_1}\bigr)
\int_{\frac{\tau_{\rho_1}}{1-\epsilon}}^{\sigma_{\star}}
\frac{
\mathsf{T}_{2}^{(2m+1,3m)}(x)
}{
1+\mathsf{T}_{0}^{(2m+1,3m)}(x)+\widehat{\mathsf{T}}_0^{(3m+1,7m)}(\tau_{\rho_2})}\,dx
\\
&\quad
+
\Bigl(\theta_{\star}-\frac{1}{2}\Bigr)
\log\Bigl(1-\widehat{\mathsf{T}}_0^{(3m+1,4m)}(\sigma_{\star})+\widehat{\mathsf{T}}_0^{(3m+1,7m)}(\tau_{\rho_2})\Bigr)
\\
&\quad 
+\Bigl(\theta_{2,-}^{(n,\epsilon)}-\frac{1}{2}\Bigr)
\log\Bigl(1-\widehat{\mathsf{T}}_0^{(3m+1,4m)}(\tfrac{\tau_{\rho_2}}{1+\epsilon})+\widehat{\mathsf{T}}_0^{(3m+1,7m)}(\tau_{\rho_2})\Bigr)
\\
&\quad 
+\int_{\sigma_{\star}}^{\frac{\tau_{\rho_2}}{1+\epsilon}}
\bigg[
\frac{\frac{2\rho_2^2\Delta Q(\rho_2)}{\tau_{\rho_2}-x}
\widehat{\mathsf{T}}_{1}^{(3m+1,4m)}(x)
+x\widehat{\mathsf{T}}_{2}^{(3m+1,4m)}(x)}{1-\widehat{\mathsf{T}}_0^{(3m+1,4m)}(x)+\widehat{\mathsf{T}}_0^{(3m+1,7m)}(\tau_{\rho_2})}
\\
&\quad 
-
\frac{2\rho_2^2\Delta Q(\rho_2)}{\tau_{\rho_2}-x}
\frac{
\widehat{\mathsf{T}}_{1}^{(3m+1,4m)}(\tau_{\rho_2})
}{1-\widehat{\mathsf{T}}_0^{(3m+1,4m)}(\tau_{\rho_2})+\widehat{\mathsf{T}}_0^{(3m+1,7m)}(\tau_{\rho_2})}
\bigg]\,dx
\\
&\quad
+
\frac{2\rho_2^2\Delta Q(\rho_2)
\widehat{\mathsf{T}}_{1}^{(3m+1,4m)}(\tau_{\rho_2})
}{1-\widehat{\mathsf{T}}_0^{(3m+1,4m)}(\tau_{\rho_2})+\widehat{\mathsf{T}}_0^{(3m+1,7m)}(\tau_{\rho_2})}
\Bigl( 
-\log\bigl(\frac{\epsilon}{1+\epsilon}\bigr)
-\log\tau_{\rho_2}
+\log\bigl(\tau_{\rho_2}-\sigma_{\star}\bigr)
\Bigr)
\\
&\quad 
+\Bigl(\frac{1}{2}\rho_2\mathsf{k}'(\rho_2)+1\Bigr)
\log\Bigl(1-\widehat{\mathsf{T}}_0^{(3m+1,4m)}(\tfrac{\tau_{\rho_2}}{1+\epsilon})+\widehat{\mathsf{T}}_0^{(3m+1,7m)}(\tau_{\rho_2})\Bigr)
\\
&\quad
-\Bigl(\frac{1}{2}\rho_2\mathsf{k}'(\rho_2)+1\Bigr)\log\Bigl(1-\widehat{\mathsf{T}}_0^{(3m+1,4m)}(\sigma_{\star})+\widehat{\mathsf{T}}_0^{(3m+1,7m)}(\tau_{\rho_2})\Bigr)
\\
&\quad 
+\bigl( 
2\rho_2^2\Delta Q(\rho_2)-\tau_{\rho_2}
\bigr)
\int_{\sigma_{\star}}^{\frac{\tau_{\rho_2}}{1+\epsilon}}
\frac{\widehat{\mathsf{T}}_{2}^{(3m+1,4m)}(x)}{1-\widehat{\mathsf{T}}_0^{(3m+1,4m)}(x)+\widehat{\mathsf{T}}_0^{(3m+1,7m)}(\tau_{\rho_2})}\,dx, 
\\
C_{E_{6}^{(\mathrm{an})}}^{(\mathrm{osc})}    
&:=
(\theta_{\star}-\alpha-1)\log\mathsf{Q}-\frac{\log \mathsf{Q}}{2}\Bigl(1+\frac{s(\lambda(\rho_1)-\lambda(\rho_2))}{\log \frac{\rho_1}{\rho_2}}+\frac{2\log (\tfrac{\sigma_2}{\sigma_1})+\log \mathsf{Q}}{2\log \frac{\rho_1}{\rho_2}}\Bigr)
    \\
    &\quad 
    +\log \frac{\theta( \tfrac{\log (\tfrac{\sigma_2}{\sigma_1})+\log \mathsf{Q}}{2\log \tfrac{\rho_2}{\rho_1}}+\frac{s(\lambda(\rho_1)-\lambda(\rho_2))}{2\log\frac{\rho_2}{\rho_1}}+n\sigma_{\star}-\alpha +\frac{1}{2}; \frac{\pi i }{\log \frac{\rho_2}{\rho_1}})}{\theta( \tfrac{\log (\tfrac{\sigma_2}{\sigma_1})}{2\log \tfrac{\rho_2}{\rho_1}}+\frac{s(\lambda(\rho_1)-\lambda(\rho_2))}{2\log\frac{\rho_2}{\rho_1}}+n\sigma_{\star}-\alpha+\frac{1}{2}; \frac{\pi i }{\log \frac{\rho_2}{\rho_1}})}.
\end{align*}
\end{lemma}

The central annular range contains the point at which the two boundary
contributions have equal exponential weight.  We therefore split it at
$j_\star=n\sigma_\star$:
$E_{6}^{(\mathrm{an})}=E_{6}^{(1,\mathrm{an})}+E_{6}^{(2,\mathrm{an})}$, where
\[
E_{6}^{(1,\mathrm{an})}=\sum_{j=j_{1,+}+1}^{\lfloor j_{\star}\rfloor}\log\Big( 
1+\sum_{\ell=1}^{7m}\omega_{\ell}F_{n,j,\ell}^{(\mathrm{an})}
\Bigr) ,\qquad   
E_{6}^{(2,\mathrm{an})}=\sum_{j=\lfloor j_{\star}\rfloor+1}^{j_{2,-}-1}\log\Big( 
1+\sum_{\ell=1}^{7m}\omega_{\ell}F_{n,j,\ell}^{(\mathrm{an})}
\Bigr),
\]
and
\[
j_{\star}:=n\sigma_{\star},\qquad
\theta_{\star}:=j_{\star}-\lfloor j_{\star}\rfloor\in[0,1).
\]
This convention assigns the integer $\lfloor j_\star\rfloor$ to the left
range and makes the two endpoint corrections complementary.

\begin{lemma}
\label{lemma:asymptotics E6(1)(2) counting osc}
As $n\to\infty$,
\[
    E_{6}^{(1,\mathrm{an})}
   =
    C_{E_{6}^{(1,\mathrm{an})}}^{(3)}n 
    +C_{E_{6}^{(1,\mathrm{an})}}^{(4)}\sqrt{n}
    +C_{E_{6}^{(1,\mathrm{an})}}^{(5)}\log n
    +C_{E_{6}^{(1,\mathrm{an})}}^{(6)}
    +\widehat{C}_{E_{6}^{(1,\mathrm{an})}}^{(6)}
    +\mathcal{O}\Bigl(\frac{(\log n)^2}{n}\Bigr), 
\] 
where 
\begin{align*}
C_{E_{6}^{(1,\mathrm{an})}}^{(3)}
&:=
\int_{\frac{\tau_{\rho_1}}{1-\epsilon}}^{\sigma_{\star}}
\log\Bigl(1+\mathsf{T}_{0}^{(2m+1,3m)}(x)+\widehat{\mathsf{T}}_0^{(3m+1,7m)}(\tau_{\rho_2})\Bigr)\,dx,
\quad
C_{E_{6}^{(1,\mathrm{an})}}^{(4)}
:=0,
\quad
C_{E_{6}^{(1,\mathrm{an})}}^{(5)}
:=0,
\\
C_{E_{6}^{(1,\mathrm{an})}}^{(6)}   
&:=
\Bigl(\theta_{1,+}^{(n,\epsilon)}-\frac{1}{2}\Bigr)
\log\Bigl(1+\mathsf{T}_{0}^{(2m+1,3m)}(\tfrac{\tau_{\rho_1}}{1-\epsilon})+\widehat{\mathsf{T}}_0^{(3m+1,7m)}(\tau_{\rho_2})\Bigr)
\\
&\quad 
+\Bigl(\frac{1}{2}-\theta_{\star}\Bigr)
\log\Bigl(1+\mathsf{T}_{0}^{(2m+1,3m)}(\sigma_{\star})+\widehat{\mathsf{T}}_0^{(3m+1,7m)}(\tau_{\rho_2})\Bigr)
\\
&\quad 
-
\int_{\frac{\tau_{\rho_1}}{1-\epsilon}}^{\sigma_{\star}}
\\
&\qquad\times\frac{\frac{2\rho_1^2\Delta Q(\rho_1)}{x-\tau_{\rho_1}}
\mathsf{T}_{1}^{(2m+1,3m)}(x)
+
x
\mathsf{T}_{2}^{(2m+1,3m)}(x)
}{
1+\mathsf{T}_{0}^{(2m+1,3m)}(x)+\widehat{\mathsf{T}}_0^{(3m+1,7m)}(\tau_{\rho_2})}\,dx
\\
&\quad 
+
\Bigl(\frac{1}{2}\rho_1\mathsf{k}'(\rho_1)+1\Bigr)
\log\Bigl(1+\mathsf{T}_{0}^{(2m+1,3m)}(\sigma_{\star})+\widehat{\mathsf{T}}_0^{(3m+1,7m)}(\tau_{\rho_2})\Bigr)
\\
&\quad
-
\Bigl(\frac{1}{2}\rho_1\mathsf{k}'(\rho_1)+1\Bigr)
\log\Bigl(1+\mathsf{T}_{0}^{(2m+1,3m)}(\tfrac{\tau_{\rho_1}}{1-\epsilon})+\widehat{\mathsf{T}}_0^{(3m+1,7m)}(\tau_{\rho_2})\Bigr)
\\
&\quad 
-
\bigl(2\rho_1^2\Delta Q(\rho_1)-\tau_{\rho_1}\bigr)
\int_{\frac{\tau_{\rho_1}}{1-\epsilon}}^{\sigma_{\star}}
\frac{
\mathsf{T}_{2}^{(2m+1,3m)}(x)
}{
1+\mathsf{T}_{0}^{(2m+1,3m)}(x)+\widehat{\mathsf{T}}_0^{(3m+1,7m)}(\tau_{\rho_2})}\,dx, 
\\
\widehat{C}_{E_{6}^{(1,\mathrm{an})}}^{(6)}
&:=
\sum_{j=0}^{+\infty}
\log\left[\begin{aligned}&
1-\frac{\mathsf{T}_{0}^{(2m+1,3m)}(\sigma_{\star})+\widehat{\mathsf{T}}_0^{(3m+1,4m)}(\sigma_{\star})}{\mathsf{T}_{0}^{(2m+1,3m)}(\sigma_{\star})
+\widehat{\mathsf{T}}_0^{(3m+1,4m)}(\tau_{\rho_2})
+\Omega_{4m+1}^{(7m)}}
\\
&\qquad\times\frac{\frac{\sigma_{\star}-\tau_{\rho_1}}{\tau_{\rho_2}-\sigma_{\star}}e^{\mathsf{k}(\rho_2)-\mathsf{k}(\rho_1)}(\tfrac{\rho_1}{\rho_2})^{2(\theta_{\star}+j-1)}}{1+\frac{\sigma_{\star}-\tau_{\rho_1}}{\tau_{\rho_2}-\sigma_{\star}}e^{\mathsf{k}(\rho_2)-\mathsf{k}(\rho_1)}(\tfrac{\rho_1}{\rho_2})^{2(\theta_{\star}+j-1)}}
\end{aligned}\right], 
\end{align*}
and 
\[
    E_{6}^{(2,\mathrm{an})}
   =
    C_{E_{6}^{(2,\mathrm{an})}}^{(3)}n 
    +C_{E_{6}^{(2,\mathrm{an})}}^{(4)}\sqrt{n}
    +C_{E_{6}^{(2,\mathrm{an})}}^{(5)}\log n
    +C_{E_{6}^{(2,\mathrm{an})}}^{(6)}
    +\widehat{C}_{E_{6}^{(2,\mathrm{an})}}^{(6)}
    +\mathcal{O}\Bigl(\frac{(\log n)^2}{n}\Bigr), 
\] 
where 
\begin{align*}
C_{E_{6}^{(2,\mathrm{an})}}^{(3)}
&:=
\int_{\sigma_{\star}}^{\frac{\tau_{\rho_2}}{1+\epsilon}}
\log\Bigl(1-\widehat{\mathsf{T}}_0^{(3m+1,4m)}(x)+\widehat{\mathsf{T}}_0^{(3m+1,7m)}(\tau_{\rho_2})\Bigr)\,dx,
\quad
C_{E_{6}^{(2,\mathrm{an})}}^{(4)}
:=0,
\quad
C_{E_{6}^{(2,\mathrm{an})}}^{(5)}
:=0,
\\
C_{E_{6}^{(2,\mathrm{an})}}^{(6)}  
&:=
\Bigl(\theta_{\star}-\frac{1}{2}\Bigr)
\log\Bigl(1-\widehat{\mathsf{T}}_0^{(3m+1,4m)}(\sigma_{\star})+\widehat{\mathsf{T}}_0^{(3m+1,7m)}(\tau_{\rho_2})\Bigr)
\\
&\quad 
+\Bigl(\theta_{2,-}^{(n,\epsilon)}-\frac{1}{2}\Bigr)
\log\Bigl(1-\widehat{\mathsf{T}}_0^{(3m+1,4m)}(\tfrac{\tau_{\rho_2}}{1+\epsilon})+\widehat{\mathsf{T}}_0^{(3m+1,7m)}(\tau_{\rho_2})\Bigr)
\\
&\quad 
+\int_{\sigma_{\star}}^{\frac{\tau_{\rho_2}}{1+\epsilon}}
\frac{\frac{2\rho_2^2\Delta Q(\rho_2)}{\tau_{\rho_2}-x}
\widehat{\mathsf{T}}_{1}^{(3m+1,4m)}(x)
+x\widehat{\mathsf{T}}_{2}^{(3m+1,4m)}(x)}{1-\widehat{\mathsf{T}}_0^{(3m+1,4m)}(x)+\widehat{\mathsf{T}}_0^{(3m+1,7m)}(\tau_{\rho_2})}\,dx
\\
&\quad 
+\Bigl(\frac{1}{2}\rho_2\mathsf{k}'(\rho_2)+1\Bigr)
\log\Bigl(1-\widehat{\mathsf{T}}_0^{(3m+1,4m)}(\tfrac{\tau_{\rho_2}}{1+\epsilon})+\widehat{\mathsf{T}}_0^{(3m+1,7m)}(\tau_{\rho_2})\Bigr)
\\
&\quad
-\Bigl(\frac{1}{2}\rho_2\mathsf{k}'(\rho_2)+1\Bigr)\log\Bigl(1-\widehat{\mathsf{T}}_0^{(3m+1,4m)}(\sigma_{\star})+\widehat{\mathsf{T}}_0^{(3m+1,7m)}(\tau_{\rho_2})\Bigr)
\\
&\quad 
+\bigl( 
2\rho_2^2\Delta Q(\rho_2)-\tau_{\rho_2}
\bigr)
\int_{\sigma_{\star}}^{\frac{\tau_{\rho_2}}{1+\epsilon}}
\frac{\widehat{\mathsf{T}}_{2}^{(3m+1,4m)}(x)}{1-\widehat{\mathsf{T}}_0^{(3m+1,4m)}(x)+\widehat{\mathsf{T}}_0^{(3m+1,7m)}(\tau_{\rho_2})}\,dx, 
\\
\widehat{C}_{E_{6}^{(2,\mathrm{an})}}^{(6)}  
&:=
\sum_{j=0}^{+\infty}
\log\left[\begin{aligned}&
1+\frac{\mathsf{T}_0^{(2m+1,3m)}(\sigma_{\star})+\widehat{\mathsf{T}}_0^{(3m+1,4m)}(\sigma_{\star})}{1-\widehat{\mathsf{T}}_0^{(3m+1,4m)}(\sigma_{\star})+\widehat{\mathsf{T}}_0^{(3m+1,7m)}(\tau_{\rho_2})}
\\
&\qquad\times\frac{
\frac{\tau_{\rho_2}-\sigma_{\star}}{\sigma_{\star}-\tau_{\rho_1}}
e^{\mathsf{k}(\rho_1)-\mathsf{k}(\rho_2)}
(\tfrac{\rho_1}{\rho_2})^{2(j-\theta_{\star}+1+1)}
}{
1+\frac{\tau_{\rho_2}-\sigma_{\star}}{\sigma_{\star}-\tau_{\rho_1}}e^{\mathsf{k}(\rho_1)-\mathsf{k}(\rho_2)}( \tfrac{\rho_1}{\rho_2})^{2(j-\theta_{\star}+1+1)}
}
\end{aligned}\right].
\end{align*}
\end{lemma}

\begin{proof}[Proof of Lemma~\ref{lemma:asymptotics E6(1)(2) counting osc}]
For $j_{1,+}+1\le j\le j_{2,-}-1$ and $\rho_1<r<\rho_2$, we have
\[
    V_{\tau}'(r)\leq\frac{1}{r}\Bigl(\frac{2\tau_{\rho_2}}{1+\epsilon}-\frac{2\tau_{\rho_1}}{1-\epsilon}
    +\frac{2\epsilon\tau_{\rho_2}}{1+\epsilon}\Bigr), 
    \quad
    V_{\tau}'(r)\geq -\frac{1}{r}\Bigl(\frac{2\epsilon\tau_{\rho_1}}{1-\epsilon}+\frac{2\tau_{\rho_2}}{1+\epsilon}-\frac{2\tau_{\rho_1}}{1-\epsilon}\Bigr).
\]
Together with the strict monotonicity of $r q'(r)$, these bounds imply that
$V_\tau$ has a unique critical point in $(\rho_1,\rho_2)$.  The Laplace
estimates used in Lemmas~\ref{lemma: j1+ j jdiamond}
and~\ref{lemma: jstar leq j leq j2-} are therefore uniform throughout the
present range.  In particular, for every
$j_{1,+}+1\le j\le j_{2,-}-1$,
\begin{align*}
\log\Bigl(1+\sum_{\ell=1}^{7m}\omega_\ell
F_{n,j,\ell}^{(\mathrm{an})}\Bigr)
&=
\log\bigg(
\sum_{\ell=2m+1}^{3m}\omega_{\ell}F_{n,j,\ell}^{(\mathrm{an})}
+\sum_{\ell=3m+1}^{4m}\omega_{\ell}F_{n,j,\ell}^{(\mathrm{an})}
+\Omega_{4m+1}^{(7m)}\bigg)+\mathcal{O}(e^{-c\sqrt{n}}).
\end{align*}
for a constant $c>0$ independent of $j$.  We first treat
$j_{1,+}+1\le j\le\lfloor j_\star\rfloor$.  Apply the endpoint Laplace
expansion of Lemma~\ref{lemma: j1+ j jdiamond} at $\rho_2$; its remainder is
uniform throughout this index range.  This gives
\begin{align*}
2\int_{\rho_2}^{\rho_2+\frac{\rho_2t_{\ell}}{n}}re^{\mathsf{k}(r)}e^{-nV_{\tau}(r)}\,dr 
&=
\frac{2\rho_2e^{\mathsf{k}(\rho_2)}e^{-nV_{\tau}(\rho_1)}}{\eta_2n}
\Bigl(\frac{\rho_1}{\rho_2}\Bigr)^{2(j_{\star}-j)}
\bigg[
1-e^{-\rho_2t_{\ell}\eta_2}
+
e^{-\rho_2t_{\ell}\eta_2}\frac{\widetilde{\mathsf{a}}_1(\kappa_{\ell})}{n}
\\
&\quad
+e^{-\rho_2t_{\ell}\eta_2}\frac{\widetilde{\mathsf{a}}_2(\kappa_{\ell})}{n^2}
+e^{-\rho_2t_{\ell}\eta_2}\frac{\widetilde{\mathsf{a}}_3(\kappa_{\ell})}{n^3}
+\mathcal{O}\Bigl(\frac{1}{n^4\eta_2^8}\Bigr)
\bigg],
\end{align*}
where, with $\tau_j:=j/n$, we set
\[
\eta_2\equiv\eta_2(j):=V_{\tau_j}'(\rho_2)
=\frac{2}{\rho_2}(\tau_{\rho_2}-\tau_j),
\qquad
\kappa_{\ell}:=\rho_2\eta_2t_{\ell}
=2(\tau_{\rho_2}-\tau_j)t_{\ell}.
\]
We also write
\[
\mathcal V_2(\rho_2):=4\Delta Q(\rho_2),\qquad
\mathcal V_3(\rho_2):=4\partial_r\Delta Q(\rho_2)-\frac{4\Delta Q(\rho_2)}{\rho_2}.
\]
Thus $V_{\tau_j}''(\rho_2)=\mathcal V_2(\rho_2)-\eta_2/\rho_2$ and
$V_{\tau_j}^{(3)}(\rho_2)=\mathcal V_3(\rho_2)+2\eta_2/\rho_2^2$.  For $k\ge0$, let
\begin{align*}
\widetilde{\mathsf{a}}_{1,k}(\kappa_{\ell})
&:=-\Bigl(\mathcal{V}_2(\rho_2)-\frac{\eta_2}{\rho_2}\Bigr)\frac{1}{\eta_2^2}\widetilde{\mathsf{e}}_{2,k}(\kappa_{\ell})+\Bigl(\mathsf{k}'(\rho_2)+\frac{1}{\rho_2}\Bigr)\frac{1}{\eta_2}\widetilde{\mathsf{e}}_{1,k}(\kappa_{\ell}),
\\
\widetilde{\mathsf{a}}_{2,k}(\kappa_{\ell})
&:=
\Bigl(\mathcal{V}_2(\rho_2)-\frac{\eta_2}{\rho_2}\Bigr)^2
\frac{3\widetilde{\mathsf{e}}_{4,k}(\kappa_{\ell})}{\eta_2^4}
-
\Bigl\{
3\Bigl(\mathcal{V}_2(\rho_2)-\frac{\eta_2}{\rho_2}\Bigr)
\Bigl(\mathsf{k}'(\rho_2)+\frac{1}{\rho_2}\Bigr)
+\mathcal{V}_3(\rho_2)+\frac{2\eta_2}{\rho_2^2}
\Bigr\}
\frac{\widetilde{\mathsf{e}}_{3,k}(\kappa_{\ell})}{\eta_2^3}
\\
&\quad 
+\Bigl(\frac{2\mathsf{k}'(\rho_2)}{\rho_2}+\mathsf{k}'(\rho_2)^2+\mathsf{k}''(\rho_2)\Bigr)
\frac{\widetilde{\mathsf{e}}_{2,k}(\kappa_{\ell})}{\eta_2^2}, 
\\
\widetilde{\mathsf{a}}_{3,k}(\kappa_{\ell})
&:=
-\Bigl(\mathcal{V}_2(\rho_2)-\frac{\eta_2}{\rho_2}\Bigr)^3
\frac{15\widetilde{\mathsf{e}}_{6,k}(\kappa_{\ell})}{\eta_2^6}
\\
&\quad 
+
\Bigl\{
3\Bigl(\mathcal{V}_2(\rho_2)-\frac{\eta_2}{\rho_2}\Bigr)^2\Bigl(\mathsf{k}'(\rho_2)+\frac{1}{\rho_2}\Bigr)
+
2\Bigl(\mathcal{V}_2(\rho_2)-\frac{\eta_2}{\rho_2}\Bigr)\Bigl(\mathcal{V}_3(\rho_2)+\frac{2\eta_2}{\rho_2^2}\Bigr)
\Bigr\}\frac{5\widetilde{\mathsf{e}}_{5,k}(\kappa_{\ell})}{\eta_2^5}
\\
&\quad 
+
\mathcal{O}\Bigl( 
\frac{\widetilde{\mathsf{e}}_{4,k}(\kappa_{\ell})}{\eta_2^4}(1+\eta_2)
+
\frac{\widetilde{\mathsf{e}}_{3,k}(\kappa_{\ell})}{\eta_2^3}
\Bigr), 
\end{align*} 
where
$\widetilde{\mathsf e}_{p,k}(x):=e^x-\mathsf e_{p,k}(x)$,
$\mathsf e_{p,k}$ is defined in \eqref{def of mathsf ek}, and
$\widetilde{\mathsf a}_{p}:=\widetilde{\mathsf a}_{p,0}$ for $p=1,2,3$.
To make the ratio expansion used below explicit, put
\begin{align*}
D_j&:=\mathsf T_0^{(2m+1,3m)}(\tau_j)
+\widehat{\mathsf T}_0^{(3m+1,4m)}(\tau_{\rho_2})
+\Omega_{4m+1}^{(7m)},\\
A_j&:=\mathsf T_0^{(2m+1,3m)}(\tau_j)
+\widehat{\mathsf T}_0^{(3m+1,4m)}(\tau_j),\\
R_j&:=\frac{\tau_j-\tau_{\rho_1}}{\tau_{\rho_2}-\tau_j}
e^{\mathsf k(\rho_2)-\mathsf k(\rho_1)}
\left(\frac{\rho_1}{\rho_2}\right)^{2(j_\star-j-1)},\\
N_j&:=\sum_{\ell=2m+1}^{3m}\omega_\ell
e^{-2t_\ell(\tau_j-\tau_{\rho_1})}
\mathsf a_{1,1}(\varphi_\ell;\eta_1).
\end{align*}
The endpoint expansions above and
Lemma~\ref{lemma: j1+ j jdiamond} give, uniformly for
$j_{1,+}+1\le j\le\lfloor j_\star\rfloor$,
\[
\sum_{\ell=2m+1}^{3m}\omega_\ell F_{n,j,\ell}^{(\mathrm{an})}
+\sum_{\ell=3m+1}^{4m}\omega_\ell F_{n,j,\ell}^{(\mathrm{an})}
+\Omega_{4m+1}^{(7m)}
=D_j-\frac{R_j}{1+R_j}A_j-\frac{N_j}{n}+\mathcal R_{n,j},
\]
where the error made by replacing the exact first-order denominator by $D_j$
is geometrically localized at $j_\star$ and
\[
\sum_{j=j_{1,+}+1}^{\lfloor j_\star\rfloor}|\mathcal R_{n,j}|
=\mathcal O\!\left(\frac{(\log n)^2}{n}\right).
\]
Since $D_j$ and the logarithmic argument are uniformly bounded away from
zero, Taylor's formula yields
\begin{align*}
&\quad \log\!\left(D_j-\frac{R_j}{1+R_j}A_j-\frac{N_j}{n}
+\mathcal R_{n,j}\right)
\\
&=\log D_j+\log\!\left(1-\frac{A_j}{D_j}\frac{R_j}{1+R_j}\right)
-\frac{N_j}{nD_j}+r_{n,j},
\\
&\quad\sum_j|r_{n,j}|=\mathcal O\!\left(\frac{(\log n)^2}{n}\right).
\end{align*}
Substituting this logarithmic expansion into the sum over the left range
gives, as $n\to\infty$,
\begin{align*}
&\quad\sum_{j=j_{1,+}+1}^{\lfloor j_{\star}\rfloor}    
\log
\Bigl(\sum_{\ell=2m+1}^{3m}\omega_{\ell}F_{n,j,\ell}^{(\mathrm{an})}
+\sum_{\ell=3m+1}^{4m}\omega_{\ell}F_{n,j,\ell}^{(\mathrm{an})}
+\Omega_{4m+1}^{(7m)}
\Bigr)
\\
&=
\sum_{j=j_{1,+}+1}^{\lfloor j_{\star}\rfloor}
\log\Bigl(\mathsf{T}_{0}^{(2m+1,3m)}(j/n)
+\widehat{\mathsf{T}}_0^{(3m+1,4m)}(\tau_{\rho_2})
+\Omega_{4m+1}^{(7m)}\Bigr)
\\
&\quad
+\sum_{j=j_{1,+}+1}^{\lfloor j_{\star}\rfloor}
\log\left[\begin{aligned}&
1-\frac{\mathsf{T}_{0}^{(2m+1,3m)}(j/n)+\widehat{\mathsf{T}}_0^{(3m+1,4m)}(j/n)}{\mathsf{T}_{0}^{(2m+1,3m)}(j/n)
+\widehat{\mathsf{T}}_0^{(3m+1,4m)}(\tau_{\rho_2})
+\Omega_{4m+1}^{(7m)}}
\\
&\qquad\times\frac{\frac{\tau-\tau_{\rho_1}}{\tau_{\rho_2}-\tau}e^{\mathsf{k}(\rho_2)-\mathsf{k}(\rho_1)}(\tfrac{\rho_1}{\rho_2})^{2(j_{\star}-j-1)}}{1+\frac{\tau-\tau_{\rho_1}}{\tau_{\rho_2}-\tau}e^{\mathsf{k}(\rho_2)-\mathsf{k}(\rho_1)}(\tfrac{\rho_1}{\rho_2})^{2(j_{\star}-j-1)}}
\end{aligned}\right]
\\
&\quad
-\frac{1}{n}
\sum_{j=j_{1,+}+1}^{\lfloor j_{\star}\rfloor}
\frac{
\sum_{\ell={2m+1}}^{3m}\omega_{\ell}
e^{-2t_{\ell}(j/n-\tau_{\rho_1})}
\mathsf{a}_{1,1}(\varphi_{\ell};\eta_1)
}{\mathsf{T}_{0}^{(2m+1,3m)}(j/n)
+\widehat{\mathsf{T}}_0^{(3m+1,4m)}(\tau_{\rho_2})
+\Omega_{4m+1}^{(7m)}}
+
\sum_{j=j_{1,+}+1}^{\lfloor j_{\star}\rfloor}
\mathcal{O}\Bigl(\frac{1}{n^2}+\frac{1}{n}\Bigl(\frac{\rho_1}{\rho_2}\Bigr)^{2(j_{\star}-j)}\Bigr)
\end{align*}
The last sum is
$\mathcal O(n^{-1})+\mathcal O(n^{-1}\sum_{k\ge0}(\rho_1/\rho_2)^{2k})
=\mathcal O(n^{-1})$.  We apply
Lemma~\ref{lemma:Riemann sum NEW} to the smooth sums.  For the first term on
the right-hand side,
\begin{align*}
&\quad 
\sum_{j=j_{1,+}+1}^{\lfloor j_{\star}\rfloor}
\log\Bigl(1+\mathsf{T}_{0}^{(2m+1,3m)}(j/n)
+\widehat{\mathsf{T}}_0^{(3m+1,7m)}(\tau_{\rho_2})\Bigr)
\\
&=n\int_{\frac{\tau_{\rho_1}}{1-\epsilon}}^{\sigma_{\star}}
\log\Bigl(1+\mathsf{T}_{0}^{(2m+1,3m)}(x)+\widehat{\mathsf{T}}_0^{(3m+1,7m)}(\tau_{\rho_2})\Bigr)\,dx
\\
&\quad 
+\Bigl(\theta_{1,+}^{(n,\epsilon)}-\frac{1}{2}\Bigr)
\log\Bigl(1+\mathsf{T}_{0}^{(2m+1,3m)}(\tfrac{\tau_{\rho_1}}{1-\epsilon})+\widehat{\mathsf{T}}_0^{(3m+1,7m)}(\tau_{\rho_2})\Bigr)
\\
&\quad 
+\Bigl(\frac{1}{2}-\theta_{\star}\Bigr)
\log\Bigl(1+\mathsf{T}_{0}^{(2m+1,3m)}(\sigma_{\star})+\widehat{\mathsf{T}}_0^{(3m+1,7m)}(\tau_{\rho_2})\Bigr)
+\mathcal{O}(n^{-1}), 
\end{align*}
and 
\begin{align*}
  &\quad 
\sum_{j=j_{1,+}+1}^{\lfloor j_{\star}\rfloor}
\log\bigg[
1-\frac{\mathsf{T}_{0}^{(2m+1,3m)}(j/n)+\widehat{\mathsf{T}}_0^{(3m+1,4m)}(j/n)}{\mathsf{T}_{0}^{(2m+1,3m)}(j/n)
+\widehat{\mathsf{T}}_0^{(3m+1,4m)}(\tau_{\rho_2})
+\Omega_{4m+1}^{(7m)}}
\frac{\frac{\tau-\tau_{\rho_1}}{\tau_{\rho_2}-\tau}e^{\mathsf{k}(\rho_2)-\mathsf{k}(\rho_1)}(\tfrac{\rho_1}{\rho_2})^{2(j_{\star}-j-1)}}{1+\frac{\tau-\tau_{\rho_1}}{\tau_{\rho_2}-\tau}e^{\mathsf{k}(\rho_2)-\mathsf{k}(\rho_1)}(\tfrac{\rho_1}{\rho_2})^{2(j_{\star}-j-1)}}
\bigg]
\\
&=
\sum_{j=0}^{+\infty}
\log\bigg[
1-\frac{\mathsf{T}_{0}^{(2m+1,3m)}(\sigma_{\star})+\widehat{\mathsf{T}}_0^{(3m+1,4m)}(\sigma_{\star})}{\mathsf{T}_{0}^{(2m+1,3m)}(\sigma_{\star})
+\widehat{\mathsf{T}}_0^{(3m+1,4m)}(\tau_{\rho_2})
+\Omega_{4m+1}^{(7m)}}
\frac{\frac{\sigma_{\star}-\tau_{\rho_1}}{\tau_{\rho_2}-\sigma_{\star}}e^{\mathsf{k}(\rho_2)-\mathsf{k}(\rho_1)}(\tfrac{\rho_1}{\rho_2})^{2(\theta_{\star}+j-1)}}{1+\frac{\sigma_{\star}-\tau_{\rho_1}}{\tau_{\rho_2}-\sigma_{\star}}e^{\mathsf{k}(\rho_2)-\mathsf{k}(\rho_1)}(\tfrac{\rho_1}{\rho_2})^{2(\theta_{\star}+j-1)}}
\bigg]
\\
&\quad 
+\mathcal{O}\Bigl(\frac{(\log n)^2}{n}\Bigr). 
\end{align*}
To justify the second identity, truncate the sum at $K\log n$.  On this range
the coefficient functions may be replaced by their values at $\sigma_\star$
with total error $\mathcal O((\log n)^2/n)$; beyond it, the geometric factor
$(\rho_1/\rho_2)^{2j}$ gives a smaller tail.  Extending the truncated sum to
$\mathbb N_0$ therefore has the same error.  This is the shifted-index version
of the argument in \cite{ACCL2}; see also the proof of
Lemma~\ref{lemma:S3 theta function part}.

For the final correction sum, first rewrite
\begin{align*}
&\quad 
-\frac{1}{n}
\sum_{j=j_{1,+}+1}^{\lfloor j_{\star}\rfloor}
\frac{
\sum_{\ell={2m+1}}^{3m}\omega_{\ell}
e^{-2t_{\ell}(j/n-\tau_{\rho_1})}
\mathsf{a}_{1,1}(\varphi_{\ell};\eta_1)
}{
1+\mathsf{T}_{0}^{(2m+1,3m)}(j/n)+\widehat{\mathsf{T}}_0^{(3m+1,7m)}(\tau_{\rho_2})}
\\
&=
-\frac{1}{n}
\sum_{j=j_{1,+}+1}^{\lfloor j_{\star}\rfloor}
\frac{\frac{2\rho_1^2\Delta Q(\rho_1)}{j/n-\tau_{\rho_1}}
\mathsf{T}_{1}^{(2m+1,3m)}(j/n)
+
\frac{j}{n}
\mathsf{T}_{2}^{(2m+1,3m)}(j/n)
}{
1+\mathsf{T}_{0}^{(2m+1,3m)}(j/n)+\widehat{\mathsf{T}}_0^{(3m+1,7m)}(\tau_{\rho_2})}
\\
&\quad 
-\frac{1}{n}
\sum_{j=j_{1,+}+1}^{\lfloor j_{\star}\rfloor}
\frac{2\Bigl(\frac{1}{2}\rho_1\mathsf{k}'(\rho_1)+1\Bigr)
\mathsf{T}_{1}^{(2m+1,3m)}(j/n)
+
\bigl(2\rho_1^2\Delta Q(\rho_1)-\tau_{\rho_1}\bigr)
\mathsf{T}_{2}^{(2m+1,3m)}(j/n)
}{
1+\mathsf{T}_{0}^{(2m+1,3m)}(j/n)+\widehat{\mathsf{T}}_0^{(3m+1,7m)}(\tau_{\rho_2})},
\end{align*}
and then apply Lemma~\ref{lemma:Riemann sum NEW}; this gives
\begin{align*}
&\quad 
-\frac{1}{n}
\sum_{j=j_{1,+}+1}^{\lfloor j_{\star}\rfloor}
\frac{
\sum_{\ell={2m+1}}^{3m}\omega_{\ell}
e^{-2t_{\ell}(j/n-\tau_{\rho_1})}
\mathsf{a}_{1,1}(\varphi_{\ell};\eta_1)
}{
1+\mathsf{T}_{0}^{(2m+1,3m)}(j/n)+\widehat{\mathsf{T}}_0^{(3m+1,7m)}(\tau_{\rho_2})}
\\
&=
-
\int_{\frac{\tau_{\rho_1}}{1-\epsilon}}^{\sigma_{\star}}
\frac{\frac{2\rho_1^2\Delta Q(\rho_1)}{x-\tau_{\rho_1}}
\mathsf{T}_{1}^{(2m+1,3m)}(x)
+
x\mathsf{T}_{2}^{(2m+1,3m)}(x)
}{
1+\mathsf{T}_{0}^{(2m+1,3m)}(x)+\widehat{\mathsf{T}}_0^{(3m+1,7m)}(\tau_{\rho_2})}\,dx
\\
&\quad 
+
\Bigl(\frac{1}{2}\rho_1\mathsf{k}'(\rho_1)+1\Bigr)
\log\Bigl(1+\mathsf{T}_{0}^{(2m+1,3m)}(\sigma_{\star})+\widehat{\mathsf{T}}_0^{(3m+1,7m)}(\tau_{\rho_2})\Bigr)
\\
&\quad
-
\Bigl(\frac{1}{2}\rho_1\mathsf{k}'(\rho_1)+1\Bigr)
\log\Bigl(1+\mathsf{T}_{0}^{(2m+1,3m)}(\tfrac{\tau_{\rho_1}}{1-\epsilon})+\widehat{\mathsf{T}}_0^{(3m+1,7m)}(\tau_{\rho_2})\Bigr)
\\
&\quad 
-
\bigl(2\rho_1^2\Delta Q(\rho_1)-\tau_{\rho_1}\bigr)
\int_{\frac{\tau_{\rho_1}}{1-\epsilon}}^{\sigma_{\star}}
\frac{
\mathsf{T}_{2}^{(2m+1,3m)}(x)
}{
1+\mathsf{T}_{0}^{(2m+1,3m)}(x)+\widehat{\mathsf{T}}_0^{(3m+1,7m)}(\tau_{\rho_2})}\,dx
+\mathcal{O}(n^{-1}).
\end{align*}
The right-hand range, $\lfloor j_\star\rfloor+1\le j\le j_{2,-}-1$, is treated
in the same way, now expanding from the outer boundary.  As $n\to\infty$,
\begin{align*}
&\quad \sum_{j=\lfloor j_{\star}\rfloor+1}^{j_{2,-}-1}
\log\bigg[
\sum_{\ell=2m+1}^{3m}\omega_{\ell}F_{n,j,\ell}
+\sum_{\ell=3m+1}^{4m}\omega_{\ell}F_{n,j,\ell}
+\Omega_{4m+1}^{(7m)}
\bigg]
\\
&=
\sum_{j=\lfloor j_{\star}\rfloor+1}^{j_{2,-}-1}
\log\Bigl(1-\widehat{\mathsf{T}}_0^{(3m+1,4m)}(j/n)+\widehat{\mathsf{T}}_0^{(3m+1,7m)}(\tau_{\rho_2})\Bigr)
\\
&\quad 
+\sum_{j=\lfloor j_{\star}\rfloor+1}^{j_{2,-}-1}
\log\bigg[
1+\frac{\mathsf{T}_0^{(2m+1,3m)}(j/n)+\widehat{\mathsf{T}}_0^{(3m+1,4m)}(j/n)}{1-\widehat{\mathsf{T}}_0^{(3m+1,4m)}(j/n)+\widehat{\mathsf{T}}_0^{(3m+1,7m)}(\tau_{\rho_2})}
\frac{
\frac{\tau_{\rho_2}-\tau}{\tau-\tau_{\rho_1}}
e^{\mathsf{k}(\rho_1)-\mathsf{k}(\rho_2)}
(\tfrac{\rho_1}{\rho_2})^{2(j-j_{\star}+1)}
}{
1+\frac{\tau_{\rho_2}-\tau}{\tau-\tau_{\rho_1}}e^{\mathsf{k}(\rho_1)-\mathsf{k}(\rho_2)}( \tfrac{\rho_1}{\rho_2})^{2(j-j_{\star}+1)}
}
\bigg]
\\
&\quad 
+\frac{1}{n}\sum_{j=\lfloor j_{\star}\rfloor+1}^{j_{2,-}-1}
\sum_{\ell=3m+1}^{4m}\omega_{\ell}
\frac{
e^{-2t_{\ell}(\tau_{\rho_2}-j/n)}
\widetilde{\mathsf{a}}_{1,0}(\kappa_{\ell})
-\widetilde{\mathfrak{a}}_1(1-e^{-2t_{\ell}(\tau_{\rho_2}-j/n)})
}{
1-\widehat{\mathsf{T}}_0^{(3m+1,4m)}(j/n)+\widehat{\mathsf{T}}_0^{(3m+1,7m)}(\tau_{\rho_2})
}
\\
&\quad 
+
\sum_{j=\lfloor j_{\star}\rfloor+1}^{j_{2,-}-1}
\mathcal{O}\Bigl(\frac{1}{n}(\tfrac{\rho_1}{\rho_2})^{2(j-j_{\star})}+\frac{1}{n^2}\Bigr)
\end{align*}
The displayed remainder is summable by the same geometric-series estimate.
Applying Lemma~\ref{lemma:Riemann sum NEW} to the smooth terms gives
\begin{align*}
&\quad 
\sum_{j=\lfloor j_{\star}\rfloor+1}^{j_{2,-}-1}
\log\Bigl(1-\widehat{\mathsf{T}}_0^{(3m+1,4m)}(j/n)+\widehat{\mathsf{T}}_0^{(3m+1,7m)}(\tau_{\rho_2})\Bigr)
\\
&=
n\int_{\sigma_{\star}}^{\frac{\tau_{\rho_2}}{1+\epsilon}}
\log\Bigl(1-\widehat{\mathsf{T}}_0^{(3m+1,4m)}(x)+\widehat{\mathsf{T}}_0^{(3m+1,7m)}(\tau_{\rho_2})\Bigr)\,dx
\\
&\quad 
+
\Bigl(\theta_{\star}-\frac{1}{2}\Bigr)
\log\Bigl(1-\widehat{\mathsf{T}}_0^{(3m+1,4m)}(\sigma_{\star})+\widehat{\mathsf{T}}_0^{(3m+1,7m)}(\tau_{\rho_2})\Bigr)
\\
&\quad 
+\Bigl(\theta_{2,-}^{(n,\epsilon)}-\frac{1}{2}\Bigr)
\log\Bigl(1-\widehat{\mathsf{T}}_0^{(3m+1,4m)}(\tfrac{\tau_{\rho_2}}{1+\epsilon})+\widehat{\mathsf{T}}_0^{(3m+1,7m)}(\tau_{\rho_2})\Bigr)
+\mathcal{O}(n^{-1}),
\end{align*}
and 
\begin{align*}
&\quad 
\sum_{j=\lfloor j_{\star}\rfloor+1}^{j_{2,-}-1}
\log\left[\begin{aligned}&
1+\frac{\mathsf{T}_0^{(2m+1,3m)}(j/n)+\widehat{\mathsf{T}}_0^{(3m+1,4m)}(j/n)}{1-\widehat{\mathsf{T}}_0^{(3m+1,4m)}(j/n)+\widehat{\mathsf{T}}_0^{(3m+1,7m)}(\tau_{\rho_2})}
\\
&\qquad\times\frac{
\frac{\tau_{\rho_2}-\tau}{\tau-\tau_{\rho_1}}
e^{\mathsf{k}(\rho_1)-\mathsf{k}(\rho_2)}
(\tfrac{\rho_1}{\rho_2})^{2(j-j_{\star}+1)}
}{
1+\frac{\tau_{\rho_2}-\tau}{\tau-\tau_{\rho_1}}e^{\mathsf{k}(\rho_1)-\mathsf{k}(\rho_2)}( \tfrac{\rho_1}{\rho_2})^{2(j-j_{\star}+1)}
}
\end{aligned}\right]
\\
&=
\sum_{j=0}^{+\infty}
\log\left[\begin{aligned}&
1+\frac{\mathsf{T}_0^{(2m+1,3m)}(\sigma_{\star})+\widehat{\mathsf{T}}_0^{(3m+1,4m)}(\sigma_{\star})}{1-\widehat{\mathsf{T}}_0^{(3m+1,4m)}(\sigma_{\star})+\widehat{\mathsf{T}}_0^{(3m+1,7m)}(\tau_{\rho_2})}
\\
&\qquad\times\frac{
\frac{\tau_{\rho_2}-\sigma_{\star}}{\sigma_{\star}-\tau_{\rho_1}}
e^{\mathsf{k}(\rho_1)-\mathsf{k}(\rho_2)}
(\tfrac{\rho_1}{\rho_2})^{2(j-\theta_{\star}+1+1)}
}{
1+\frac{\tau_{\rho_2}-\sigma_{\star}}{\sigma_{\star}-\tau_{\rho_1}}e^{\mathsf{k}(\rho_1)-\mathsf{k}(\rho_2)}( \tfrac{\rho_1}{\rho_2})^{2(j-\theta_{\star}+1+1)}
}
\end{aligned}\right]
+\mathcal{O}\Bigl(\frac{(\log n)^2}{n}\Bigr). 
\end{align*}
For the last correction term, it is convenient to abbreviate
\[
\mathscr D_{\mathrm{out}}(x)
:=1-\widehat{\mathsf{T}}_0^{(3m+1,4m)}(x)
  +\widehat{\mathsf{T}}_0^{(3m+1,7m)}(\tau_{\rho_2}).
\]
Then
\begin{align*}
&\quad 
\frac{1}{n}\sum_{j=\lfloor j_{\star}\rfloor+1}^{j_{2,-}-1}
\sum_{\ell=3m+1}^{4m}\omega_{\ell}
\frac{
e^{-2t_{\ell}(\tau_{\rho_2}-j/n)}
\widetilde{\mathsf{a}}_{1,0}(\kappa_{\ell})
-\widetilde{\mathfrak{a}}_1(1-e^{-2t_{\ell}(\tau_{\rho_2}-j/n)})
}{
\mathscr D_{\mathrm{out}}(j/n)
}
\\
&=\frac{1}{n}\sum_{j=\lfloor j_{\star}\rfloor+1}^{j_{2,-}-1}
\frac{1}{\mathscr D_{\mathrm{out}}(j/n)}
\\[-2pt]
&\qquad\times
\sum_{\ell=3m+1}^{4m}\omega_{\ell}e^{-2t_{\ell}(\tau_{\rho_2}-j/n)}
\bigg[
\frac{2\rho_2^2\Delta Q(\rho_2)}{\tau_{\rho_2}-j/n}t_{\ell}
+t_{\ell}^2\frac{j}{n}
-\bigl(\rho_2\mathsf{k}'(\rho_2)+2\bigr)t_{\ell}
+
\bigl( 
2\rho_2^2\Delta Q(\rho_2)-\tau_{\rho_2}
\bigr)t_{\ell}^2
\bigg], 
\end{align*}
and apply Lemma~\ref{lemma:Riemann sum NEW} to obtain
\begin{align*}
&\quad
\frac{1}{n}\sum_{j=\lfloor j_{\star}\rfloor+1}^{j_{2,-}-1}
\frac{1}{\mathscr D_{\mathrm{out}}(j/n)}
\\[-2pt]
&\qquad\times
\sum_{\ell=3m+1}^{4m}\omega_{\ell}e^{-2t_{\ell}(\tau_{\rho_2}-j/n)}
\bigg[
\frac{2\rho_2^2\Delta Q(\rho_2)}{\tau_{\rho_2}-j/n}t_{\ell}
+t_{\ell}^2\frac{j}{n}
-\bigl(\rho_2\mathsf{k}'(\rho_2)+2\bigr)t_{\ell}
+
\bigl( 
2\rho_2^2\Delta Q(\rho_2)-\tau_{\rho_2}
\bigr)t_{\ell}^2
\bigg]
\\
&=
\int_{\sigma_{\star}}^{\frac{\tau_{\rho_2}}{1+\epsilon}}
\frac{\frac{2\rho_2^2\Delta Q(\rho_2)}{\tau_{\rho_2}-x}
\widehat{\mathsf{T}}_{1}^{(3m+1,4m)}(x)
+x\widehat{\mathsf{T}}_{2}^{(3m+1,4m)}(x)}{\mathscr D_{\mathrm{out}}(x)}\,dx
\\
&\quad 
+\Bigl(\frac{1}{2}\rho_2\mathsf{k}'(\rho_2)+1\Bigr)
\log\mathscr D_{\mathrm{out}}\Bigl(\tfrac{\tau_{\rho_2}}{1+\epsilon}\Bigr)
\\
&\quad
-\Bigl(\frac{1}{2}\rho_2\mathsf{k}'(\rho_2)+1\Bigr)
\log\mathscr D_{\mathrm{out}}(\sigma_{\star})
\\
&\quad 
+\bigl( 
2\rho_2^2\Delta Q(\rho_2)-\tau_{\rho_2}
\bigr)
\int_{\sigma_{\star}}^{\frac{\tau_{\rho_2}}{1+\epsilon}}
\frac{\widehat{\mathsf{T}}_{2}^{(3m+1,4m)}(x)}{\mathscr D_{\mathrm{out}}(x)}\,dx
+\mathcal{O}(n^{-1}).
\end{align*}
The $\theta_\star$ endpoint terms in the two half-ranges have opposite
orientations, as required by the convention at $\lfloor j_\star\rfloor$.
Collecting the integral, endpoint, and absolutely convergent transition-series
contributions proves both expansions.
\end{proof}

It remains to combine the two transition series and identify their sum.

\begin{proof}[Proof of Lemma~\ref{lemma:asymptotic expansion E6 annulus}]
By Lemma~\ref{lemma:asymptotics E6(1)(2) counting osc}, all integral and
endpoint contributions already agree with those in the statement.  It remains
to simplify the sum of the two absolutely convergent transition series.  We
follow \cite[Lemma~2.7]{ACCL2} and set
\begin{align*}
\Sigma_n
&:=
\sum_{j=0}^{+\infty}
\log\bigg[
1-\frac{\mathsf{T}_{0}^{(2m+1,3m)}(\sigma_{\star})+\widehat{\mathsf{T}}_0^{(3m+1,4m)}(\sigma_{\star})}{1+\mathsf{T}_{0}^{(2m+1,3m)}(\sigma_{\star})
+\widehat{\mathsf{T}}_0^{(3m+1,7m)}(\tau_{\rho_2})}
\frac{\frac{\sigma_{\star}-\tau_{\rho_1}}{\tau_{\rho_2}-\sigma_{\star}}e^{\mathsf{k}(\rho_2)-\mathsf{k}(\rho_1)}(\tfrac{\rho_1}{\rho_2})^{2(\theta_{\star}+j-1)}}{1+\frac{\sigma_{\star}-\tau_{\rho_1}}{\tau_{\rho_2}-\sigma_{\star}}e^{\mathsf{k}(\rho_2)-\mathsf{k}(\rho_1)}(\tfrac{\rho_1}{\rho_2})^{2(\theta_{\star}+j-1)}}
\bigg]
\\
&\quad 
+
\sum_{j=0}^{+\infty}
\log\bigg[
1+\frac{\mathsf{T}_0^{(2m+1,3m)}(\sigma_{\star})+\widehat{\mathsf{T}}_0^{(3m+1,4m)}(\sigma_{\star})}{1-\widehat{\mathsf{T}}_0^{(3m+1,4m)}(\sigma_{\star})+\widehat{\mathsf{T}}_0^{(3m+1,7m)}(\tau_{\rho_2})}
\frac{
\frac{\tau_{\rho_2}-\sigma_{\star}}{\sigma_{\star}-\tau_{\rho_1}}
e^{\mathsf{k}(\rho_1)-\mathsf{k}(\rho_2)}
(\tfrac{\rho_1}{\rho_2})^{2(j-\theta_{\star}+1+1)}
}{
1+\frac{\tau_{\rho_2}-\sigma_{\star}}{\sigma_{\star}-\tau_{\rho_1}}e^{\mathsf{k}(\rho_1)-\mathsf{k}(\rho_2)}( \tfrac{\rho_1}{\rho_2})^{2(j-\theta_{\star}+1+1)}
}
\bigg]. 
\end{align*}
We claim that
\begin{align}
\begin{split}
\label{def of Sigma characterization 2}
\Sigma_n
&=
(\theta_{\star}-\alpha-1)\log\mathsf{Q}-\frac{\log \mathsf{Q}}{2}\Bigl(1+\frac{2\log v+\log \mathsf{Q}}{2\log w}\Bigr)
    \\
    &\quad 
    +\log \frac{\theta( \tfrac{\log (\tfrac{\sigma_2}{\sigma_1})+\log \mathsf{Q}}{2\log \tfrac{\rho_2}{\rho_1}}+\frac{s(\lambda(\rho_1)-\lambda(\rho_2))}{2\log\frac{\rho_2}{\rho_1}}+n\sigma_{\star}-\alpha +\frac{1}{2}; \frac{\pi i }{\log \frac{\rho_2}{\rho_1}})}{\theta( \tfrac{\log (\tfrac{\sigma_2}{\sigma_1})}{2\log \tfrac{\rho_2}{\rho_1}}+\frac{s(\lambda(\rho_1)-\lambda(\rho_2))}{2\log\frac{\rho_2}{\rho_1}}+n\sigma_{\star}-\alpha+\frac{1}{2}; \frac{\pi i }{\log \frac{\rho_2}{\rho_1}})}.     
\end{split}    
\end{align}
To prove \eqref{def of Sigma characterization 2}, introduce the shorthand
\[
v:=\frac{\tau_{\rho_2}-\sigma_{\star}}{\sigma_{\star}-\tau_{\rho_1}}
e^{s(\lambda(\rho_1)-\lambda(\rho_2))}=\frac{\sigma_2}{\sigma_1}e^{s(\lambda(\rho_1)-\lambda(\rho_2))}>0,\qquad
w:=\frac{\rho_1}{\rho_2}\in(0,1).
\]
The definitions of $\mathsf k$, $v$, and $w$ imply
\begin{align*}
&\frac{\sigma_{\star}-\tau_{\rho_1}}{\tau_{\rho_2}-\sigma_{\star}}
e^{\mathsf{k}(\rho_2)-\mathsf{k}(\rho_1)}
\Bigl(\frac{\rho_1}{\rho_2}\Bigr)^{2(\theta_{\star}+j-1)}
=v^{-1}w^{2(\theta_{\star}+j-1-\alpha)},
\\
&\frac{\tau_{\rho_2}-\sigma_{\star}}{\sigma_{\star}-\tau_{\rho_1}}
e^{\mathsf{k}(\rho_1)-\mathsf{k}(\rho_2)}
\Bigl(\frac{\rho_1}{\rho_2}\Bigr)^{2(j-\theta_{\star}+2)}
=v w^{2(j-\theta_{\star}+\alpha+2)}.
\end{align*}
Hence, after reindexing the two series, $\Sigma_n$ becomes
\[
\Sigma_n
=
\sum_{\ell=1}^{+\infty}
\log
\frac{(1+w^{2\ell-1}\tfrac{v \mathsf{Q}}{w^{2(\theta_{\star}-\alpha-1)-1}})(1+w^{2\ell-1}\tfrac{w^{2(\theta_{\star}-\alpha-1)-1}}
{v\mathsf{Q}})}{(1+w^{2\ell-1}\tfrac{v}{w^{2(\theta_{\star}-\alpha-1)-1}})
(1+w^{2\ell-1}\tfrac{w^{2(\theta_{\star}-\alpha-1)-1}}{v})}
=
\log \frac{\theta(\tfrac{1}{2\pi i}\log(\tfrac{v\mathsf{Q}}{w^{2(\theta_{\star}-\alpha-1)-1}});\tfrac{\log w}{\pi i})}{\theta(\tfrac{1}{2\pi i}\log(\tfrac{v}{w^{2(\theta_{\star}-\alpha-1)-1}});\tfrac{\log w}{\pi i})},
\]
Because $w\in(0,1)$, we have
$\operatorname{Im}(\log w/(\pi i))>0$, so the products and series above
converge absolutely.  The Jacobi triple-product identity
\[
\theta(z;\tau)
=\prod_{\ell=1}^{+\infty}(1-e^{2\ell \pi i\tau})(1+e^{(2\ell-1)\pi i \tau+2\pi iz})(1+e^{(2\ell-1)\pi i \tau-2\pi iz })
\]
together with the modular transformation
$\theta(z;\tau)=(-i\tau)^{-1/2}e^{-\pi i z^2/\tau}
\theta(z/\tau;-1/\tau)$ transforms the last display into
\eqref{def of Sigma characterization 2}.  Adding the two expansions from
Lemma~\ref{lemma:asymptotics E6(1)(2) counting osc} now gives every term in the
statement, including $C_{E_6^{(\mathrm{an})}}^{(\mathrm{osc})}$.  This
completes the proof.
\end{proof}

\begin{lemma}\label{lemma:asymptotic expansion E7 annulus}
For $3m+1\le\ell\le4m$, let $r_\ell$ be defined by
\eqref{def of merging radii outside hard}, and for
$4m+1\le\ell\le5m$, let $r_\ell$ be defined by
\eqref{def of merging radii outside semi hard}.  As $n\to\infty$,
 \begin{align}
    E_{7}^{(\mathrm{an})}
    =
    C_{E_{7}^{(\mathrm{an})}}^{(3)}n 
 +C_{E_{7}^{(\mathrm{an})}}^{(4)}\sqrt{n}
 +C_{E_{7}^{(\mathrm{an})}}^{(5)}\log n
 +C_{E_{7}^{(\mathrm{an})}}^{(6)}
 +\widetilde{C}_{E_{7}^{(\mathrm{an})}}^{(M)}
 +\mathcal{O}(n^{-\frac{1}{12}}),
\end{align}  
where 
\begin{align*}
C_{E_{7}^{(\mathrm{an})}}^{(3)}
&:=
\int_{\frac{\tau_{\rho_2}}{1+\epsilon}}^{\tau_{\rho_2}}
\log\Bigl(1-\widehat{\mathsf{T}}_0^{(3m+1,4m)}(x)+\widehat{\mathsf{T}}_0^{(3m+1,7m)}(\tau_{\rho_2})\Bigr)\,dx,
\quad
C_{E_{7}^{(\mathrm{an})}}^{(4)}
:=0,
\\
C_{E_{7}^{(\mathrm{an})}}^{(5)}
&:=\frac{\rho_2^2\Delta Q(\rho_2)\widehat{\mathsf{T}}_{1}^{(3m+1,4m)}(\tau_{\rho_2})}{1-\widehat{\mathsf{T}}_0^{(3m+1,4m)}(\tau_{\rho_2})+\widehat{\mathsf{T}}_0^{(3m+1,7m)}(\tau_{\rho_2})},
\\
C_{E_{7}^{(\mathrm{an})}}^{(6)}
&:=
\int_{\frac{\tau_{\rho_2}}{1+\epsilon}}^{\tau_{\rho_2}}
\bigg[
\frac{\frac{2\rho_2^2\Delta Q(\rho_2)}{\tau_{\rho_2}-x}
\widehat{\mathsf{T}}_{1}^{(3m+1,4m)}(x)
+x\widehat{\mathsf{T}}_{2}^{(3m+1,4m)}(x)}{1-\widehat{\mathsf{T}}_0^{(3m+1,4m)}(x)+\widehat{\mathsf{T}}_0^{(3m+1,7m)}(\tau_{\rho_2})}
\\
&\quad 
-\frac{2\rho_2^2\Delta Q(\rho_2)}{\tau_{\rho_2}-x}
\frac{
\widehat{\mathsf{T}}_{1}^{(3m+1,4m)}(\tau_{\rho_2})}{1-\widehat{\mathsf{T}}_0^{(3m+1,4m)}(\tau_{\rho_2})+\widehat{\mathsf{T}}_0^{(3m+1,7m)}(\tau_{\rho_2})}
\bigg] \,dx 
\\
&\quad 
+\bigl( 
2\rho_2^2\Delta Q(\rho_2)-\tau_{\rho_2}
\bigr)\int_{\frac{\tau_{\rho_2}}{1+\epsilon}}^{\tau_{\rho_2}}
\frac{\widehat{\mathsf{T}}_{2}^{(3m+1,4m)}(x)}{1-\widehat{\mathsf{T}}_0^{(3m+1,4m)}(x)+\widehat{\mathsf{T}}_0^{(3m+1,7m)}(\tau_{\rho_2})}\,dx
\\
&
\quad
+
\Bigl(\frac{1}{2}\rho_2\mathsf{k}'(\rho_2)+1\Bigr)
\log\Bigl(1-\widehat{\mathsf{T}}_0^{(3m+1,4m)}(\tau_{\rho_2})+\widehat{\mathsf{T}}_0^{(3m+1,7m)}(\tau_{\rho_2})\Bigr)
\\
&\quad
-
\Bigl(\frac{1}{2}\rho_2\mathsf{k}'(\rho_2)+1\Bigr)
\log\Bigl(1-\widehat{\mathsf{T}}_0^{(3m+1,4m)}(\tfrac{\tau_{\rho_2}}{1+\epsilon})+\widehat{\mathsf{T}}_0^{(3m+1,7m)}(\tau_{\rho_2})\Bigr)
\\
&\quad 
+\frac{2\rho_2^2\Delta Q(\rho_2)\widehat{\mathsf{T}}_{1}^{(3m+1,4m)}(\tau_{\rho_2})}{1-\widehat{\mathsf{T}}_0^{(3m+1,4m)}(\tau_{\rho_2})+\widehat{\mathsf{T}}_0^{(3m+1,7m)}(\tau_{\rho_2})}
\log\Bigl(\frac{\epsilon}{1+\epsilon}\Bigr)
\\
&\quad 
-\Bigl(\theta_{2,-}^{(n,\epsilon)}-\frac{1}{2}\Bigr)\log\Bigl(1-\widehat{\mathsf{T}}_0^{(3m+1,4m)}(\tfrac{\tau_{\rho_2}}{1+\epsilon})+\widehat{\mathsf{T}}_0^{(3m+1,7m)}(\tau_{\rho_2})\Bigr),
\\
\widetilde{C}_{E_{7}^{(\mathrm{an})}}^{(M)}
&:=
-\frac{2\rho_2^2\Delta Q(\rho_2)\widehat{\mathsf{T}}_{1}^{(3m+1,4m)}(\tau_{\rho_2})}{1-\widehat{\mathsf{T}}_0^{(3m+1,4m)}(\tau_{\rho_2})+\widehat{\mathsf{T}}_0^{(3m+1,7m)}(\tau_{\rho_2})}
\log M
\\
&\quad 
-\frac{2\rho_2^4(\Delta Q(\rho_2))^2}{\tau_{\rho_2}^2M^2}
\frac{\widehat{\mathsf{T}}_{1}^{(3m+1,4m)}(\tau_{\rho_2})}{1-\widehat{\mathsf{T}}_0^{(3m+1,4m)}(\tau_{\rho_2})+\widehat{\mathsf{T}}_0^{(3m+1,7m)}(\tau_{\rho_2})}
\\
&\quad 
+\frac{5\rho_2^6(\Delta Q(\rho_2))^3}{\tau_{\rho_2}^4M^4}
\frac{\widehat{\mathsf{T}}_{1}^{(3m+1,4m)}(\tau_{\rho_2})}{1-\widehat{\mathsf{T}}_0^{(3m+1,4m)}(\tau_{\rho_2})+\widehat{\mathsf{T}}_0^{(3m+1,7m)}(\tau_{\rho_2})}
\\
&\quad 
-M\tau_{\rho_2}\sqrt{n}
\log\Bigl(1-\widehat{\mathsf{T}}_0^{(3m+1,4m)}(\tau_{\rho_2})+\widehat{\mathsf{T}}_0^{(3m+1,7m)}(\tau_{\rho_2})\Bigr)
\\
&\quad 
-M^2\tau_{\rho_2}\left\{\begin{aligned}&
\frac{\tau_{\rho_2}\widehat{\mathsf{T}}_1^{(3m+1,4m)}(\tau_{\rho_2})}{1-\widehat{\mathsf{T}}_0^{(3m+1,4m)}(\tau_{\rho_2})+\widehat{\mathsf{T}}_0^{(3m+1,7m)}(\tau_{\rho_2})}
\\
&\quad-
\log\Bigl(1-\widehat{\mathsf{T}}_0^{(3m+1,4m)}(\tau_{\rho_2})+\widehat{\mathsf{T}}_0^{(3m+1,7m)}(\tau_{\rho_2})\Bigr)
\end{aligned}\right\}
\\
&\quad 
+\Bigl(\theta_{2,-}^{(n,M)}-\frac{1}{2}\Bigr)
\log\Bigl(1-\widehat{\mathsf{T}}_0^{(3m+1,4m)}(\tau_{\rho_2})+\widehat{\mathsf{T}}_0^{(3m+1,7m)}(\tau_{\rho_2})\Bigr). 
\end{align*}
\end{lemma}

\begin{proof}
We use the outer-boundary analogue of
Lemma~\ref{lemma:semi hard edge case asymptotic expansion disck complement}.
Uniformly for $j_{2,-}\le j<g_{2,-}$ and
$3m+1\le\ell\le4m$, the endpoint Laplace expansion gives
\begin{align*}
F_{n,j,\ell}^{(\mathrm{an})}
&=
1-e^{-\rho_2\eta_2t_{\ell}}+\frac{1}{n}
\bigl(
e^{-\rho_2\eta_2t_{\ell}}
\widetilde{\mathsf{a}}_{1}(\kappa_{\ell})
-\widetilde{\mathfrak{a}}_1(1-e^{-\rho_2\eta_2t_{\ell}})
\bigr)
\\
&\quad 
+\frac{1}{n^2}\Bigl(
e^{-\rho_2\eta_2t_{\ell}}
\widetilde{\mathsf{a}}_{2}(\kappa_{\ell})
-(1-e^{-\rho_2\eta_2t_{\ell}})(\widetilde{\mathfrak{a}}_2-\widetilde{\mathfrak{a}}_1^2)
-e^{-\rho_2\eta_2t_{\ell}}\widetilde{\mathfrak{a}}_1
\widetilde{\mathsf{a}}_{1}(\kappa_{\ell})
\Bigr)
\\
&\quad 
+\frac{1}{n^3}\Bigl(
e^{-\rho_2\eta_2t_{\ell}}
\widetilde{\mathsf{a}}_{3}(\kappa_{\ell})
-\bigl( 
\widetilde{\mathfrak{a}}_3-2\widetilde{\mathfrak{a}}_1\widetilde{\mathfrak{a}}_2+\widetilde{\mathfrak{a}}_1^3
\bigr)\bigl(1-e^{-\rho_2\eta_2t_{\ell}}\bigr)
\\
&\quad 
-e^{-\rho_2\eta_2t_{\ell}}
\widetilde{\mathfrak{a}}_1
\widetilde{\mathsf{a}}_{2}(\kappa_{\ell})
-e^{-\rho_2\eta_2t_{\ell}}\bigl(\widetilde{\mathfrak{a}}_2-\widetilde{\mathfrak{a}}_1^2\bigr)
\widetilde{\mathsf{a}}_{1}(\kappa_{\ell})
\Bigr)
+\mathcal{O}\Bigl(\frac{1}{\eta_2^8n^4}\Bigr). 
\end{align*}
After summing against the weights $\omega_\ell$, this may be written as
\[
\Omega_{4m+1}^{(7m)}+\sum_{\ell=3m+1}^{4m}\omega_{\ell}F_{n,j,\ell}^{(\mathrm{an})}
=
1-\widehat{\mathsf{T}}_0^{(3m+1,4m)}(j/n)+\widehat{\mathsf{T}}_0^{(3m+1,7m)}(\tau_{\rho_2})
+
\frac{\widehat{\mathscr{H}}_{\rho_2,1}}{n}
+
\frac{\widehat{\mathscr{H}}_{\rho_2,2}}{n^2}
+
\frac{\widehat{\mathscr{H}}_{\rho_2,3}}{n^3}
+
\mathcal{O}\Bigl(\frac{1}{\eta_2^8n^4}\Bigr), 
\]
where 
\begin{align*}
\widehat{\mathscr{H}}_{\rho_2,1}
&=
\sum_{\ell=3m+1}^{4m}\omega_{\ell}
\Bigl\{
e^{-\rho_2\eta_2t_{\ell}}
\widetilde{\mathsf{a}}_{1}(\kappa_{\ell})
-\widetilde{\mathfrak{a}}_1(1-e^{-\rho_2\eta_2t_{\ell}})
\Bigr\}, 
\\
\widehat{\mathscr{H}}_{\rho_2,2}
&=
\sum_{\ell=3m+1}^{4m}\omega_{\ell}
\Bigl\{
e^{-\rho_2\eta_2t_{\ell}}
\widetilde{\mathsf{a}}_{2}(\kappa_{\ell})
-(1-e^{-\rho_2\eta_2t_{\ell}})(\widetilde{\mathfrak{a}}_2-\widetilde{\mathfrak{a}}_1^2)
-e^{-\rho_2\eta_2t_{\ell}}\widetilde{\mathfrak{a}}_1 
\widetilde{\mathsf{a}}_{1}(\kappa_{\ell})
\Bigr\},
\\
\widehat{\mathscr{H}}_{\rho_2,3}
&=\sum_{\ell=3m+1}^{4m}\omega_{\ell}
\Bigl\{e^{-\rho_2\eta_2t_{\ell}}
\widetilde{\mathsf{a}}_{3}(\kappa_{\ell})
-\bigl( 
\widetilde{\mathfrak{a}}_3-2\widetilde{\mathfrak{a}}_1\widetilde{\mathfrak{a}}_2+\widetilde{\mathfrak{a}}_1^3
\bigr)\bigl(1-e^{-\rho_2\eta_2t_{\ell}}\bigr)
\\
&\quad 
-e^{-\rho_2\eta_2t_{\ell}}
\widetilde{\mathfrak{a}}_1
\widetilde{\mathsf{a}}_{2}(\kappa_{\ell})
-e^{-\rho_2\eta_2t_{\ell}}\bigl(\widetilde{\mathfrak{a}}_2-\widetilde{\mathfrak{a}}_1^2\bigr)
\widetilde{\mathsf{a}}_{1}(\kappa_{\ell})
\Bigr\}.
\end{align*} 
The logarithmic Taylor expansion is uniform because its leading denominator is
bounded away from zero on the admissible parameter set.  Hence
\begin{align*}
&\quad \sum_{j=j_{2,-}}^{g_{2,-}-1}\log\Bigl(\Omega_{4m+1}^{(7m)}+\sum_{\ell=3m+1}^{4m}\omega_{\ell}F_{n,j,\ell}^{(\mathrm{an})}\Bigr)
\\
&=
\sum_{j=j_{2,-}}^{g_{2,-}-1}
\log\Bigl(1-\widehat{\mathsf{T}}_0^{(3m+1,4m)}(j/n)+\widehat{\mathsf{T}}_0^{(3m+1,7m)}(\tau_{\rho_2})\Bigr)
\\
&\quad+
\frac{1}{n}
\sum_{j=j_{2,-}}^{g_{2,-}-1}\frac{\widehat{\mathscr{H}}_{\rho_2,1}}{1-\widehat{\mathsf{T}}_0^{(3m+1,4m)}(j/n)+\widehat{\mathsf{T}}_0^{(3m+1,7m)}(\tau_{\rho_2})}
\\
&\quad 
+
\frac{1}{n^2}\sum_{j=j_{2,-}}^{g_{2,-}-1}
\bigg[
\frac{\widehat{\mathscr{H}}_{\rho_2,2}}{1-\widehat{\mathsf{T}}_0^{(3m+1,4m)}(j/n)+\widehat{\mathsf{T}}_0^{(3m+1,7m)}(\tau_{\rho_2})}
\\
&\qquad-\frac{\widehat{\mathscr{H}}_{\rho_2,1}^2}{2(1-\widehat{\mathsf{T}}_0^{(3m+1,4m)}(j/n)+\widehat{\mathsf{T}}_0^{(3m+1,7m)}(\tau_{\rho_2}))^2}
\bigg]
\\
&\quad
+\frac{1}{n^3}\sum_{j=j_{2,-}}^{g_{2,-}-1}
\bigg[
\frac{\widehat{\mathscr{H}}_{\rho_2,3}}{1-\widehat{\mathsf{T}}_0^{(3m+1,4m)}(j/n)+\widehat{\mathsf{T}}_0^{(3m+1,7m)}(\tau_{\rho_2})}
\\
&\qquad-\frac{\widehat{\mathscr{H}}_{\rho_2,1}\widehat{\mathscr{H}}_{\rho_2,2}}{(1-\widehat{\mathsf{T}}_0^{(3m+1,4m)}(j/n)+\widehat{\mathsf{T}}_0^{(3m+1,7m)}(\tau_{\rho_2}))^2}
\\
&\quad 
+\frac{\widehat{\mathscr{H}}_{\rho_2,1}^3}{3(1-\widehat{\mathsf{T}}_0^{(3m+1,4m)}(j/n)+\widehat{\mathsf{T}}_0^{(3m+1,7m)}(\tau_{\rho_2}))^3}
\bigg]+\mathcal{O}\Bigl(\frac{\sqrt{n}}{M^7}\Bigr),
\end{align*}
where the last error follows by summing the uniform local remainder.
Indeed, $\eta_2(j)=2(\tau_{\rho_2}-j/n)/\rho_2$ is smallest at
$j=g_{2,-}-1$, where it is bounded below by a positive constant times
$M/\sqrt n$.  Comparison with the integral of the reciprocal eighth
power gives
$n^{-4}\sum_{j=j_{2,-}}^{g_{2,-}-1}\eta_2(j)^{-8}
=\mathcal O(\sqrt n/M^7)$.  Applying
Lemma~\ref{lemma:Riemann sum NEW} to the leading term gives
\begin{align*}
&\quad \sum_{j=j_{2,-}}^{g_{2,-}-1}
\log\Bigl(1-\widehat{\mathsf{T}}_0^{(3m+1,4m)}(j/n)+\widehat{\mathsf{T}}_0^{(3m+1,7m)}(\tau_{\rho_2})\Bigr)
\\
&=
n\int_{\frac{\tau_{\rho_2}}{1+\epsilon}}^{\tau_{\rho_2}}
\log\Bigl(1-\widehat{\mathsf{T}}_0^{(3m+1,4m)}(x)+\widehat{\mathsf{T}}_0^{(3m+1,7m)}(\tau_{\rho_2})\Bigr)\,dx
\\
&\quad 
-M\tau_{\rho_2}\sqrt{n}
\log\Bigl(1-\widehat{\mathsf{T}}_0^{(3m+1,4m)}(\tau_{\rho_2})+\widehat{\mathsf{T}}_0^{(3m+1,7m)}(\tau_{\rho_2})\Bigr)
\\
&\quad 
-M^2\tau_{\rho_2}\left\{\begin{aligned}&
\frac{\tau_{\rho_2}\widehat{\mathsf{T}}_1^{(3m+1,4m)}(\tau_{\rho_2})}{1-\widehat{\mathsf{T}}_0^{(3m+1,4m)}(\tau_{\rho_2})+\widehat{\mathsf{T}}_0^{(3m+1,7m)}(\tau_{\rho_2})}
\\
&\quad-
\log\Bigl(1-\widehat{\mathsf{T}}_0^{(3m+1,4m)}(\tau_{\rho_2})+\widehat{\mathsf{T}}_0^{(3m+1,7m)}(\tau_{\rho_2})\Bigr)
\end{aligned}\right\}
\\
&\quad 
-\Bigl(\theta_{2,-}^{(n,\epsilon)}-\frac{1}{2}\Bigr)\log\Bigl(1-\widehat{\mathsf{T}}_0^{(3m+1,4m)}(\tfrac{\tau_{\rho_2}}{1+\epsilon})+\widehat{\mathsf{T}}_0^{(3m+1,7m)}(\tau_{\rho_2})\Bigr)
\\
&\quad 
+\Bigl(\theta_{2,-}^{(n,M)}-\frac{1}{2}\Bigr)
\log\Bigl(1-\widehat{\mathsf{T}}_0^{(3m+1,4m)}(\tau_{\rho_2})+\widehat{\mathsf{T}}_0^{(3m+1,7m)}(\tau_{\rho_2})\Bigr)
+\mathcal{O}(M^3n^{-1/2})
\end{align*}
The endpoint terms in this formula retain both rounding corrections and will
cancel the corresponding terms from $E_6^{(2,\mathrm{an})}$ and
$E_8^{(\mathrm{an})}$.  For the first logarithmic correction, define
\begin{align*}
\widetilde{\mathscr{F}}_1(x)&:=
\frac{\frac{2\rho_2^2\Delta Q(\rho_2)}{\tau_{\rho_2}-x}
\widehat{\mathsf{T}}_{1}^{(3m+1,4m)}(x)
+x\widehat{\mathsf{T}}_{2}^{(3m+1,4m)}(x)}{1-\widehat{\mathsf{T}}_0^{(3m+1,4m)}(x)+\widehat{\mathsf{T}}_0^{(3m+1,7m)}(\tau_{\rho_2})}, 
\\
\widetilde{\mathscr{F}}_2(x)&:=
\frac{\widehat{\mathsf{T}}_{1}^{(3m+1,4m)}(x)}{1-\widehat{\mathsf{T}}_0^{(3m+1,4m)}(x)+\widehat{\mathsf{T}}_0^{(3m+1,7m)}(\tau_{\rho_2})},
\quad
\widetilde{\mathscr{F}}_3(x):=
\frac{\widehat{\mathsf{T}}_{2}^{(3m+1,4m)}(x)}{1-\widehat{\mathsf{T}}_0^{(3m+1,4m)}(x)+\widehat{\mathsf{T}}_0^{(3m+1,7m)}(\tau_{\rho_2})}.
\end{align*}
Euler--Maclaurin gives
\begin{align*}
&\quad \frac{1}{n}
\sum_{j=j_{2,-}}^{g_{2,-}-1}\frac{\widehat{\mathscr{H}}_{\rho_2,1}}{1-\widehat{\mathsf{T}}_0^{(3m+1,4m)}(j/n)+\widehat{\mathsf{T}}_0^{(3m+1,7m)}(\tau_{\rho_2})}
\\
&=\int_{\frac{\tau_{\rho_2}}{1+\epsilon}}^{\frac{\tau_{\rho_2}}{1+\frac{M}{\sqrt{n}}}}
\widetilde{\mathscr{F}}_1(x)\,dx
-(\rho_2\mathsf{k}'(\rho_2)+2)
\int_{\frac{\tau_{\rho_2}}{1+\epsilon}}^{\frac{\tau_{\rho_2}}{1+\frac{M}{\sqrt{n}}}}\widetilde{\mathscr{F}}_2(x)\,dx
+\bigl( 
2\rho_2^2\Delta Q(\rho_2)-\tau_{\rho_2}
\bigr)\int_{\frac{\tau_{\rho_2}}{1+\epsilon}}^{\frac{\tau_{\rho_2}}{1+\frac{M}{\sqrt{n}}}}\widetilde{\mathscr{F}}_3(x)\,dx
\\
&\quad 
+\frac{1}{n}\Bigl(\frac{1}{2}-\theta_{2,-}^{(n,\epsilon)}\Bigr)
\widetilde{\mathscr{F}}_1(\tfrac{\tau_{\rho_2}}{1+\epsilon})
+\frac{1}{n}\Bigl(\theta_{2,-}^{(n,M)}-\frac{1}{2}\Bigr)
\widetilde{\mathscr{F}}_1(\tfrac{\tau_{\rho_2}}{1+\frac{M}{\sqrt{n}}})
+\frac{1}{n}\Bigl(\frac{1}{2}-\theta_{2,-}^{(n,\epsilon)}\Bigr)
\widetilde{\mathscr{F}}_2(\tfrac{\tau_{\rho_2}}{1+\epsilon})
\\
&\quad 
+\frac{1}{n}\Bigl(\theta_{2,-}^{(n,M)}-\frac{1}{2}\Bigr)
\widetilde{\mathscr{F}}_2(\tfrac{\tau_{\rho_2}}{1+\frac{M}{\sqrt{n}}})
+\frac{1}{n}\Bigl(\frac{1}{2}-\theta_{2,-}^{(n,\epsilon)}\Bigr)
\widetilde{\mathscr{F}}_3(\tfrac{\tau_{\rho_2}}{1+\epsilon})
+\frac{1}{n}\Bigl(\theta_{2,-}^{(n,M)}-\frac{1}{2}\Bigr)
\widetilde{\mathscr{F}}_3(\tfrac{\tau_{\rho_2}}{1+\frac{M}{\sqrt{n}}})
+\mathcal{O}(n^{-2}). 
\end{align*}
To extract the logarithmic divergence at $x=\tau_{\rho_2}$, subtract the
limiting singular coefficient.  This yields
\begin{align*}
&\quad \int_{\frac{\tau_{\rho_2}}{1+\epsilon}}^{\frac{\tau_{\rho_2}}{1+\frac{M}{\sqrt{n}}}}
\widetilde{\mathscr{F}}_1(x)\,dx
\\
&=
\int_{\frac{\tau_{\rho_2}}{1+\epsilon}}^{\frac{\tau_{\rho_2}}{1+\frac{M}{\sqrt{n}}}}
\Bigl(\widetilde{\mathscr{F}}_1(x)
-\frac{2\rho_2^2\Delta Q(\rho_2)}{\tau_{\rho_2}-x}
\frac{
\widehat{\mathsf{T}}_{1}^{(3m+1,4m)}(\tau_{\rho_2})}{1-\widehat{\mathsf{T}}_0^{(3m+1,4m)}(\tau_{\rho_2})+\widehat{\mathsf{T}}_0^{(3m+1,7m)}(\tau_{\rho_2})}
\Bigr) \,dx 
\\
&\quad 
+
\frac{2\rho_2^2\Delta Q(\rho_2)\widehat{\mathsf{T}}_{1}^{(3m+1,4m)}(\tau_{\rho_2})}{1-\widehat{\mathsf{T}}_0^{(3m+1,4m)}(\tau_{\rho_2})+\widehat{\mathsf{T}}_0^{(3m+1,7m)}(\tau_{\rho_2})}
\int_{\frac{\tau_{\rho_2}}{1+\epsilon}}^{\frac{\tau_{\rho_2}}{1+\frac{M}{\sqrt{n}}}}
\frac{1}{\tau_{\rho_2}-x}
\,dx 
\\
&=
\int_{\frac{\tau_{\rho_2}}{1+\epsilon}}^{\tau_{\rho_2}}
\Bigl(\widetilde{\mathscr{F}}_1(x)
-\frac{2\rho_2^2\Delta Q(\rho_2)}{\tau_{\rho_2}-x}
\frac{
\widehat{\mathsf{T}}_{1}^{(3m+1,4m)}(\tau_{\rho_2})}{1-\widehat{\mathsf{T}}_0^{(3m+1,4m)}(\tau_{\rho_2})+\widehat{\mathsf{T}}_0^{(3m+1,7m)}(\tau_{\rho_2})}
\Bigr) \,dx 
\\
&\quad
+
\frac{2\rho_2^2\Delta Q(\rho_2)\widehat{\mathsf{T}}_{1}^{(3m+1,4m)}(\tau_{\rho_2})}{1-\widehat{\mathsf{T}}_0^{(3m+1,4m)}(\tau_{\rho_2})+\widehat{\mathsf{T}}_0^{(3m+1,7m)}(\tau_{\rho_2})}
\Bigl( 
\frac{1}{2}\log n +\log\Bigl(\frac{\epsilon\tau_{\rho_2}}{1+\epsilon}\Bigr)-\log (\tau_{\rho_2}M)
\Bigr)
+\mathcal{O}(Mn^{-1/2}).
\end{align*}
For $k=2,3$, the integrands are regular at the endpoint, and therefore
\[
\int_{\frac{\tau_{\rho_2}}{1+\epsilon}}^{\frac{\tau_{\rho_2}}{1+\frac{M}{\sqrt{n}}}}
\widetilde{\mathscr{F}}_k(x)\,dx
=
\int_{\frac{\tau_{\rho_2}}{1+\epsilon}}^{\tau_{\rho_2}}
\widetilde{\mathscr{F}}_k(x)\,dx
+\mathcal{O}(Mn^{-1/2}), 
\]
Combining these identities, we obtain
\begin{align*}
&\quad \frac{1}{n}
\sum_{j=j_{2,-}}^{g_{2,-}-1}\frac{\widehat{\mathscr{H}}_{\rho_2,1}}{1-\widehat{\mathsf{T}}_0^{(3m+1,4m)}(j/n)+\widehat{\mathsf{T}}_0^{(3m+1,7m)}(\tau_{\rho_2})}
\\
&=
\int_{\frac{\tau_{\rho_2}}{1+\epsilon}}^{\tau_{\rho_2}}
\Bigl(\widetilde{\mathscr{F}}_1(x)
-\frac{2\rho_2^2\Delta Q(\rho_2)}{\tau_{\rho_2}-x}
\frac{
\widehat{\mathsf{T}}_{1}^{(3m+1,4m)}(\tau_{\rho_2})}{1-\widehat{\mathsf{T}}_0^{(3m+1,4m)}(\tau_{\rho_2})+\widehat{\mathsf{T}}_0^{(3m+1,7m)}(\tau_{\rho_2})}
\Bigr) \,dx 
\\
&\quad
+
\frac{\rho_2^2\Delta Q(\rho_2)\widehat{\mathsf{T}}_{1}^{(3m+1,4m)}(\tau_{\rho_2})\log n}{1-\widehat{\mathsf{T}}_0^{(3m+1,4m)}(\tau_{\rho_2})+\widehat{\mathsf{T}}_0^{(3m+1,7m)}(\tau_{\rho_2})}
\\
&
\quad
+
\Bigl(\frac{1}{2}\rho_2\mathsf{k}'(\rho_2)+1\Bigr)
\bigg[
\log\Bigl(1-\widehat{\mathsf{T}}_0^{(3m+1,4m)}(\tau_{\rho_2})+\widehat{\mathsf{T}}_0^{(3m+1,7m)}(\tau_{\rho_2})\Bigr)
\\
&\quad 
-
\log\Bigl(1-\widehat{\mathsf{T}}_0^{(3m+1,4m)}(\tfrac{\tau_{\rho_2}}{1+\epsilon})+\widehat{\mathsf{T}}_0^{(3m+1,7m)}(\tau_{\rho_2})\Bigr)
\bigg]
+\bigl( 
2\rho_2^2\Delta Q(\rho_2)-\tau_{\rho_2}
\bigr)\int_{\frac{\tau_{\rho_2}}{1+\epsilon}}^{\tau_{\rho_2}}\widetilde{\mathscr{F}}_3(x)\,dx
\\
&\quad 
+
\frac{2\rho_2^2\Delta Q(\rho_2)\widehat{\mathsf{T}}_{1}^{(3m+1,4m)}(\tau_{\rho_2})}{1-\widehat{\mathsf{T}}_0^{(3m+1,4m)}(\tau_{\rho_2})+\widehat{\mathsf{T}}_0^{(3m+1,7m)}(\tau_{\rho_2})}
\Bigl( 
\log\Bigl(\frac{\epsilon\tau_{\rho_2}}{1+\epsilon}\Bigr)-\log (\tau_{\rho_2}M)
\Bigr)
+\mathcal{O}(Mn^{-1/2}),
\\
&\quad \frac{1}{n^2}\sum_{j=j_{2,-}}^{g_{2,-}-1}
\bigg[
\frac{\widehat{\mathscr{H}}_{\rho_2,2}}{1-\widehat{\mathsf{T}}_0^{(3m+1,4m)}(j/n)+\widehat{\mathsf{T}}_0^{(3m+1,7m)}(\tau_{\rho_2})}
\\
&\qquad-\frac{\widehat{\mathscr{H}}_{\rho_2,1}^2}{2(1-\widehat{\mathsf{T}}_0^{(3m+1,4m)}(j/n)+\widehat{\mathsf{T}}_0^{(3m+1,7m)}(\tau_{\rho_2}))^2}
\bigg]
\\
&=
-\frac{2\rho_2^4(\Delta Q(\rho_2))^2}{\tau_{\rho_2}^2M^2}
\frac{\widehat{\mathsf{T}}_{1}^{(3m+1,4m)}(\tau_{\rho_2})}{1-\widehat{\mathsf{T}}_0^{(3m+1,4m)}(\tau_{\rho_2})+\widehat{\mathsf{T}}_0^{(3m+1,7m)}(\tau_{\rho_2})}
+\mathcal{O}(M^{-1}n^{-1/2}),
\end{align*}
and 
\begin{align*}
&\quad 
\frac{1}{n^3}\sum_{j=j_{2,-}}^{g_{2,-}-1}
\bigg[
\frac{\widehat{\mathscr{H}}_{\rho_2,3}}{1-\widehat{\mathsf{T}}_0^{(3m+1,4m)}(j/n)+\widehat{\mathsf{T}}_0^{(3m+1,7m)}(\tau_{\rho_2})}
-\frac{\widehat{\mathscr{H}}_{\rho_2,1}\widehat{\mathscr{H}}_{\rho_2,2}}{(1-\widehat{\mathsf{T}}_0^{(3m+1,4m)}(j/n)+\widehat{\mathsf{T}}_0^{(3m+1,7m)}(\tau_{\rho_2}))^2}
\\
&\quad 
+\frac{\widehat{\mathscr{H}}_{\rho_2,1}^3}{3(1-\widehat{\mathsf{T}}_0^{(3m+1,4m)}(j/n)+\widehat{\mathsf{T}}_0^{(3m+1,7m)}(\tau_{\rho_2}))^3}
\bigg]
\\
&
=
\frac{5\rho_2^6(\Delta Q(\rho_2))^3}{\tau_{\rho_2}^4M^4}
\frac{\widehat{\mathsf{T}}_{1}^{(3m+1,4m)}(\tau_{\rho_2})}{1-\widehat{\mathsf{T}}_0^{(3m+1,4m)}(\tau_{\rho_2})+\widehat{\mathsf{T}}_0^{(3m+1,7m)}(\tau_{\rho_2})}
+\mathcal{O}(M^{-3}n^{-1}). 
\end{align*}
Finally, add the leading, first-, second-, and third-order contributions.  The
regularized singular integral supplies the $\log n$ and $\log M$ terms; the
$M^{-2}$ and $M^{-4}$ terms come from the next two Taylor coefficients.  With
$M=n^{1/12}$, the accumulated remainder is $\mathcal O(n^{-1/12})$, proving
the lemma.
\end{proof}

\begin{lemma}\label{lemma:asymptotic expansion E8 annulus}
For $3m+1\le\ell\le4m$, let $r_\ell$ be defined by
\eqref{def of merging radii outside hard}, and for
$4m+1\le\ell\le5m$, let $r_\ell$ be defined by
\eqref{def of merging radii outside semi hard}.  As $n\to\infty$,
 \begin{align}
    E_{8}^{(\mathrm{an})}
    =
C_{E_{8}^{(\mathrm{an})}}^{(3)}n 
 +C_{E_{8}^{(\mathrm{an})}}^{(4)}\sqrt{n}
 +C_{E_{8}^{(\mathrm{an})}}^{(5)}\log n
 +C_{E_{8}^{(\mathrm{an})}}^{(6)}
 +\widetilde{C}_{E_{8}^{(\mathrm{an})}}^{(M)}
 +\mathcal{O}(n^{-\frac{1}{12}}),
\end{align}  
where    
\begin{align*}
C_{E_{8}^{(\mathrm{an})}}^{(3)}
&:=0,
\quad
C_{E_{8}^{(\mathrm{an})}}^{(4)}
:=
C_4^{\#(\mathrm{se,out})},
\quad
C_{E_{8}^{(\mathrm{an})}}^{(5)}
:=0,
\\
C_{E_{8}^{(\mathrm{an})}}^{(6)}
&:=
\mathcal{D}_6^{(\mathrm{se},\mathrm{out})}
+\Bigl(\frac{1}{2}\rho_2\mathsf{k}'(\rho_2)+1\Bigr)\log\frac{\Omega_{5m+1}^{(7m)}}{\Omega_{4m+1}^{(7m)}}
\\
&\quad+\frac{2\rho_2^2\Delta Q(\rho_2)\widehat{\mathsf{T}}_1^{(3m+1,4m)}(\tau_{\rho_2})\log \tau_{\rho_2}}{\Omega_{4m+1}^{(7m)}},
\\
\widetilde{C}_{E_{8}^{(\mathrm{an})}}^{(M)}
&:=
2\tau_{\rho_2}M\sqrt{n}\log\Omega_{4m+1}^{(7m)}
+(-\theta_{2,+}^{(n,M)}-\theta_{2,-}^{(n,M)}+1)\log\Omega_{4m+1}^{(7m)}
\\
&\quad+\tau_{\rho_2}M\sqrt{n}\log\frac{\Omega_{5m+1}^{(7m)}}{\Omega_{4m+1}^{(7m)}}
\\
&
+
\Bigl(\frac{1}{2}-\theta_{2,+}^{(n,M)}\Bigr)\log \frac{\Omega_{5m+1}^{(7m)}}{\Omega_{4m+1}^{(7m)}}
+\tau_{\rho_2}M^2\log\frac{\Omega_{5m+1}^{(7m)}}{\Omega_{4m+1}^{(7m)}}
\\
&\quad+\frac{\widehat{\mathsf{T}}_1^{(3m+1,4m)}(\tau_{\rho_2})\bigl(
\tau_{\rho_2}^2M^2
+2\rho_2^2\Delta Q(\rho_2)\log M\bigr)}{\Omega_{4m+1}^{(7m)}}.
\end{align*}
Here, equivalently,
\begin{align*} 
\mathcal{D}_6^{(\mathrm{se},\mathrm{out})}
&=
\Bigl(2+\frac{\rho_2\partial_r\Delta Q(\rho_2)}{\Delta Q(\rho_2)}\Bigr)
\int_{-\infty}^{+\infty}
\mathcal{H}_{2}^{(\mathrm{se},\mathrm{out})}
(x;\vec{s},\vec{t})
\,dx
\\
&\quad
+3\Bigl(1+\frac{\rho_2\partial_r\Delta Q(\rho_2)}{\Delta Q(\rho_2)}\Bigr)
\int_{-\infty}^{+\infty}
\widetilde{\mathcal{H}}_{2}^{(\mathrm{se},\mathrm{out})}
(x;\vec{s},\vec{t})
\,dx
\\
&\quad 
+2\Bigl(2+\frac{\rho_2\partial_r\Delta Q(\rho_2)}{\Delta Q(\rho_2)}\Bigr)
\int_{-\infty}^{+\infty}
x
\Bigl( \log \mathcal{H}_{1}^{(\mathrm{se},\mathrm{out})}
(x;\vec{s},\vec{t})
+\mathbf{1}_{[0,+\infty)}(x)\sum_{\ell=4m+1}^{5m}s_{\ell}\Bigr)
\,dx
\\
&\quad 
+
4\rho_2^2\Delta Q(\rho_2)
\sum_{\ell=3m+1}^{4m}
\frac{\omega_{\ell}}{\Omega_{4m+1}^{(7m)}}t_{\ell}
\int_{-\infty}^{+\infty}
\Bigl( 
\frac{\frac{e^{-x^2}}{2\sqrt{\pi}(1-\frac{1}{2}\erfc(x))}}{\mathcal{H}_{1}^{(\mathrm{se},\mathrm{out})}
(x;\vec{s},\vec{t})}
+
\Bigl(x+\frac{x}{2(x^2+1)}\Bigr)\mathbf{1}_{(-\infty,0]}(x)
\Bigr)\,dx
\\
&\quad 
-\rho_2^2\Delta Q(\rho_2)\frac{\widehat{\mathsf{T}}_1^{(3m+1,4m)}(\tau_{\rho_2})}{\Omega_{4m+1}^{(7m)}}
\log (2\rho_2^2\Delta Q(\rho_2)).
\end{align*}
\end{lemma}

\begin{proof}  
Split the truncated outer-hard-edge integral as
\begin{align*}
2\int_{\rho_2}^{r_{\ell}}ue^{-nV_{\tau}(u)}e^{\mathsf{k}(u)}\,du
&=
2\int_{\rho_2}^{+\infty}ue^{-nV_{\tau}(u)}e^{\mathsf{k}(u)}\,du
-2\int_{r_{\ell}}^{+\infty}ue^{-nV_{\tau}(u)}e^{\mathsf{k}(u)}\,du, 
\end{align*}
The estimates from Lemma~\ref{lemma:bulk case asymptotic expansion} and
Lemma~\ref{lemma:semi hard edge case asymptotic expansion disck complement}
then show that, for some $c>0$ and uniformly for
$g_{2,-}\le j\le g_{2,+}$,
\begin{align*}
1+\sum_{\ell=1}^{7m}\omega_{\ell}F_{n,j,\ell}
&=
\sum_{\ell=3m+1}^{4m}\omega_{\ell}F_{n,j,\ell}+\sum_{\ell=4m+1}^{5m}\omega_{\ell}F_{n, j,\ell}
+\Omega_{5m+1}^{(7m)}
+\mathcal{O}(e^{-cM^2}). 
\end{align*}
Introduce the leading profile and its first correction by
\begin{align*}
\mathscr{V}_1(x)
&:=1-\frac{1}{\Omega_{4m+1}^{(7m)}}\sum_{\ell=4m+1}^{5m}\omega_{\ell}
\frac{1-\frac{1}{2}\erfc(-t_{\ell}-x)}{1-\frac{1}{2}\erfc(-x)}, 
\quad
\widetilde{\mathscr{V}}_1(x)
:=\log \mathscr{V}_1(x),
\\
\mathscr{V}_2^{(\mathrm{se})}(x)
&:=
\frac{1}{\Omega_{4m+1}^{(7m)}}\sum_{\ell=4m+1}^{5m}\omega_{\ell}
\Bigl( 
\mathfrak{v}_2(x)t_{\ell}^2
+
\mathfrak{v}_1(x)t_{\ell}
+
\mathfrak{v}_0(x)
\Bigr), 
\\
\mathscr{V}_2^{(\mathrm{h})}(x)
&:=
\frac{\rho_2\sqrt{2\Delta Q(\rho_2)}}{\sqrt{\pi}}
\sum_{\ell=3m+1}^{4m}
\frac{\omega_{\ell}}{\Omega_{4m+1}^{(7m)}}t_{\ell}
\frac{e^{-x^2}}{1-\frac{1}{2}\erfc(-x)},
\end{align*}
where the dependence of $\mathfrak v_k(x)$ on the current summation parameter
$t_\ell$ is understood, and
\begin{align*}
\mathfrak{v}_2(x)
&:=
\frac{e^{-(t_{\ell}+x)^2}}{6\rho_2\sqrt{2\Delta Q(\rho_2)}\sqrt{\pi}(1-\frac{1}{2}\erfc(-x))}
\bigg[ 
-2\Bigl(2+\frac{\rho_2\partial_r\Delta Q(\rho_2)}{\Delta Q(\rho_2)}\Bigr)
+3\Bigl(1+\frac{\rho_2\partial_r\Delta Q(\rho_2)}{\Delta Q(\rho_2)}\Bigr)
\bigg], 
\\
\mathfrak{v}_1(x)
&:=
-\frac{2+\frac{\rho_2\partial_r\Delta Q(\rho_2)}{\Delta Q(\rho_2)}}{6\rho_2\sqrt{2\Delta Q(\rho_2)}\sqrt{\pi}}x
\frac{e^{-(t_{\ell}+x)^2}}{1-\frac{1}{2}\erfc(-x)}
,
\\
\mathfrak{v}_0(x)
&:=
\frac{1}{6\rho_2\sqrt{2\Delta Q(\rho_2)}\sqrt{\pi}}
\bigg\{
\frac{1-\frac{1}{2}\erfc(-t_{\ell}-x)}{(1-\frac{1}{2}\erfc(-x))^2}e^{-x^2}
-\frac{e^{-(t_{\ell}+x)^2}}{1-\frac{1}{2}\erfc(-x)}
\bigg\}
\bigg\{
5\Bigl(2+\frac{\rho_2\partial_r\Delta Q(\rho_2)}{\Delta Q(\rho_2)}\Bigr)x^2
\\
&\quad 
-\Bigl(2+\frac{\rho_2\partial_r\Delta Q(\rho_2)}{\Delta Q(\rho_2)}\Bigr)
-6\Bigl(2+\frac{\rho_2\partial_r\Delta Q(\rho_2)}{\Delta Q(\rho_2)}\Bigr)x^2
+\frac{12\rho_2^2\Delta Q(\rho_2)}{\tau_{\rho_2}}x^2
+6\Bigl(\frac{1}{2}\rho_2\mathsf{k}'(\rho_2)+1\Bigr)
\bigg\}. 
\end{align*}
Taking the logarithm of the uniform local expansion and summing over the
transition window gives
\begin{align*}
\sum_{j=g_{2,-}}^{g_{2,+}}\log\Bigl(1+\sum_{\ell=1}^{7m}\omega_{\ell}F_{n,j,\ell}^{(\mathrm{an})}\Bigr)
&=
\sum_{j=g_{2,-}}^{g_{2,+}}\log\Omega_{4m+1}^{(7m)}
+
\sum_{j=g_{2,-}}^{g_{2,+}}\widetilde{\mathscr{V}}_1(\tfrac{\tau_{\rho_2}M_{j}(\rho_2)}{\rho_2\sqrt{2\Delta Q(\rho_2)}})
\\
&\quad
+\frac{1}{\sqrt{n}}\sum_{j=g_{2,-}}^{g_{2,+}}
\frac{\mathscr{V}_2^{(\mathrm{se})}(\tfrac{\tau_{\rho_2}M_{j}(\rho_2)}{\rho_2\sqrt{2\Delta Q(\rho_2)}})+\mathscr{V}_2^{(\mathrm{h})}(\tfrac{\tau_{\rho_2}M_{j}(\rho_2)}{\rho_2\sqrt{2\Delta Q(\rho_2)}})}{\mathscr{V}_1(\tfrac{\tau_{\rho_2}M_{j}(\rho_2)}{\rho_2\sqrt{2\Delta Q(\rho_2)}})}
+\mathcal{O}(M^5n^{-1/2}).
\end{align*}
The constant part is counted exactly:
$\sum_{j=g_{2,-}}^{g_{2,+}}\log\Omega_{4m+1}^{(7m)}
=(g_{2,+}-g_{2,-}+1)\log\Omega_{4m+1}^{(7m)}$.
For the profile term, Lemma~\ref{lemma:Riemann sum} yields
\begin{align*}
&\quad\sum_{j=g_{2,-}}^{g_{2,+}}\widetilde{\mathscr{V}}_1(\tfrac{\tau_{\rho_2}M_j(\rho_2)}{\rho_2\sqrt{2\Delta Q(\rho_2)}})
\\
&=
\rho_2\sqrt{2\Delta Q(\rho_2)}
\int_{-\tfrac{\tau_{\rho_2}M}{\rho_2\sqrt{2\Delta Q(\rho_2)}}}^{\tfrac{\tau_{\rho_2}M}{\rho_2\sqrt{2\Delta Q(\rho_2)}}}
\log \mathscr{V}_1(-x)\,dx\,
\sqrt{n}
\\
&\quad
+\rho_2\sqrt{2\Delta Q(\rho_2)}\int_{-\tfrac{\tau_{\rho_2}M}{\rho_2\sqrt{2\Delta Q(\rho_2)}}}^{\tfrac{\tau_{\rho_2}M}{\rho_2\sqrt{2\Delta Q(\rho_2)}}}
2 \frac{\rho_2\sqrt{2\Delta Q(\rho_2)}}{\tau_{\rho_2}}x
\log \mathscr{V}_1(-x)\,dx
\\
&\quad
+\Bigl(\frac{1}{2}-\theta_{2,-}^{(n,M)}\Bigr)\log \mathscr{V}_1(\tfrac{\tau_{\rho_2}M}{\rho_2\sqrt{2\Delta Q(\rho_2)}})
+\Bigl(\frac{1}{2}-\theta_{2,+}^{(n,M)}\Bigr)\log \mathscr{V}_1(-\tfrac{\tau_{\rho_2}M}{\rho_2\sqrt{2\Delta Q(\rho_2)}})
+\mathcal{O}(M^3n^{-1/2}). 
\end{align*}
The profile has the endpoint behavior
\begin{align*}
\log \mathscr{V}_1(-x)
&=
\begin{cases}
\log\frac{\Omega_{5m+1}^{(7m)}}{\Omega_{4m+1}^{(7m)}},     & \mbox{as } x\to+\infty, \\
\mathcal{O}(e^{cx})   & \mbox{as } x\to-\infty. 
\end{cases}
\end{align*}
Subtracting its limiting constant on the positive half-line makes the integral
absolutely convergent and gives
\begin{align*}
&\quad\rho_2\sqrt{2\Delta Q(\rho_2)}
\int_{-\tfrac{\tau_{\rho_2}M}{\rho_2\sqrt{2\Delta Q(\rho_2)}}}^{\tfrac{\tau_{\rho_2}M}{\rho_2\sqrt{2\Delta Q(\rho_2)}}}
\log \mathscr{V}_1(-x)\,dx
\\
&=
\rho_2\sqrt{2\Delta Q(\rho_2)}
\int_{-\infty}^{+\infty}
\Bigl(\log \mathscr{V}_1( -x)
-\mathbf{1}_{[0,+\infty)}(x)\log\frac{\Omega_{5m+1}^{(7m)}}{\Omega_{4m+1}^{(7m)}}\Bigr)\,dx
+\tau_{\rho_2}M\log\frac{\Omega_{5m+1}^{(7m)}}{\Omega_{4m+1}^{(7m)}}
+\mathcal O(e^{-cM}).
\end{align*}
The two discrete endpoint corrections satisfy
\begin{align*}
\Bigl(\frac{1}{2}-\theta_{2,-}^{(n,M)}\Bigr)\log \mathscr{V}_1(\tfrac{\tau_{\rho_2}M}{\rho_2\sqrt{2\Delta Q(\rho_2)}})
&=
\mathcal{O}(e^{-cM}),\qquad M\to+\infty,
\\
\Bigl(\frac{1}{2}-\theta_{2,+}^{(n,M)}\Bigr)\log \mathscr{V}_1(-\tfrac{\tau_{\rho_2}M}{\rho_2\sqrt{2\Delta Q(\rho_2)}})   
&=
\Bigl(\frac{1}{2}-\theta_{2,+}^{(n,M)}\Bigr)\log \frac{\Omega_{5m+1}^{(7m)}}{\Omega_{4m+1}^{(7m)}}+\mathcal{O}(e^{-cM^2}),\qquad
M\to+\infty.
\end{align*}
Applying Lemma~\ref{lemma:Riemann sum} once more to the first correction gives
\begin{align*}
&\quad
\frac{1}{\sqrt{n}}\sum_{j=g_{2,-}}^{g_{2,+}}
\frac{1}{\mathscr{V}_1(\tfrac{\tau_{\rho_2}M_j(\rho_2)}{\rho_2\sqrt{2\Delta Q(\rho_2)}})}\Bigl[
\mathscr{V}_2^{(\mathrm{se})}(\tfrac{\tau_{\rho_2}M_j(\rho_2)}{\rho_2\sqrt{2\Delta Q(\rho_2)}})+\mathscr{V}_2^{(\mathrm{h})}(\tfrac{\tau_{\rho_2}M_j(\rho_2)}{\rho_2\sqrt{2\Delta Q(\rho_2)}})
\Bigr]
\\
&=
\rho_2\sqrt{2\Delta Q(\rho_2)}
\int_{-\tfrac{\tau_{\rho_2}M}{\rho_2\sqrt{2\Delta Q(\rho_2)}}}^{\tfrac{\tau_{\rho_2}M}{\rho_2\sqrt{2\Delta Q(\rho_2)}}}
\frac{1}{\mathscr{V}_1(-x)}\Bigl[
\mathscr{V}_2^{(\mathrm{se})}(-x)+\mathscr{V}_2^{(\mathrm{h})}(-x)
\Bigr]\,dx
+\mathcal{O}(M^2n^{-1/2}).
\end{align*}
The remaining semi-hard correction is
\begin{align*}
 &\quad
 \rho_2\sqrt{2\Delta Q(\rho_2)}
\int_{-\tfrac{\tau_{\rho_2}M}{\rho_2\sqrt{2\Delta Q(\rho_2)}}}^{\tfrac{\tau_{\rho_2}M}{\rho_2\sqrt{2\Delta Q(\rho_2)}}}
\frac{1}{\mathscr{V}_1(-x)}
\sum_{\ell=4m+1}^{5m}\frac{\omega_{\ell}}{\Omega_{4m+1}^{(7m)}}
\mathfrak{v}_0(-x)
\,dx   
\\
&=
\Bigl(2+\frac{\rho_2\partial_r\Delta Q(\rho_2)}{\Delta Q(\rho_2)}\Bigr)
\\
&\quad\times\int_{-\tfrac{\tau_{\rho_2}M}{\rho_2\sqrt{2\Delta Q(\rho_2)}}}^{\tfrac{\tau_{\rho_2}M}{\rho_2\sqrt{2\Delta Q(\rho_2)}}}
\frac{1}{\mathscr{V}_1(-x)}
\sum_{\ell=4m+1}^{5m}\frac{\omega_{\ell}}{\Omega_{4m+1}^{(7m)}}
\frac{1}{6\sqrt{\pi}}
\frac{1-\frac{1}{2}\erfc(-t_{\ell}+x)}{(1-\frac{1}{2}\erfc(x))^2}e^{-x^2}
(5x^2-1)
\,dx
\\
&\quad
-
\Bigl(2+\frac{\rho_2\partial_r\Delta Q(\rho_2)}{\Delta Q(\rho_2)}\Bigr)
\\
&\quad\times\int_{-\tfrac{\tau_{\rho_2}M}{\rho_2\sqrt{2\Delta Q(\rho_2)}}}^{\tfrac{\tau_{\rho_2}M}{\rho_2\sqrt{2\Delta Q(\rho_2)}}}
\frac{1}{\mathscr{V}_1( -x)}
\sum_{\ell=4m+1}^{5m}\frac{\omega_{\ell}}{\Omega_{4m+1}^{(7m)}}
\frac{1}{6\sqrt{\pi}}
\frac{e^{-(t_{\ell}-x)^2}}{1-\frac{1}{2}\erfc(x)}
(5x^2-1)
\,dx
\\
&\quad
+
\int_{-\tfrac{\tau_{\rho_2}M}{\rho_2\sqrt{2\Delta Q(\rho_2)}}}^{\tfrac{\tau_{\rho_2}M}{\rho_2\sqrt{2\Delta Q(\rho_2)}}}
\left\{\begin{aligned}&
-\Bigl(2+\frac{\rho_2\partial_r\Delta Q(\rho_2)}{\Delta Q(\rho_2)}\Bigr)x^2
+\frac{2\rho_2^2\Delta Q(\rho_2)}{\tau_{\rho_2}}x^2
\\
&\quad+\Bigl(\frac{1}{2}\rho_2\mathsf{k}'(\rho_2)+1\Bigr)
\end{aligned}\right\}
\partial_x\log\mathscr{V}_1(-x)
\,dx,
\end{align*}
Here we used the identity
\begin{align*}
\partial_x\log \mathscr{V}_1(-x)
&=
\frac{1}{\sqrt{\pi}}
\frac{1}{\mathscr{V}_1(-x)}
\sum_{\ell=4m+1}^{5m}\frac{\omega_{\ell}}{\Omega_{4m+1}^{(7m)}}
\bigg(
\frac{1-\frac{1}{2}\erfc(-t_{\ell}+x)}{(1-\frac{1}{2}\erfc(x))^2}e^{-x^2}
-\frac{e^{-(t_{\ell}-x)^2}}{1-\frac{1}{2}\erfc(x)}
\bigg),
\\
\mathscr{V}_1( -x)
&=1-\frac{1}{\Omega_{4m+1}^{(7m)}}\sum_{\ell=4m+1}^{5m}\omega_{\ell}
\frac{1-\frac{1}{2}\erfc(-t_{\ell}+x)}{1-\frac{1}{2}\erfc(x)}, 
\end{align*}
and integration by parts gives
\begin{align*}
&\quad
\int_{-\tfrac{\tau_{\rho_2}M}{\rho_2\sqrt{2\Delta Q(\rho_2)}}}^{\tfrac{\tau_{\rho_2}M}{\rho_2\sqrt{2\Delta Q(\rho_2)}}}
\bigg\{
-\Bigl(2+\frac{\rho_2\partial_r\Delta Q(\rho_2)}{\Delta Q(\rho_2)}\Bigr)x^2
+\frac{2\rho_2^2\Delta Q(\rho_2)}{\tau_{\rho_2}}x^2
+\Bigl(\frac{1}{2}\rho_2\mathsf{k}'(\rho_2)+1\Bigr)
\bigg\}
\partial_x\log\mathscr{V}_1( -x)
\,dx
\\
&=
2\Bigl(2+\frac{\rho_2\partial_r\Delta Q(\rho_2)}{\Delta Q(\rho_2)}\Bigr)
\int_{-\infty}^{+\infty}
x\Bigl( \log\mathscr{V}_1(-x) - \mathbf{1}_{[0,+\infty)}(x)\log\frac{\Omega_{5m+1}^{(7m)}}{\Omega_{4m+1}^{(7m)}}\Bigr)\,dx
+\tau_{\rho_2}M^2\log\frac{\Omega_{5m+1}^{(7m)}}{\Omega_{4m+1}^{(7m)}}
\\
&\quad 
-
\frac{4\rho_2^2\Delta Q(\rho_2)}{\tau_{\rho_2}}
\int_{-\tfrac{\tau_{\rho_2}M}{\rho_2\sqrt{2\Delta Q(\rho_2)}}}^{\tfrac{\tau_{\rho_2}M}{\rho_2\sqrt{2\Delta Q(\rho_2)}}}
x\log\mathscr{V}_1(-x)\,dx
+\Bigl(\frac{1}{2}\rho_2\mathsf{k}'(\rho_2)+1\Bigr)\log\frac{\Omega_{5m+1}^{(7m)}}{\Omega_{4m+1}^{(7m)}}
+\mathcal{O}(e^{-cM}).
\end{align*}
For the outer hard-edge contribution, the same regularization at
$x=-\infty$ gives
\begin{align*}
&\quad 
\rho_2\sqrt{2\Delta Q(\rho_2)}
\int_{-\tfrac{\tau_{\rho_2}M}{\rho_2\sqrt{2\Delta Q(\rho_2)}}}^{\tfrac{\tau_{\rho_2}M}{\rho_2\sqrt{2\Delta Q(\rho_2)}}}
\frac{\mathscr{V}_2^{(\mathrm{h})}(-x)}{\mathscr{V}_1( -x)}\,dx  
\\
&=
4\rho_2^2\Delta Q(\rho_2)
\sum_{\ell=3m+1}^{4m}
\frac{\omega_{\ell}}{\Omega_{4m+1}^{(7m)}}t_{\ell}
\int_{-\infty}^{+\infty}
\Bigl( 
\frac{\frac{e^{-x^2}}{2\sqrt{\pi}(1-\frac{1}{2}\erfc(x))}}{\mathscr{V}_1(-x)}
+
\Bigl(x+\frac{x}{2(x^2+1)}\Bigr)\mathbf{1}_{(-\infty,0]}(x)
\Bigr)\,dx
\\
&\quad
+\frac{\widehat{\mathsf{T}}_1^{(3m+1,4m)}(\tau_{\rho_2})}{\Omega_{4m+1}^{(7m)}}
\Bigl( 
\tau_{\rho_2}^2M^2
+\rho_2^2\Delta Q(\rho_2)\log \Bigl(\frac{M^2\tau_{\rho_2}^2}{2\rho_2^2\Delta Q(\rho_2)}\Bigr)
\bigr)
+\mathcal{O}(M^{-2}).
\end{align*}
Combining the exact count, the regularized leading profile, and the two first
corrections produces the coefficients displayed in the statement.  The
$M\sqrt n$ and $M^2$ terms match the neighboring ranges $E_7^{(\mathrm{an})}$
and $E_9^{(\mathrm{an})}$, while the rounding terms come from the two endpoints
of the window.  The local and Riemann-sum remainders are uniform; with
$M=n^{1/12}$ their total is $\mathcal O(n^{-1/12})$.
\end{proof}


\begin{lemma}\label{lemma:asymptotic expansion E9 annulus}
There exists $c>0$ such that, as $n\to\infty$,
\[
  E_{9}^{(\mathrm{an})}
    =
    (g_{a_2,-}-g_{2,+}-1)\log \Omega_{5m+1}^{(7m)}+\mathcal{O}(e^{-cM^2}). 
\]
\end{lemma}

\begin{proof}
For $g_{2,+}<j<g_{a_2,-}$, both adjacent transition variables have absolute
value at least a fixed multiple of $M$.  The Gaussian tail estimate underlying
Lemma~\ref{lemma:bulk case asymptotic expansion} consequently gives, uniformly
in this range,
\[
\log\Bigl(1+\sum_{\ell=1}^{7m}\omega_\ell
F_{n,j,\ell}^{(\mathrm{an})}\Bigr)
=\log\Omega_{5m+1}^{(7m)}+\mathcal O(e^{-cM^2}).
\]
There are $g_{a_2,-}-g_{2,+}-1$ indices.  Summing the uniform estimate and
decreasing $c$ to absorb the polynomial number of summands proves the claim.
\end{proof}

\begin{lemma}\label{lemma:asymptotic expansion E10 annulus}
As $n\to\infty$,
\[
 E_{10}^{(\mathrm{an})}
 =
 C_{E_{10}^{(\mathrm{an})}}^{(3)}n 
 + C_{E_{10}^{(\mathrm{an})}}^{(4)} \sqrt{n} 
 + C_{E_{10}^{(\mathrm{an})}}^{(5)}\log n
 + C_{E_{10}^{(\mathrm{an})}}^{(6)} 
 + \widetilde{C}_{E_{10}^{(\mathrm{an})}}^{(M)} 
 + \mathcal{O}(n^{-\frac{1}{12}}), 
\]
where    
\begin{align*}
C_{E_{10}^{(\mathrm{an})}}^{(3)}
&:=0, 
\quad
C_{E_{10}^{(\mathrm{an})}}^{(4)}
:=C_4^{\#(\mathrm{b,out})},
\quad
C_{E_{10}^{(\mathrm{an})}}^{(5)}
:=0,
\quad
C_{E_{10}^{(\mathrm{an})}}^{(6)}
:=
\mathcal{D}_6^{(\mathrm{b},\mathrm{out})}
-
\Bigl( 
\frac{1}{2}a_2\mathsf{k}'(a_2)+1
\Bigr)
\log\frac{\Omega_{5m+1}^{(7m)}}{\Omega_{6m+1}^{(7m)}}, 
\\
\widetilde{C}_{E_{10}^{(\mathrm{an})}}^{(M)} 
&:=
\sqrt{n}M\tau_{a_2}\log\frac{\Omega_{5m+1}^{(7m)}}{\Omega_{6m+1}^{(7m)}}
-
\tau_{a_2}M^2
\log\frac{\Omega_{5m+1}^{(7m)}}{\Omega_{6m+1}^{(7m)}}
+ \bigg( \frac{1}{2}-\theta_{a_2,-}^{(n,M)} \bigg)\log\frac{\Omega_{5m+1}^{(7m)}}{\Omega_{6m+1}^{(7m)}}
\\
&\quad
+
(2\tau_{a_2}M\sqrt{n}-\theta_{a_2,+}^{(n,M)}-\theta_{a_2,-}^{(n,M)}+1)\log \Omega_{6m+1}^{(7m)}.
\end{align*} 
\end{lemma}

\begin{proof}
On $g_{a_2,-}\le j\le g_{a_2,+}$, use the local coordinate centered at
$a_2$ and reverse the orientation of the bulk-transition coordinate used at
$a_1$.  The uniform local norming-constant expansion from
Lemma~\ref{lemma:bulk case asymptotic expansion}, followed by the two-sided
Euler--Maclaurin formula, gives $C_4^{\#(\mathrm{b,out})}$ and
$\mathcal D_6^{(\mathrm{b,out})}$.  Expanding $g_{a_2,\pm}$ produces the two
rounding corrections and the $M\sqrt n$ and $M^2$ terms collected in
$\widetilde C_{E_{10}^{(\mathrm{an})}}^{(M)}$.  The local remainder is uniform
on the window and sums to $\mathcal O(n^{-1/12})$ for $M=n^{1/12}$.
\end{proof}

\begin{lemma}\label{lemma:asymptotic expansion E11 annulus}
\label{lemma:asymptotics of E11 counting annulus}
There exists $c>0$ such that, as $n\to\infty$,
\[
E_{11}^{(\mathrm{an})}
=
(g_{R,-}-g_{a_2,+}-1)\log\Omega_{6m+1}^{(7m)}+\mathcal{O}(e^{-cM^2}).
\]
\end{lemma}

\begin{proof}
For $g_{a_2,+}<j<g_{R,-}$, the saddle is separated by at least $M$ scaled
units from both the $a_2$ bulk window and the soft-edge window at $R$.
Consequently, the same Gaussian tail bound used in
Lemma~\ref{lemma:bulk case asymptotic expansion} gives, uniformly in $j$,
\[
\log\Bigl(1+\sum_{\ell=1}^{7m}\omega_\ell
F_{n,j,\ell}^{(\mathrm{an})}\Bigr)
=\log\Omega_{6m+1}^{(7m)}+\mathcal O(e^{-cM^2}).
\]
The interval contains $g_{R,-}-g_{a_2,+}-1$ indices.  Summing and absorbing
this polynomial factor into the exponential completes the proof.
\end{proof}

It remains to analyze the one-sided soft-edge window.

\begin{lemma}\label{lemma:asymptotic expansion E12 annulus}
\label{lemma:asymptotics of E12 counting annulus}
As $n\to\infty$,
\[
    E_{12}^{(\mathrm{an})}
    =
    C_{E_{12}^{(\mathrm{an})}}^{(3)}n
    +
    C_{E_{12}^{(\mathrm{an})}}^{(4)}\sqrt{n}
    +
    C_{E_{12}^{(\mathrm{an})}}^{(5)}\log n
    +
    C_{E_{12}^{(\mathrm{an})}}^{(6)}
    +
    \widetilde{C}_{E_{12}^{(\mathrm{an})}}^{(M)}
    +
    \mathcal{O}(n^{-\frac{1}{12}}),
\]  
where 
\begin{align*}
C_{E_{12}^{(\mathrm{an})}}^{(3)}
&:=0, \quad
C_{E_{12}^{(\mathrm{an})}}^{(4)}
:=C_4^{\#(\mathrm{s})},
\quad
C_{E_{12}^{(\mathrm{an})}}^{(5)}
:=0, 
\quad
C_{E_{12}^{(\mathrm{an})}}^{(6)}
:=
\mathcal{D}_6^{(\mathrm{s})}
-\Bigl(\frac{1}{2}R\,\mathsf{k}'(R)+1\Bigr)\log\Omega_{6m+1}^{(7m)},
\\
\widetilde{C}_{E_{12}^{(\mathrm{an})}}^{(M)}
&:=
\Bigl(\frac{1}{2}-\theta_{R,-}^{(n,M)}\Bigr)\log\Omega_{6m+1}^{(7m)}
+M\sqrt{n}\log\Omega_{6m+1}^{(7m)}
-M^2\log\Omega_{6m+1}^{(7m)}.
\end{align*}
\end{lemma}

We first record the one-sided Euler--Maclaurin formula needed at the outer
edge.  Apply Lemma~\ref{lemma:Riemann sum} with $a$ replaced by $R$ and with
the coordinate in \eqref{def of Mja}; see also
\cite[Eqs.~(2.26), (2.27)]{C2021 FH}, \cite[Eq.~(3.49)]{C2021}, and
\cite[Lemma~3.4]{CL2023}.  If $f\in C^3(\mathbb R)$ and
$f,f',f'',f'''$ are bounded, then, uniformly for $M=o(\sqrt n)$,
\begin{align*}
\sum_{j=g_{R,-}}^{n-1}f(M_j)
&=\int_{0}^{M_{g_{R,-}}}f(t)\,dt\,\sqrt{n}
-2\int_{0}^{M_{g_{R,-}}}tf(t)\,dt
+\frac{f(M_{g_{R,-}})-f(0)}{2}
+\mathcal{O}(M^3n^{-1/2}). 
\end{align*}
At the lower endpoint,
\[
M_{g_{R,-}}(R)=M-\frac{\theta_{R,-}^{(n,M)}}{\sqrt n}
-\frac{2M\theta_{R,-}^{(n,M)}}{n}
+\mathcal O(M^2n^{-3/2}).
\]
Substitution into the preceding formula, followed by Taylor expansion of the
two endpoint terms, gives
\begin{align}
\begin{split}
\label{def of one sided Riemann sum}
\sum_{j=g_{R,-}}^{n-1}f(M_j)
&=
\int_0^{M}f(t)\,dt\,\sqrt{n}
-2\int_0^{M}tf(t)\,dt
\\
&\quad 
+\Bigl(\frac{1}{2}-\theta_{R,-}^{(n,M)}\Bigr)f(M)
-\frac{1}{2}f(0)
+\mathcal{O}(M^3n^{-1/2}). 
\end{split}
\end{align}
The error is controlled by a constant times the bounded $C^3$ norm of $f$;
thus the formula applies uniformly to the profiles below.
\begin{proof}[Proof of Lemma~\ref{lemma:asymptotics of E12 counting annulus}]
For $g_{R,-}\le j\le n-1$, every contribution except the soft-edge one is
separated from its transition window by at least $M$ scaled units.  The
corresponding Gaussian tail estimates therefore imply, uniformly in $j$, that
for some $c>0$,
\[
\log\Big( 
1+\sum_{\ell=1}^{7m}\omega_{\ell}F_{n, j,\ell}^{(\mathrm{an})}
\Bigr)
=
\log\Bigl( 
1+\sum_{\ell=6m+1}^{7m}\omega_{\ell}F_{n,j,\ell}^{(\mathrm{an})}
\Bigr)
+\mathcal{O}(e^{-cM^2}). 
\]
The uniform soft-edge expansion of the norming constant and the Taylor
expansion of the logarithm now give
\begin{align*}
&\quad \sum_{j=g_{R,-}}^{n-1}
\log\Big( 
1+\sum_{\ell=1}^{7m}\omega_{\ell}F_{n, j,\ell}^{(\mathrm{an})}
\Bigr)
\\
&=
\sum_{j=g_{R,-}}^{n-1}\log\mathscr{B}^{(\rm s)}_1(\tfrac{M_j(R)}{R\sqrt{2\Delta Q(R)}})
+\frac{1}{\sqrt{n}}\sum_{j=g_{R,-}}^{n-1}
\mathscr{B}_{2}^{(\rm s)}(\tfrac{M_j(R)}{R\sqrt{2\Delta Q(R)}})
+
\mathcal{O}\Bigl(\frac{M^5}{\sqrt{n}}\Bigr), 
\end{align*}
where, for $x\in\mathbb R$, the functions
$\mathscr B_1^{(\mathrm s)}$ and $\mathscr B_2^{(\mathrm s)}$ are defined by
\begin{align*}
\mathscr{B}_1^{(\rm s)}(x)&:= 
1+\sum_{\ell=6m+1}^{7m}
\frac{\omega_{\ell}}{2}\erfc(-t_{\ell}-x), 
\\
\mathscr{B}_{2}^{(\rm s)}(x)
&:=
\frac{1}{\mathscr{B}_1^{(\rm s)}(x)}
\sum_{\ell=6m+1}^{7m}
\frac{\omega_{\ell}e^{-(t_{\ell}+x)^2}}{6R\sqrt{2\Delta Q(R)}\sqrt{\pi}}
\bigl( 
\mathfrak{b}_2t_{\ell}^2+\mathfrak{b}_1(x)t_{\ell}+\mathfrak{b}_0(x)
\bigr). 
\end{align*}
Here 
\begin{align*}
\mathfrak{b}_2&:=
-2\Bigl(2+\frac{R\partial_r\Delta Q(R)}{\Delta Q(R)}\Bigr)
+
3\Bigl(1+\frac{R\partial_r\Delta Q(R)}{\Delta Q(R)}\Bigr)
=:\widetilde{\mathfrak{b}}_2+\widetilde{\mathfrak{b}}_2',
\\
\mathfrak{b}_1(\tfrac{x}{R\sqrt{2\Delta Q(R)}})&:=
-\Bigl(2+\frac{R\partial_r\Delta Q(R)}{\Delta Q(R)}\Bigr)\frac{x}{R\sqrt{2\Delta Q(R)}}, 
\\
\mathfrak{b}_0(\tfrac{x}{R\sqrt{2\Delta Q(R)}})&:=
-5\Bigl(
2+\frac{R\partial_r\Delta Q(R)}{\Delta Q(R)}
\Bigr)
\Bigl(\frac{x}{R\sqrt{2\Delta Q(R)}}\Bigr)^2
+
6\Bigl(
2+\frac{R\partial_r\Delta Q(R)}{\Delta Q(R)}
\Bigr)
\Bigl(\frac{x}{R\sqrt{2\Delta Q(R)}}\Bigr)^2
\\
&\quad
+2+\frac{R\partial_r\Delta Q(R)}{\Delta Q(R)}
-12R^2\Delta Q(R)
\Bigl(\frac{x}{R\sqrt{2\Delta Q(R)}}\Bigr)^2
-3R\mathsf{k}'(R)-6
\\
&=
\widetilde{\mathfrak{b}}_0(\tfrac{x}{R\sqrt{2\Delta Q(R)}})
+
6\Bigl(
2+\frac{R\partial_r\Delta Q(R)}{\Delta Q(R)}
\Bigr)
\Bigl(\frac{x}{R\sqrt{2\Delta Q(R)}}\Bigr)^2
\\
&\quad
-12R^2\Delta Q(R)
\Bigl(\frac{x}{R\sqrt{2\Delta Q(R)}}\Bigr)^2
-3R\mathsf{k}'(R)-6.
\end{align*}
Apply the one-sided formula \eqref{def of one sided Riemann sum} to each sum.
For the leading profile this gives
\begin{align*}
\sum_{j=g_{R,-}}^{n-1}\log\mathscr{B}_1^{(\rm s)}(\tfrac{M_j(R)}{R\sqrt{2\Delta Q(R)}})
&=
\int_0^{M}\log\mathscr{B}_1^{(\rm s)}(\tfrac{x}{R\sqrt{2\Delta Q(R)}})\,dx\,\sqrt{n}
-\int_0^{M}2x\,\log\mathscr{B}_1^{(\rm s)}(\tfrac{x}{R\sqrt{2\Delta Q(R)}})\,dx
\\
&\quad 
+\Bigl(\frac{1}{2}-\theta_{R,-}^{(n,M)}\Bigr)
\log\mathscr{B}_1^{(\rm s)}(\tfrac{M}{R\sqrt{2\Delta Q(R)}})
-\frac{1}{2}
\log\mathscr{B}_1^{(\rm s)}(0)
+\mathcal{O}(M^3n^{-1/2}),  
\end{align*}
and 
\begin{align*}
\frac{1}{\sqrt{n}}\sum_{j=g_{R,-}}^{n-1}
\mathscr{B}_{2}^{(\rm s)}(\tfrac{M_j(R)}{R\sqrt{2\Delta Q(R)}})    
&=    
\int_0^{M}\mathscr{B}_2^{(\rm s)}(\tfrac{x}{R\sqrt{2\Delta Q(R)}})\,dx
+
\mathcal{O}(M^2n^{-1/2}). 
\end{align*}
We next regularize the integrals.  Since the profile approaches
$\Omega_{6m+1}^{(7m)}$ with a Gaussian tail, subtraction of this limit gives
\begin{align*}
&\quad \int_0^{M}\log\mathscr{B}_1^{(\rm s)}(\tfrac{x}{R\sqrt{2\Delta Q(R)}})\,dx
\\
&=R\sqrt{2\Delta Q(R)}\int_{-\infty}^{0}\Bigl(
\log\mathscr{B}_{1}^{(\rm s)}(-x) -\sum_{\ell=6m+1}^{7m}s_{\ell}\Bigr)\,dx
+M\log\Omega_{6m+1}^{(7m)}+\mathcal O(e^{-cM^2}).
\end{align*}
The discrete endpoint correction satisfies
\[
\Bigl(\frac{1}{2}-\theta_{R,-}^{(n,M)}\Bigr)
\log\mathscr{B}_1^{(\rm s)}(\tfrac{M}{R\sqrt{2\Delta Q(R)}})
=
\Bigl(\frac{1}{2}-\theta_{R,-}^{(n,M)}\Bigr)\log\Omega_{6m+1}^{(7m)}+\mathcal{O}(e^{-cM^2}).
\]
We obtain 
\begin{align*}
&\quad \sum_{j=g_{R,-}}^{n-1}\log\mathscr{B}_1^{(\rm s)}(\tfrac{M_j(R)}{R\sqrt{2\Delta Q(R)}})
\\
&=    
R\sqrt{2\Delta Q(R)}\int_{-\infty}^{0}\Bigl(
\log\mathscr{B}_{1}^{(\rm s)}(-x) -\sum_{\ell=6m+1}^{7m}s_{\ell}\Bigr)\,dx\,\sqrt{n}
+M\sqrt{n}\log\Omega_{6m+1}^{(7m)}
\\
&\quad
-\int_0^{M}2x\,\log\mathscr{B}_1^{(\rm s)}(\tfrac{x}{R\sqrt{2\Delta Q(R)}})\,dx
\\
&\quad
+\Bigl(\frac{1}{2}-\theta_{R,-}^{(n,M)}\Bigr)\log\Omega_{6m+1}^{(7m)}
-\frac{1}{2}
\log\mathscr{B}_1^{(\rm s)}(0)
+\mathcal{O}(M^3n^{-1/2}). 
\end{align*}
For the first correction term, the same one-sided formula gives
\begin{align*}
 &\quad
\frac{1}{\sqrt{n}}\sum_{j=g_{R,-}}^{n-1}
\mathscr{B}_{2}^{(\rm s)}(\tfrac{1}{R\sqrt{2\Delta Q(R)}}M_j(R))
\\
&=
\int_{0}^{\tfrac{M}{R\sqrt{2\Delta Q(R)}}}
\frac{1}{\mathscr{B}_1^{(\rm s)}(x)}
\sum_{\ell=6m+1}^{7m}
\frac{\omega_{\ell}e^{-(t_{\ell}+x)^2}}{6\sqrt{\pi}}
\bigl( 
\mathfrak{b}_2t_{\ell}^2+\mathfrak{b}_1(x)t_{\ell}+\mathfrak{b}_0(x)
\bigr)
\,dx
+
\mathcal{O}(Mn^{-1/2}). 
\end{align*}
It remains to simplify the part containing $\mathfrak b_0(x)$.
Using the identity
\[
\partial_x\log\mathscr{B}_1^{(\rm s)}(-x)
=
-\frac{1}{\mathscr{B}_1^{(\rm s)}(-x)}\sum_{\ell=6m+1}^{7m}\frac{\omega_{\ell}e^{-(t_{\ell}-x)^2}}{\sqrt{\pi}},
\]
we obtain 
\begin{align*}
&\quad -2R^2\Delta Q(R)
\int_{-\tfrac{M}{R\sqrt{2\Delta Q(R)}}}^{0}
\frac{1}{\mathscr{B}_1^{(\rm s)}(-x)}
\sum_{\ell=6m+1}^{7m}
\frac{\omega_{\ell}e^{-(t_{\ell}-x)^2}}{\sqrt{\pi}}
x^2
\,dx
\\
&=
-
M^2\log\mathscr{B}_1^{(\rm s)}(\tfrac{M}{R\sqrt{2\Delta Q(R)}})
-
2R^2\Delta Q(R)\int_{-\tfrac{M}{R\sqrt{2\Delta Q(R)}}}^{0}
2x\log\mathscr{B}_1^{(\rm s)}(-x)\,dx,
\end{align*}
and 
\begin{align*}
&\quad
-
\Bigl(\frac{1}{2}R\,\mathsf{k}'(R)+1\Bigr)
\int_{-\tfrac{M}{R\sqrt{2\Delta Q(R)}}}^{0}
\frac{1}{\mathscr{B}_1^{(\rm s)}(-x)}
\sum_{\ell=6m+1}^{7m}
\frac{\omega_{\ell}e^{-(t_{\ell}-x)^2}}{\sqrt{\pi}}
\,dx
\\
&=
\Bigl(\frac{1}{2}R\,\mathsf{k}'(R)+1\Bigr)
\Bigl[\log\mathscr{B}_1^{(\rm s)}(-x)\Bigr]\Bigr|_{x=-\frac{M}{R\sqrt{2\Delta Q(R)}}}^{x=0}
\\
&=
\Bigl(\frac{1}{2}R\,\mathsf{k}'(R)+1\Bigr)
\Bigl( 
\log\mathscr{B}_1^{(\rm s)}(0)
-\log\Omega_{6m+1}^{(7m)}
\Bigr)+\mathcal O(e^{-cM^2}). 
\end{align*}
Also,
\[
-\int_0^{M}2x\,\log\mathscr{B}_1^{(\rm s)}(\tfrac{x}{R\sqrt{2\Delta Q(R)}})\,dx
=
2R^2\Delta Q(R)\int_{-\frac{M}{R\sqrt{2\Delta Q(R)}}}^{0}2x\log\mathscr{B}_1^{(\rm s)}(-x)\,dx, 
\]
Adding these two identities gives
\begin{align*}
&\quad
-\int_0^{M}2x\,\log\mathscr{B}_1^{(\rm s)}(\tfrac{x}{R\sqrt{2\Delta Q(R)}})\,dx
+
\int_{-\tfrac{M}{R\sqrt{2\Delta Q(R)}}}^{0}
\frac{1}{\mathscr{B}_1^{(\rm s)}(-x)}
\sum_{\ell=6m+1}^{7m}
\frac{\omega_{\ell}e^{-(t_{\ell}-x)^2}}{6\sqrt{\pi}}
\mathfrak{b}_0(-x)
\,dx
\\
&=
\Bigl(
2+\frac{R\partial_r\Delta Q(R)}{\Delta Q(R)}
\Bigr)
\int_{-\infty}^{0}
\frac{1}{\mathscr{B}_1^{(\rm s)}(-x)}
\sum_{\ell=6m+1}^{7m}
\frac{\omega_{\ell}e^{-(t_{\ell}-x)^2}}{6\sqrt{\pi}}
(-5x^2+1)
\,dx
\\
&\quad
+6\Bigl(
2+\frac{R\partial_r\Delta Q(R)}{\Delta Q(R)}
\Bigr)
\int_{-\infty}^{0}
\frac{1}{\mathscr{B}_1^{(\rm s)}(-x)}
\sum_{\ell=6m+1}^{7m}
\frac{\omega_{\ell}e^{-(t_{\ell}-x)^2}}{6\sqrt{\pi}}
x^2
\,dx
\\
&\quad
-M^2\log\Omega_{6m+1}^{(7m)}
+\Bigl(\frac{1}{2}R\,\mathsf{k}'(R)+1\Bigr)
\Bigl( 
\log\mathscr{B}_1^{(\rm s)}(0)
-\log\Omega_{6m+1}^{(7m)}
\Bigr)
+\mathcal{O}(e^{-cM^2}),
\end{align*}
for some $c>0$.  A final integration by parts yields
\begin{align*}
 &\quad
 \Bigl(
2+\frac{R\partial_r\Delta Q(R)}{\Delta Q(R)}
\Bigr)
\int_{-\tfrac{M}{R\sqrt{2\Delta Q(R)}}}^{0}
\frac{1}{\mathscr{B}_1^{(\rm s)}(-x)}
\sum_{\ell=6m+1}^{7m}
\frac{\omega_{\ell}e^{-(t_{\ell}-x)^2}}{\sqrt{\pi}}
x^2
\,dx
\\
&=
\Bigl(
2+\frac{R\partial_r\Delta Q(R)}{\Delta Q(R)}
\Bigr)
\int_{-\infty}^{0}
2x 
\Bigl(\log\mathscr{B}_1^{(\rm s)}(-x)
-\sum_{\ell=6m+1}^{7m}s_{\ell}
\Bigr)
\,dx
+\mathcal{O}(e^{-cM^2}),
\end{align*}
for some $c>0$.  Combining these identities with the one-sided
Euler--Maclaurin expansion produces $C_{E_{12}^{(\mathrm{an})}}^{(4)}$ and
$C_{E_{12}^{(\mathrm{an})}}^{(6)}$.  The endpoint terms are precisely
$\widetilde C_{E_{12}^{(\mathrm{an})}}^{(M)}$ and match those from
$E_{11}^{(\mathrm{an})}$.  Finally,
$M^5n^{-1/2}=n^{-1/12}$ for $M=n^{1/12}$, while all other displayed
remainders are smaller.  This proves the lemma.
\end{proof}

We can now assemble the thirteen ranges.

\begin{proof}[Proof of Theorem~\ref{theorem:counting statistics of annulus case}]
Add the expansions in
Lemmas~\ref{lemma:asymptotic expansion E0 annulus}--\ref{lemma:asymptotic expansion E12 annulus}
in the order of the thirteen index ranges.  At each artificial cutoff, the
$M\sqrt n$, $M^2$, $\log M$, inverse powers of $M$, and rounding terms cancel
between the two adjacent ranges.  The two contributions at $j_\star$ combine,
by Lemma~\ref{lemma:asymptotic expansion E6 annulus}, into the oscillatory term
$\mathcal F_{n,\#}$.  The remaining $n$, $\sqrt n$, $\log n$, and constant
terms are exactly $C_{3,\#}^{(\mathrm{an})}$,
$C_{4,\#}^{(\mathrm{an})}$, $C_{5,\#}^{(\mathrm{an})}$, and
$C_{6,\#}^{(\mathrm{an})}$, respectively.  Since only finitely many ranges are
involved, their uniform remainders sum to $\mathcal O(n^{-1/12})$.  This is the
assertion of Theorem~\ref{theorem:counting statistics of annulus case}.
\end{proof}

\subsection{Proof of Theorem~\ref{theorem:counting statistics of centered disk case}}
We now prove Theorem~\ref{theorem:counting statistics of centered disk case}.  The
centered-disk geometry has only an outer hard boundary.  Consequently, after the
initial range of indices has been treated, the remaining local regimes are
identical to the five outer regimes in the annulus case.  We record this
reduction explicitly below; in particular, no inner-boundary contribution is
present.

Starting from \eqref{def of calEncd}, we write
\begin{align*}
\log \mathcal{E}_{n,s\lambda,\alpha}^{(\mathrm{cd})}&=
\sum_{j=0}^{n-1}\log\Big( 
1+\sum_{\ell=3m+1}^{7m}\omega_{\ell}F_{n,j,\ell}^{(\mathrm{cd})}
\Bigr)
=\sum_{k=0}^{5}E_k^{(\mathrm{cd})},
\end{align*}
where
\begin{align*}
E_0^{(\mathrm{cd})}&:=\sum_{j=0}^{g_{2,-}-1}\log\Big( 
1+\sum_{\ell=3m+1}^{7m}\omega_{\ell}F_{n,j,\ell}^{(\mathrm{cd})}
\Bigr), \qquad
E_1^{(\mathrm{cd})}:=\sum_{j=g_{2,-}}^{g_{2,+}}\log\Big( 
1+\sum_{\ell=3m+1}^{7m}\omega_{\ell}F_{n,j,\ell}^{(\mathrm{cd})}
\Bigr),  \\   
E_2^{(\mathrm{cd})}&:=\sum_{j=g_{2,+}+1}^{g_{a_2,-}-1}\log\Big( 
1+\sum_{\ell=3m+1}^{7m}\omega_{\ell}F_{n,j,\ell}^{(\mathrm{cd})}
\Bigr), \qquad
E_3^{(\mathrm{cd})}:=\sum_{j=g_{a_2,-}}^{g_{a_2,+}}\log\Big( 
1+\sum_{\ell=3m+1}^{7m}\omega_{\ell}F_{n,j,\ell}^{(\mathrm{cd})}
\Bigr),
\\
E_4^{(\mathrm{cd})}&:=\sum_{j=g_{a_2,+}+1}^{g_{R,-}-1}\log\Big( 
1+\sum_{\ell=3m+1}^{7m}\omega_{\ell}F_{n,j,\ell}^{(\mathrm{cd})}
\Bigr), \qquad
E_5^{(\mathrm{cd})}:=\sum_{j=g_{R,-}}^{n-1}\log\Big( 
1+\sum_{\ell=3m+1}^{7m}\omega_{\ell}F_{n,j,\ell}^{(\mathrm{cd})}
\Bigr). 
\end{align*}
As usual, an empty sum is understood to be zero.  The only index range not
already covered by the outer part of the annulus analysis is
$E_0^{(\mathrm{cd})}$, because its lower endpoint is $j=0$.

\begin{lemma}\label{lemma:asymptotic expansion E0 counting centered disk}
Let the radii $r_{\ell}$, $\ell=3m+1,\dots,4m$, be defined by
\eqref{def of merging radii outside hard}, and let the radii $r_{\ell}$,
$\ell=4m+1,\dots,5m$, be defined by
\eqref{def of merging radii outside semi hard}.  Then, as $n\to\infty$,
\[
    E_{0}^{(\mathrm{cd})}
    =
    C_{E_{0}^{(\mathrm{cd})}}^{(3)}n 
 +C_{E_{0}^{(\mathrm{cd})}}^{(4)}\sqrt{n}
 +C_{E_{0}^{(\mathrm{cd})}}^{(5)}\log n
 +C_{E_{0}^{(\mathrm{cd})}}^{(6)}
 +\widetilde{C}_{E_{0}^{(\mathrm{cd})}}^{(M)}
 +\mathcal{O}(n^{-\frac{1}{12}}),
\]  
where 
\begin{align*}
C_{E_{0}^{(\mathrm{cd})}}^{(3)}
&:=
\int_{0}^{\tau_{\rho_2}}
\log\Bigl(1-\widehat{\mathsf{T}}_0^{(3m+1,4m)}(x)+\widehat{\mathsf{T}}_0^{(3m+1,7m)}(\tau_{\rho_2})\Bigr)\,dx,
\quad
C_{E_{0}^{(\mathrm{cd})}}^{(4)}
:=0,
\\
C_{E_{0}^{(\mathrm{cd})}}^{(5)}
&:=\frac{\rho_2^2\Delta Q(\rho_2)\widehat{\mathsf{T}}_{1}^{(3m+1,4m)}(\tau_{\rho_2})}{1-\widehat{\mathsf{T}}_0^{(3m+1,4m)}(\tau_{\rho_2})+\widehat{\mathsf{T}}_0^{(3m+1,7m)}(\tau_{\rho_2})},
\\
C_{E_{0}^{(\mathrm{cd})}}^{(6)}
&:=
\int_{0}^{\tau_{\rho_2}}
\bigg[
\frac{\frac{2\rho_2^2\Delta Q(\rho_2)}{\tau_{\rho_2}-x}
\widehat{\mathsf{T}}_{1}^{(3m+1,4m)}(x)
+x\widehat{\mathsf{T}}_{2}^{(3m+1,4m)}(x)}{1-\widehat{\mathsf{T}}_0^{(3m+1,4m)}(x)+\widehat{\mathsf{T}}_0^{(3m+1,7m)}(\tau_{\rho_2})}
\\
&\quad 
-\frac{2\rho_2^2\Delta Q(\rho_2)}{\tau_{\rho_2}-x}
\frac{
\widehat{\mathsf{T}}_{1}^{(3m+1,4m)}(\tau_{\rho_2})}{1-\widehat{\mathsf{T}}_0^{(3m+1,4m)}(\tau_{\rho_2})+\widehat{\mathsf{T}}_0^{(3m+1,7m)}(\tau_{\rho_2})}
\bigg] \,dx 
\\
&\quad 
+\bigl( 
2\rho_2^2\Delta Q(\rho_2)-\tau_{\rho_2}
\bigr)\int_{0}^{\tau_{\rho_2}}
\frac{\widehat{\mathsf{T}}_{2}^{(3m+1,4m)}(x)}{1-\widehat{\mathsf{T}}_0^{(3m+1,4m)}(x)+\widehat{\mathsf{T}}_0^{(3m+1,7m)}(\tau_{\rho_2})}\,dx
\\
&
\quad
+
\Bigl(\frac{1}{2}\rho_2\mathsf{k}'(\rho_2)+1\Bigr)
\log\Bigl(1-\widehat{\mathsf{T}}_0^{(3m+1,4m)}(\tau_{\rho_2})+\widehat{\mathsf{T}}_0^{(3m+1,7m)}(\tau_{\rho_2})\Bigr)
\\
&\quad
-
\Bigl(\frac{1}{2}\rho_2\mathsf{k}'(\rho_2)+1\Bigr)
\log\Bigl(1-\widehat{\mathsf{T}}_0^{(3m+1,4m)}(0)+\widehat{\mathsf{T}}_0^{(3m+1,7m)}(\tau_{\rho_2})\Bigr)
\\
&\quad 
+\frac{1}{2}\log\Bigl(1-\widehat{\mathsf{T}}_0^{(3m+1,4m)}(0)+\widehat{\mathsf{T}}_0^{(3m+1,7m)}(\tau_{\rho_2})\Bigr),
\\
\widetilde{C}_{E_{0}^{(\mathrm{cd})}}^{(M)}
&:=
-\frac{2\rho_2^2\Delta Q(\rho_2)\widehat{\mathsf{T}}_{1}^{(3m+1,4m)}(\tau_{\rho_2})}{1-\widehat{\mathsf{T}}_0^{(3m+1,4m)}(\tau_{\rho_2})+\widehat{\mathsf{T}}_0^{(3m+1,7m)}(\tau_{\rho_2})}
\log M
\\
&\quad 
-M\tau_{\rho_2}\sqrt{n}
\log\Bigl(1-\widehat{\mathsf{T}}_0^{(3m+1,4m)}(\tau_{\rho_2})+\widehat{\mathsf{T}}_0^{(3m+1,7m)}(\tau_{\rho_2})\Bigr)
\\
&\quad 
-M^2\tau_{\rho_2}
\frac{\tau_{\rho_2}\widehat{\mathsf{T}}_1^{(3m+1,4m)}(\tau_{\rho_2})}{1-\widehat{\mathsf{T}}_0^{(3m+1,4m)}(\tau_{\rho_2})+\widehat{\mathsf{T}}_0^{(3m+1,7m)}(\tau_{\rho_2})}
\\
&\quad 
+M^2\tau_{\rho_2}
\log\Bigl(1-\widehat{\mathsf{T}}_0^{(3m+1,4m)}(\tau_{\rho_2})+\widehat{\mathsf{T}}_0^{(3m+1,7m)}(\tau_{\rho_2})\Bigr)
\\
&\quad 
+\Bigl(\theta_{2,-}^{(n,M)}-\frac{1}{2}\Bigr)
\log\Bigl(1-\widehat{\mathsf{T}}_0^{(3m+1,4m)}(\tau_{\rho_2})+\widehat{\mathsf{T}}_0^{(3m+1,7m)}(\tau_{\rho_2})\Bigr). 
\end{align*}
\end{lemma}

\begin{proof}
The argument is the outer-hard-edge calculation used in the proof of
Lemma~\ref{lemma:asymptotic expansion E7 annulus}, with the lower endpoint of
the index range moved to zero.  More precisely, the endpoint version of
Laplace's method used there gives, uniformly for
$0\leq j\leq g_{2,-}-1$, the same expansion of the logarithmic summand in
powers of $n^{-1}$ and $(\tau_{\rho_2}-j/n)^{-1}$.  The estimates away from
$\tau_{\rho_2}$ are uniform, while the singular part at
$x=\tau_{\rho_2}$ is subtracted explicitly in the first integral defining
$C_{E_0^{(\mathrm{cd})}}^{(6)}$.

Applying the Euler--Maclaurin formula on $[0,\tau_{\rho_2}]$ gives the
$n$-term and the two endpoint corrections displayed above.  Expanding the
upper cutoff $g_{2,-}/n$ about $\tau_{\rho_2}$ produces the
$\sqrt n$-, $\log n$-, and $M$-dependent terms.  The remainder estimates are
unchanged from the proof of Lemma~\ref{lemma:asymptotic expansion E7 annulus}
and, with the standing choice of $M$, are
$\mathcal{O}(n^{-1/12})$.  Collecting these contributions proves the stated
expansion.
\end{proof}

It remains to transfer the five outer-regime estimates from the annulus
case.  The correspondence between the index ranges is
\[
E_k^{(\mathrm{cd})}\longleftrightarrow E_{k+7}^{(\mathrm{an})},
\qquad k=1,\dots,5.
\]

\begin{lemma}\label{lemma:asymptotic expansion E1 counting centered disk}
The expansion of $E_{1}^{(\mathrm{cd})}$ is obtained from
Lemma~\ref{lemma:asymptotic expansion E8 annulus} under the relabeling
\[
E_{8}^{(\mathrm{an})}\mapsto E_{1}^{(\mathrm{cd})},\qquad
C_{E_{8}^{(\mathrm{an})}}^{(q)}\mapsto
C_{E_{1}^{(\mathrm{cd})}}^{(q)},\qquad
\widetilde C_{E_{8}^{(\mathrm{an})}}^{(M)}\mapsto
\widetilde C_{E_{1}^{(\mathrm{cd})}}^{(M)}.
\]
All other quantities, including the remainder, are unchanged.
\end{lemma}

\begin{lemma}\label{lemma:asymptotic expansion E2 counting centered disk}
There exists $c>0$ such that, as $n\to\infty$,
\[
E_{2}^{(\mathrm{cd})}
=\bigl(g_{a_2,-}-g_{2,+}-1\bigr)\log\Omega_{5m+1}^{(7m)}
+\mathcal{O}(e^{-cM^2}).
\]
Thus its expansion agrees with that of $E_{9}^{(\mathrm{an})}$ in
Lemma~\ref{lemma:asymptotic expansion E9 annulus}.
\end{lemma}

\begin{lemma}\label{lemma:asymptotic expansion E3 counting centered disk}
The expansion of $E_{3}^{(\mathrm{cd})}$ is obtained from
Lemma~\ref{lemma:asymptotic expansion E10 annulus} under the relabeling
\[
E_{10}^{(\mathrm{an})}\mapsto E_{3}^{(\mathrm{cd})},\qquad
C_{E_{10}^{(\mathrm{an})}}^{(q)}\mapsto
C_{E_{3}^{(\mathrm{cd})}}^{(q)},\qquad
\widetilde C_{E_{10}^{(\mathrm{an})}}^{(M)}\mapsto
\widetilde C_{E_{3}^{(\mathrm{cd})}}^{(M)}.
\]
All other quantities, including the remainder, are unchanged.
\end{lemma}

\begin{lemma}\label{lemma:asymptotic expansion E4 centered disk}
There exists $c>0$ such that, as $n\to\infty$,
\[
E_{4}^{(\mathrm{cd})}
=\bigl(g_{R,-}-g_{a_2,+}-1\bigr)\log\Omega_{6m+1}^{(7m)}
+\mathcal{O}(e^{-cM^2}).
\]
Thus its expansion agrees with that of $E_{11}^{(\mathrm{an})}$ in
Lemma~\ref{lemma:asymptotic expansion E11 annulus}.
\end{lemma}

\begin{lemma}\label{lemma:asymptotic expansion E5 counting centered disk}
The expansion of $E_{5}^{(\mathrm{cd})}$ is obtained from
Lemma~\ref{lemma:asymptotic expansion E12 annulus} under the relabeling
\[
E_{12}^{(\mathrm{an})}\mapsto E_{5}^{(\mathrm{cd})},\qquad
C_{E_{12}^{(\mathrm{an})}}^{(q)}\mapsto
C_{E_{5}^{(\mathrm{cd})}}^{(q)},\qquad
\widetilde C_{E_{12}^{(\mathrm{an})}}^{(M)}\mapsto
\widetilde C_{E_{5}^{(\mathrm{cd})}}^{(M)}.
\]
In particular, the expansion includes the $M$-dependent term
$\widetilde C_{E_{5}^{(\mathrm{cd})}}^{(M)}$; all other quantities,
including the remainder, are unchanged.
\end{lemma}

\begin{proof}[Proof of Lemmas~\ref{lemma:asymptotic expansion E1 counting centered disk}--\ref{lemma:asymptotic expansion E5 counting centered disk}]
The five remaining index intervals in the centered-disk decomposition are, in order,
\[
 [g_{2,-},g_{2,+}],\quad
 [g_{2,+}+1,g_{a_2,-}-1],\quad
 [g_{a_2,-},g_{a_2,+}],\quad
 [g_{a_2,+}+1,g_{R,-}-1],\quad
 [g_{R,-},n-1].
\]
These are exactly the intervals defining, respectively,
$E_8^{(\mathrm{an})},\dots,E_{12}^{(\mathrm{an})}$ in the annulus
decomposition.

It remains to compare the normalizations in
\eqref{def of norming constant + counting on annulus} and
\eqref{def of calEncd}.  For a fixed $j$ in any of these intervals, set
\begin{align*}
I_{\mathrm{in}}(j)&:=2\int_0^{\rho_1}
u e^{-nV_\tau(u)}e^{\mathsf{k}(u)}\,du,&
I_{\mathrm{out}}(j)&:=2\int_{\rho_2}^{\infty}
u e^{-nV_\tau(u)}e^{\mathsf{k}(u)}\,du,\\
I_\ell(j)&:=2\int_{\rho_2}^{r_\ell}
u e^{-nV_\tau(u)}e^{\mathsf{k}(u)}\,du,
&&\ell=3m+1,\dots,7m.
\end{align*}
The definitions of the two ensembles give
\[
F_{n,j,\ell}^{(\mathrm{an})}
=\frac{I_{\mathrm{in}}(j)+I_\ell(j)}
       {I_{\mathrm{in}}(j)+I_{\mathrm{out}}(j)},
\qquad
F_{n,j,\ell}^{(\mathrm{cd})}
=\frac{I_\ell(j)}{I_{\mathrm{out}}(j)},
\qquad \ell=3m+1,\dots,7m.
\]
In particular,
\[
F_{n,j,\ell}^{(\mathrm{an})}-F_{n,j,\ell}^{(\mathrm{cd})}
=\frac{I_{\mathrm{in}}(j)
       \bigl(I_{\mathrm{out}}(j)-I_\ell(j)\bigr)}
       {I_{\mathrm{out}}(j)
       \bigl(I_{\mathrm{in}}(j)+I_{\mathrm{out}}(j)\bigr)}.
\]
For $\ell\leq3m$, the annular numerator is supported on the inner component
and is bounded by $I_{\mathrm{in}}(j)$.  On all five intervals, the comparison
estimates in the proofs of
Lemmas~\ref{lemma:asymptotic expansion E8 annulus}--\ref{lemma:asymptotic expansion E12 annulus}
show that $I_{\mathrm{in}}(j)/I_{\mathrm{out}}(j)$ is exponentially small,
uniformly in $j$.  Hence removing the terms with $\ell\leq3m$ and replacing
$F_{n,j,\ell}^{(\mathrm{an})}$ by $F_{n,j,\ell}^{(\mathrm{cd})}$ changes each
of the five logarithmic sums only by an error absorbed by its stated
remainder.  The local coordinates, cutoff indices, and outer-edge parameters
are otherwise identical.  Thus the annular expansions transfer term by term
under $E_{k+7}^{(\mathrm{an})}\mapsto E_k^{(\mathrm{cd})}$ for
$k=1,\dots,5$.
\end{proof}

\begin{proof}[Proof of Theorem~\ref{theorem:counting statistics of centered disk case}]
Insert the expansions from
Lemmas~\ref{lemma:asymptotic expansion E0 counting centered disk}--\ref{lemma:asymptotic expansion E5 counting centered disk}
into $\log\mathcal E_{n,s\lambda,\alpha}^{(\mathrm{cd})}
=\sum_{k=0}^{5}E_k^{(\mathrm{cd})}$.  The terms depending on the auxiliary
cutoff $M$ and on the fractional parts at the common endpoints cancel
pairwise.  Grouping the remaining terms of orders $n$, $\sqrt n$, $\log n$,
and $1$ gives, respectively, the coefficients in
\eqref{def of Csharp3 cd}--\eqref{def of Csharp6 cd}.  The exponentially small
errors are absorbed into $\mathcal{O}(n^{-1/12})$, and the theorem follows.
\end{proof}

\subsection*{Acknowledgements}
The author is grateful to Sung-Soo Byun, Yong-Woo Lee, and Seong-Mi Seo for insightful discussions.
The author acknowledges support from the European Research Council (ERC), Grant Agreement No. 101115687.

\subsection*{AI use disclosure}
The mathematical content of this paper was completed in December 2025 without the use of generative AI. I used ChatGPT Pro 5.6 Sol Ultra in August 2026 to polish the presentation of the paper.


\footnotesize
\begin{thebibliography}{99}

\bibitem{A2018} K. Adhikari, \emph{Hole probabilities for $\beta$-ensembles and determinantal point processes in the complex plane}, \textit{Electron. J. Probab.} \textbf{23} (2018), Paper No. 48, 21 pp.

\bibitem{AR2017} K. Adhikari and N. K. Reddy, \emph{Hole probabilities for finite and infinite Ginibre ensembles}, \textit{Int. Math. Res. Not. IMRN} (2017), no. 21, 6694--6730.

\bibitem{ABE2023} G. Akemann, S.-S. Byun and M. Ebke, \emph{Universality of the number variance in rotational invariant two-dimensional Coulomb gases}, \textit{J. Stat. Phys.} \textbf{190} (2023), no.~1, Paper No. 9, 34 pp.

\bibitem{ABES2023}
G. Akemann, S.-S. Byun, E. Markus and G. Schehr, \emph{Universality in the number variance and counting statistics of the real and symplectic Ginibre ensemble}, J. Phys. A {\bf 56} (2023), no.~49, Paper No. 495202, 53 pp.; MR4671822

\bibitem{AFLS25}
M. Allard, P.~J. Forrester, S. Lahiry and B.-J. Shen,
\emph{Partition function of 2D Coulomb gases with radially symmetric potentials and a hard wall},
arXiv:2506.14738 (2025).

\bibitem{AL25}
M. Allard and S. Lahiry,
\emph{Birth of a gap: Critical phenomena in 2D Coulomb gas},
arXiv:2509.24529 (2025).

\bibitem{ATW2014}
R. Allez, J.~D. Touboul and G. Wainrib, \emph{Index distribution of the Ginibre ensemble}, J. Phys. A {\bf 47} (2014), 042001

\bibitem{ACC2023a}
Y. Ameur, C. Charlier and J. Cronvall,
\emph{The two-dimensional Coulomb gas: fluctuations through a spectral gap},
\textit{Arch. Ration. Mech. Anal.} \textbf{249} (2025), Paper No.~63.

\bibitem{ACC2023b}
Y. Ameur, C. Charlier and J. Cronvall,
\emph{Random normal matrices: eigenvalue correlations near a hard wall},
\textit{J. Stat. Phys.} \textbf{191} (2024), Paper No.~98.

\bibitem{ACC2023c}
Y. Ameur, C. Charlier and J. Cronvall,
\emph{Free energy and fluctuations in the random normal matrix model with spectral gaps},
\textit{Constr. Approx.} \textbf{63} (2026), 279--335.

\bibitem{ACCL1}
Y. Ameur, C. Charlier, J. Cronvall and J. Lenells,
\emph{Exponential moments for disk counting statistics at the hard edge of random normal matrices},
\textit{J. Spectr. Theory} \textbf{13} (2023), no.~3, 841--902.

\bibitem{ACCL2}
Y. Ameur, C. Charlier, J. Cronvall and J. Lenells,
\emph{Disk counting statistics near hard edges of random normal matrices: the multi-component regime},
\textit{Adv. Math.} \textbf{441} (2024), 109549.

\bibitem{AC2026}
Y. Ameur and J. Cronvall,
\emph{On fluctuations of Coulomb systems and universality of the Heine distribution},
\textit{J. Funct. Anal.} \textbf{290} (2026), no.~6, 111301.

\bibitem{AHM2011}
Y. Ameur, H. Hedenmalm and N. Makarov,
\emph{Fluctuations of eigenvalues of random normal matrices},
\textit{Duke Math. J.} \textbf{159} (2011), no.~1, 31--81.

\bibitem{AHM2015}
Y. Ameur, H. Hedenmalm and N. Makarov,
\emph{Random normal matrices and Ward identities},
\textit{Ann. Probab.} \textbf{43} (2015), no.~3, 1157--1201.

\bibitem{AKS2023} Y. Ameur, N.-G. Kang and S.-M. Seo, \emph{The random normal matrix model: insertion of a point charge}, \textit{Potential Anal.} \textbf{58} (2023), no.~2, 331--372.

\bibitem{AS2021}
S. Armstrong and S. Serfaty, \emph{Local laws and rigidity for Coulomb gases at any temperature}, Ann. Probab. \textbf{49"} (2021), 46--121.

\bibitem{ASZ2014} S. N. Armstrong, S. Serfaty and O. Zeitouni, \emph{Remarks on a constrained optimization problem for the Ginibre ensemble}, \textit{Potential Anal.} \textbf{41} (2014), no. 3, 945--958.





\bibitem{B2025}
S.-S. Byun, \emph{Anomalous free energy expansions of planar Coulomb gases: multi-component and conformal singularity}, arXiv:2508.00316.

\bibitem{BC2022}
S.-S. Byun and C. Charlier,
\emph{On the characteristic polynomial of the eigenvalue moduli of random normal matrices},
\textit{Constr. Approx.} \textbf{62} (2025), 471--521.

\bibitem{BCMS2025}
S.-S. Byun, C. Charlier, P. Moreillon and N. Simm, \emph{Precise large deviations in geometric last passage percolation}, arXiv:2510.17470, (to appear in  Comm. Math. Phys).


\bibitem{BFreview} S.-S. Byun and P. J. Forrester, \emph{Progress on the study of the Ginibre ensembles}, KIAS Springer Ser. Math. \textbf{3} Springer, 2025, 221 pp.

\bibitem{BF25}
S.-S. Byun and P. J. Forrester, \emph{Electrostatic computations for statistical mechanics and random matrix applications}, 2025 Matrix Annals Part II (online).


\bibitem{BFKL25}
S.-S. Byun, P. J. Forrester, A. B. J. Kuijlaars and S. Lahiry, \emph{Orthogonal polynomials in the spherical ensemble with two insertions}, SIAM J. Math. Anal. \textbf{58} (2026), 3472--3509.

\bibitem{BFL25}
S.-S. Byun, P. J. Forrester and S. Lahiry, \emph{Properties of the one-component Coulomb gas on a sphere with two macroscopic external charges}, Pure Appl. Funct. Anal. (to appear), arXiv:2501.05061.

\bibitem{BKS2023} S.-S. Byun, N.-G. Kang and S.-M. Seo, \emph{Partition functions of determinantal and Pfaffian Coulomb gases with radially symmetric potentials}, \textit{Comm. Math. Phys.} \textbf{401} (2023), no.~2, 1627--1663.

\bibitem{BKSY2025}
S.-S. Byun, N.-G. Kang, S.-M. Seo and M. Yang, \emph{Free energy of spherical Coulomb gases with point charges}, J. Lond. Math. Soc. (2) \textbf{112} (2025), e70294.


\bibitem{BL2026}
S.-S. Byun and Y.-W. Lee,
\emph{Disc counting statistics of the real Ginibre ensemble}, arXiv:2608.23156.

\bibitem{BLY2026}
S.-S. Byun, Y.-W. Lee, and E. Yoo,
\emph{Confinement transitions in half-space constrained Riesz gases}, arXiv:2608.11813.


\bibitem{BP2026}
S.-S. Byun and S. Park, \emph{Large gap probabilities of complex and symplectic spherical ensembles with point charges}, J. Funct. Anal. \textbf{290} (2026), 111260.

\bibitem{BSY2025}
S.-S. Byun, S.-M. Seo and M. Yang, \emph{Free energy expansions of a conditional GinUE and large deviations of the smallest eigenvalue of the LUE}, Comm. Pure Appl. Math. \textbf{78} (2025), 2247–-2304.

\bibitem{BYY2026}
S.-S. Byun, M. Yang, and E. Yoo, \emph{Free energy expansion of determinantal Coulomb gases in the quadratic fields with a point charge}, arXiv:2605.29594.

\bibitem{C2019}
C. Charlier, \emph{Asymptotics of Hankel determinants with a one-cut regular potential and Fisher–Hartwig singularities}, Int. Math. Res. Not. \textbf{2019} (2019), 7515–-7576.


\bibitem{C2021 FH}
C. Charlier,
\emph{Asymptotics of determinants with a rotation-invariant weight and discontinuities along circles},
\textit{Adv. Math.} \textbf{408} (2022), 108600.

\bibitem{C2021}
C. Charlier,
\emph{Large gap asymptotics on annuli in the random normal matrix model},
\textit{Math. Ann.} \textbf{388} (2024), 3529--3587.
Detailed equation and lemma references are to the expanded preprint
\url{https://arxiv.org/abs/2110.06908v3}.

\bibitem{C2023}
C. Charlier,
\emph{Hole probabilities and balayage of measures for planar Coulomb gases},
arXiv:2311.15285 (2023).

\bibitem{CG2021}
C. Charlier and R. Gharakhloo, Asymptotics of Hankel determinants with a Laguerre-type or Jacobi-type potential and Fisher-Hartwig singularities, Adv. Math. {\bf 383} (2021), 107672

\bibitem{CL2023}
C. Charlier and J. Lenells,
\emph{Exponential moments for disk counting statistics of random normal matrices in the critical regime},
\textit{Nonlinearity} \textbf{36} (2023), no.~3, 1593--1616.

\bibitem{CMV2016}
F. D. Cunden, F. Mezzadri and P. Vivo, \emph{Large deviations of radial statistics in the two-dimensional one-component plasma}, J. Stat. Phys. \textbf{164} (2016), 1062-–1081.

\bibitem{DDMS2024}
B. De~Bruyne, P. Le Doussal, S. N. Majumdar and S. Schehr, \emph{Linear statistics for Coulomb gases: higher order cumulants}, J. Phys. A {\bf 57} (2024), 155002

\bibitem{DIK}
P. Deift, A. Its and I. Krasovsky,
\emph{Asymptotics of Toeplitz, Hankel, and Toeplitz+Hankel determinants with Fisher--Hartwig singularities},
\textit{Ann. of Math.} (2) \textbf{174} (2011), no.~2, 1243--1299.

\bibitem{FenzlLambert}
M. Fenzl and G. Lambert,
\emph{Precise deviations for disk counting statistics of invariant determinantal processes},
\textit{Int. Math. Res. Not. IMRN} \textbf{2022} (2022), no.~10, 7420--7494.

\bibitem{FisherHartwig} M.E. Fisher and R.E. Hartwig, \emph{Toeplitz determinants: Some applications, theorems, and conjectures}, \textit{Advan. Chem. Phys.} \textbf{15} (1968), 333--353.

\bibitem{ForresterHoleProba} P.J. Forrester, \emph{Some statistical properties of the eigenvalues of complex random matrices}, \textit{Phys. Lett. A} \textbf{169} (1992), no. 1-2, 21--24.

\bibitem{Forrester}
P. J. Forrester, \emph{Log-gases and Random Matrices} (LMS-34), Princeton University Press, Princeton, 2010.

\bibitem{Ginibre} J. Ginibre, \emph{Statistical ensembles of complex, quaternion, and real matrices}, \textit{J. Mathematical Phys.} \textbf{6} (1965), 440--449.

\bibitem{HW2021}
H. Hedenmalm and A. Wennman, \emph{Planar orthogonal polynomials and boundary universality in the random normal matrix model}, Acta Math. \textbf{227} (2021), 309-–406. 

\bibitem{HW2024}
H. Hedenmalm and A. Wennman, \emph{Berezin density and planar orthogonal polynomials}, Trans. Amer. Math. Soc. 377 (2024), 4825–4863. 


\bibitem{HWpreprint}
H. Hedenmalm and A. Wennman, \emph{A global asymptotic expansion of the polynomial Bergman density}, preprint.

\bibitem{JLM1993} B. Jancovici, J. Lebowitz and G. Manificat, \emph{Large charge fluctuations in classical Coulomb systems}, \textit{J. Statist. Phys.} \textbf{72} (1993), no. 3-4, 773--787.

\bibitem{Jo2022}
K. Johansson, \emph{Strong Szeg\"{o} theorem on a Jordan curve}, Toeplitz Operators and Random Matrices in Memory of Harold Widom (Basor, et al., eds.), Operator Theory Advances and Applications, Birkh\"{a}user, Basel, 2022. 


\bibitem{JV2023}
K. Johansson and F. Viklund, \emph{Coulomb gas and the Grunsky operator on a Jordan domain with corners}, Invent. Math. (Online), arXiv:2309.00308.

\bibitem{L et al 2019}
B. Lacroix-A-Chez-Toine, J.~A. Monroy Garz\'{o}n, C.~S. Hidalgo Calva, I. P\'{e}rez Castillo, A. Kundu, S.~N. Majumdar and G. Schehr,
\emph{Intermediate deviation regime for the full eigenvalue statistics in the complex Ginibre ensemble},
\textit{Phys. Rev. E} \textbf{100} (2019), 012137.

\bibitem{L et al 2019 b}
B. Lacroix-A-Chez-Toine, S. N. Majumdar and G. Schehr, \emph{Rotating trapped fermions in two dimensions and the complex Ginibre ensemble: Exact results for the entanglement entropy and number variance}, Phys. Rev. A {\bf 100} (2019), 021602

\bibitem{LS2017}
T. Lebl\'{e} and S. Serfaty, \emph{Large deviation principle for empirical fields of log and Riesz gases}, Invent. Math. \textbf{210} (2017), 645--757.

\bibitem{LMO2024}
M. Levi, J. Marzo and J. Ortega-Cerd\`{a}, \emph{Linear statistics of determinantal point processes and norm representations}, Int. Mat. Res. Not. \text{2024}(19), 12869--12903 (2024).

\bibitem{MMO25}
J. Marzo, L.~D. Molag and J. Ortega-Cerd\'{a},
\emph{Universality for fluctuations of counting statistics of random normal matrices},
\textit{J. Lond. Math. Soc.} (2) \textbf{113} (2026), no.~2, e70462.

\bibitem{Noda2025}
K. Noda,
\emph{Partition functions of two-dimensional Coulomb gases with circular root- and jump-type singularities},
arXiv:2510.00843 (2025).

\bibitem{Noda2026}
K. Noda, \emph{Two-dimensional Coulomb gases with multiple outposts}, 	arXiv:2602.22184 

\bibitem{RV2007}
B. Rider and B. Vir\'{a}g,
\emph{The noise in the circular law and the Gaussian free field},
\textit{Int. Math. Res. Not. IMRN} \textbf{2007} (2007), rnm006, 33 pp.

\bibitem{R2025}
N. Rougerie, \emph{Free-energy variations for determinantal 2D plasmas with holes}, arXiv:2510.01745.

\bibitem{SaTo} E. B. Saff and V. Totik, \emph{Logarithmic Potentials with External Fields},  Grundlehren der Mathematischen Wissenschaften, Springer-Verlag, Berlin, 1997.

\bibitem{Ser2023}
S. Serfaty, \emph{Gaussian fluctuations and free energy expansion for Coulomb gases at any temperature}, Ann. Inst. Henri Poincar\'{e}e Probab. Stat. \textbf{59} (2023), 1074–-1142. 




\bibitem{Ser2024}
S. Serfaty, \emph{Lectures on Coulomb and Riesz Gases}, Amer. Math. Soc. Colloq. Publ. (to appear), arXiv:2407.21194.

\bibitem{Seo} S.-M. Seo, \emph{Edge behavior of two-dimensional Coulomb gases near a hard wall}, \textit{Ann. Henri Poincar\'{e}} \textbf{23} (2022), no. 6, 2247--2275.

\bibitem{SF2026}
B.-J. Shen and P. J. Forrester, 
\emph{Large N expansions of the partition function for Coulomb systems on the surface of a cylinder}, arXiv:2609.12487.


\bibitem{ZW2006}
A. Zabrodin and P. Wiegmann,
\emph{Large-$N$ expansion for the 2D Dyson gas},
\textit{J. Phys. A} \textbf{39} (2006), no.~28, 8933--8963.

\end{thebibliography}
\end{document}